\documentclass[12pt]{article}

\usepackage{amsmath, amsthm, amssymb}
\usepackage{graphicx, float}

\usepackage[hidelinks]{hyperref}

\usepackage{epstopdf}

\usepackage{pdflscape}

\usepackage{natbib}

\newtheorem{theorem}{Theorem}
\newtheorem{proposition}{Proposition}

\allowdisplaybreaks

\usepackage{authblk}

\title{Fitting tails affected by truncation}

\author[1,3]{Jan Beirlant\footnote{Corresponding author. Address: Celestijnenlaan 200B, 3001 Leuven, Belgium; Email: \href{mailto:jan.beirlant@kuleuven.be}{jan.beirlant@kuleuven.be}.
}}
\author[2]{Isabel~Fraga Alves}
\author[1]{Tom Reynkens}

\affil[1]{Department of Mathematics and LStat, KU Leuven}
\affil[2]{Department of Statistics and Operations Research, University of Lisbon}
\affil[3]{Department of Mathematical Statistics and Actuarial Science, University of the Free State}

\begin{document}

 \maketitle
\begin{abstract}
{\noindent 
In several applications, ultimately at the largest data, truncation effects can be observed when analysing tail characteristics of statistical  distributions. In some cases truncation effects are forecasted through physical models such as the Gutenberg-Richter relation in geophysics, while at other instances the nature of the measurement process itself may cause under recovery of large values, for instance due to flooding in river discharge readings.
Recently, \citet{Truncation} discussed tail fitting for truncated Pareto-type distributions.
Using examples from earthquake analysis, hydrology and diamond valuation we demonstrate the need for a unified treatment of extreme value analysis for truncated heavy and light tails.
We generalise the classical Peaks over Threshold approach for the different max-domains of attraction with shape parameter $\xi > -1/2$ to allow for truncation effects. We use a pseudo maximum likelihood approach to estimate the model parameters and consider extreme quantile estimation and reconstruction of quantile levels before truncation whenever appropriate. We report on some simulation experiments and provide some basic asymptotic results.
 }
\end{abstract}


%\noindent {\bf AMS 2000 subject classifications.} 62G32, %62H12, 62G20.

\section{Introduction}

Modelling extreme events has recently received a lot of interest. Assessing the risk of rare events through estimation of extreme quantiles or corresponding return periods has been developed extensively and was applied to a wide variety of fields such as meteorology, finance, insurance and geology, among others.
The methodology on modelling the univariate upper tail of the distribution of such quantities $Y$ relies on the fact that the maximum of independent measurements $Y_i, \; i=1,\ldots,n,$  can be approximated by the generalised extreme value distribution: as $n\to \infty$
\begin{equation}
\mathbb{P}\left(\tfrac{\displaystyle\max_{i=1,\ldots,n}Y_i -b_n}{a_n} \leq y \right) \to
G_{\xi} (y) = \exp \left( - (1 + \xi y)^{-1/\xi} \right),   \;\; 1+\xi y>0,
\label{eq:maxd}
\end{equation}
where $b_n \in \mathbb{R}$, $a_n >0$ and $\xi \in \mathbb{R}$ are the location, scale and shape parameters, respectively. For $\xi =0$, $G_0(y)$ has to be read as $\exp \{- \exp (-y)\}$.
In fact, \eqref{eq:maxd} represents the only possible non-degenerate limits for maxima of independent and identically distributed sequences $Y_i$. Condition \eqref{eq:maxd} is equivalent to the convergence of the  distribution of excesses (or peaks)  over high thresholds $t$ to the generalised Pareto distribution (GPD):
as $t$ tends to the endpoint of the distribution of $Y$, then, with $\bar{F}$ the right tail function (RTF) of a given distribution,
{\small\begin{equation}
\mathbb{P} \left(\frac{Y -t}{\sigma_Y (t)}>y \,\middle|\, Y>t\right) =
\frac{\bar{F}_Y (t+y \sigma_Y(t) )}{\bar{F}_Y (t)}
 \to H_{\xi}(y) = -\log G_{\xi} (y) = \left( 1+\xi y\right)^{-1/\xi},
\label{eq:pot}
\end{equation}}%
where $\sigma_Y (t) >0$. Below we set $ \sigma_Y (t)= \sigma_t$.
%In case $\xi =0$, $( 1+\xi y)^{-1/\xi} $ is to be read as $e^{-y}$.
Setting $t$ at the $(k+1)$th largest observation $y_{n-k,n}$ for some $k \in \{1\ldots,n-1\}$ so that $k$ data points are larger than the threshold $t$, \eqref{eq:pot} leads  to the  estimator
\begin{equation}
\hat{p}_c = \frac{k}{n}H_{\hat\xi}\left( \frac{c-y_{n-k,n}}{\hat\sigma}\right)
\label{eq:phat}
\end{equation}
of the tail probability $\mathbb{P}(Y>c)$ for $c>0$ large, where ($\hat\xi,\hat\sigma $) denote estimators for ($\xi,\sigma_t$). The modelling of extreme values and the estimation of tail parameters through the peaks over threshold (POT) methodology has been discussed for instance in \citet{Coles}, \citet{Embrechts}, \citet{SoE}, and \citet{dHF}.

\vspace{0.5cm}\noindent
Recently, \citet{trHill}, \citet{Chakra} and \citet{Truncation} have addressed the problem of using unbounded  probability mass leading to levels that are unreasonably large or physically impossible. All of these papers consider cases with shape parameter $\xi >0$. In \citet{Truncation} it was observed that the above mentioned extreme value methods, even when using a negative extreme value index, are not able to capture truncation at high levels. However, in several other fields, such as hydrology and earthquake magnitude modelling, the underlying distributions appear to be lighter tailed than Pareto.
In this paper we will propose an adaptation of the classical approach to truncated tails over the whole range of max-convergence \eqref{eq:maxd} with $\xi >-0.5$ as in the original POT approach.

\vspace{0.5cm} \noindent
First, we consider recent {\bf magnitude data} (expressed on the Richter scale) of the 200 largest earthquakes in the Groningen area (the Netherlands), in the period 2003--2015, which are caused by gas extraction. In Figure~\ref{fig:EarthquakeQQ}, we present the time plot and the exponential QQ-plot $\left(x_{n-j+1,n},\log (j/n) \right)$ ($j=1,2, \ldots,n$) where $x_{1,n} \leq \ldots \leq x_{n-j+1,n}\leq \ldots \leq x_{n,n}$ denote the ordered data.
\begin{figure}[!ht]
    \begin{center}
		\includegraphics[height=0.485\textwidth, angle=270]{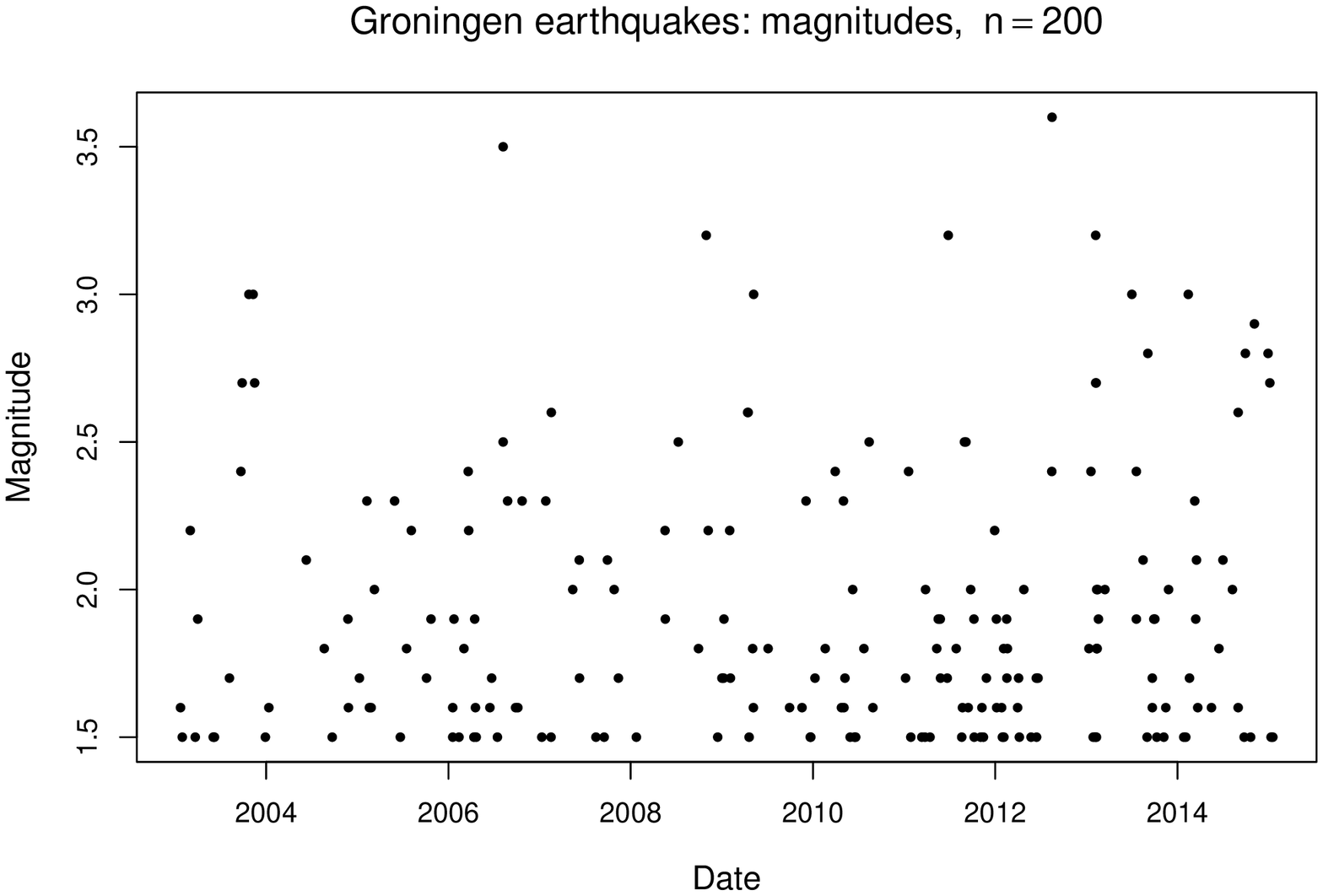}
    \includegraphics[height=0.485\textwidth, angle=270]{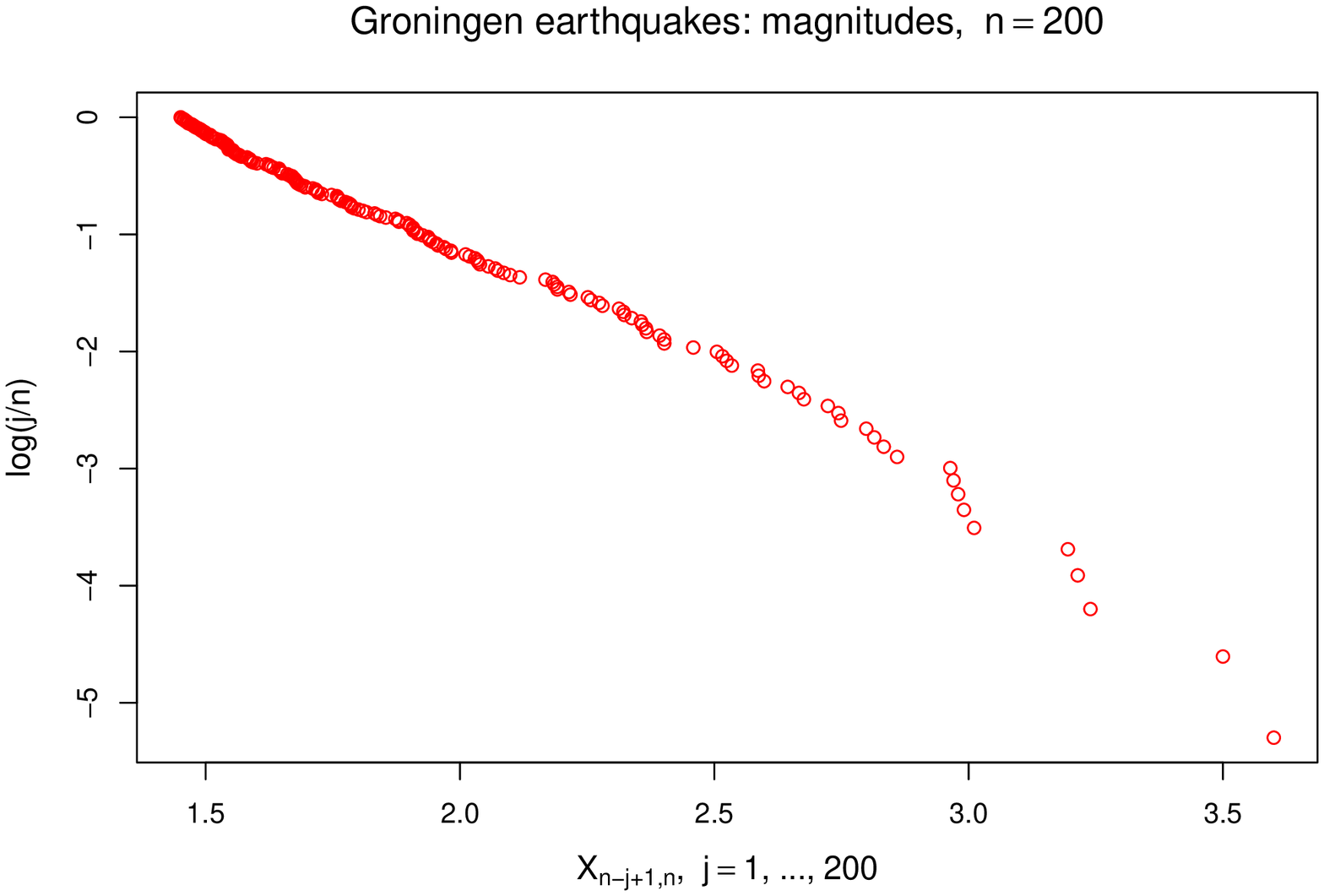}
 \caption {Time plot and exponential QQ-plot of earthquake magnitude data from the Groningen area.}\label{fig:EarthquakeQQ}
    \end{center}
   \end{figure}
Along the Gutenberg-Richter (\citeyear{GR}) law the magnitudes of independent earthquakes are drawn from a doubly truncated exponential distribution
\[
\mathbb{P}(M>m) = \frac{e^{-\lambda m}-e^{-\lambda T_M}}{e^{-\lambda m_0}-e^{-\lambda T_M}},\; m_0 < m < T_M.
\]
\citet{KS}
%and Holschneider {\it et al.} (2014)
provide a review of the vast literature on  estimating the maximum possible magnitude $T_M$.    The energy $E$ released by earthquakes, expressed in megajoule ($MJ$), relates to the magnitude $M$, expressed on the Richter scale, by
\[
\mbox{ M } = \log_{10} \left( E/2\right)/1.5 +1
.\]
In Figure~\ref{fig:EarthquakeQQ}, a linear pattern is visible for a large section of the magnitudes data, while some curvature appears at the largest values. The data set was  tested for serial correlation and no significance could be detected.

\vspace{0.5cm} \noindent
Secondly, we revisit the {\bf diamond size data} considered in \citet{VWSP}. The nature of metallurgical recovery processes in diamond mining may cause under recovery of large diamonds between 30 and 60 cts per stone. If stones are not recovered during this process they are discarded onto tailing dumps from which they can be recovered during future re-mining programs. Because even a small number of large diamonds can have a large value, the question arises whether re-mining a mine dump can be made profitable by recovering these large diamonds. Therefore, the expected number of large diamonds above certain carat values $c$ is of interest and the original non-truncated  values are to be reconstructed from the data, which exhibit truncation. In Figure~\ref{fig:DiamondQQ}, the Pareto QQ-plot or log-log plot $\left(\log x_{n-j+1,n},\log (j/n) \right)$ ($j=1,2, \ldots,n$) of the available carat data is presented. Again, a curvature near the top data is visible.
\begin{figure}[!ht]
    \begin{center}
    \includegraphics[height=0.70\textwidth, angle=270]{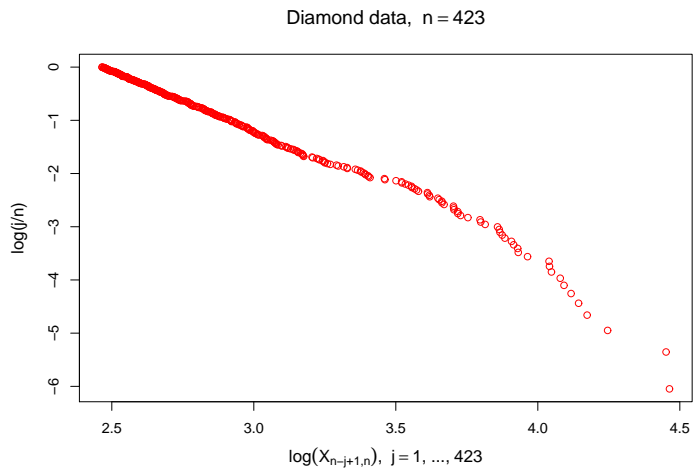}
 \caption{Log-log plot of diamond size data from \citet{VWSP}.}\label{fig:DiamondQQ}
    \end{center}
   \end{figure}

\vspace{0.5cm} \noindent
Thirdly, we study the river flows of the Molenbeek river at Erpe-Mere in Belgium (in $m^3/s, n=426$) obtained between 1986 and 1996. The data are peaks over threshold values taken from a complete series of hourly flow measurements which was filtered in order to satisfy hydrological independence as discussed in \citet{Willems}. This river is prone to flooding at high flow levels and hence the measurements can be truncated. In Figure~\ref{fig:MolenbeekQQ} the exponential QQ-plot
 is given, which exhibits a linear (i.e.~exponential) pattern with a downward curvature near the largest floods.
\begin{figure}[!ht]
    \begin{center}
        \includegraphics[height=0.70\textwidth, angle=270]{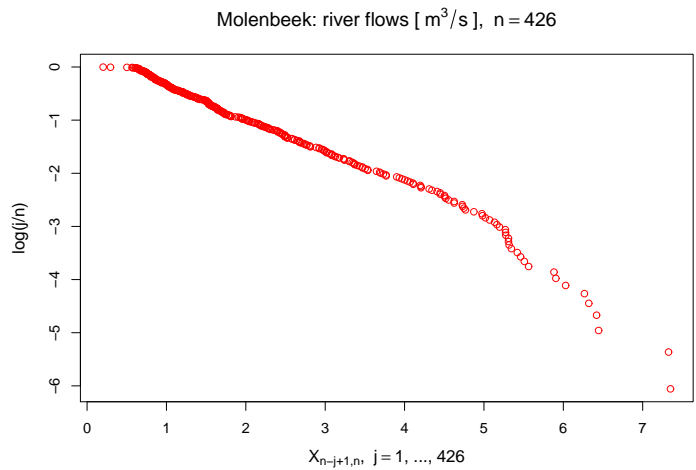}
 \caption{Exponential QQ-plot of the Molenbeek flow data.}\label{fig:MolenbeekQQ}
    \end{center}
   \end{figure}

\vspace{0.2cm}\noindent
In this paper, we aim to provide a statistical model being able to approximate tail characteristics of distributions truncated at high levels. Moreover, the statistical estimation methods should also include the case of no-truncation in order for these methods to be useful and competitive both in cases with and without truncation.
In the case of Pareto-type tails with $\xi>0$ the proposed methods should also be compared with the methods which have been developed specifically for that sub-case.
\\
To this purpose we extend the classical POT technique with maximum likelihood estimation of the GPD parameters $\xi$ and $\sigma$. Of course estimators for tail probabilities and extreme quantiles of a truncated distribution are to be discussed. Estimation of the endpoint $T$ of a truncated distribution is of particular importance as discussed above in earthquake applications. Motivated by the river flow and diamond valuation examples, we finally consider the problem of reconstructing quantiles of the underlying unobserved variable $Y$ before truncation.

\section{Model}

Let $Y$ denote a parent random variable with distribution function
$F_Y (y) = \mathbb{P}(Y \leq y)$, RTF $\bar{F}_Y(y) = 1-F_Y(y)$, quantile function $Q_Y (p) = \inf \{y: F_Y(y) \geq p \}$ ($0<p<1$), and tail quantile function $U_Y (v) = Q_Y(1-\frac{1}{v})$ ($v>1$). We consider the right truncated distribution from which independent and identically distributed data $X_1, X_2,\ldots, X_n$ are observed with, for some $T>0$,
\begin{equation}
X =_d Y\,|\,Y<T.
\label{eq:T}
\end{equation}
The corresponding RTF is denoted with $\bar{F}_T (x) = \mathbb{P}(X >x)$ and the tail quantile function is given by $U_T (u) = Q_T (1-\frac{1}{u})$ ($u>1$).
Then,
\begin{eqnarray}
 \bar{F}_T (x) &=& \frac{\bar{F}_Y(x)-\bar{F}_Y(T)}{1-\bar{F}_Y(T)} = (1+D_T)\bar{F}_Y(x)-D_T , \label{eq:FYT}\\
 U_T (u) &=& U_Y \left( \frac{u}{F_Y (T)} \left[1 +uD_T \right]^{-1}\right) \label{eq:UYCD} \\
% &=& U_Y \left( u {1+D_T \over 1 +uD_T }\right)
 %\label{UYD}\\
 &= &U_Y \left( \frac{1}{\bar{F}_Y (T)}\left[1 +\frac{1}{uD_T} \right]^{-1}\right), \label{eq:UYTD}
\end{eqnarray}
where $D_T = \bar{F}_Y (T)/F_Y(T)$ equals the odds  of the truncated probability mass under the untruncated distribution $Y$.

\vspace{0.3cm} \noindent
The goal of this paper is to provide a test for truncation and to estimate
\begin{itemize}
\item the model parameters $\xi$ and $\sigma=\sigma_t$,
\item the odds $D_T$,
\item quantiles  $Q_T (1-p)$ ($p$ small) of the truncated distribution and the truncation point $T=Q_T(1)$,
\item tail probabilities $\mathbb{P}(X >c)$ ($c$ large) of the truncated distribution,
\item and reconstruct quantile levels $Q_Y (1-p)$ of the parent variable $Y$ before truncation,
\end{itemize}
all on the basis of a pure random sample from $X$ (possibly) truncated at some large $T$.

\vspace{0.3cm} \noindent
We assume that the distribution of $Y$ satisfies \eqref{eq:maxd} or, equivalently, \eqref{eq:pot}.
Condition \ref{eq:pot} is also known to be equivalent to the following condition relating extreme quantile levels at $1-\frac{1}{vy}$ and $1-\frac{1}{y}$ close to the endpoint of the distribution: there exists a positive measurable function $a$ such that
\begin{equation}
\frac{U_Y(vy) -U_Y(y)}{a(y) } \to \frac{v^{\xi}-1}{\xi} \mbox{ when } y \to \infty,
\label{eq:maxdU}
\end{equation}
with $a(1/\bar{F}_Y (t_{k,n}))=\sigma_t$ where $t=t_{k,n}=U_T(n/k)$. The right hand side of \eqref{eq:maxdU} is to be read as $\log v$ for $\xi=0$.

\vspace{0.3cm} \noindent
The specific case $\xi >0$ of Pareto-type distributions satisfies
\begin{equation}
\frac{U_Y (vy)}{U_Y (y)} \to_{y\to \infty} v^{\xi} \mbox{ and }
\mathbb{P} (Y/t>y \,|\, Y>t) = \frac{\bar{F}_Y (ty )}{\bar{F}_Y (t)}
\to_{t\to \infty} y^{-1/\xi}.
\label{eq:Pa}
\end{equation}
Also when $\xi >0$, $\sigma_t \sim \xi t$  as $t \to \infty$. Furthermore, it is known that $\sigma_t/t \to 0$ when $\xi \leq 0$.
\\
Note that for a given $T$ fixed, the tail of a truncated model $X$ defined through \eqref{eq:T} has an extreme value index $\xi_X=-1$, see for instance Figure~2.8 in \citet{SoE}.

\vspace{0.3cm} \noindent
Truncation of a  distribution $Y$ satisfying \eqref{eq:pot} at a value $T$ necessarily requires $t<T \to \infty$. The threshold $t$ is mostly taken at the theoretical quantile $Q_T (1-\frac{k}{n})\\=U_T(n/k)$, which in practice is estimated by the empirical quantile $X_{n-k,n}$. Given the fact that our model is only defined choosing $t=t_n,T=T_n \to \infty$ as the sample size $n \to \infty$, the underlying model depends on  $n$ and a triangular array formulation  $X_{n1}, \ldots, X_{nn}$ of the observations should be used in order to emphasise the  nature of the model. However, in statistical procedures as presented here, when a single sample is given, the notation $X_1,\ldots,X_n$ is more natural and will be used throughout.
\\
The model considered in this paper is then given by
\begin{itemize}
\item[($\mathcal{M}$)] For a sequence $T_n \to \infty$, $\{X_{n1},\ldots,X_{nn}\}=\{X_{1},\ldots,X_{n}\}$ are independent copies of a random variable $X=X_{T_n}$ where $X=X_{T_n}$ is distributed as $Y|Y<T_n$, with $Y$ satisfying \eqref{eq:pot} or equivalently \eqref{eq:maxdU}.
\end{itemize}

\vspace{0.3cm} \noindent
Now we consider the distribution of the POT values for the data  of the truncated distribution under ($\mathcal{M}$):
\begin{eqnarray}
\mathbb{P}\left(\frac{X-t}{\sigma_t } >x \,\middle|\, X>t\right) &=& \mathbb{P}\left(\frac{Y-t}{ \sigma_t} >x \,\middle|\, t<Y<T\right)  \nonumber \\
&=& \frac{\mathbb{P}(Y>t+x \sigma_t) - \mathbb{P}(Y>T)}{\mathbb{P}(Y>t)-\mathbb{P}(Y>T)}
\nonumber
\\
&= & \frac{\frac{\mathbb{P}(Y>t+x \sigma_t)}{\mathbb{P}(Y>t)}-\frac{\mathbb{P}(Y>T)}{\mathbb{P}(Y>t)}} {1-\frac{\mathbb{P}(Y>T)}{\mathbb{P}(Y>t)}}.
\label{eq:GPDT}
\end{eqnarray}
One can now consider two cases as $t,T \to \infty$:
\begin{itemize}
\item
($\mathcal{T}_t$)  {\it Rough truncation with the threshold $t=t_n$}:
\begin{equation}
(T-t)/\sigma_t \to \kappa >0,
\label{eq:C}
\end{equation}
and hence from \eqref{eq:pot} and with local  uniform convergence in \eqref{eq:pot}
\begin{equation}
\frac{\mathbb{P}(Y>T)}{\mathbb{P}(Y>t)} \to (1+\xi\kappa)^{-1/\xi}.
\label{eq:C2}
\end{equation}
This entails that for $x\in (0,\kappa)$
\begin{equation}
\mathbb{P}\left(\frac{X-t}{\sigma_t} >x \,\middle|\, X>t\right) \to \frac{(1+ \xi  x)^{-1/\xi}  -  (1+ \xi \kappa )^{-1/\xi}}
{1-(1+ \xi \kappa)^{-1/\xi}} =: \bar{F}_{\xi,\kappa}(x).
\label{eq:FT}
\end{equation}
This corresponds to situations where  the deviation from the Pareto behaviour due to truncation at a high value $T$ will be visible in the data from $t$ on, and the approximation of the POT distribution  using the limit distribution in \eqref{eq:FT} appears more appropriate than with a simple GPD.
\item ($\bar{\mathcal{T}}_t$)  {\it Light truncation with the threshold $t=t_n$ : $\frac{\mathbb{P}(Y>T)}{\mathbb{P}(Y>t)} \to 0$}.\\ This entails
\begin{equation}
\mathbb{P}\left(\frac{X-t}{\sigma_t} > x \,\middle|\, X>t\right)  \to (1+ \xi  x)^{-1/\xi} , \;\; 1+\xi x>0.
\label{eq:LT}
\end{equation}
Light truncation is introduced for mathematical completeness.
But ($\bar{\mathcal{T}}_t$) means that the  truncation is not really visible in the data above $t$, and the classical extreme value modelling without truncation is appropriate. Hence, it will be practically impossible to discriminate light truncation from no truncation (i.e.~$T=\infty$).
\end{itemize}

 \vspace{0.3cm}\noindent
 Under ($\mathcal{T}_t$) with $t=t_{k,n}= U_T (n/k)$ we find from applying $F_Y$ to both sides of \eqref{eq:UYCD} with $u=n/k$ that
 \begin{equation*}
\bar{F}_Y (t) = F_Y(T) \frac{1+ (n/k)D_T}{n/k}= F_Y (T) \left(\frac{k}{n}+D_T \right),
 \end{equation*}
 from which, dividing by $\bar{F}_Y (T)$, we obtain
\[
\frac{\bar{F}_Y (t) }{\bar{F}_Y (T)} = \frac{1}{D_T}\left(\frac{k}{n} +D_T \right),
\]
while, using \eqref{eq:pot} and ($\mathcal{T}_t$),
\[
\frac{\bar{F}_Y (T)}{\bar{F}_Y (t)}  \to (1+\xi \kappa)^{-1/\xi},
\]
and hence under ($\mathcal{T}_t$)
\begin{equation}
\frac{k}{nD_T} \to (1+\xi\kappa)^{1/\xi}-1.
\label{eq:knDT}
\end{equation}

\vspace{0.3cm}\noindent
Now in order to be able to construct extreme quantile estimators under ($\mathcal{T}_t$), remark that from \eqref{eq:maxdU} with $vy=1/p$, $y=1/\bar{F}_Y(t)$ and $k_{\xi}(u)= (u^{\xi}-1)/\xi$, we have as $t \to \infty$ and $\bar{F}_Y (t)/p \to C$ for some constant $C>0$ that
\[
\frac{Q_Y (1-p) - t}{\sigma_t} - k_{\xi}\left( \frac{\bar{F}_Y (t)}{p}\right) \to 0.
\]
Hence, with \eqref{eq:UYTD} and $p = \bar{F}_Y (T)(1+\frac{1}{ uD_T})$
we obtain
\begin{eqnarray*}
\frac{U_T (u) - t}{\sigma_t} &=& \frac{U_Y \left( \frac{1}{\bar{F}_Y (t)}[1+ \frac{1}{uD_T}]^{-1}\right) -t}{\sigma_t}\\
&= & k_\xi  \left( \frac{\bar{F}_Y (t)}{\bar{F}_Y (T)[1+\frac{1}{uD_T}]}\right) + o(1).
\end{eqnarray*}
Using \eqref{eq:knDT} and \eqref{eq:pot} with $y=\kappa$ we obtain under ($\mathcal{T}_t$) that
\[
\frac{\bar{F}_Y (t)}{\bar{F}_Y (T)} \sim (1+\xi\kappa)^{1/\xi} \sim 1+ \frac{k}{nD_T}.
\]
Hence, we conclude that under ($\mathcal{T}_t$) for $1/(uD_T) \to 0$
\begin{equation}
\frac{U_T(u) - t}{\sigma_t} - k_{\xi}\left( \frac{1+ \frac{k}{nD_T}}{1+\frac{1}{uD_T}}\right) \to 0.
\label{eq:UTy}
\end{equation}
These derivations will motivate the proposed estimators of $D_T$ and extreme quantiles $Q_T (1-p)$.

\section{Inference}

\subsection{Estimators and goodness-of-fit} Estimation of the parameters ($\xi,\sigma$) in the classical POT without truncation is well-developed \citep{Coles, SoE}. Fitting the scaled GPD  with RTF $(1+\frac{\xi}{\sigma} x)^{-1/\xi}$ to the excesses $X-t$ given $X>t$ (based on \eqref{eq:FT}) using maximum likelihood is by far the most popular method in this respect. Here we rely on the generalisation \eqref{eq:FT} under ($\mathcal{T}_t$), with $t$ replaced by a random threshold $X_{n-k,n}$ and using the exceedances $E_{j,k}= X_{n-j+1,n}-X_{n-k,n}$  ($j=1,2,\ldots,k$) for some $k \geq 2$. Substituting
$ E_{1,k}/\sigma$ for $\kappa$ following \eqref{eq:C}, the log-likelihood is  given by
\begin{eqnarray*}
\log L_{k,n} (\xi,\sigma ) &=& \log \left(\prod_{j=2}^k \frac{\sigma^{-1}\left(1+ \frac{\xi}{\sigma} E_{j,k}\right)^{-(1/\xi) -1} }
{1-\left(1+ \frac{\xi}{\sigma}E_{1,k}\right)^{-1/\xi}}\right) \\
 &=&  -(k-1) \log \sigma - \left( 1 + \frac{1}{\xi }\right)\sum_{j=2}^k \log \left(1+\frac{\xi}{\sigma} E_{j,k}\right) \\
&& \vspace{1cm}
-(k-1) \log \left( 1- \left(1+ \frac{\xi}{\sigma}E_{1,k}\right)^{-1/\xi} \right),
 \end{eqnarray*}
or, by reparametrising ($\xi,\sigma$) to  ($\xi,\tau$) with $\tau= \xi/\sigma$,
\begin{eqnarray*}
\log L_{k,n} (\xi,\tau )
 &=&   (k-1) \log \tau - (k-1) \log \xi- \left( 1 + \frac{1}{\xi }\right)\sum_{j=2}^k \log (1+\tau E_{j,k}) \\
 && \vspace{1cm}
-(k-1) \log \left( 1- (1+ \tau E_{1,k})^{-1/\xi} \right).
 \end{eqnarray*}
The partial derivatives are given by 
\newpage
\begin{align*}
\frac{1}{k-1}\frac{ \partial \log L_{k,n} (\xi,\tau )}{\partial  \xi } &=  -\frac{1}{\xi}  + \frac{1}{\xi^2 } \frac{1}{k-1}\sum_{j=2}^k \log (1+ \tau E_{j,k}) \nonumber \\
&\quad+ \frac{1}{\xi^2 } \frac{(1+\tau E_{1,k})^{-1/\xi} \log (1+\tau E_{1,k})}{1- (1+\tau E_{1,k})^{-1/\xi}}, \\
\frac{1}{k-1} \frac{\partial \log L_{k,n} (\xi,\tau )}{\partial  \tau } &= \frac{1}{\tau }  -
 \left( 1 + \frac{1}{\xi } \right)\frac{1}{k-1} \sum_{j=2}^k \frac{E_{j,k}}{1+ \tau E_{j,k}} \nonumber \\
 &\quad- \frac{1}{\xi}E_{1,k} \frac{(1+\tau E_{1,k})^{-1-1/\xi}}{1- (1+\tau E_{1,k})^{-1/\xi} },
\end{align*}
from which the likelihood equations defining the pseudo maximum likelihood estimators ($\hat\xi _k, \hat\tau _k$) are obtained:
\begin{align}
&\frac{1}{k-1}\sum_{j=2}^k \log (1+ \hat{\tau}_k E_{j,k})+ \frac{(1+\hat{\tau}_k E_{1,k})^{-1/\hat{\xi}_k} \log (1+\hat{\tau}_k E_{1,k})}{1- (1+\hat{\tau}_k E_{1,k})^{-1/\hat{\xi}_k}}=\hat{\xi}_k 
\label{eq:lik_xitau1} \\
& \frac{1}{k-1} \sum_{j=2}^k \frac{1}{1+ \hat{\tau}_k E_{j,k}} =\frac{1}{1+\hat{\xi}_k} 
\frac{1-(1+\hat{\tau}_k E_{1,k})^{-1-1/\hat{\xi}_k}}{1- (1+\hat{\tau}_k E_{1,k})^{-1/\hat{\xi}_k} }.
\label{eq:lik_xitau2}
\end{align}
When computing ($\hat\xi _k,\hat\tau _k$), one has to impose the model restrictions. In order to meet the restrictions $\sigma=\xi/\tau>0$ and  $1+\tau E_{j,k}>0$ for $j=1,\ldots,k$, in our implementation we require the estimates of these quantities to be larger than the numerical tolerance value $10^{-10}$.

\vspace{0.3cm}\noindent
An estimator of $D_T$ now follows from taking $u=n$ in \eqref{eq:UTy}:
\[
U_T(n) - U_T (n/k) \approx \sigma k_{\xi}\left( \frac{1+ \frac{k}{nD_T}}{1+\frac{1}{nD_T}}\right).
\]
Estimating $U_T(n) - U_T (n/k)$ by $E_{1,k}$ we obtain
\begin{equation}
\hat{D}_{T,k}:= \max\left\{0,\frac{k}{n}
\frac{(1+ \hat\tau _k E_{1,k})^{-1/\hat\xi _k}-\frac{1}{k}}
{1-(1+ \hat\tau _k E_{1,k})^{-1/\hat\xi _k}}\right\}.
\label{eq:DTk}
\end{equation}
Similarly taking $u=1/p$ in \eqref{eq:UTy} with $np/k \to 0$, we obtain estimators for $Q_T (1-p)$:
\begin{equation}
{\hat{Q}_{T,k}(1-p)} = X_{n-k,n} +
\frac{1}{\hat\tau _k}\left[\left\{ \frac{\hat{D}_{T,k} + \frac{k}{n}}{\hat{D}_{T,k}+p} \right\}^{\hat\xi _k} -1 \right].
\label{eq:QTp}
\end{equation}
Setting $p=0$ in \eqref{eq:QTp} one obtains an estimator for the truncation point $T$:
\begin{equation}
\hat{T} _{k} = X_{n-k,n} + \frac{1}{\hat\tau _k}\left[\left\{ \frac{1-k^{-1}}{(1+ \hat\tau _k E_{1,k})^{-1/\hat\xi _k}-k^{-1}  } \right\}^{\hat\xi _k} -1 \right].
\label{eq:hatT}
\end{equation}
Based on \eqref{eq:phat} and \eqref{eq:FYT} an estimator for tail probabilities $\mathbb{P}(X>c)$ can be derived:
\begin{equation}
{\hat{p}_{T,k}(c)} = (1+ \hat{D}_{T,k})\; \frac{k}{n}\left(1+\hat\tau _k(c-X_{n-k,n}) \right)^{-1/\hat\xi _k} -\hat{D}_{T,k}.
\label{eq:pTc}
\end{equation}
Note that all proposed estimators from \eqref{eq:lik_xitau1}, \eqref{eq:lik_xitau2},  \eqref{eq:QTp} and \eqref{eq:pTc} are direct generalisations of the classical POT estimators under no-truncation which are obtained by setting $\hat{D}_{T,k}$ equal to 0.

\vspace{0.3cm}\noindent
From \eqref{eq:UYCD} it follows that when $p-(1-p)D_T >0$, or $p > D_T/(1+D_T) = \bar{F}_Y (T)$
\[
Q_Y (1-p) = Q_T ((1-p)(1+D_T))= Q_T (1-\{p-(1-p)D_T \}),
\]
from which the following estimator reconstructing $Q_Y (1-p)$ of the parent distribution $Y$ follows:
\begin{eqnarray}
\hat{Q}_{Y,k}(1-p) &=& \hat{Q}_{T,k} \left( 1-[p -(1-p)\hat{D}_{T,k}] \right)\nonumber \\
&=& X_{n-k,n} +
\frac{1}{\hat\tau _k}\left[\left\{ \frac{\hat{D}_{T,k} + \frac{k}{n}}{p(\hat{D}_{T,k}+1)} \right\}^{\hat\xi _k} -1 \right].
\label{eq:QYp}
\end{eqnarray}

\vspace{0.3cm}\noindent
In the specific case $\xi>0$ the estimators developed above can be compared with those developed in \citet{Truncation} for this special Pareto-type case:
\begin{eqnarray*}
H_{k,n} & = & \hat\xi ^+ _{k} + \frac{R_{k,n}^{1/\hat\xi ^+ _{k}}\log R_{k,n}}{1- R_{k,n}^{1/\hat\xi ^+_{k}}},\\
\hat{D}^+_{T,k} &=&\max \left\{0,\frac{k}{n}\frac{R_{k,n}^{1/\hat{\xi}^+_{k}}-\frac{1}{k}}{1-R_{k,n}^{1/\hat{\xi}^+_{k}}}\right\},\\
\log \hat{Q}^+_{T,k}(1-p) &=& \log X_{n-k,n} + \hat{\xi}^+_{k}
\log\left( \frac{\hat{D}^+_{T,k} + \frac{k}{n}}{\hat{D}^+_{T,k} + p}\right),
\end{eqnarray*}
with $H_{k,n} = \frac{1}{k}\sum_{j=1}^k \log X_{n-j+1,n}-\log X_{n-k,n}$ the \citet{Hill} statistic, and $R_{k,n} = X_{n-k,n}/X_{n,n}$.

\vspace{0.3cm}\noindent
Of course, in practice there is a clear need for detecting rough truncation. Let $(\bar{\mathcal{T}}_k)$ and $(\mathcal{T}_k)$ denote light and rough truncation with the thresholds $X_{n-k,n}$.  A test for 
$$H_{0,k}:\, (\bar{\mathcal{T}}_k) \mbox{ versus } H_{1,k}:\, (\mathcal{T}_k)$$ can be constructed generalising the goodness-of-fit test which was proposed by \citet{trHill} within a Pareto context, rejecting $H_{0,k}$ at asymptotic level $q \in (0,1)$ when
\begin{equation}
T_{k,n} := k \left(1+\hat{\tau}E_{1,k} \right)^{-1/\hat{\xi}_k}
> \log (1/q),
\label{eq:gof}
\end{equation}
while the P-value is given by $e^{-T_{k,n}}$, as under $H_{0,k}$, $T_{k,n}$ approximately  follows a standard exponential distribution as will be shown in Theorem~\ref{thm:trTest_mle} below.

\subsection{Simulation study}

The authors have performed an extensive simulation study concerning all the proposed estimators for different distributions of $Y$.
We compare the results with the results from a Pareto analysis $\hat\xi ^+ _{k}$ and $\hat{Q}^+_{T,k}(1-p)$ \citep{trHill,Truncation}, with the classical POT maximum likelihood results denoted by
 $\hat\xi ^{\infty} _{k}$, $\hat{Q}^{\infty}_{k}(1-p)$, and with the classical moment estimators \citep{Dek}
\begin{eqnarray*}
\hat \xi_{k}^{M} &=& M_{k}^{(1)}+1-\frac{1}{2} \left[1-\frac{\left(M_{k}^{(1)}\right)^2}{M_{k}^{(2)}}\right]^{-1}, \\
\hat{Q}_{k}^{M}(1-p) &=& X_{n-k,n}+X_{n-k,n}M_{k}^{(1)}\left(1-\hat \xi_{k}^{M}\right)\frac{\left(\frac{k}{np}\right)^{\hat \xi_{k}^{M}} -1}{\hat \xi_{k}^{M}},
\end{eqnarray*}
with
$M_{k}^{(j)}:=\frac{1}{k}\sum_{l=1}^{k}\log^j \left( X_{n-l+1,n}/X_{n-k,n}\right)$, $j=1,2$. In the Appendix we give a selection from these simulation results for $Y$ following the standard Pareto distribution, the standard lognormal distribution, the standard exponential distribution, and the GPD with RTF $H_{-0.2}$. For each setting, 1000 samples for $X$ of size 500 were generated where we consider different levels of truncation: $T=Q_Y(0.975)$, $T=Q_Y(0.99)$ and $T=Q_Y(1)$. Note that  the last case corresponds to no truncation, or $X =_d Y$. The samples were generated using inverse transform sampling with the quantile function $Q_T(p)=Q_Y(pF_Y(T))$ (which can easily be deduced from \eqref{eq:FYT}).

\vspace{0.3cm}\noindent To show the performance of the test for truncation, we plot the average P-values over the 1000 simulations as a function of $k$ in the first columns of Figures~\ref{fig:sim_xi_first}--\ref{fig:sim_xi_last} (full line). Additionally, the median (dashed line), first quartile (dotted line) and third quartile (dotted line) of the P-values over the 1000 simulations are also plotted as a function of $k$. This corresponds to the box of the boxplot of P-values as a function of $k$. Finally, we add blue horizontal lines (dash-dotted line) indicating the standard significance levels of 1\% and 5\%. When truncation is present ($T=Q_Y(0.975)$ or $T=Q_Y(0.99)$), the average P-values show that the test rejects the null hypothesis of no truncation when $k$ is large enough. For the standard exponential, standard lognormal and GPD(-0.2,1) truncated at $T=Q_Y(0.99)$, the average P-value is higher than, or just below, the 5\% significance level, even for high values of $k$. However, when looking at the median values and the third quartile, we see that the majority, and sometimes more than 75\%, of the P-values are below the 5\% significance level.
When the data are not truncated, i.e.~$X =_d Y$, the P-values are on average always well above the considered significance levels, hence correctly not rejecting the null hypothesis. The first quartile of the P-values is also above the 5\% significance level, except for smaller values of $k$.  Note that when we look at $Y\sim \mbox{GPD}(-0.2,1)$, $Y$ itself is upper truncated at $-\sigma/\xi=5$, but still $X =_d Y$ when we set $T=Q_Y(1)$. The simulation results show that the test performs as expected: rejecting the null hypothesis when $T=Q_Y(0.975)$ or $T=Q_Y(0.99)$, and not rejecting the null hypothesis when $T=Q_Y(1)$.

\vspace{0.3cm}\noindent
Concerning the estimation of $\xi$, see the second and third columns in Figures~\ref{fig:sim_xi_first}--\ref{fig:sim_xi_last}, the behaviour of $\hat\xi _{k}$ in the standard Pareto case exhibits a slightly smaller bias but quite a larger variance compared to $\hat\xi ^+_{T,k}$ from \citet{trHill,Truncation}  which was constructed exclusively for the case $\xi>0$. The classical POT and moment estimators exhibit large bias under truncation, as they tend to -1 when the threshold tends to  $x_{n,n}$. The mean squared error of $\hat\xi _{k}$ is comparable to the mean squared error (MSE) of these estimators for $k \geq 200$. In  case of no truncation the bias of $\hat\xi _{k}$ is the smallest for $k \geq 100$ while the mean squared error is the worst of the four estimators.\\
 When $\xi \leq 0$, the estimator  $\hat\xi ^+_{T,k}$ from the Pareto analysis is breaking down as can be expected whereas the difference between the classical estimators and the newly proposed POT estimator is small for $k \geq 200$ in case $\xi =0$ and $k \geq 300$ in the case $\xi <0$. In all cases presented $\hat\xi _{k}$ compares well for $k$ sufficiently large with the classical estimators when there is no truncation. Note that all estimators have a large bias for the (truncated) log-normal distribution. As can clearly be seen, the bias of all estimators decreases as truncation becomes lighter, or when there is no truncation, as expected. Moreover, the stable area of the $\hat{\xi}_k$ estimates starts for smaller values of $k$ when the truncation point gets larger.

\vspace{0.3cm}\noindent
Concerning the estimation of $Q_T(1-p)$, see Figures~\ref{fig:sim_Q_first}--\ref{fig:sim_Q_last} with $p=0.01$ and $0.005$ and $T= Q_Y(0.975), Q_Y(0.99)$, the estimator $\hat{Q}_{T,k}(1-p)$ has the smallest bias, uniformly over all distributions and values of $p$ considered, while the MSE values are always comparable with the best performing estimators. Even in case of no truncation $\hat{Q}_{T,k}(1-p)$ does not lose too much accuracy in comparison with the classical MLE estimator. 
%Note also that the classical POT quantile estimator and $%\hat{Q}_{T,k}(1-p)$ are comparable when the underlying tail of %$Y$ is lighter.

%\vspace{0.3cm}\noindent
%Finally, with the aim of reconstructing $Q_Y (1-p)$, see see Figures~\ref{fig:sim_QY_first}--\ref{fig:sim_QY_last}, the bias of the classical POT estimator $\hat{Q}_{k}^{\infty} (1-p)$ is remarkably small when $\xi \leq 0$, while in the Pareto case $\xi >0$ it exhibits a large bias but the smallest MSE. The proposed estimator $\hat{Q}_{Y,k}$ is appropriate from a certain level of $k$ on: $k\geq 200$ for Pareto, lognormal and exponential tails, and $k \geq 300$ for $\xi <0$. In case $\xi >0$ the bias compares well with that of $\hat{Q}^+_{Y,k}$ designed for this case,  while the MSE becomes larger for smaller values of $p$.  In the lognormal and exponential case, $\hat{Q}_{Y,k}(1-p)$ is competing with the classical POT estimator, while for $\xi <0$ the classical POT estimator is superior.

\subsection{Asymptotic results}

Here we present the asymptotic normality of ($\hat\xi, \hat\tau $) and $\hat{Q}_{T,k}(1-p)$ under rough truncation, and the asymptotic null distribution of the goodness-of-fit test statistic $T_{k,n}$. The proofs are provided in the Appendix. 
\\
We assume a second-order remainder relation in \eqref{eq:maxdU} as in Theorem~3.4.2 in \citet{dHF}: with $\xi>-\frac12$,
\begin{equation}
\lim_{t\to\infty} \frac{\frac{U_Y(tx)-U_Y(t)}{a_Y(t)}-\frac{x^{\xi}-1}{\xi}}{A(t)}=\Psi_{\xi,\rho}(x)  \mbox{ for all } x>0,
\label{eq:thm342}
\end{equation}
where
\[\Psi_{\xi,\rho}(x) = \int_1^x s^{\xi-1} \int_1^s u^{\rho-1} \, du\,ds,\]
with $\rho\leq 0$.
Furthermore, we introduce the notations $b_{T,k,n}:=\frac{k+1}{(n+1)D_T}$, $a_{T,k,n}:=a_Y\left(1/(\bar{F}_Y(T)(1+b_{T,k,n}))\right)$, and we denote the limit of $k/(nD_T)$ under rough truncation as derived in \eqref{eq:knDT} by $\beta := (1+\xi\kappa)^{1/\xi}-1$. 

\vspace{0.3cm}\noindent
\begin{theorem}\label{thm:trMLE}
Let $X_1,X_2,\ldots,$ be i.i.d.~random variables with distribution function $F_T$ following \eqref{eq:FYT} where $U_Y$ satisfies \eqref{eq:thm342}. Let $n,k=k_n\to\infty$, $\frac{k}{n}\to0$, $T\to\infty$. Then, under $(\mathcal{T}_t)$ we have that as $\sqrt{k}A(1/[\bar{F}_Y(T)(1+b_{T,k,n})]) \to \lambda \in \mathbb{R}$
\[
\sqrt{k} \left(\hat\xi _k -\xi, \hat \tau_k a_{T,k,n} -\xi \right)'
=  \mathcal{I}_{\beta} ^{-1}{\bf N}_{\xi,\beta}+ \lambda\mathcal{I}_{\beta} ^{-1}{\bf f}_{\xi,\beta,\rho} + o_p(1){\bf 1},
\]
where 
\[
\mathcal{I}_{\beta} = 
\begin{pmatrix}
1- \frac{1+\beta}{\beta^2}\log^2 (1+\beta) &
\frac{1}{\xi}\left[-\frac{\xi}{1+\xi}\frac{1+\beta}{\beta}(1-(1+\beta)^{-1-\xi})
\right.\\
&\left.\qquad
+\frac{1+\beta}{\beta^2}\log (1+\beta) (1-(1+\beta)^{-\xi})\right] 
 \\ \\
 -\frac{1}{\xi}\left[ -\frac{\xi}{1+\xi}\frac{1+\beta}{\beta}(1-(1+\beta)^{-1-\xi})\right.
& -\frac{1}{\xi\beta}\left[\frac{\xi}{1+2\xi}(1+\beta)(1-(1+\beta)^{-1-2\xi}) \right.
\\ \left. \quad
+\frac{1+\beta}{\beta^2}\log (1+\beta) (1-(1+\beta)^{-\xi})\right] & \hspace{1.2cm}\left. - \frac{1+\beta}{\beta}\frac{1}{\xi}(1-(1+\beta)^{-\xi})^2 \right]
\end{pmatrix},
\]
\[
{\bf N}_{\xi,\beta}= \frac{\beta}{1+\beta}
\begin{pmatrix}
\xi \int_0^1 W_n(u) \left( \frac{1+u\beta}{1+\beta}\right)^{-1}du \\
\hspace{0.5cm}
-\xi W_n(1) \left(-\frac{(1+\beta)^{1-\xi}\log (1+\beta)}{\beta^2} +\frac{\xi (1+\beta)^{-\xi}+(1+\beta)}{(1+\xi)\beta}\right)\\ \\
\xi(1+\xi)\int_0^1W_n (u) \left( \frac{1+u\beta}{1+\beta}\right)^{-1+\xi}du \\
\hspace{0.5cm}
-W_n(1) \left( \frac{\xi(1+\xi)(1+\beta)}{(1+2\xi)\beta}(1-(1+\beta)^{-1-2\xi})\right.\\
\left. \hspace{2.5cm}-\frac{(1+\beta)^{1-\xi}}{\beta^2} (1-(1+\beta)^{-\xi})\right)
\end{pmatrix},
\]
and
\[
{\bf f}_{\xi,\beta,\rho}= 
\begin{pmatrix}
\xi \int_0^1 \Psi_{\xi,\rho}(\frac{1+\beta}{1+u\beta}) \left( \frac{1+u\beta}{1+\beta}\right)^{\xi}du \\
\hspace{1cm}
-\xi  \Psi_{\xi,\rho}(1+\beta)(1+\beta)^{-\xi}
\left(\frac{(1+\beta)\log (1+\beta)}{\beta^2} -\frac{1}{\beta}\right)\\ \\
\xi(1+\xi)\int_0^1 \Psi_{\xi,\rho}(\frac{1+\beta}{1+u\beta}) \left( \frac{1+u\beta}{1+\beta}\right)^{2\xi}du\\
\hspace{1cm}
-
\Psi_{\xi,\rho}(1+\beta) \frac{(1+\beta)^{1-\xi}}{\beta^2} (1-(1+\beta)^{-\xi})
\end{pmatrix},
\]
for a sequence of Brownian motions $\{ W_n (s); s \geq 0\}$.
\end{theorem}

\vspace{0.3cm} \noindent
Under $(\bar{\mathcal{T}}_t)$ the asymptotic result for $(\hat\xi _k, \hat\tau _k)$ can be checked to be identical to that of the classical MLE estimators under no truncation as given in Theorem 3.4.2 in \citet{dHF}. 
\\
Note that  the information matrix $\mathcal{I}_{\beta}$ equals 0 when $\kappa=0$, or equivalently $\beta =0$, so that the asymptotic variances are unbounded in such case. In practice this induces large variances for smaller values of $k$. This also appears in Figures~\ref{fig:sim_xi_first}--\ref{fig:sim_xi_last}. Fortunately, the bias stays reasonably small for  larger values of $k$, as can be deduced for instance in case of the lognormal distribution. 

\vspace{0.3cm} \noindent
In order to state the asymptotic result for the quantile estimator $\hat{Q}_{T,k}(1-p)$ with $p=p_n \to 0$, we use the notation $d_n = k/(np_n)$. Furthermore, we will use the result that when $U_Y$ satisfies \eqref{eq:thm342}, we have that 
\begin{equation}\label{eq:dHF}
\lim_{t \to \infty}\frac{\frac{a_Y (tx)}{a_Y(t)}-x^{\xi}}{A(t)}
= Cx^{\xi}\frac{x^{\rho}-1}{\rho}
\end{equation}
for some constant $C$ (see B.3.4 in \citet{dHF}).
\noindent
\begin{theorem}\label{thm:Q_T}
Let $X_1,X_2,\ldots,$ be i.i.d.~random variables with distribution function $F_T$ following \eqref{eq:FYT} where $U_Y$ satisfies \eqref{eq:thm342}. Let $n,k=k_n\to\infty$, $\frac{k}{n}\to0$, $T\to\infty$, $p=p_n \to 0$ and $np_n/\sqrt{k} \to 0$. Then, under $(\mathcal{T}_t)$  we have that
\begin{align*}
& \hspace{-1cm}
\frac{\left( \hat{Q}_{T,k}(1-p) -Q_T (1-p) \right)}{a_Y\left(\frac1{\bar{F}_Y(T)}\right)} \\
& = -\frac{\beta}{k}(E-1) + O_p\left(\frac{1}{k^2}\vee \frac{1}{d_n^2}\right) \\
& \quad -\beta\left(\frac{1}{d_n}-\frac{1}{k} \right)\left[ A\left(\frac1{\bar{F}_Y(T)}\right)C \frac{(1+\beta)^{-\rho}-1}{\rho} \right. \\
& \hspace{2cm} +\left( \frac{ \hat{\xi}_k}{\hat{\tau}_k}\frac{1}{a_{T,k,n}}-1 \right)\\
& \hspace{2cm} - \left(\hat{\xi}_k -\xi \right)\frac{1}{\xi}
\frac{(1+\beta)\log (1+\beta)}{\beta} \\
& \hspace{2cm} +\left(\hat{\tau}_k a_{T,k,n}-\xi \right)
\frac{1- (1+\beta)^{-\xi}}{\xi}\left(1+\frac{1+\beta}{\xi\beta}\right)\displaybreak\\
&\hspace{2cm}
+(1+\beta)^{-\xi}\left(\frac{1+\beta}{\beta}+\xi\right) \\
& \hspace{2.5cm} \left. \times\left(-\frac{W_n(1)}{\sqrt{k}}
+ A\left(\frac1{\bar{F}_Y(T)}\right)(1+\beta)^{-\xi}\Psi_{\xi,\rho}(1+\beta)
\right)
\right],
\end{align*}
where $E$ is a standard exponential random variable and  $\{ W_n (s); s \geq 0\}$ a sequence of Brownian motions.
\end{theorem}
This result should be compared with Theorem~4.3.1 in \citet{dHF} stating the basic asymptotic result for the quantile estimator based on the classical ML estimators under no truncation. Note that under $(\mathcal{T}_t)$ the rate of the stochastic part in the asymptotic representation is $O_p(1/k)$ rather than the classical $O_p(1/\sqrt{k})$.

%\vspace{0.3cm}\noindent
\begin{theorem}\label{thm:trTest_mle}
Let $X_1,X_2,\ldots,$ be i.i.d.~random variables with distribution function $F_T$ following \eqref{eq:FYT} where $U_Y$ satisfies \eqref{eq:thm342}. Let $n,k=k_n\to\infty$, $\frac{k}{n}\to0$, $T\to\infty$. Then, under $(\bar{\mathcal{T}}_k)$ with $nD_T \to 0$ we have that
\[
T_{k,n}=_d E(1+o_p(1))
\]
where $E$ is a standard exponential random variable.
\end{theorem}

\section{Case studies}

Analysing the magnitude data from the Groningen area, it appears that the given 200 top data confirm the Gutenberg-Richter law with $\hat\xi _k$ clearly indicating that $Y=M$ belongs to the Gumbel $\xi=0$ domain. The goodness-of-fit test rejects light truncation for $k\geq 40$ and the proposed truncation model fits well to the top 50 data as indicated on the exponential QQ-plot in Figure~\ref{fig:Groningen2}. Furthermore, $\hat{D}_{T,k}$ indicates a truncation volume $D_T$ between 0.01 and 0.02. Finally, the endpoint can be estimated in two ways: directly on the magnitude data using $\hat{T} _{M,k} = \hat{Q}_{T,k}(1)$, or using a Pareto analysis $\hat{T}^+_{E,k} = \hat{Q}^+_{T,k}(1)$ based on the energy data and transforming back to the magnitude scale with a logarithmic transformation. Both approaches lead to a value around 3.75.

   \begin{figure}[!ht]
	\centering
	    \includegraphics[height=0.45\textwidth, angle=270]{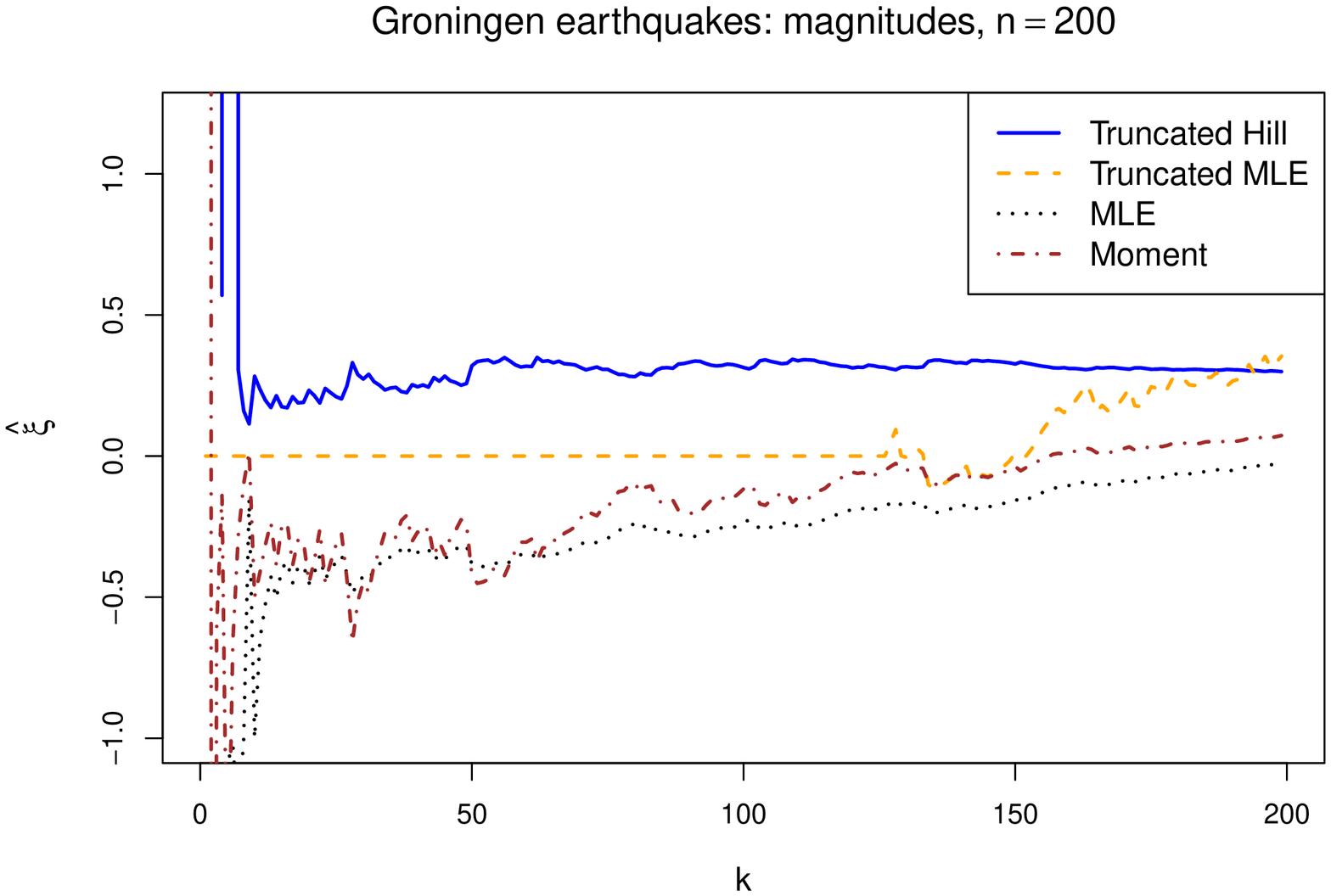}      
			\includegraphics[height=0.45\textwidth, angle=270]{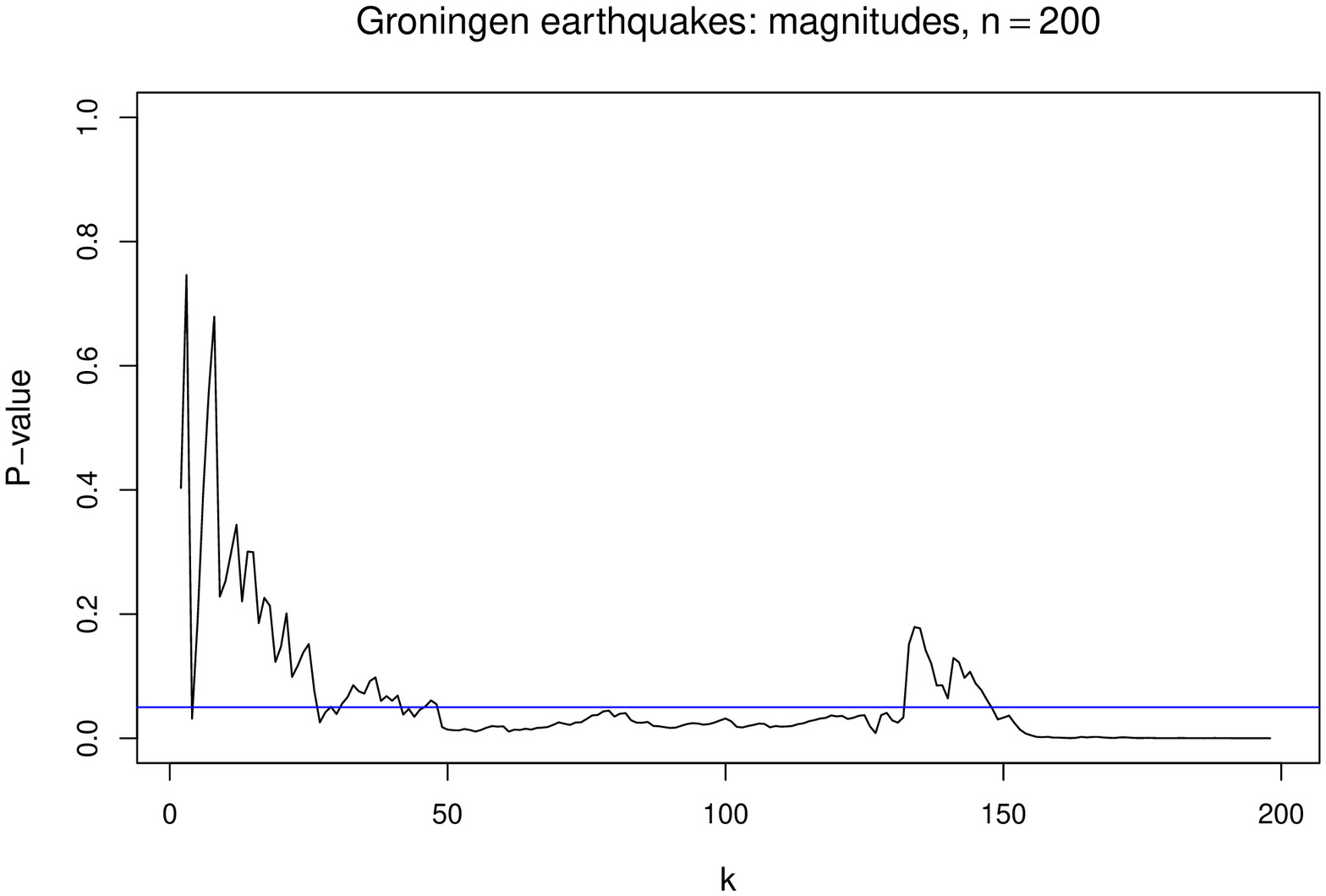}\\
      \includegraphics[height=0.45\textwidth, angle=270]{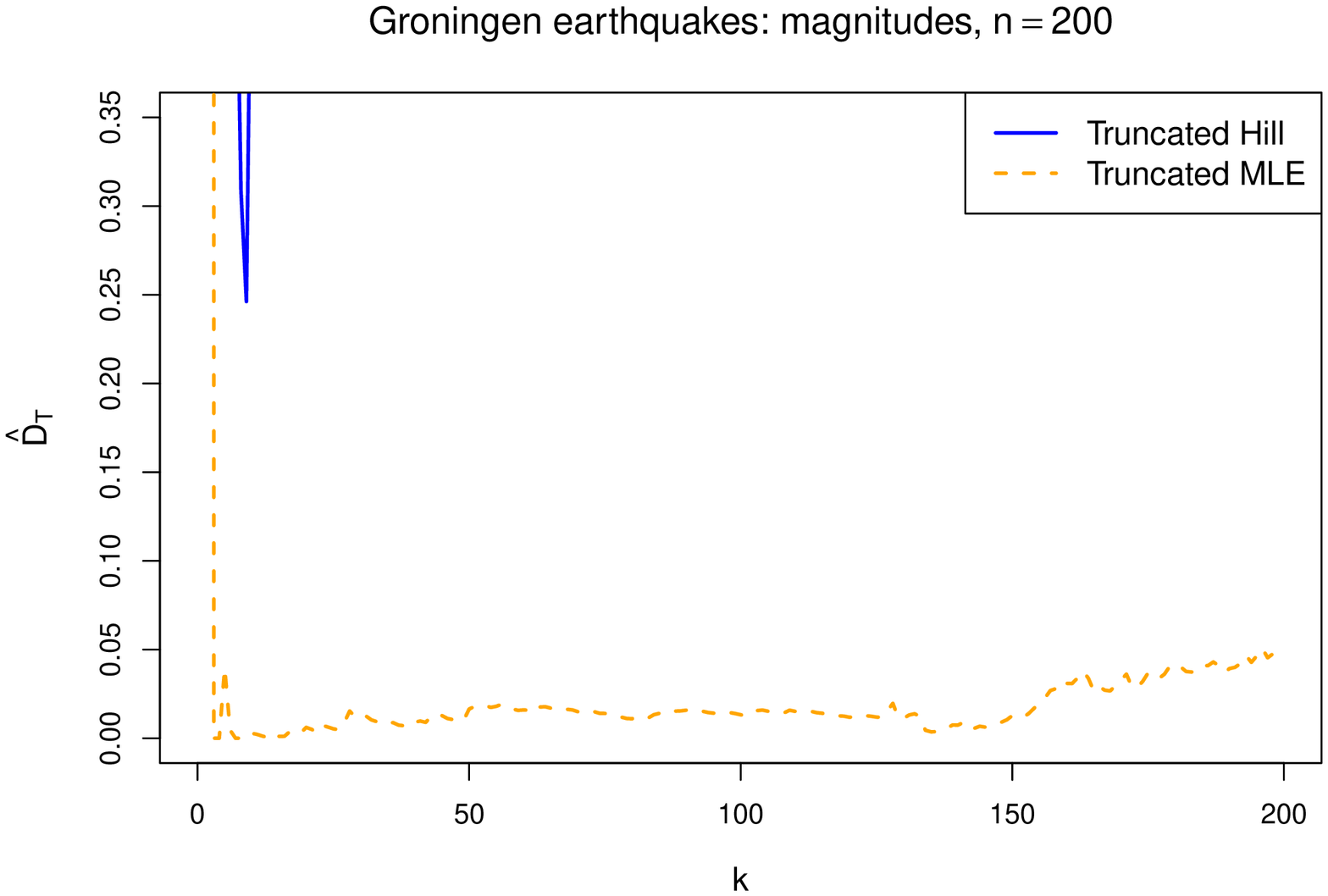}
			\includegraphics[height=0.45\textwidth, angle=270]{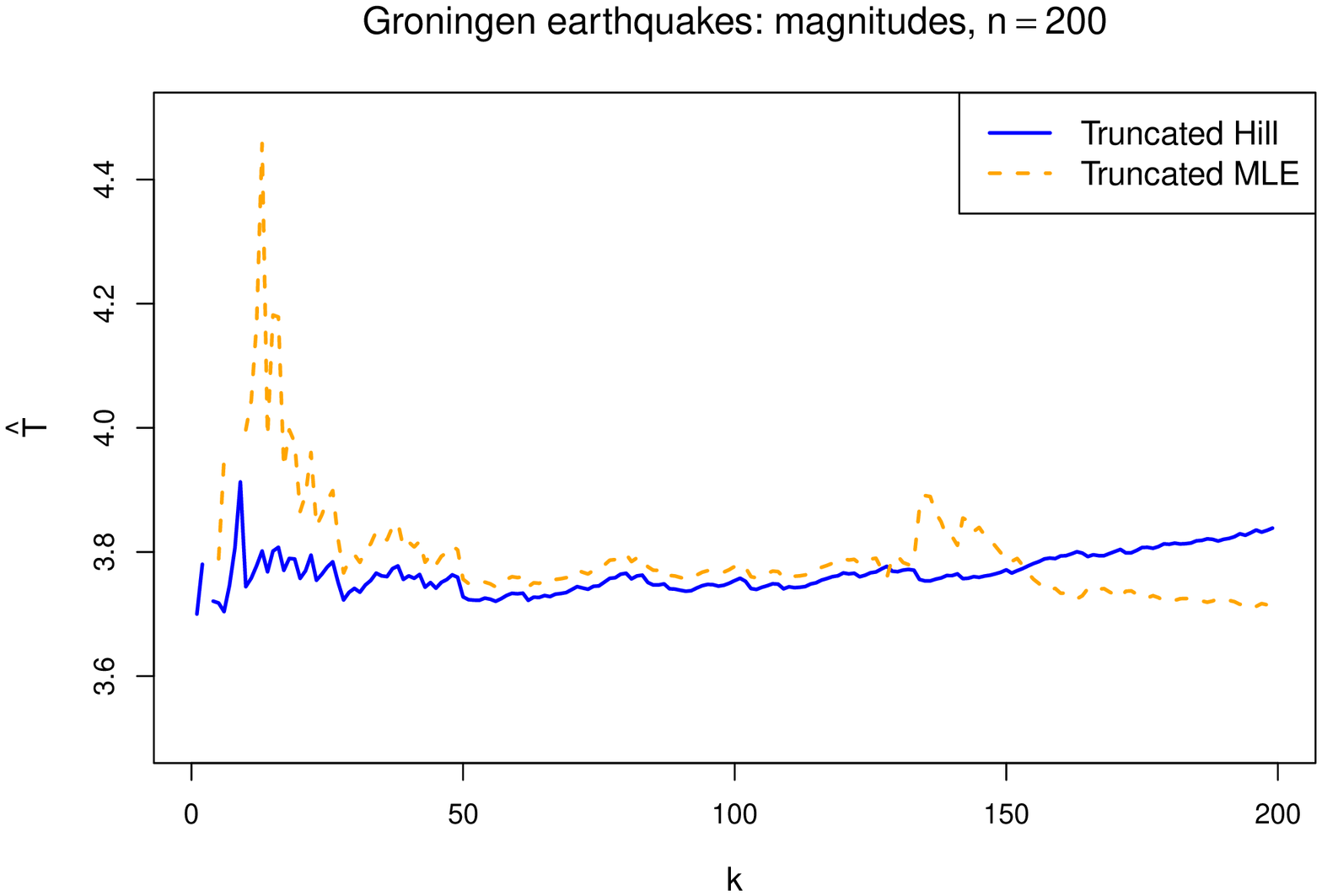}\\
			\includegraphics[height=0.45\textwidth, angle=270]{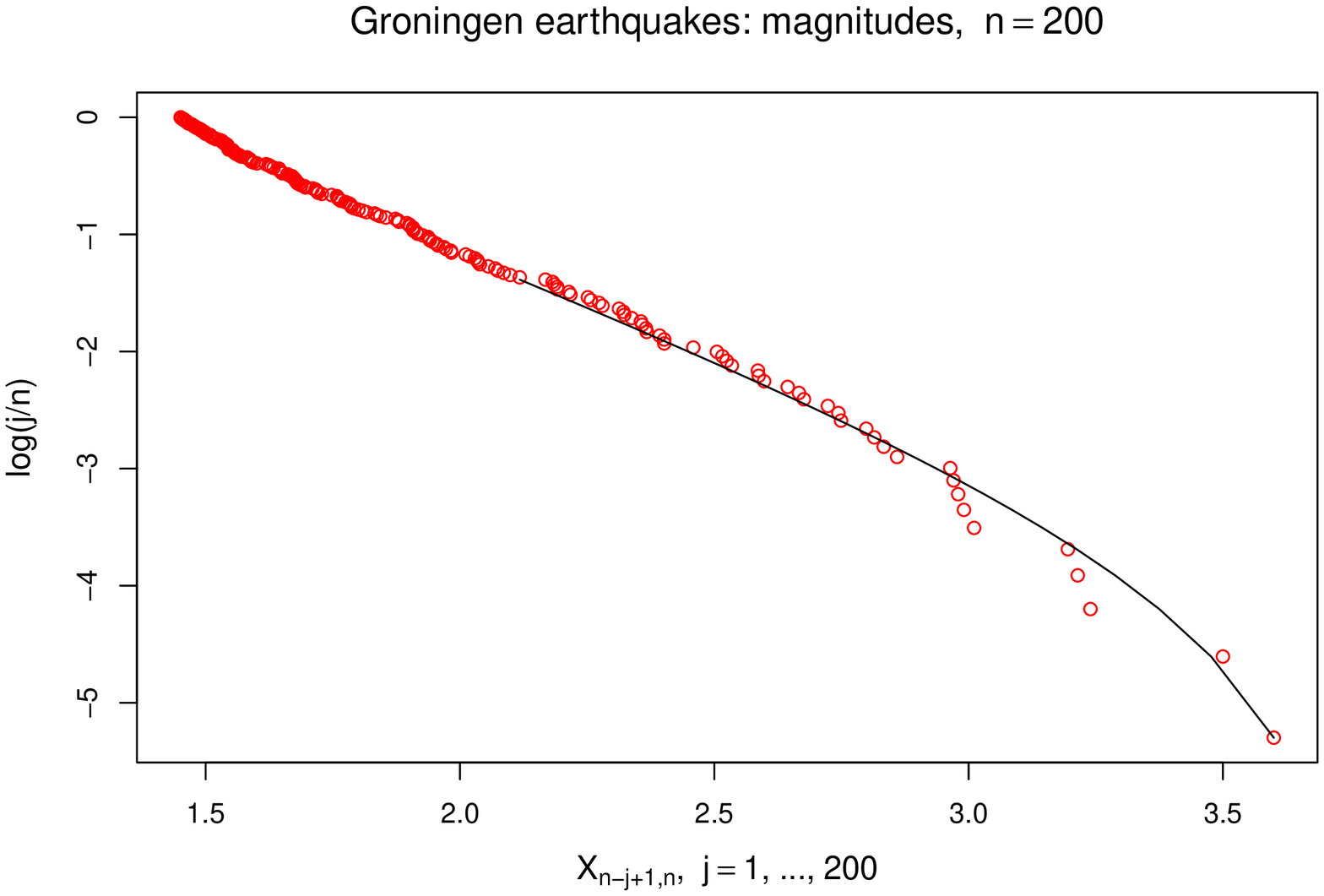}
  \caption {Groningen earthquake magnitude data: $\hat{\xi}^+_{k}$, $\hat{\xi}_{k}$, $\hat{\xi}^{\infty}_{k}$ and $\hat{\xi}^M_{k}$ (top left); P-values for test for truncation (top right); $\hat{D}^+_{T,k}$ and $\hat{D}_{T,k}$ (middle left); the logarithmic transform  of $\hat{T}^+_{E,k}$ and $\hat{T}_{M,k}$ (middle right); exponential QQ-plot with fit based on $k=50$ largest magnitudes (bottom).}
  \label{fig:Groningen2}
       \end{figure}

\vspace{0.3cm}\noindent
Concerning the diamond data introduced in Figure~\ref{fig:DiamondQQ}, $\hat{\xi}_{k}$ and $\hat{\xi}^+_{k}$, respectively $\hat{D}_{T,k}$ and $\hat{D}^+_{T,k}$,  correspond well for $k \geq 250$ and lead to a Pareto fit with extreme value index around 0.5 and a truncation odds $D_T$ around 0.02. The goodness-of-fit test now rejects light truncation for $k\geq 110$. Reconstructing $Q_Y (0.99)$ with $\hat{Q}_{Y,k}(0.99)$ and  $\hat{Q}^+_{Y,k}(0.99)$ leads to a value of 120 cts at $k=250$.

   \begin{figure}[!ht]
	\centering
	      \includegraphics[height=0.45\textwidth, angle=270]{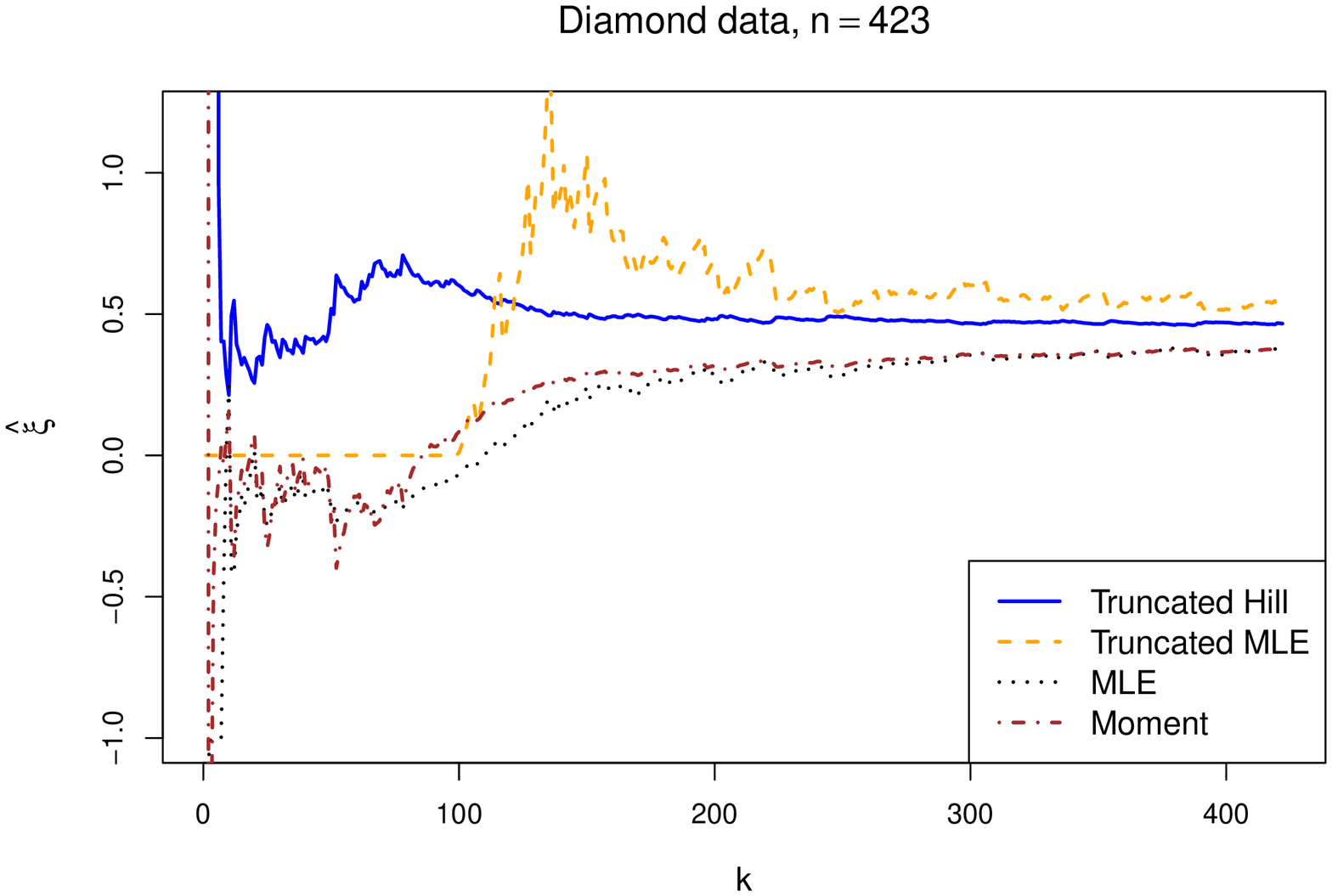}
			\includegraphics[height=0.45\textwidth, angle=270]{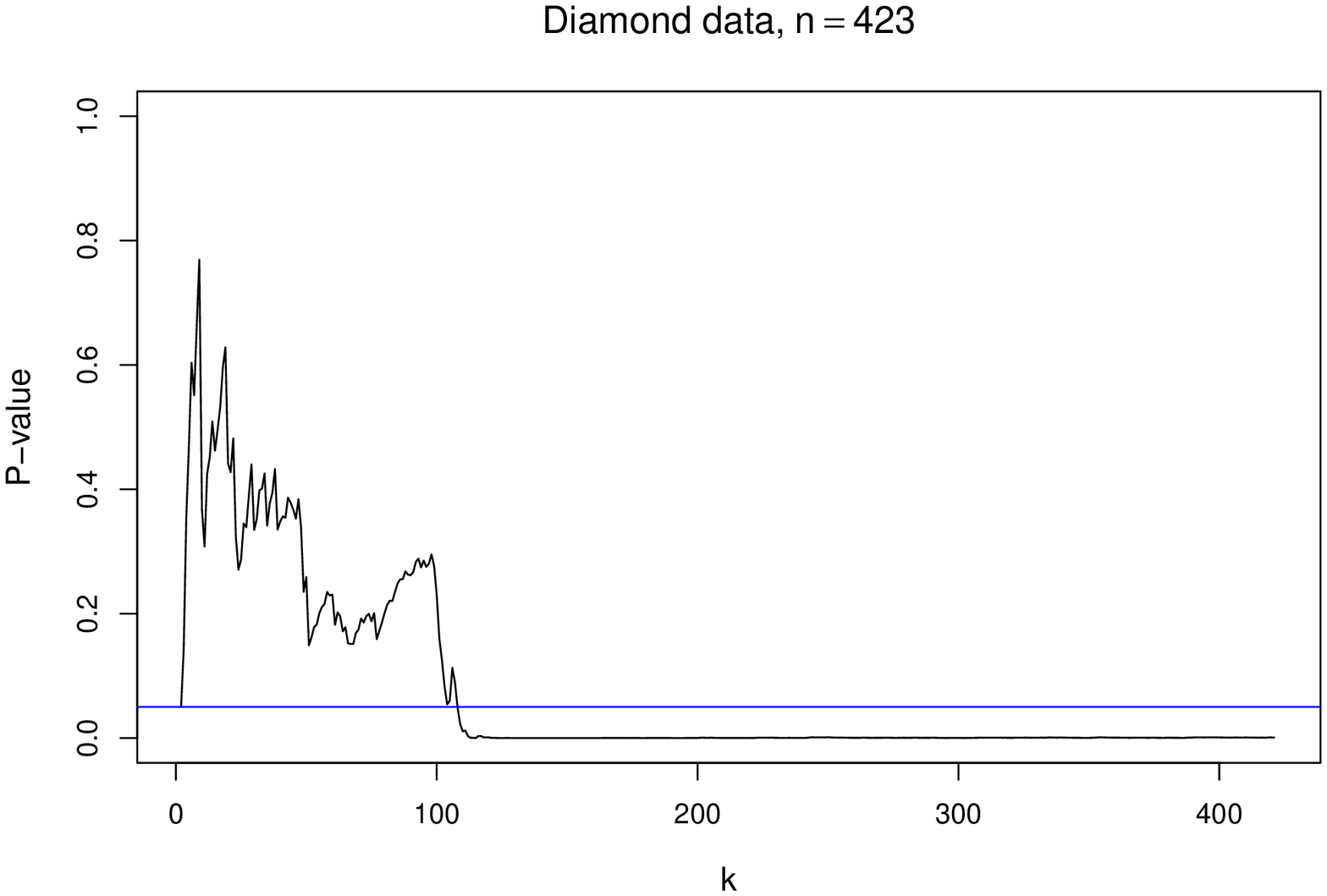}\\
      \includegraphics[height=0.45\textwidth, angle=270]{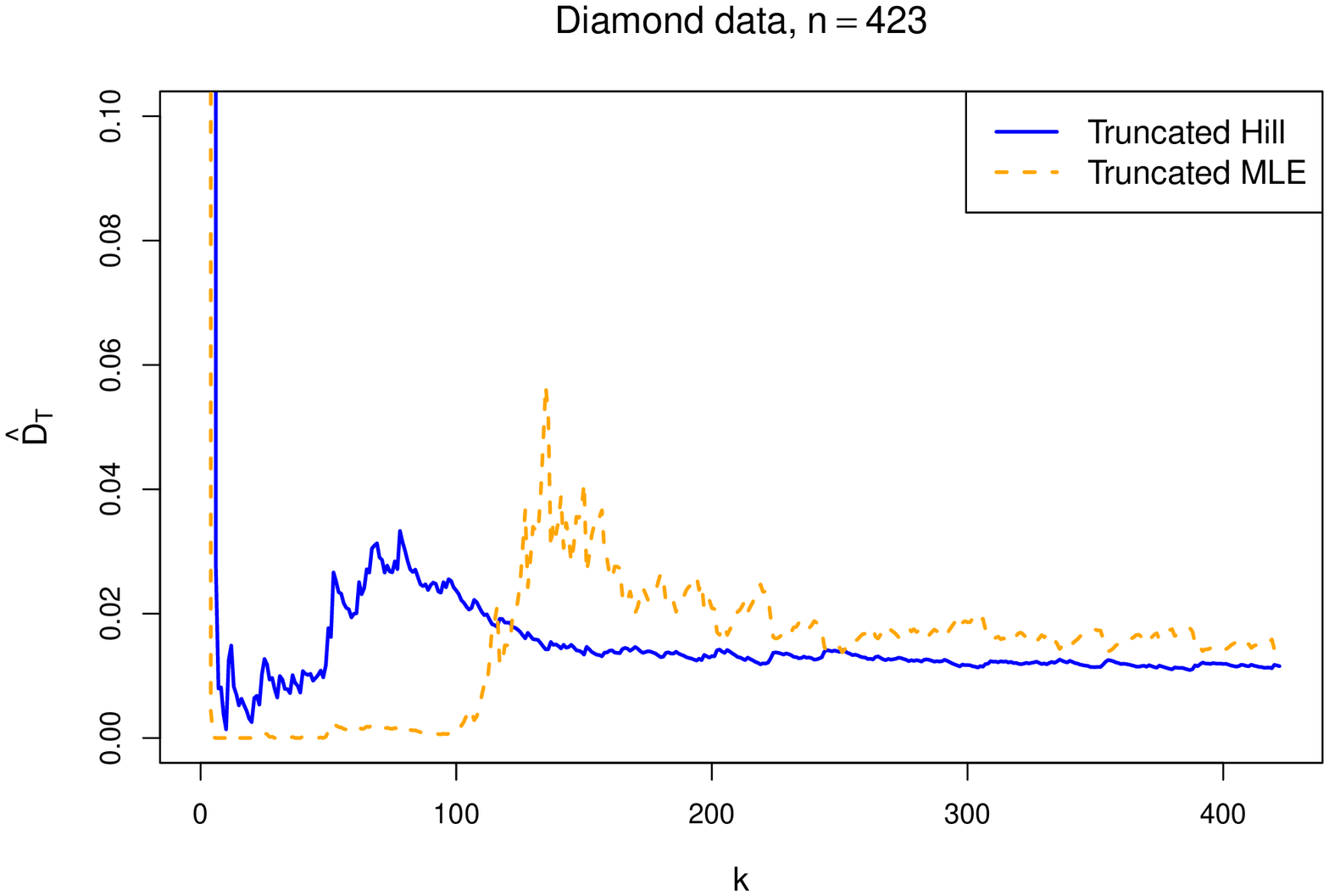}
			\includegraphics[height=0.45\textwidth, angle=270]{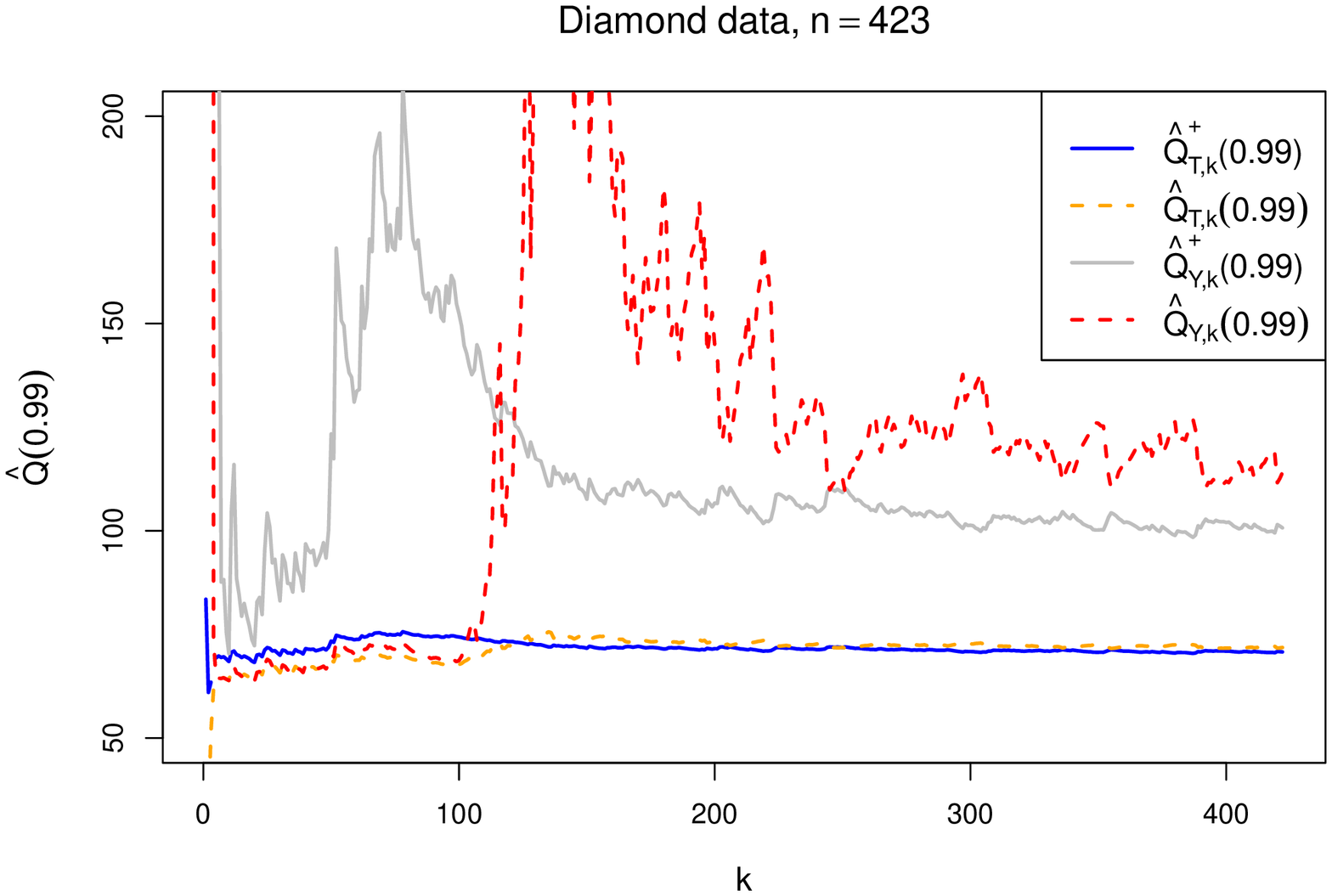}\\
			\includegraphics[height=0.45\textwidth, angle=270]{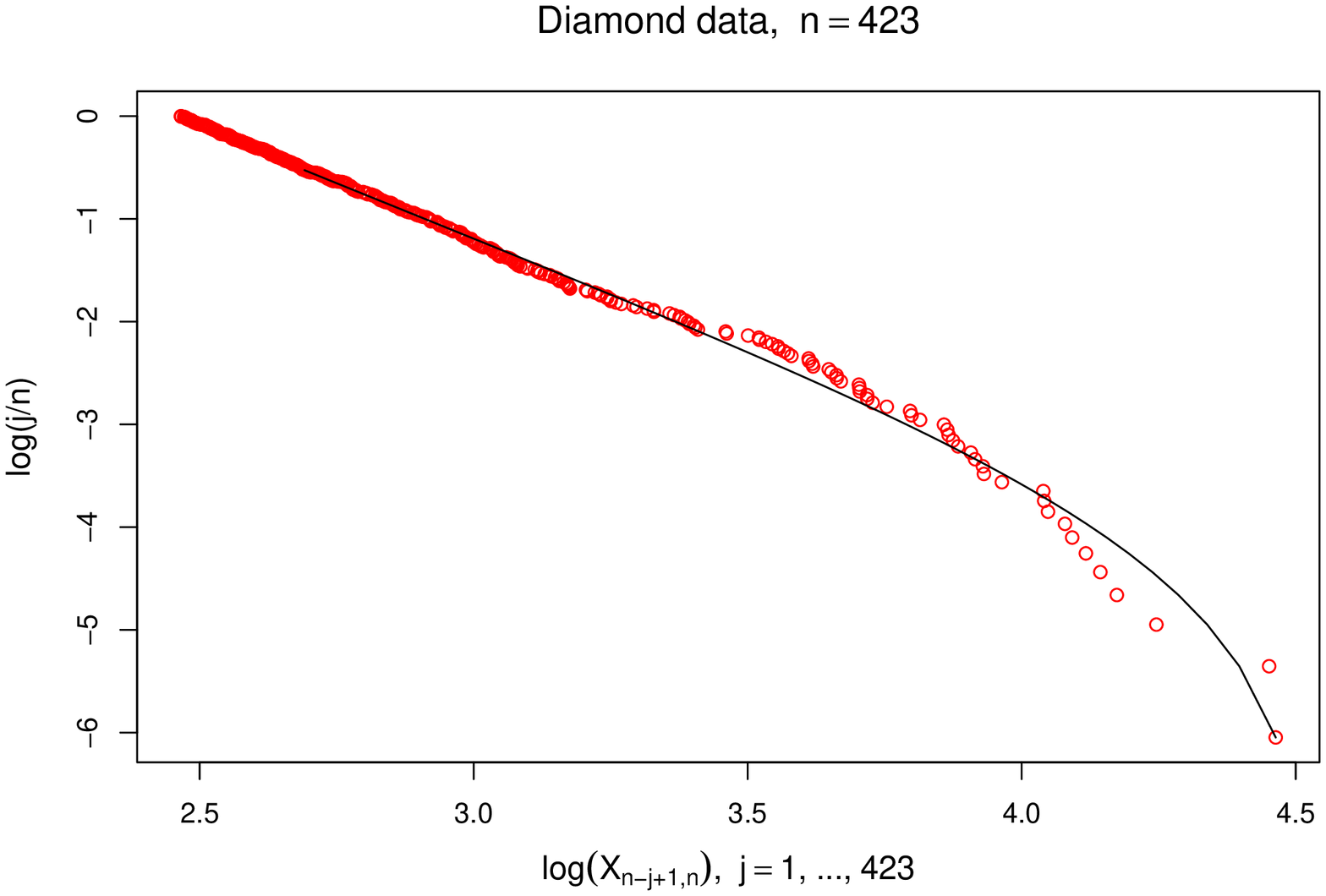}
  \caption {Diamond data: $\hat{\xi}^+_{k}$, $\hat{\xi}_{k}$, $\hat{\xi}^{\infty}_{k}$ and $\hat{\xi}^M_{k}$ (top left); P-values for test for truncation (top right); $\hat{D}^+_{T,k}$ and $\hat{D}_{T,k}$ (middle left); $\hat{Q}^+_{T,k}(0.99)$, $\hat{Q}_{T,k}(0.99)$, $\hat{Q}^+_{Y,k}(0.99)$ and $\hat{Q}_{Y,k}(0.99)$ (middle right); log-log plot with fit based on $k=250$ largest sizes (bottom).}
       \end{figure}

%\vspace{0.3cm}\noindent
Finally with the Molenbeek data, the goodness-of-fit test and  the fit of the proposed truncation model on the exponential QQ-plot on the top 100 data, again indicate that this $Y$ belongs to the Gumbel domain with an odds $D_T$ around 0.02. Here, the Pareto domain estimators 
$\hat\xi ^+_k$ and $\hat{D} ^+_{T,k}$ clearly do not show a stable pattern as a function of $k$. Reconstructing $Q_Y(0.97)$ leads to a value $\hat{Q}_{Y,100}(0.97)= 6.5\ m^3/s$.
   \begin{figure}[!ht]
	\centering
	     \includegraphics[height=0.45\textwidth, angle=270]{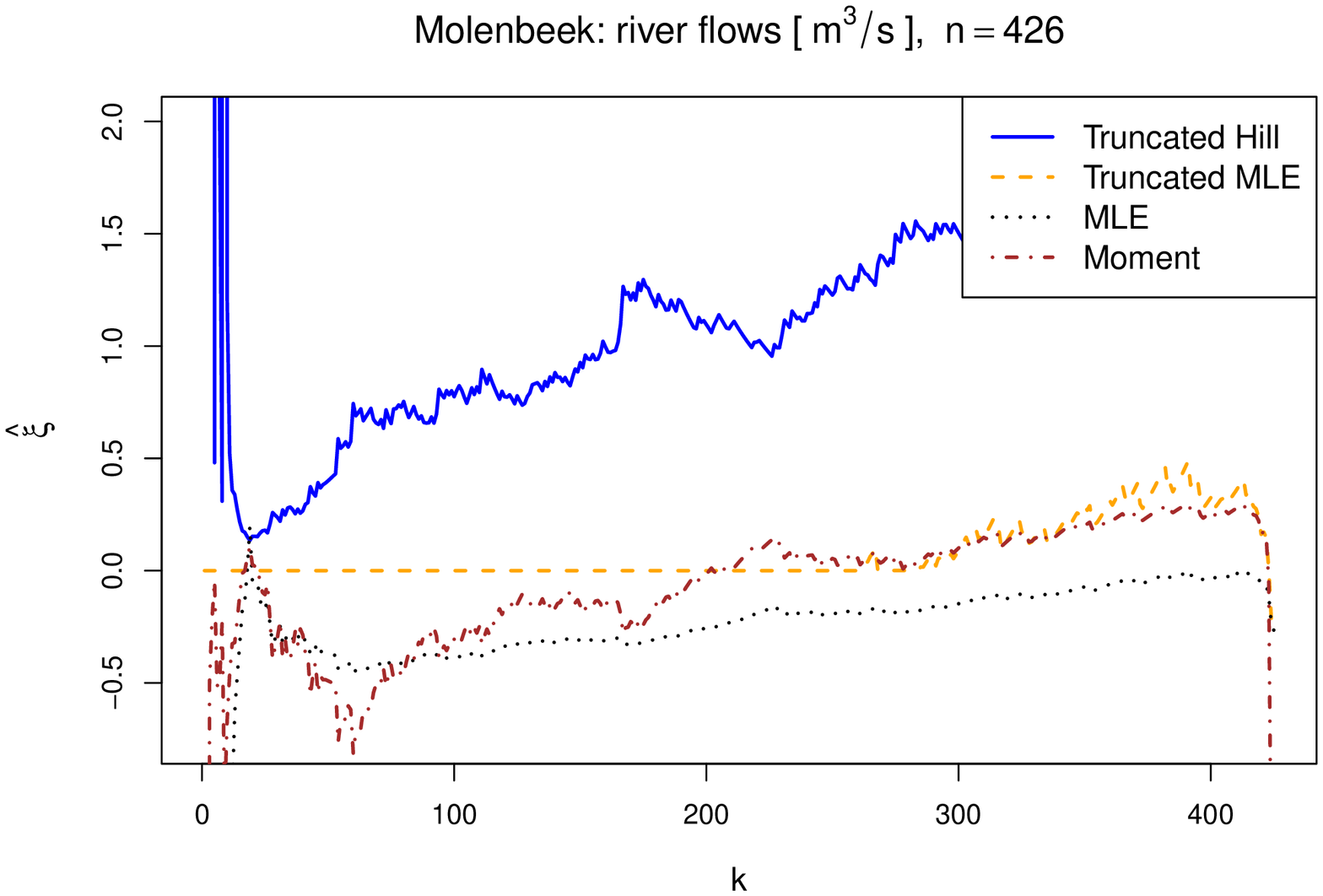}
			\includegraphics[height=0.45\textwidth, angle=270]{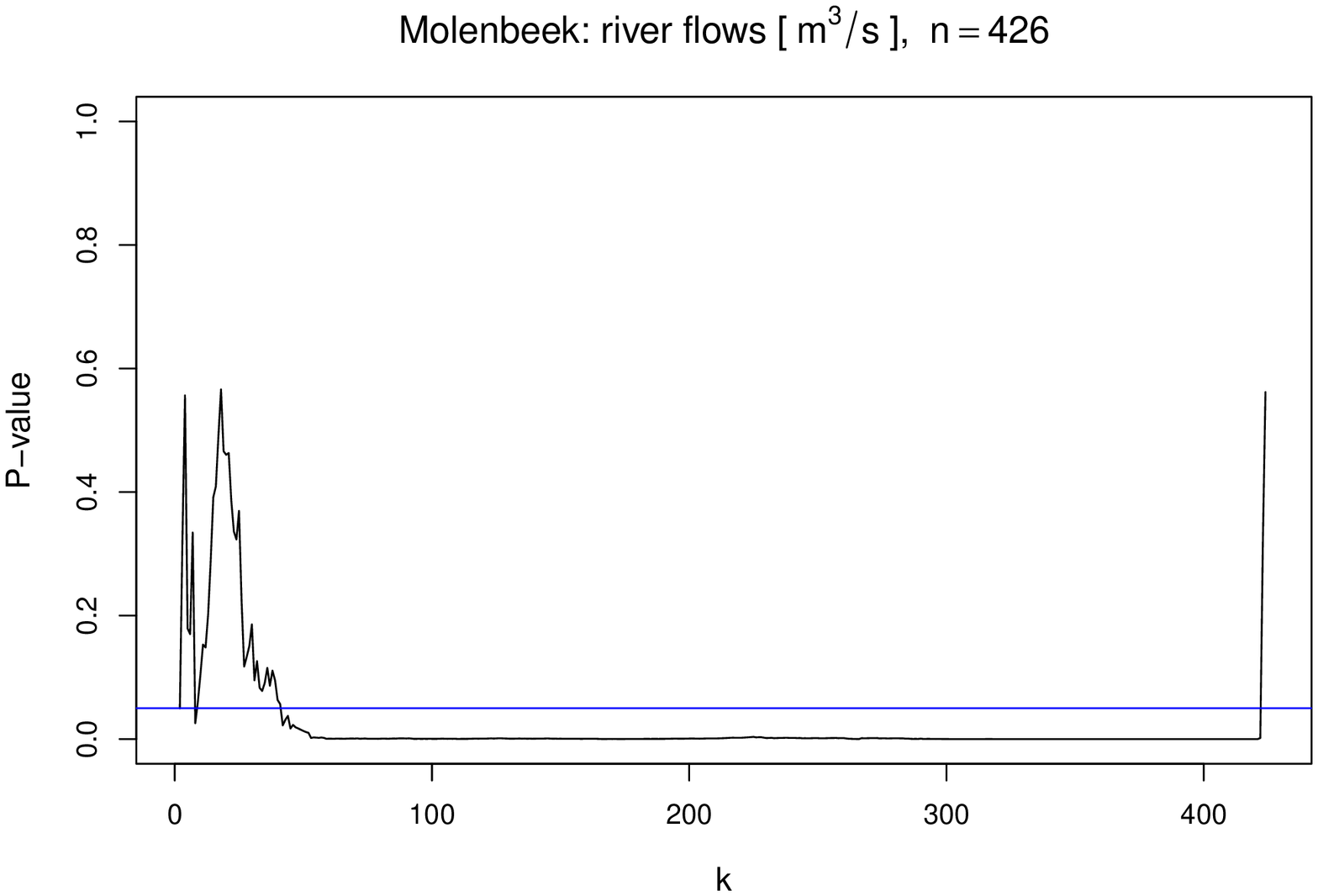}\\
      \includegraphics[height=0.45\textwidth, angle=270]{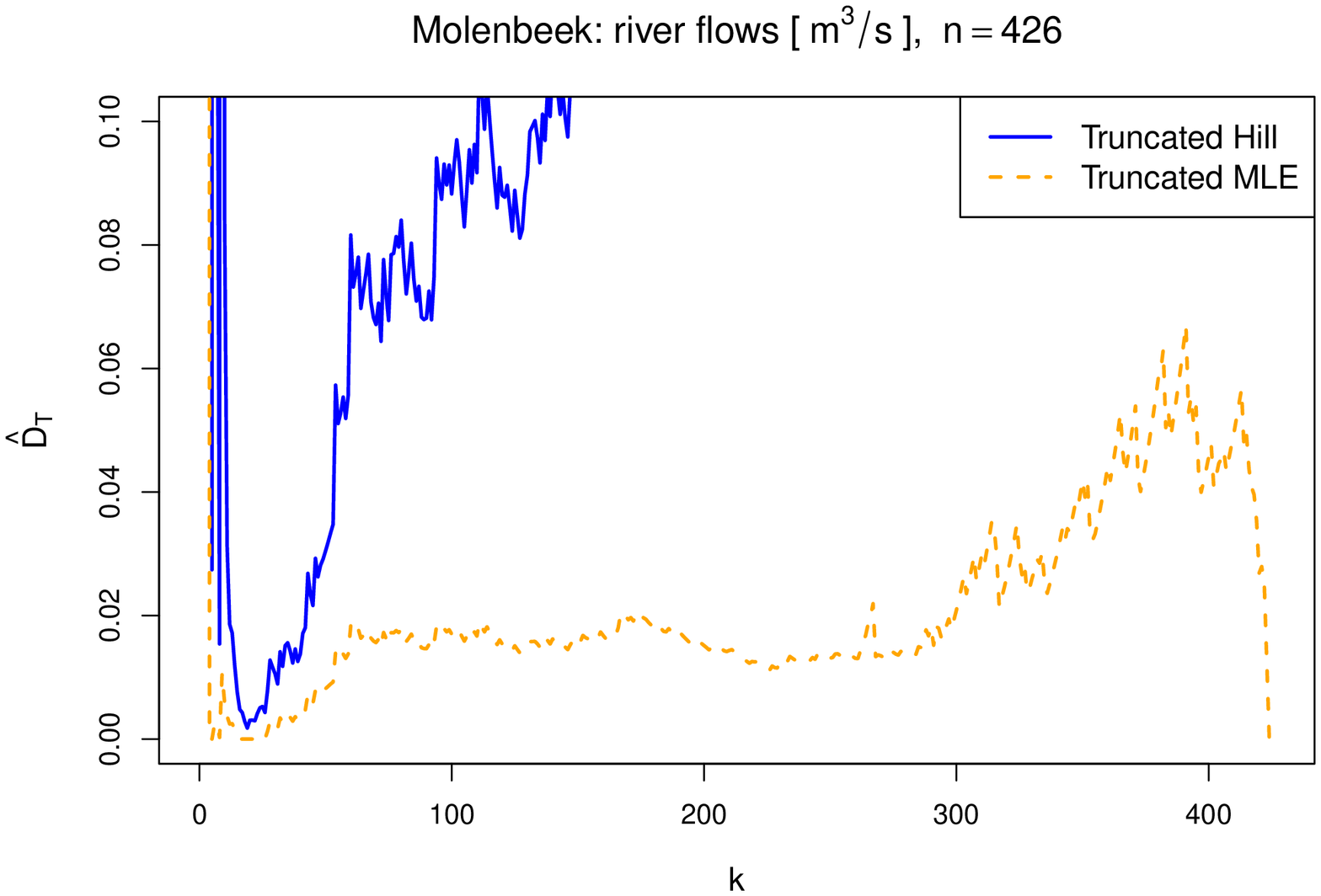}
			\includegraphics[height=0.45\textwidth, angle=270]{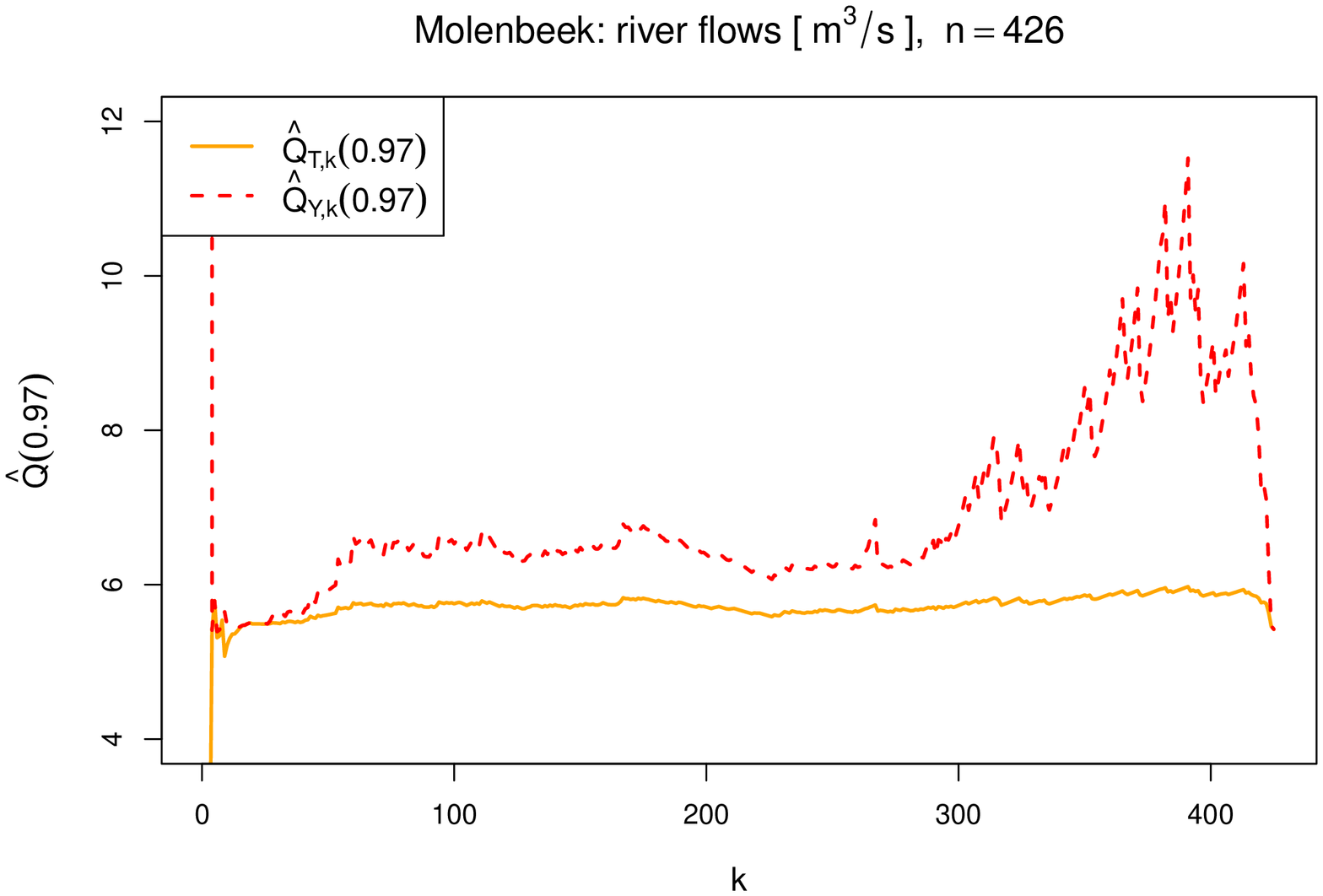}\\
			\includegraphics[height=0.45\textwidth, angle=270]{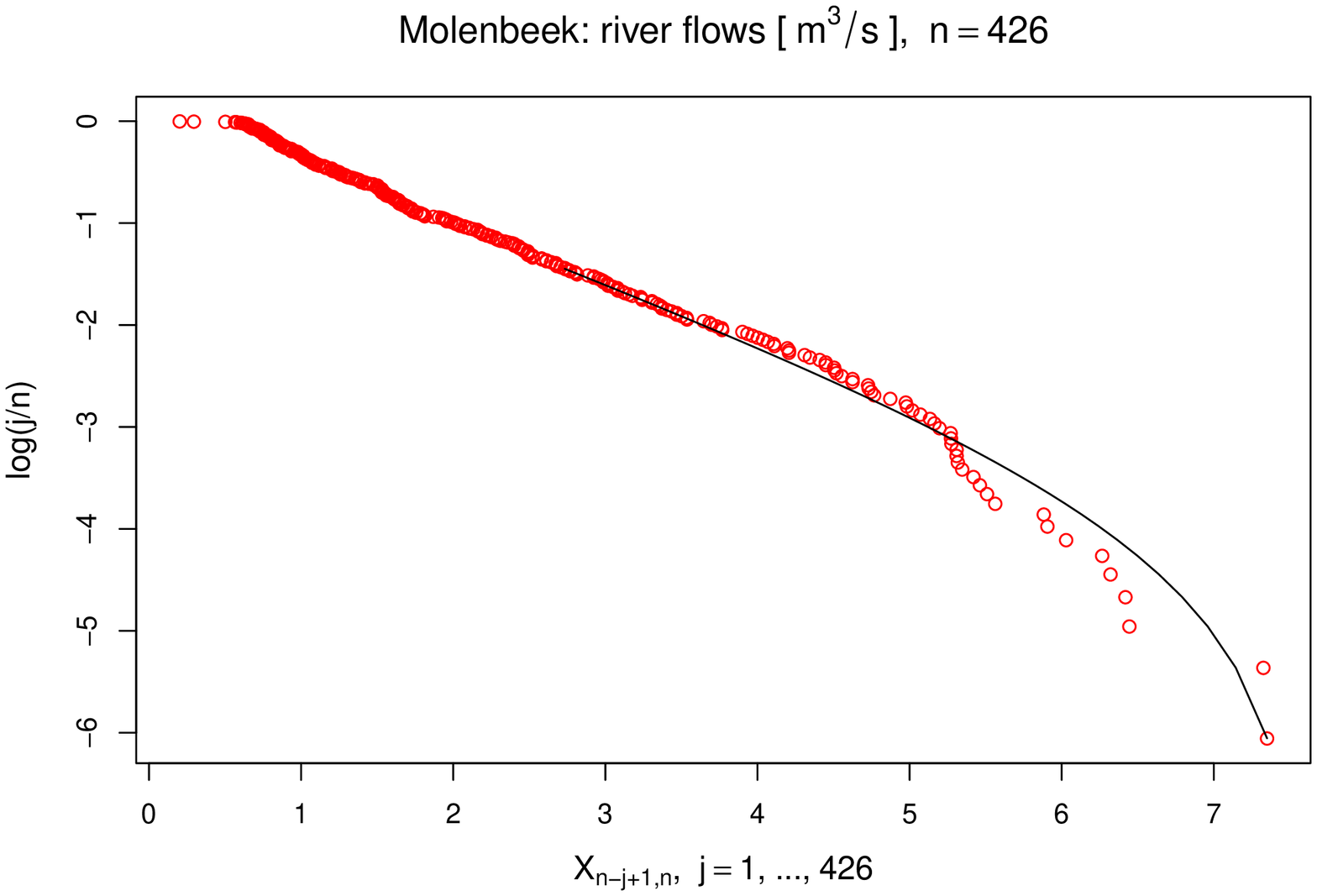}
  \caption {Molenbeek flow data: $\hat{\xi}^+_{k}$, $\hat{\xi}_{k}$, $\hat{\xi}^{\infty}_{k}$ and $\hat{\xi}^M_{k}$ (top left); P-values for test for truncation (top right); $\hat{D}^+_{T,k}$ and $\hat{D}_{T,k}$ (middle left); $\hat{Q}_{T,k}(0.97)$ and  $\hat{Q}_{Y,k}(0.97)$ (middle right); exponential QQ-plot with fit based on $k=100$ largest flows (bottom).}
       \end{figure}

\section{Discussion}

We proposed a general tail estimation approach for cases where truncation affects the ultimate right tail of the distribution. Using applications from geophysics, hydrology and geology we motivated the importance of this problem.
The proposed estimators of the extreme value index, and quantiles of the truncated and underlying non-truncated distribution, in most cases compare well with the best performing alternatives, even in case there is no truncation. The proposed estimator of extreme quantiles of a truncated distribution is performing uniformly best. While the alternative procedures sometimes break down in at least one situation, our proposals remain always useful for large enough $k$. Hence, in addition to the existing methods, this method can be an interesting extra tool when analysing tails.

\section*{Acknowledgements}
We thank the anonymous referee for the careful reading and the several suggestions that helped improving the paper.

  \bibliographystyle{agsm}
\bibliography{Truncating_allmaxdomains_bib}

\appendix

\section*{Appendix: proofs of Theorems}

%\section{New theorem and proof}

%\small
\begin{proposition}\label{thm:prop_trMLE} Under the condition of Theorem~\ref{thm:trMLE}, one can define a sequence of Brownian motions $\{W_n(s) \, |\, s>0\}$, such that for $\epsilon >0$
\begin{enumerate}
\renewcommand{\theenumi}{(\alph{enumi})}
\renewcommand\labelenumi{\theenumi}
	\item\label{itm:propa} \begin{align*}
	\max_{j=1,\ldots,k} \left(\frac{j}{k+1}\right)^{0.5+\varepsilon} \Bigg| &\sqrt{k} \left[\frac{X_{n-j+1,n}-U_T\left(\frac{n+1}{k+1}\right)}{a_{T,k,n}}-\frac1{\xi}\left(\left(\frac{1+\frac{j}{k+1}b_{T,k,n}}{1+b_{T,k,n}}\right)^{-\xi}-1\right)\right]\\
	&+\frac{b_{T,k,n}}{1+b_{T,k,n}} \left(\frac{1+\frac{j}{k+1}b_{T,k,n}}{1+b_{T,k,n}}\right)^{-1-\xi}W_n\left(\frac{j}{k+1}\right)\\
	&+\sqrt{k} A\left(\frac1{\bar{F}_Y(T)(1+b_{T,k,n})}\right)\Psi_{\xi,\rho}\left(\frac{1+b_{T,k,n}}{1+\frac{j}{k+1}b_{T,k,n}}\right)\Bigg| \to_p 0
	\end{align*}
	
	\item\label{itm:propb} \begin{align*}
	\max_{j=1,\ldots,k} \left(\frac{j}{k+1}\right)^{0.5+\varepsilon} \Bigg| &\sqrt{k} \left[\frac{X_{n-j+1,n}-X_{n-k,n}}{a_{T,k,n}}-\frac1{\xi}\left(\left(\frac{1+\frac{j}{k+1}b_{T,k,n}}{1+b_{T,k,n}}\right)^{-\xi}-1\right)\right]\\
	&+\frac{b_{T,k,n}}{1+b_{T,k,n}} \left[\left(\frac{1+\frac{j}{k+1}b_{T,k,n}}{1+b_{T,k,n}}\right)^{-1-\xi}W_n\left(\frac{j}{k+1}\right)-W_n(1)\right]\\
	&+\sqrt{k} A\left(\frac1{\bar{F}_Y(T)(1+b_{T,k,n})}\right)\Psi_{\xi,\rho}\left(\frac{1+b_{T,k,n}}{1+\frac{j}{k+1}b_{T,k,n}}\right)\Bigg| \to_p 0.
	\end{align*}
\end{enumerate}
\end{proposition}
\begin{proof} In order to derive \ref{itm:propa}, note that for $j=1,\ldots,k$,
\begin{align*}
&X_{n-j+1,n}-U_T\left(\frac{n+1}{k+1}\right) =_d U_T(Y_{n-j+1,n})-U_T\left(\frac{n+1}{k+1}\right)\\
&\quad =U_Y\left(\frac{1+b_{T,k,n}}{1+\frac1{Y_{n-j+1,n}D_T}}\frac{1}{\bar{F}_Y(T)(1+b_{T,k,n})}\right)-U_Y\left(\frac{1}{\bar{F}_Y(T)(1+b_{T,k,n})}\right)
\end{align*}
where we used \eqref{eq:UYTD}, and where $Y_{1,n}\leq Y_{2,n}\leq\ldots \leq Y_{n,n}$ denote the order statistics of an i.i.d.~sample from a standard Pareto distribution with distribution function $1-1/x$ for $x\geq1$. Hence, using \eqref{eq:thm342} with \[t=\frac1{\bar{F}_Y(T)(1+b_{T,k,n})} \quad \text{and} \quad x=\frac{1+b_{T,k,n}}{1+\frac{n+1}{jY_{n-j+1,n}} \frac{j}{k+1}b_{T,k,n}},\]
we obtain
\begin{align}\label{eq:propc}
\frac{X_{n-j+1,n}-U_T\left(\frac{n+1}{k+1}\right)}{a_{T,k,n}} &= \frac1{\xi}\left(\left(\frac{1+\frac{n+1}{jY_{n-j+1,n}}\frac{j}{k+1}b_{T,k,n}}{1+b_{T,k,n}}\right)^{-\xi}-1\right)\nonumber\\
	&+ A\left(\frac1{\bar{F}_Y(T)(1+b_{T,k,n})}\right)\Psi_{\xi,\rho}\left(\frac{1+b_{T,k,n}}{1+\frac{j}{k+1}b_{T,k,n}}\right)+o_p(1).
\end{align}
Using Lemma~2.4.10 in \citet{dHF} applied to the standard Pareto distribution one gets
\[\max_{j=1,\ldots,k} \left(\frac{j}{k+1}\right)^{0.5+\varepsilon} \Bigg|\sqrt{k} \left(Y_{n-j+1,n} \frac{j}{n}-1\right)-\left(\frac{j}{k+1}\right)^{-1} W_n\left(\frac{j}{k+1}\right)\Bigg|\to_p 0.\]
Using the mean value theorem we now obtain
\begin{align*}
&\frac1{\xi}\left(\left(\frac{1+\frac{n+1}{jY_{n-j+1,n}}\frac{j}{k+1}b_{T,k,n}}{1+b_{T,k,n}}\right)^{-\xi}-\left(\frac{1+\frac{j}{k+1}b_{T,k,n}}{1+b_{T,k,n}}\right)^{-\xi}\right)\nonumber\\
	&\quad=\frac{b_{T,k,n}}{1+b_{T,k,n}} \frac{j}{k+1}\left(\frac{1+\frac{j}{k+1}b_{T,k,n}}{1+b_{T,k,n}}\right)^{-1-\xi} \left(\frac{jY_{n-j+1,n}}n-1\right)(1+o_p(1)).
\end{align*}
Hence, combining this with \eqref{eq:propc} and the result from Lemma~2.4.10 in \citet{dHF}, we arrive at \ref{itm:propa}.
Combining \ref{itm:propa} with the analogous result for $j=k+1$, one arrives at \ref{itm:propb}. To this end note that $\Psi_{\xi,\rho}(1)=0$.
\end{proof}
\vspace{\baselineskip}
\begin{proof}[Proof of Theorem~\ref{thm:trMLE}]
This proof follows the approach of the proof of Theorem~3.4.2 in \citet{dHF}. Let $\hat{\tau}_k a_{T,k,n}=\hat{\tau}_k^s$, and
\begin{align*}
Z_{T,k,n}\left(\frac{j}{k+1}\right)&=\frac{b_{T,k,n}}{1+b_{T,k,n}} \left(\left(\frac{1+\frac{j}{k+1}b_{T,k,n}}{1+b_{T,k,n}}\right)^{-1-\xi}W_n\left(\frac{j}{k+1}\right)-W_n(1)\right)\\
&\quad+\sqrt{k}A\left(\frac1{\bar{F}_Y(T)(1+b_{T,k,n})}\right)\Psi_{\xi,\rho}\left(\frac{1+b_{T,k,n}}{1+\frac{j}{k+1}b_{T,k,n}}\right).
\end{align*}
Then, uniformly in $j\in\{1,\ldots,k\}$,
\begin{align*}
1+\hat{\tau}_k^s \frac{E_{j,k}}{a_{T,k,n}}&=\left(\frac{1+\frac{j}{k+1}b_{T,k,n}}{1+b_{T,k,n}}\right)^{-\xi}+\frac1{\xi}(\hat{\tau}_k^s-\xi)\left(\left(\frac{1+\frac{j}{k+1}b_{T,k,n}}{1+b_{T,k,n}}\right)^{-\xi}-1\right)\\
&\quad+\hat{\tau}_k^s \frac1{\sqrt{k}} Z_{T,k,n}\left(\frac{j}{k+1}\right)+o_p(1).
\end{align*}
Using $\log(1+u)= u(1+o(1))$ if $u\downarrow0$, we get
\begin{align*}
&\log\left(\left(\frac{1+\frac{j}{k+1}b_{T,k,n}}{1+b_{T,k,n}}\right)^{\xi}\left(1+\hat{\tau}_k^s \frac{E_{j,k}}{a_{T,k,n}}\right)\right)\\
&\quad=\frac1{\xi}(\hat{\tau}_k^s-\xi)\left(1-\left(\frac{1+\frac{j}{k+1}b_{T,k,n}}{1+b_{T,k,n}}\right)^{\xi}\right)+\hat{\tau}_k^s \frac1{\sqrt{k}} Z_{T,k,n}\left(\frac{j}{k+1}\right)\left(\frac{1+\frac{j}{k+1}b_{T,k,n}}{1+b_{T,k,n}}\right)^{\xi}+o_p(1).
\end{align*}
Hence, the first term on the left hand side of \eqref{eq:lik_xitau1} is given by
\begin{align}\label{eq:firstterm_lik_xitau1}
&\frac{1}{k-1}\sum_{j=2}^k \log (1+ \hat{\tau}_k E_{j,k}) = \left[-\xi \int_0^1 \log\left(\frac{1+ub_{T,k,n}}{1+b_{T,k,n}}\right)\,du \right.\nonumber\\
&\left.\quad +\frac1{\xi}(\hat{\tau}_k^s-\xi)\int_0^1 \left(1-\left(\frac{1+ub_{T,k,n}}{1+b_{T,k,n}}\right)^{\xi}\right)\,du+\hat{\tau}_k^s \frac1{\sqrt{k}} \int_0^1Z_{T,k,n}\left(u\right)\left(\frac{1+ub_{T,k,n}}{1+b_{T,k,n}}\right)^{\xi}\,du\right]\nonumber\\
&\sim \left[\xi \left(1-\frac{\log(1+b_{T,k,n})}{b_{T,k,n}}\right) +\frac1{\xi}(\hat{\tau}_k^s-\xi)\left(1-\frac{1+b_{T,k,n}}{b_{T,k,n}(1+\xi)}\left(1-(1+b_{T,k,n})^{-1-\xi}\right)\right)\right.\nonumber\\
&\left.\quad+\hat{\tau}_k^s \frac1{\sqrt{k}} \int_0^1Z_{T,k,n}\left(u\right)\left(\frac{1+ub_{T,k,n}}{1+b_{T,k,n}}\right)^{\xi}\,du\right].
\end{align}
Moreover, using Proposition~\ref{thm:prop_trMLE}\ref{itm:propb} with $j=1$, we obtain
\begin{align*}
&\left(1+\hat{\tau}_k^s \frac{E_{1,k}}{a_{T,k,n}}\right)^{-1/\hat{\xi}_k}=\left(1+\hat{\tau}_k^s \frac1{\xi}\left((1+b_{T,k,n})^{\xi}-1\right)+\hat{\tau}_k^s \frac1{\sqrt{k}} Z_{T,k,n}\left(\frac{1}{k+1}\right)\right)^{-1/\hat{\xi}_k}\\
&\quad=(1+b_{T,k,n})^{-\frac{\xi}{\hat{\xi_k}}}\left(1+(\hat{\tau}_k^s-\xi) \frac1{\xi}\left(1-(1+b_{T,k,n})^{-\xi}\right)\right.\\
&\qquad\qquad\left.+\hat{\tau}_k^s \frac1{\sqrt{k}} Z_{T,k,n}\left(\frac{1}{k+1}\right)(1+b_{T,k,n})^{-\xi}\right)^{-1/\hat{\xi}_k}\\
&\quad=(1+b_{T,k,n})^{-1}\left(1+(\hat{\xi}_k-\xi)\frac1{\xi}\log(1+b_{T,k,n})-(\hat{\tau}_k^s-\xi) \frac1{\xi^2}\left(1-(1+b_{T,k,n})^{-\xi}\right)\right.\\
&\left.\qquad\qquad -\frac{\hat{\tau}_k^s}{\hat{\xi}_k} \frac1{\sqrt{k}} Z_{T,k,n}\left(\frac{1}{k+1}\right)(1+b_{T,k,n})^{-\xi}\right)(1+o_p(1))
\end{align*}
where we used the series expansions \newline$e^{-\left(\frac{\xi}{\hat{\xi}_k}-1\right)\log(1+b_{T,n,k})}= 1-\left(\frac{\xi}{\hat{\xi}_k}-1\right)\log(1+b_{T,n,k})(1+o_p(1))$ and $(1+u)^{-1/\xi}= 1-\frac1{\xi}u(1+o(1))$. Hence, the second term on the left hand side of \eqref{eq:lik_xitau1} equals
\begin{align}\label{eq:secondterm_lik_xitau1}
&-\hat{\xi}_k\frac{\left(1+\hat{\tau}^s_k \frac{E_{1,k}}{a_{T,k,n}}\right)^{-1/\hat{\xi}_k} \log\left(1+\hat{\tau}_k^s  \frac{E_{1,k}}{a_{T,k,n}}\right)^{-1/\hat{\xi}_k}}{1- \left(1+\hat{\tau}_k^s  \frac{E_{1,k}}{a_{T,k,n}}\right)^{-1/\hat{\xi}_k}}\nonumber\\
&\quad=-\hat{\xi}_k(1+b_{T,k,n})^{-1}\left(1+\frac{(\hat{\xi}_k-\xi)}{\xi}\log(1+b_{T,k,n})-\frac{(\hat{\tau}_k^s-\xi)}{\xi^2}\left(1-(1+b_{T,k,n})^{-\xi}\right)\right.\nonumber\\
&\left.\qquad\qquad -\frac{\hat{\tau}_k^s}{\hat{\xi}_k} \frac1{\sqrt{k}} Z_{T,k,n}\left(\frac{1}{k+1}\right)(1+b_{T,k,n})^{-\xi}\right)\nonumber\\
&\qquad\times \log(1+b_{T,k,n})^{-1} \times \left(1-\frac{(\hat{\xi}_k-\xi)}{\xi}+\frac{(\hat{\tau}_k^s-\xi)}{\xi^2}\frac{\left(1-(1+b_{T,k,n})^{-\xi}\right)}{\log(1+b_{T,k,n})}\right.\nonumber\\
&\left.\qquad\qquad +\frac{\hat{\tau}_k^s}{\hat{\xi}_k} \frac1{\sqrt{k}} Z_{T,k,n}\left(\frac{1}{k+1}\right)\frac{(1+b_{T,k,n})^{-\xi}}{\log(1+b_{T,k,n})}\right)\nonumber\\
&\qquad/ \left[\frac{b_{T,k,n}}{1+b_{T,k,n}} \left(1-\frac{(\hat{\xi}_k-\xi)}{\xi}\frac{\log(1+b_{T,k,n})}{b_{T,k,n}}+ \frac{(\hat{\tau}_k^s-\xi)}{\xi^2}\frac{\left(1-(1+b_{T,k,n})^{-\xi}\right)}{b_{T,k,n}}\right.\right.\nonumber\\
&\left.\left.\qquad\qquad +\frac{\hat{\tau}_k^s}{\hat{\xi}_k} \frac1{\sqrt{k}} Z_{T,k,n}\left(\frac{1}{k+1}\right)\frac{(1+b_{T,k,n})^{-\xi}}{b_{T,k,n}}\right)\right](1+o_p(1))\nonumber\\
&\quad \sim \left[\hat{\xi}_k \frac{\log(1+b_{T,k,n})}{b_{T,k,n}}+(\hat{\xi}_k-\xi) \frac{\log(1+b_{T,k,n})}{b_{T,k,n}}\left(-1+\frac{1+b_{T,k,n}}{b_{T,k,n}}\log(1+b_{T,k,n})\right)\right.\nonumber\\
&\qquad\left. -\frac{(\hat{\tau}_k^s-\xi)}{\xi}\left(1-(1+b_{T,k,n})^{-\xi}\right)\left(\frac{1+b_{T,k,n}}{b_{T,k,n}}-\frac1{\log(1+b_{T,k,n})}\right)\frac{\log(1+b_{T,k,n})}{b_{T,k,n}}\right.\nonumber\\
&\qquad\left. -\hat{\tau}_k^s \frac1{\sqrt{k}} Z_{T,k,n}\left(\frac{1}{k+1}\right)(1+b_{T,k,n})^{-\xi}\left(\frac{1+b_{T,k,n}}{b_{T,k,n}}-\frac1{\log(1+b_{T,k,n})}\right)\frac{\log(1+b_{T,k,n})}{b_{T,k,n}}\right].
\end{align}
Combining \eqref{eq:lik_xitau1}, \eqref{eq:firstterm_lik_xitau1} and \eqref{eq:secondterm_lik_xitau1} gives
\begin{align*}
&\left[\xi \left(1-\frac{\log(1+b_{T,k,n})}{b_{T,k,n}}\right) +\frac1{\xi}(\hat{\tau}_k^s-\xi)\left(1-\frac{1+b_{T,k,n}}{b_{T,k,n}(1+\xi)}\left(1-(1+b_{T,k,n})^{-1-\xi}\right)\right)\right.\nonumber\\
&\left.\quad+\hat{\tau}_k^s \frac1{\sqrt{k}} \int_0^1Z_{T,k,n}\left(u\right)\left(\frac{1+ub_{T,k,n}}{1+b_{T,k,n}}\right)^{\xi}\,du\right](1+o_p(1))\\
&\quad+\left[\hat{\xi}_k \frac{\log(1+b_{T,k,n})}{b_{T,k,n}}+(\hat{\xi}_k-\xi) \frac{\log(1+b_{T,k,n})}{b_{T,k,n}}\left(-1+\frac{1+b_{T,k,n}}{b_{T,k,n}}\log(1+b_{T,k,n})\right)\right.\nonumber\\
&\qquad\left. -\frac{(\hat{\tau}_k^s-\xi)}{\xi}\left(1-(1+b_{T,k,n})^{-\xi}\right)\left(\frac{1+b_{T,k,n}}{b_{T,k,n}}-\frac1{\log(1+b_{T,k,n})}\right)\frac{\log(1+b_{T,k,n})}{b_{T,k,n}}\right.\nonumber\\
&\qquad\left. - \hat{\tau}_k^s \frac1{\sqrt{k}} Z_{T,k,n}\left(\frac{1}{k+1}\right)(1+b_{T,k,n})^{-\xi}\left(\frac{1+b_{T,k,n}}{b_{T,k,n}}-\frac1{\log(1+b_{T,k,n})}\right)\frac{\log(1+b_{T,k,n})}{b_{T,k,n}}\right]\\
&=\hat{\xi}_k (1+o_p(1)).
\end{align*}
This equation can be written as
\begin{align}\label{eq:final1}
&\left[(\hat{\xi}_k-\xi) \left(-1+\frac{(1+b_{T,k,n})\log^2(1+b_{T,k,n})}{b_{T,k,n}^{2}}\right)\right.\nonumber\\
& \left.+\frac1{\xi}(\hat{\tau}_k^s-\xi)\left(\frac{\xi}{1+\xi} \frac{1+b_{T,k,n}}{b_{T,k,n}}\left(1-(1+b_{T,k,n})^{-1-\xi}\right)\right.\right.\nonumber\\
&\left.\left.\qquad-\frac{(1+b_{T,k,n})}{b_{T,k,n}^2}\log(1+b_{T,k,n})\left(1-(1+b_{T,k,n})^{-\xi}\right)\right)\right.\nonumber\\
&\left. +  \frac{\hat{\tau}_k^s}{\sqrt{k}} \int_0^1Z_{T,k,n}\left(u\right)\left(\frac{1+ub_{T,k,n}}{1+b_{T,k,n}}\right)^{\xi}\,du\right.\nonumber\\
&\qquad\left. -\frac{\hat{\tau}_k^s}{\sqrt{k}} Z_{T,k,n}\left(\frac{1}{k+1}\right)(1+b_{T,k,n})^{-\xi}\left(\frac{(1+b_{T,k,n})\log(1+b_{T,k,n})}{b_{T,k,n}^2}-\frac1{b_{T,k,n}}\right)\right](1+o_p(1))\nonumber\\
& =0.
\end{align}
The left hand side of \eqref{eq:lik_xitau2} yields, using similar asymptotic methods as above,
\begin{align}\label{eq:firstterm_lik_xitau2}
&\frac1{k-1}\sum_{j=2}^k \left(\frac{1+\frac{j}{k+1}b_{T,k,n}}{1+b_{T,k,n}}\right)^{\xi}\left[1-\frac{(\hat{\tau}_k^s-\xi)}{\xi}\left(1-\left(\frac{1+\frac{j}{k+1}b_{T,k,n}}{1+b_{T,k,n}}\right)^{\xi}\right)\right.\nonumber\\
&\quad\left.-\frac{\hat{\tau}_k^s}{\sqrt{k}} Z_{T,k,n}\left(\frac{j}{k+1}\right)\left(1-\left(\frac{1+\frac{j}{k+1}b_{T,k,n}}{1+b_{T,k,n}}\right)^{\xi}\right)\right](1+o_p(1))\nonumber\\
&=\left[\frac1{1+\xi} \frac{1+b_{T,k,n}}{b_{T,k,n}}\left(1-(1+b_{T,k,n})^{-1-\xi}\right)\right.\nonumber\\
&\qquad\left.-\frac{(\hat{\tau}_k^s-\xi)}{\xi}\left(\frac{\xi(1+b_{T,k,n})}{(1+\xi)(1+2\xi)}-\frac{(1+b_{T,k,n})^{-\xi}}{1+\xi}+\frac{(1+b_{T,k,n})^{-2\xi}}{1+2\xi}\right)\right.\nonumber\\
&\qquad\left. -\frac{\hat{\tau}_k^s}{\sqrt{k}} \int_0^1 Z_{T,k,n}\left(u\right)\left(\frac{1+ub_{T,k,n}}{1+b_{T,k,n}}\right)^{2\xi}\, du\right](1+o_p(1)).
\end{align}

The right hand side of \eqref{eq:lik_xitau2} is asymptotically equivalent to (where we used again Proposition~\ref{thm:prop_trMLE}\ref{itm:propb} with $j=1$)
\newpage
\begin{align}\label{eq:secondterm_lik_xitau2}
&\frac1{1+\hat{\xi}_k}\Bigg[1-\frac1{1+b_{T,k,n}}\left(1+(\hat{\xi}_k-\xi)\frac1{\xi}\log(1+b_{T,k,n})-\frac1{\xi^2}(\hat{\tau}_k^s-\xi)\left(1-(1+b_{T,k,n})^{-\xi}\right)\right.\nonumber\\
&\qquad\qquad\left. -\hat{\tau}_k^s \frac1{\hat{\xi}_k\sqrt{k}}Z_{T,k,n}\left(\frac1{k+1}\right)(1+b_{T,k,n})^{-\xi}\right)\nonumber\\
&\qquad\times \left((1+b_{T,k,n})^{\xi}+\frac1{\xi}(\hat{\tau}_k^s-\xi)\left((1+b_{T,k,n})^{\xi}-1\right)+\hat{\tau}_k^s \frac1{\sqrt{k}}Z_{T,k,n}\left(\frac1{k+1}\right)\right)^{-1} \Bigg]\nonumber\\
&\quad\times \left[1-\frac1{1+b_{T,k,n}}\left(1+(\hat{\xi}_k-\xi)\frac1{\xi}\log(1+b_{T,k,n})-\frac1{\xi^2}(\hat{\tau}_k^s-\xi)\left(1-(1+b_{T,k,n})^{-\xi}\right)\right.\right.\nonumber\\
&\qquad\left.\left. -\hat{\tau}_k^s \frac1{\hat{\xi}_k\sqrt{k}}Z_{T,k,n}\left(\frac1{k+1}\right)(1+b_{T,k,n})^{-\xi}\right)\right]^{-1}\nonumber\\
&\sim \frac1{1+\hat{\xi}_k} \frac{(1+b_{T,k,n})\left(1-(1+b_{T,k,n})^{-1-\xi}\right)}{b_{T,k,n}}\left(1-(\hat{\xi}_k-\xi)\frac1{\xi}\log(1+b_{T,k,n})\frac{(1+b_{T,k,n})^{-1-\xi}}{1-(1+b_{T,k,n})^{-1-\xi}}\right.\nonumber\\
&\qquad\qquad\left. +(\hat{\tau}_k^s-\xi)\frac{1+\xi}{\xi^2}\frac{(1+b_{T,k,n})^{-1-\xi}}{1-(1+b_{T,k,n})^{-1-\xi}}\left(1-(1+b_{T,k,n})^{-\xi}\right)\right.\nonumber\\
&\qquad\qquad\left. +\frac{\hat{\tau}_k^s}{\hat{\xi}_k} \frac1{\sqrt{k}}Z_{T,k,n}\left(\frac1{k+1}\right) \frac{1+\hat{\xi}_k}{\hat{\xi}_k}\frac{(1+b_{T,k,n})^{-1-2\xi}}{1-(1+b_{T,k,n})^{-1-\xi}}\right)\nonumber\\
&\quad\times \Bigg[1-(\hat{\xi}_k-\xi)\frac1{\xi}\frac{\log(1+b_{T,k,n})}{b_{T,k,n}}+(\hat{\tau}_k^s-\xi)\frac{1}{\xi^2}\frac{1-(1+b_{T,k,n})^{-\xi}}{b_{T,k,n}}\nonumber\\
&\qquad\qquad+\frac{\hat{\tau}_k^s}{\hat{\xi}_k} \frac1{\sqrt{k}}Z_{T,k,n}\left(\frac1{k+1}\right) \frac{(1+b_{T,k,n})^{-\xi}}{b_{T,k,n}}\Bigg]^{-1}\nonumber\\
&\sim \frac1{1+\hat{\xi}_k} \frac{(1+b_{T,k,n})\left(1-(1+b_{T,k,n})^{-1-\xi}\right)}{b_{T,k,n}} + \frac{(\hat{\xi}_k-\xi)}{\xi(1+\xi)}\frac{1+b_{T,k,n}}{b_{T,k,n}^2}\log(1+b_{T,k,n})\left(1-(1+b_{T,k,n})^{-\xi}\right)\nonumber\\
&\qquad+(\hat{\tau}_k^s-\xi)\frac{1}{\xi^2(1+\xi)}\frac{1+b_{T,k,n}}{b_{T,k,n}}\left(1-(1+b_{T,k,n})^{-\xi}\right) \Bigg(-\frac{1-(1+b_{T,k,n})^{-1-\xi}}{b_{T,k,n}}\nonumber\\
& \qquad \qquad +(1+\xi)(1+b_{T,k,n})^{-1-\xi} \Bigg)\nonumber\\
&\qquad-\frac1{1+\xi}\frac{\hat{\tau}^s_k}{\hat{\xi}_k}\frac1{\sqrt{k}}Z_{T,k,n}\left(\frac1{k+1}\right) \frac{(1+b_{T,k,n})^{1-\xi}}{b_{T,k,n}^2}\left(1-(1+b_{T,k,n})^{-\xi}\right).
\end{align}
Combining \eqref{eq:lik_xitau2}, \eqref{eq:firstterm_lik_xitau2} and \eqref{eq:secondterm_lik_xitau2} leads to (after some lengthy calculations)
\begin{align}\label{eq:final2}
&(\hat{\xi}_k-\xi)\frac1{\xi} \frac{1+b_{T,k,n}}{b_{T,k,n}} \left(\frac{\xi}{1+\xi}\left(1-(1+b_{T,k,n})^{-1-\xi}\right)-\frac{\log(1+b_{T,k,n})}{b_{T,k,n}}\left(1-(1+b_{T,k,n})^{-\xi}\right)\right)\nonumber\\
&-(\hat{\tau}_k^s-\xi)\frac{1+b_{T,k,n}}{b_{T,k,n}} \frac1{\xi}\left(\frac{\xi}{1+2\xi}\left(1-(1+b_{T,k,n})^{-1-2\xi}\right)-\frac1{b_{T,k,n}}\frac1{\xi}\left(1-(1+b_{T,k,n})^{-\xi}\right)^2\right)\nonumber\\
&= \frac{\xi(\xi+1)}{\sqrt{k}} \int_0^1 Z_{T,k,n}\left(u\right)\left(\frac{1+ub_{T,k,n}}{1+b_{T,k,n}}\right)^{2\xi}\, du\nonumber\\
&\qquad-\frac1{\sqrt{k}}Z_{T,k,n}\left(\frac1{k+1}\right) \frac{(1+b_{T,k,n})^{1-\xi}}{b_{T,k,n}^2}\left(1-(1+b_{T,k,n})^{-\xi}\right).
\end{align}
\end{proof}

\begin{proof}[Proof of Theorem~\ref{thm:Q_T}]
\begin{align*}
\hat{Q}_{T,k}&(1-p)\\
&=X_{n-k,n}+\frac1{\hat{\tau}_k} \left(\left(1+\frac{k}{n\hat{D}_T}\right)^{\hat{\xi}_k}\left(1+\frac1{d_n}\frac{k}{n\hat{D}_T}\right)^{-\hat{\xi}_k}-1\right)\\
&=X_{n-k,n}+\frac1{\hat{\tau}_k} \left(\left(\frac{1-\frac1k}{(1+\hat{\tau}_kE_{1,k})^{-\frac1{\hat{\xi}_k}}-\frac1k}\right)^{\hat{\xi}_k}\left(1+\frac1{d_n}\frac{k}{n\hat{D}_T}\right)^{-\hat{\xi}_k}-1\right)\\
&=X_{n-k,n}+\frac1{\hat{\tau}_k} \left((1+\hat{\tau}_kE_{1,k})\left(\frac{1-\frac1k}{1-\frac1k(1+\hat{\tau}_kE_{1,k})^{\frac1{\hat{\xi}_k}}}\right)^{\hat{\xi}_k}\left(1+\frac1{d_n}\frac{k}{n\hat{D}_T}\right)^{-\hat{\xi}_k}-1\right)\\
&=X_{n-k,n}+\frac1{\hat{\tau}_k} \left((1+\hat{\tau}_kE_{1,k})\left(1-\frac{\hat{\xi}_k}k\left(1-(1+\hat{\tau}_kE_{1,k})^{\frac1{\hat{\xi}_k}}\right)(1+o_p(1))\right)\right.\\
&\qquad\left. \times \left(1-\frac{\hat{\xi}_k}{d_n}\frac{k}{n\hat{D}_T}(1+o_p(1))\right)-1\right)\\
&=X_{n-k,n}+\left(E_{1,k}+(1+\hat{\tau}_kE_{1,k})\left(-\frac{\hat{\xi}_k}{\hat{\tau}_kk}\left(1-(1+\hat{\tau}_kE_{1,k})^{\frac1{\hat{\xi}_k}}\right) -\frac{\hat{\xi}_k}{d_n\hat{\tau}_k}\frac{k}{n\hat{D}_T}\right)(1+o_p(1))\right)\\
&=X_{n,n}-\frac{\hat{\xi}_k}{\hat{\tau}_k}(1+\hat{\tau}_kE_{1,k})\left(\frac1k\left(1-(1+\hat{\tau}_kE_{1,k})^{\frac1{\hat{\xi}_k}}\right) +\frac1{d_n}\frac{k}{n\hat{D}_T}\right)(1+o_p(1)).
\end{align*}
Hence, 
\begin{align}\label{eq:Q_T1}
& \hspace{-0.2cm}
\hat{Q}_{T,k}(1-p)-Q_{T}(1-p) \nonumber \\
&=\left(X_{n,n}-Q_{T}\left(1-\frac1n\right)\right)+\left(Q_{T}\left(1-\frac1n\right)-Q_{T}(1-p)\right)\nonumber\\
&\quad -\frac{\hat{\xi}_k}{\hat{\tau}_k}(1+\hat{\tau}_kE_{1,k})\left(\frac1k\left(1-(1+\hat{\tau}_kE_{1,k})^{\frac1{\hat{\xi}_k}}\right) +\frac1{d_n}\frac{k}{n\hat{D}_T}\right)\left(1+o_p\left(\frac1{d_n}\right)\right).
\end{align}
First, using again the notation $Y_{1,n}\leq Y_{2,n}\leq\ldots \leq Y_{n,n}$ for the order statistics of an i.i.d.~sample of size $n$ from a standard Pareto distribution, we obtain using \eqref{eq:dHF}
\begin{align}\label{eq:Q_T2}
X_{n,n}&-Q_{T}\left(1-\frac1n\right)=_dU_T(Y_{n,n})-U_T(n)\nonumber\\
&=U_Y\left(\frac1{\bar{F}_Y(T)\left(1+\frac{n}{Y_{n,n}}\frac1{nD_T}\right)}\right)-U_Y\left(\frac1{\bar{F}_Y(T)\left(1+\frac1{nD_T}\right)}\right)\nonumber\\
&=U_Y\left(\frac{1+\frac1{nD_T}}{1+\frac{n}{Y_{n,n}}\frac1{nD_T}}\frac1{\bar{F}_Y(T)\left(1+\frac1{nD_T}\right)}\right)-U_Y\left(\frac1{\bar{F}_Y(T)\left(1+\frac1{nD_T}\right)}\right)\nonumber\\
&=a_Y\left(\frac1{\bar{F}_Y(T)\left(1+\frac1kb_{T,k,n}\right)}\right)\left(\frac1{\xi}\left(\left(\frac{1+\frac{b_{T,k,n}}k}{1+\frac{n}{Y_{n,n}}\frac{b_{T,k,n}}k}\right)^{\xi}-1\right)\right.\nonumber\\
&\quad+\left.A\left(\frac1{\bar{F}_Y(T)\left(1+\frac{1}{k}b_{T,k,n}\right)}\right)\Psi_{\xi,\rho}\left(\frac{1+\frac{b_{T,k,n}}{k}}{1+\frac{n}{Y_{n,n}}\frac{b_{T,k,n}}k}\right)\right)(1+o_p(1))\nonumber\\
&=a_Y\left(\frac1{\bar{F}_Y(T)}\right)\left(1+\frac{b_{T,k,n}}k\right)^{-\xi}\left(1+A\left(\frac1{\bar{F}_Y(T)}\right)C\left(\frac{\left(1+\frac{b_{T,k,n}}k\right)^{-\rho}-1}{\rho}\right)\right)\nonumber\\
&\qquad \times \left(-b_{T,k,n}\frac1k\left(\frac{n}{Y_{n,n}}-1\right)\left(1+O_p\left(\frac1k\right)\right)+O_p\left(\frac1{k^2}\right)\right)\nonumber\\
&=a_Y\left(\frac1{\bar{F}_Y(T)}\right)\left(1-\frac{\xi b_{T,k,n}}k-A\left(\frac1{\bar{F}_Y(T)}\right)C\frac{b_{T,k,n}}k+ O_p\left(\frac1{k^2}\right)\right)\nonumber\\
&\qquad \times \left(-\frac{b_{T,k,n}}k(E-1)+O_p\left(\frac1{k^2}\right)\right).
\end{align}
Here, we used that $\frac{n}{Y_{n,n}}=_d E+O_p\left(\frac1n\right)$ and that  $\Psi_{\xi,\rho}\left(1+\frac{D}k\right)=O\left(\frac1{k^2}\right)$ for any constant $D$.
Furthermore,
\begin{align}\label{eq:Q_T3}
Q&_{T}\left(1-\frac1n\right)-Q_{T}\left(1-p\right)\nonumber\\
&=U_Y\left(\frac1{\bar{F}_Y(T)\left(1+\frac1{nD_T}\right)}\right)-U_Y\left(\frac1{\bar{F}_Y(T)\left(1+\frac{p}{D_T}\right)}\right)\nonumber\\
&=U_Y\left(\frac{1+\frac{b_{T,k,n}}{d_n}}{1+\frac{b_{T,k,n}}{k}}\frac1{\bar{F}_Y(T)\left(1+\frac{b_{T,k,n}}{d_n}\right)}\right)-U_Y\left(\frac1{\bar{F}_Y(T)\left(1+\frac{p}{D_T}\right)}\right)\nonumber\\
&=a_Y\left(\frac1{\bar{F}_Y(T)\left(1+\frac{b_{T,k,n}}{d_n}\right)}\right)\left(\frac1{\xi}\left(\left(\frac{1+\frac{b_{T,k,n}}{d_n}}{1+\frac{b_{T,k,n}}{k}}\right)^{\xi}-1\right)\right.\nonumber\\
&\quad+\left.A\left(\frac1{\bar{F}_Y(T)\left(1+\frac{b_{T,k,n}}{d_n}\right)}\right)\Psi_{\xi,\rho}\left(\frac{1+\frac{b_{T,k,n}}{d_n}}{1+\frac{b_{T,k,n}}k}\right)\right)(1+o_p(1))\nonumber\\
&=a_Y\left(\frac1{\bar{F}_Y(T)}\right)\left(1+\frac{b_{T,k,n}}{d_n}\right)^{-\xi}\left(1+A\left(\frac1{\bar{F}_Y(T)}\right)C\left(\frac{\left(1+\frac{b_{T,k,n}}{d_n}\right)^{-\rho}-1}{\rho}\right)\right)\nonumber\\
&\qquad \times \left(b_{T,k,n}\left(\frac1{d_n}-\frac1k\right)\left(1+O\left(\frac1{d_n}\right)\right)+O\left(\frac1{d_n^2}\right)\right)\nonumber\\
&=a_Y\left(\frac1{\bar{F}_Y(T)}\right)\left(b_{T,k,n}\left(\frac1{d_n}-\frac1k\right)+O\left(\frac1{d_n^2}\vee\frac1{k^2}\right)\right).
\end{align}
Finally, using 
$k/(n\hat{D}_T) = \left((1+\hat{\tau}_kE_{1,k})^{\frac1{\hat{\xi}_k}}-1\right)(1+O_p(1/k))$ and 
derivations as in the proof of Theorem~\ref{thm:trMLE}, the third term in the right hand side of \eqref{eq:Q_T1} equals
\newpage
\begin{align}\label{eq:Q_T4}
&-\left(\frac{\hat{\xi}_k}{\hat{\tau}_k}\frac1{a_{T,k,n}}\right)a_{T,k,n}(1+\hat{\tau}_kE_{1,k})\left(\frac1k\left(1-(1+\hat{\tau}_kE_{1,k})^{\frac1{\hat{\xi}_k}}\right)+\frac1{d_n}\frac{k}{n\hat{D}_T}\right) \nonumber\\
&=-a_Y\left(\frac1{\bar{F}_Y(T)}\right) (1+b_{T,k,n})^{-\xi} \left(1+A\left(\frac1{\bar{F}_Y(T)}\right)C\left(\frac{\left(1+b_{T,k,n}\right)^{-\rho}-1}{\rho}\right)\right)\nonumber\\
&\quad\times \left(1+\left(\frac{\hat{\xi}_k}{\hat{\tau}_k}\frac1{a_{T,k,n}}-1\right)\right)\nonumber\\
&\quad\times (1+b_{T,k,n})^{\xi} \left(1+(\hat{\tau}_k^s-\xi)\frac1{\xi}\left(1-(1+b_{T,k,n})^{-\xi}\right)\right.\nonumber\\
&\qquad\qquad\qquad\qquad+\left.\frac{\hat{\tau}_k^s}{\sqrt{k}}Z_{T,k,n}\left(\frac1{k+1}\right)(1+b_{T,k,n})^{-\xi}\right)\nonumber\\
&\quad\times  \left(\left(1-(1+\hat{\tau}_kE_{1,k})^{\frac1{\hat{\xi}_k}}\right)\left(\frac1{d_n}-\frac1k\right) +O_p\left(\frac{1}{d_n k}\right) \right)\nonumber\\
&=-a_Y\left(\frac1{\bar{F}_Y(T)}\right) \left(1+A\left(\frac1{\bar{F}_Y(T)}\right)C\left(\frac{\left(1+b_{T,k,n}\right)^{-\rho}-1}{\rho}\right)\right)\left(1+O_p\left(\frac1k\right)\right)\nonumber\\
&\quad\times \left(1+\left(\frac{\hat{\xi}_k}{\hat{\tau}_k}\frac1{a_{T,k,n}}-1\right)\right)\nonumber\\
&\quad\times \left(1+(\hat{\tau}_k^s-\xi)\frac1{\xi}\left(1-(1+b_{T,k,n})^{-\xi}\right)+\frac{\hat{\tau}_k^s}{\sqrt{k}}Z_{T,k,n}\left(\frac1{k+1}\right)(1+b_{T,k,n})^{-\xi}\right)\nonumber\\
&\quad \times b_{T,k,n} \left(\frac1{d_n}-\frac1k\right)\nonumber \\
&\quad\times \left(1-(\hat{\xi}_k-\xi)\frac1{\xi} \frac{1+b_{T,k,n}}{b_{T,k,n}}\log(1+b_{T,k,n})+(\hat{\tau}_k^s-\xi)\frac1{\xi^2}\frac{1+b_{T,k,n}}{b_{T,k,n}}\left(1-(1+b_{T,k,n})^{-\xi}\right)\right.\nonumber\\
&\hspace{1.3cm}\left.+\frac{1}{\sqrt{k}}Z_{T,k,n}\left(\frac1{k+1}\right)\frac{(1+b_{T,k,n})^{1-\xi}}{b_{T,k,n}}\right).
\end{align}
The result follows from joining \eqref{eq:Q_T1}, \eqref{eq:Q_T2}, \eqref{eq:Q_T3} and \eqref{eq:Q_T4} and retaining terms of order $O\left(\frac1k\right)$, $O\left(\left(\frac1{d_n}-\frac1k\right)A\left(\frac1{\bar{F}_Y(T)}\right)\right)$ and $O\left(\left(\frac1{d_n}-\frac1k\right)\frac1{\sqrt{k}}\right)$.
\end{proof}

\newpage
\begin{proof}[Proof of Theorem~\ref{thm:trTest_mle}]
Note that using \eqref{eq:UYCD} and $\bar{F}_Y(T) = D_T F_Y(T)$, we obtain
\begin{align*}
T_{k,n}&= k\left( 1+ \hat{\tau}_k^s \frac{E_{1,k}}{a_{T,k,n}}\right)^{-1/\hat{\xi}_k} \\
&= k\left( 1+ \hat{\tau}_k^s \frac{U_T(Y_{n,n}) -U_T(Y_{n-k,n})}{a_{T,k,n}} \right)^{-1/\hat{\xi}_k} \\
&= k\left( 1+ \frac{\hat{\tau}_k^s}{a_Y \left(\frac{1}{D_T(1+ b_{T,k,n})F_Y(T)} \right)} \right. \\
& \qquad \left.\times \left[
U_Y\left( \frac{Y_{n,n}}{F_Y(T)(1+Y_{n,n} D_T) }\right) -  
U_Y\left(\frac{Y_{n-k,n}}{F_Y(T)(1+Y_{n-k,n} D_T)} \right)
\right]
\right)^{-1/\hat{\xi}_k} \\
&= k\left( 1+ \frac{\hat{\tau}_k^s}{a_Y \left(\frac{n/k}{(1 +nD_T/k)F_Y(T)}\right)} \right. \\
& \left.\qquad \times 
\left[
U_Y\left( \frac{\frac{Y_{n,n}}{Y_{n-k,n}}}{\frac{1+Y_{n,n}D_T}{1+Y_{n-k,n}D_T}}
\frac{\frac{kY_{n-k,n}}{n} \frac{n}{k}}{F_Y(T)\left(1+\frac{kY_{n-k,n}}{n}\frac{nD_T}{k}\right)}\right) \right. \right.\\
& \quad\qquad \left.\left. - 
U_Y\left(  \frac{\frac{kY_{n-k,n}}{n} \frac{n}{k}}{F_Y(T)\left(1+\frac{kY_{n-k,n}}{n}\frac{nD_T}{k}\right)}\right) \right]
\right)^{-1/\hat{\xi}_k}.
\end{align*}
Now one applies \eqref{eq:thm342} with $t=\frac{\frac{kY_{n-k,n}}{n} \frac{n}{k}}{F_Y(T)\left(1+\frac{kY_{n-k,n}}{n}\frac{nD_T}{k}\right)} = \frac{n}{k}(1+o_p(1))$  and $x=\frac{Y_{n,n}}{Y_{n-k,n}}\frac{1+Y_{n-k,n}D_T}{1+Y_{n,n}D_T}= U_{1,k}^{-1} (1+o_p(1))$ since
$\frac{kY_{n-k,n}}{n} = 1 + O_p(1/\sqrt{k})$, $Y_{n,n}/n = 1+o_p(1)$, $nD_T \to 0$ and 
$Y_{n-k,n}/Y_{n,n}=_d  U_{1,k}$, the minimum of an i.i.d.~sample of size $k$ from the uniform (0,1) distribution. This, with $\hat{\tau}^s_k/\xi = 1+o_p(1)$, yields
\begin{align*}
T_{k,n} &=  k \left( 1+ \frac{\hat{\tau}_k^s}{\xi} 
 \left[ 
   U_{1,k}^{-\xi}(1+o_p(1))-1
%\right. \right.   
%    \\&& \hspace{3cm} \left. \left.
   +\xi A\left(\frac{n}{k}(1+o_p(1))\right)\Psi_{\xi,\rho} \left(U_{1,k}^{-1}(1+o_p(1))\right)
 \right]
\right)^{-1/\hat{\xi}_k} \\
&= k\left(  U_{1,k}^{-\xi} (1+o_p(1)) +\xi A\left(\frac{n}{k}(1+o_p(1))\right)\Psi_{\xi,\rho} \left(U_{1,k}^{-1}(1+o_p(1))\right)\right)^{-(1/\xi)(1+O_p(1/\sqrt{k}))} \\
&= 
kU_{1,k}\left(1+o_p(1) + \xi U_{1,k}^{\xi} A\left(\frac{n}{k}(1+o_p(1))\right)\Psi_{\xi,\rho} (k (1+o_p(1)))\right)^{-1/\xi}\end{align*}
because $U_{1,k}^{-1} =O_p(k)$ and $U_{1,k}^{\xi}\Psi_{\xi,\rho} (k (1+o_p(1)))=O_p(1)$. The result now follows from $kU_{1,k}=_d E (1+o_p(1))$. 
\end{proof}

\section*{Appendix: simulation results}

\begin{landscape}
   \begin{figure}[!ht]
		\centering
		   \includegraphics[height=6.25cm, angle=270]{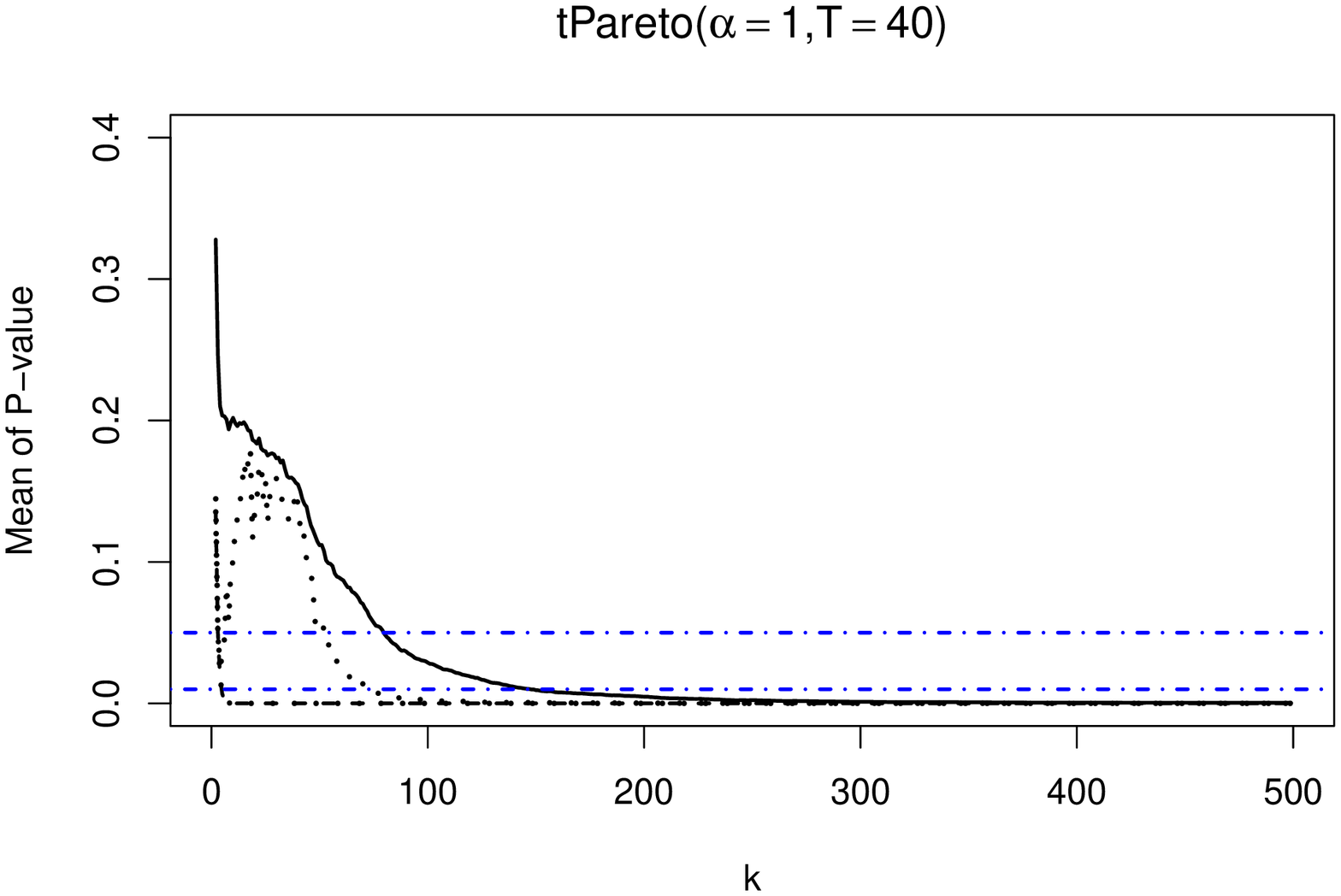}
	   \includegraphics[height=6.25cm, angle=270]{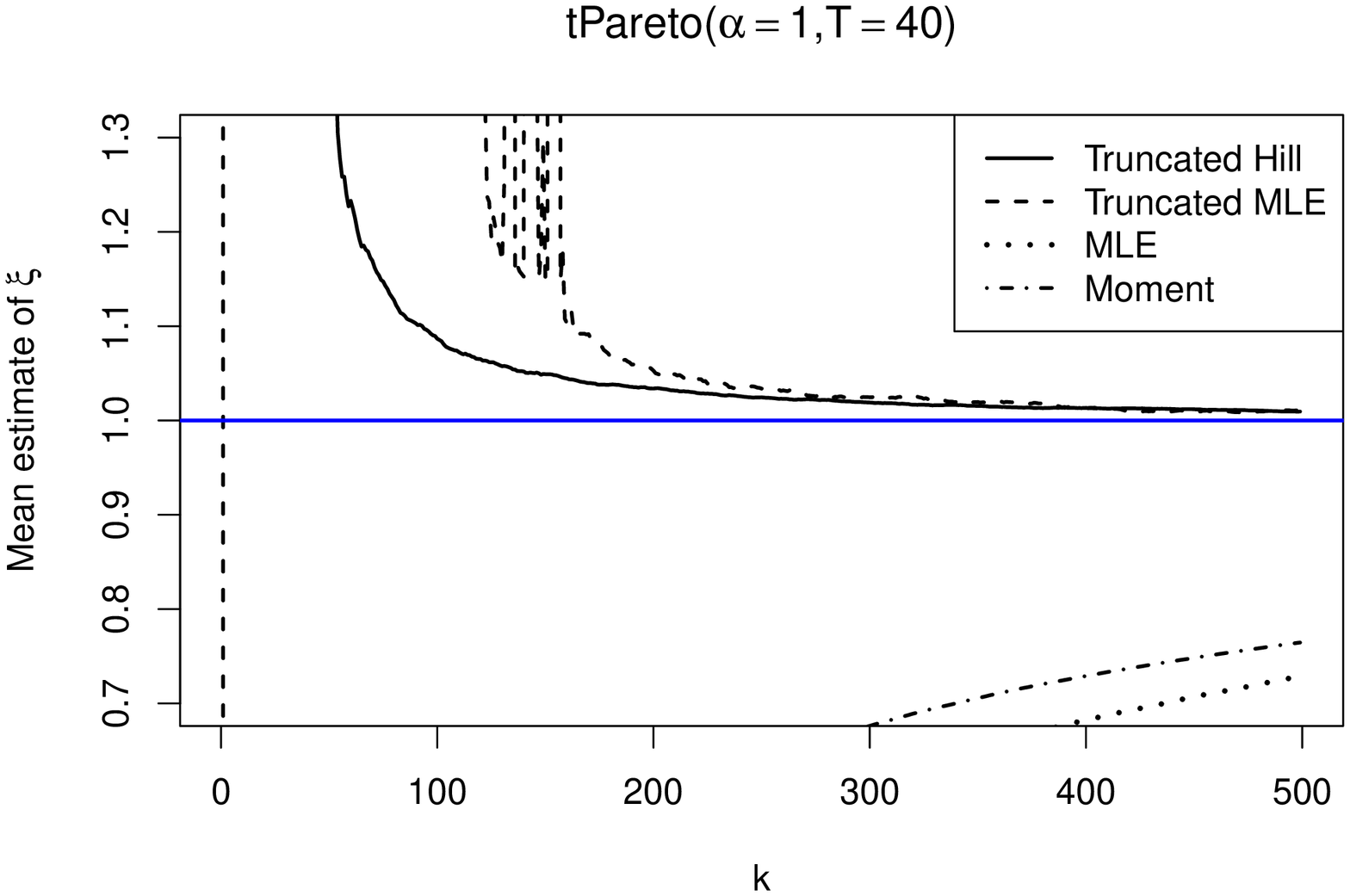}
  \includegraphics[height=6.25cm, angle=270]{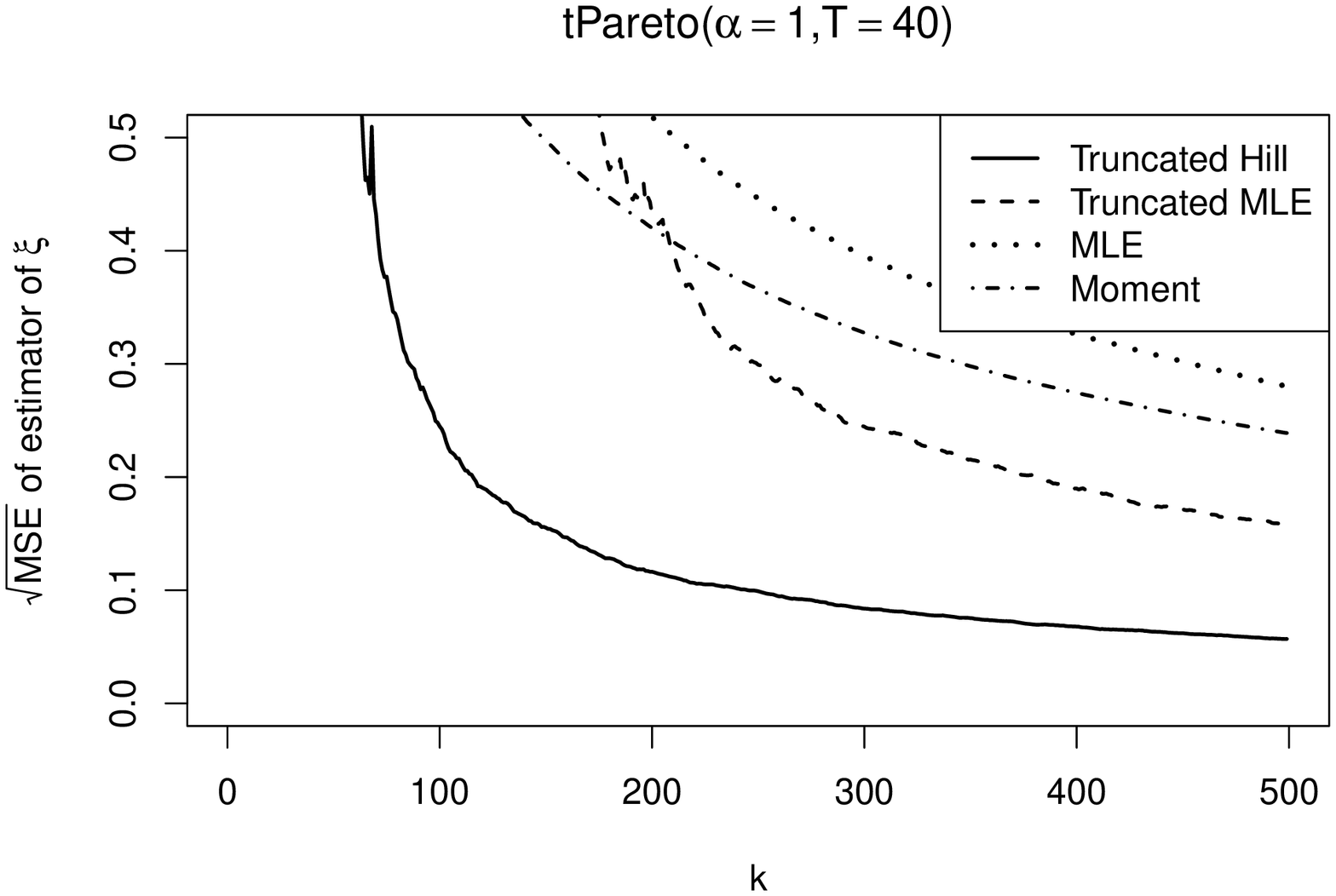}\\   
	\includegraphics[height=6.25cm, angle=270]{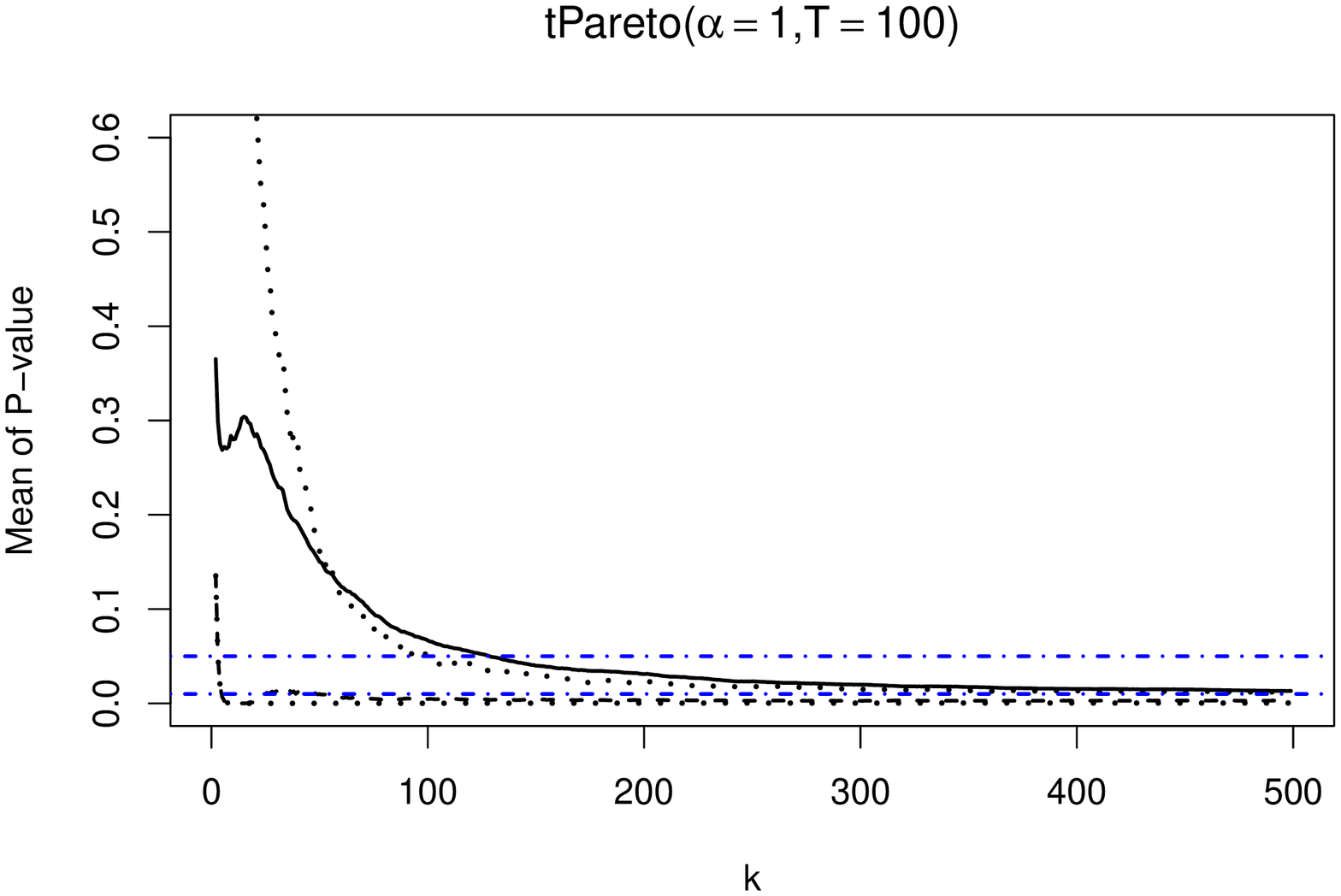}
   \includegraphics[height=6.25cm, angle=270]{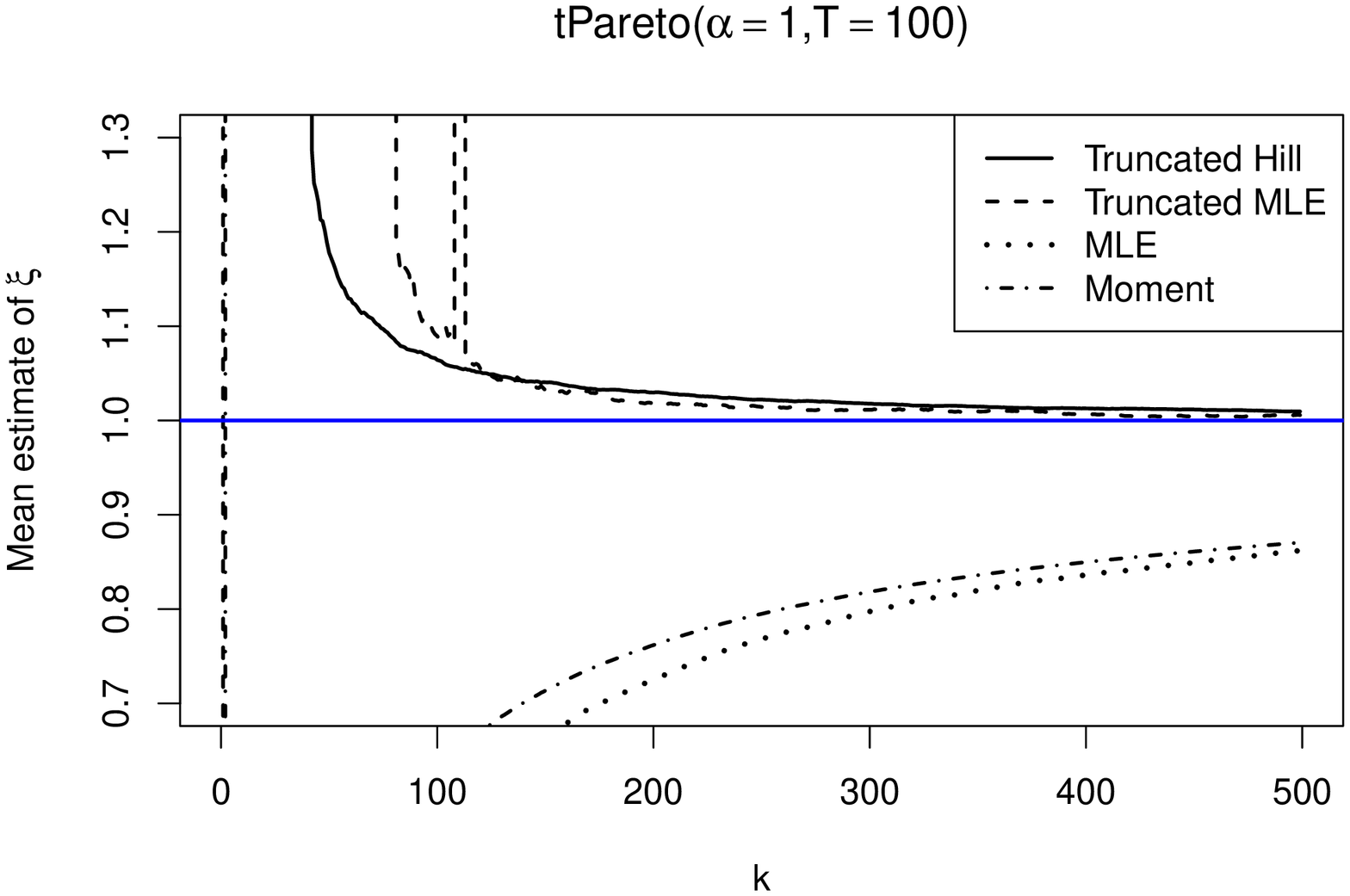}
  \includegraphics[height=6.25cm, angle=270]{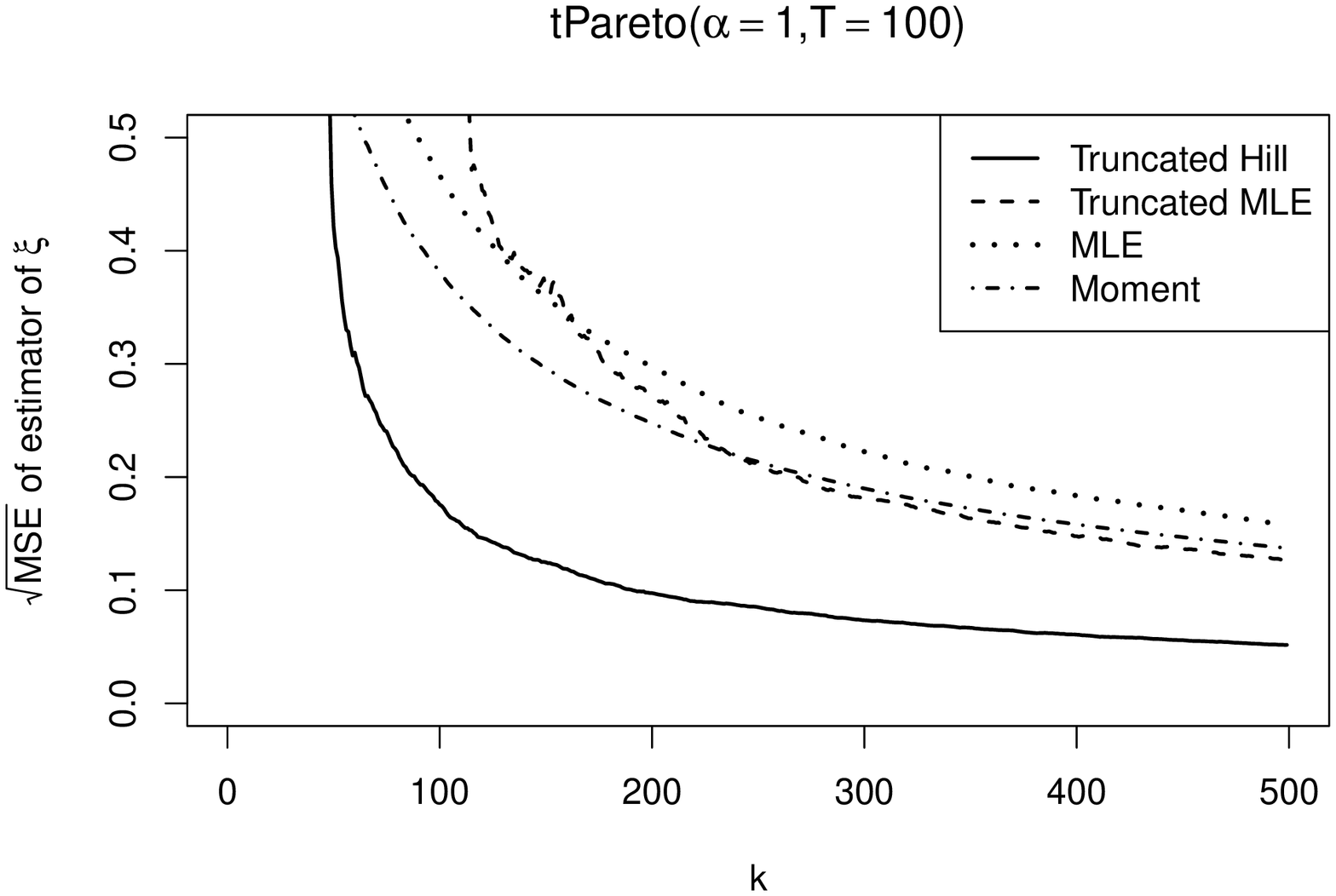}\\
	  \includegraphics[height=6.25cm, angle=270]{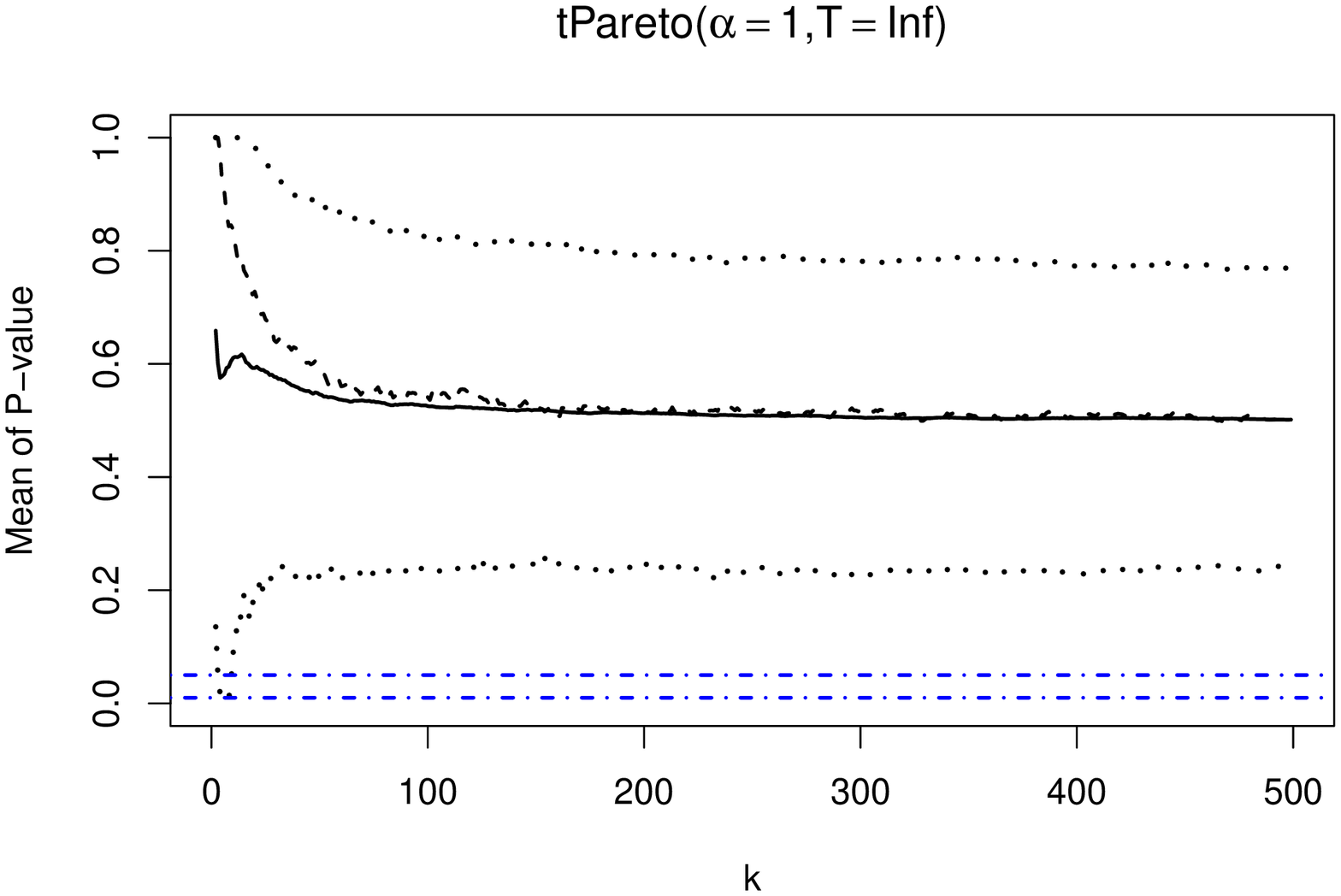}
  \includegraphics[height=6.25cm, angle=270]{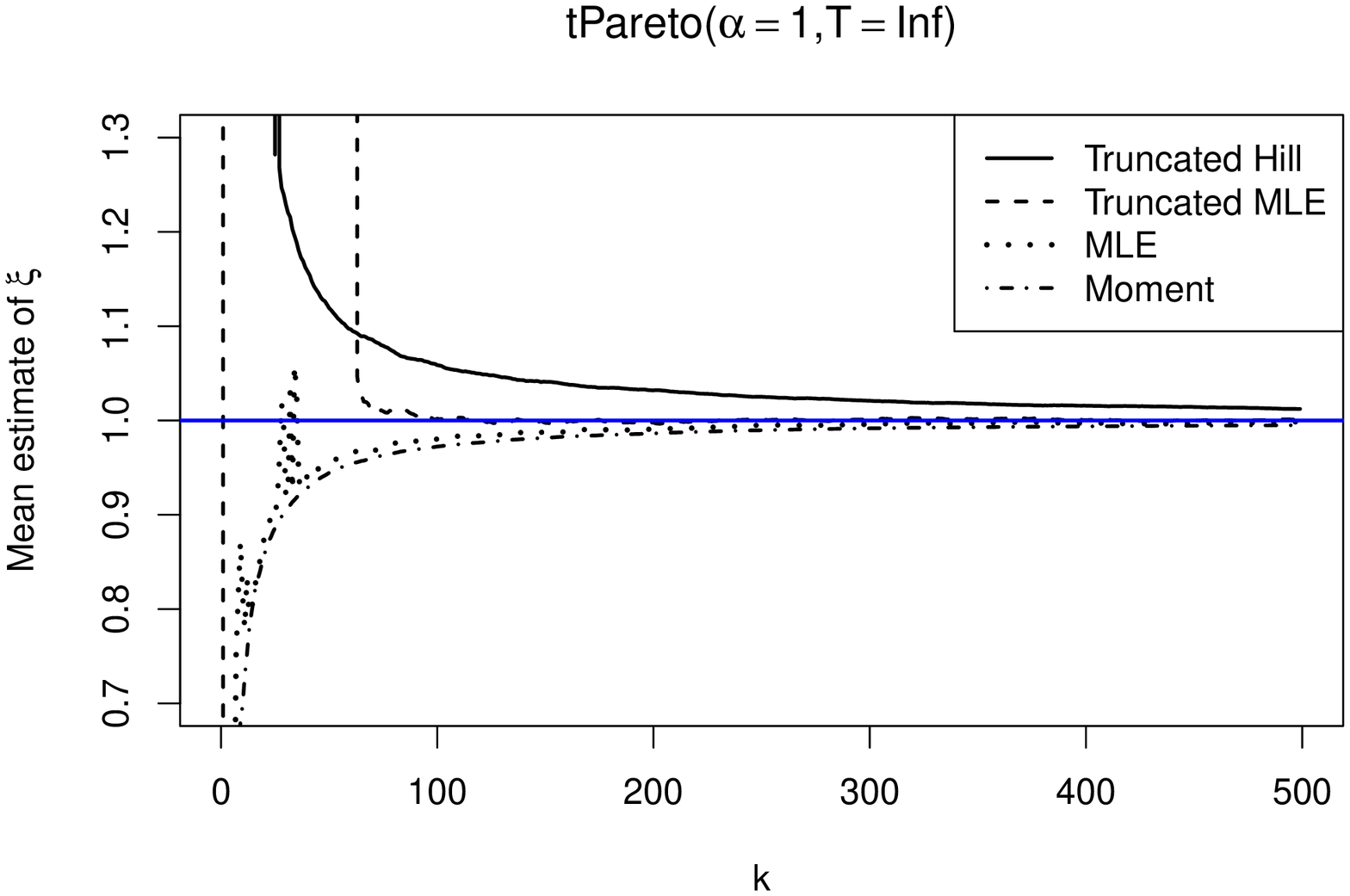}
  \includegraphics[height=6.25cm, angle=270]{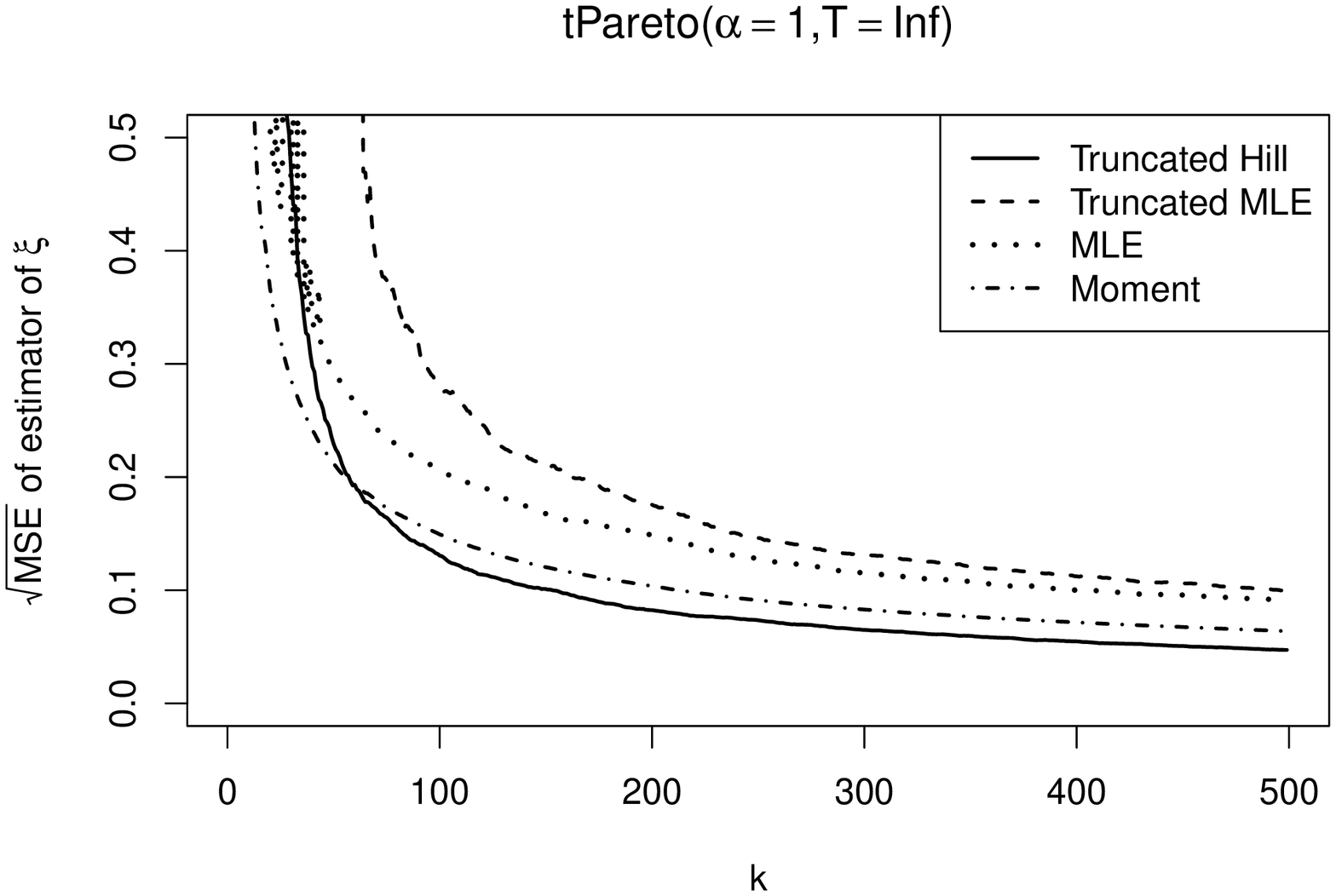}
  \caption{Means and boxplots of P-values for test (left), means (middle) and root MSE (right) of $\hat{\xi}^+_k$, $\hat{\xi}_k$, $\hat{\xi}^{\infty}_k$ and $\hat{\xi}^M_k$  from the standard Pareto distribution truncated at $Q_Y (0.975)$ (top), $Q_Y (0.99)$ (middle) and non truncated (bottom).}\label{fig:sim_xi_first}
  \end{figure}
	\end{landscape}

\begin{landscape}
   \begin{figure}[!ht]
		\centering
		   \includegraphics[height=6.25cm, angle=270]{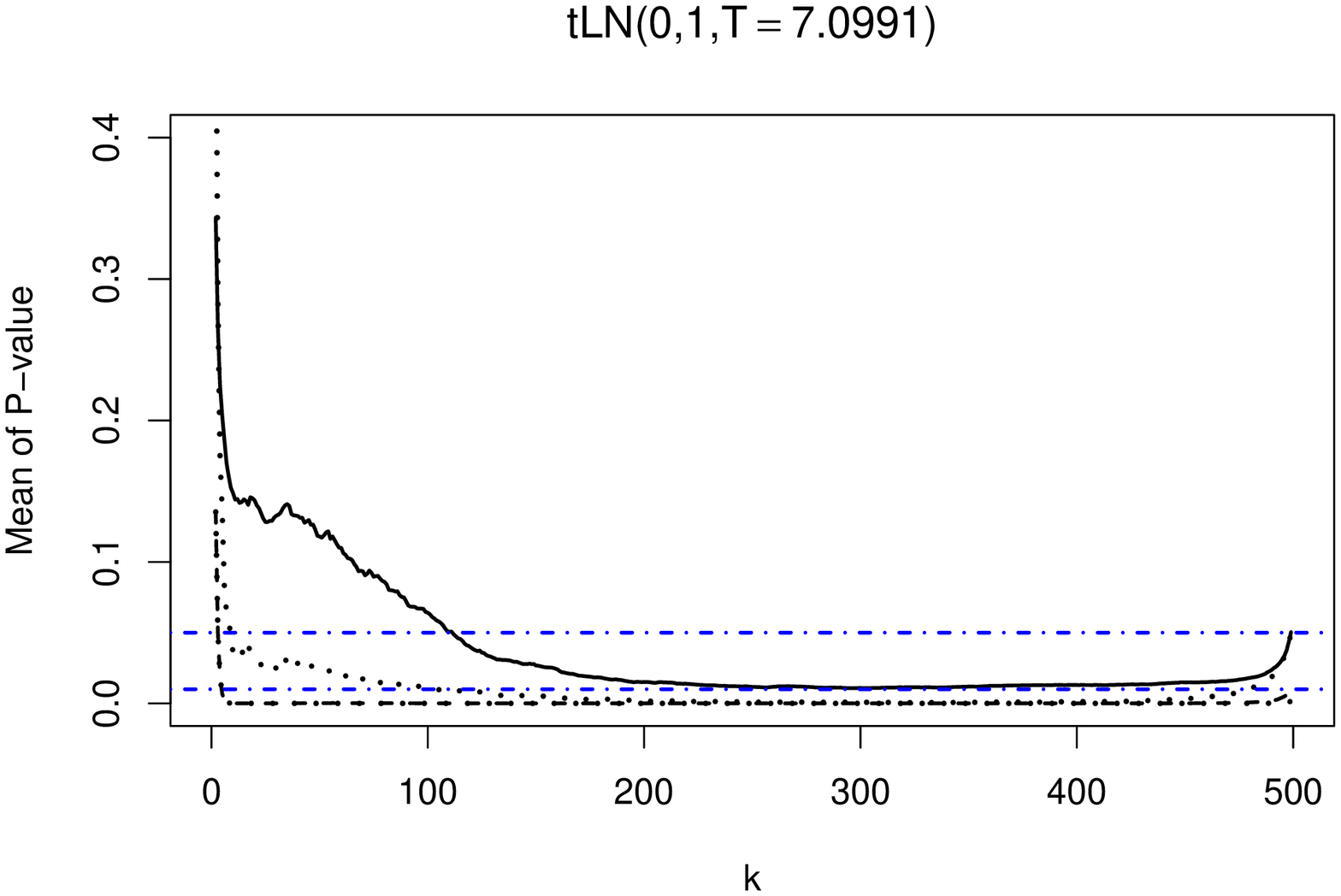}
	   \includegraphics[height=6.25cm, angle=270]{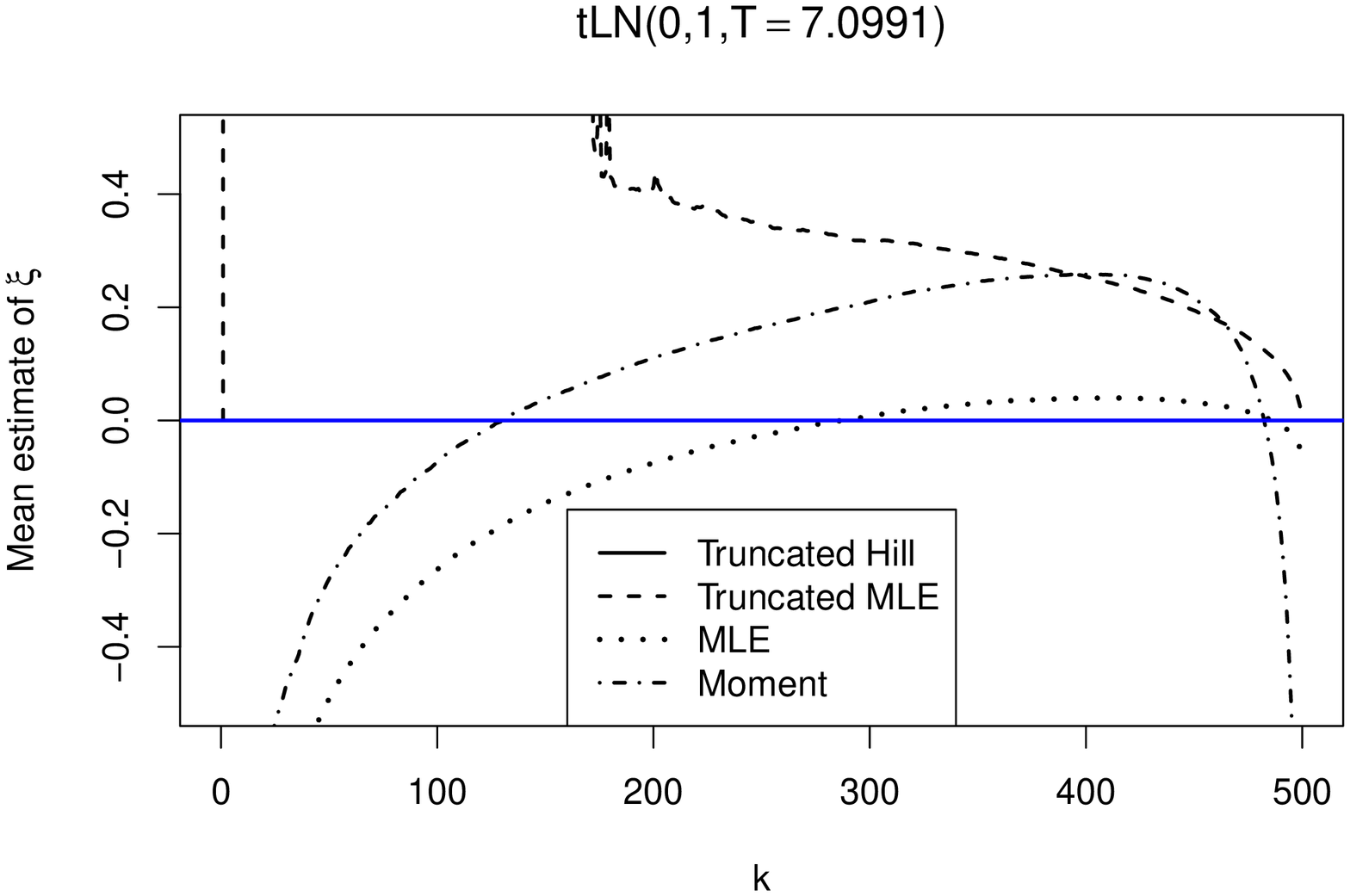}
  \includegraphics[height=6.25cm, angle=270]{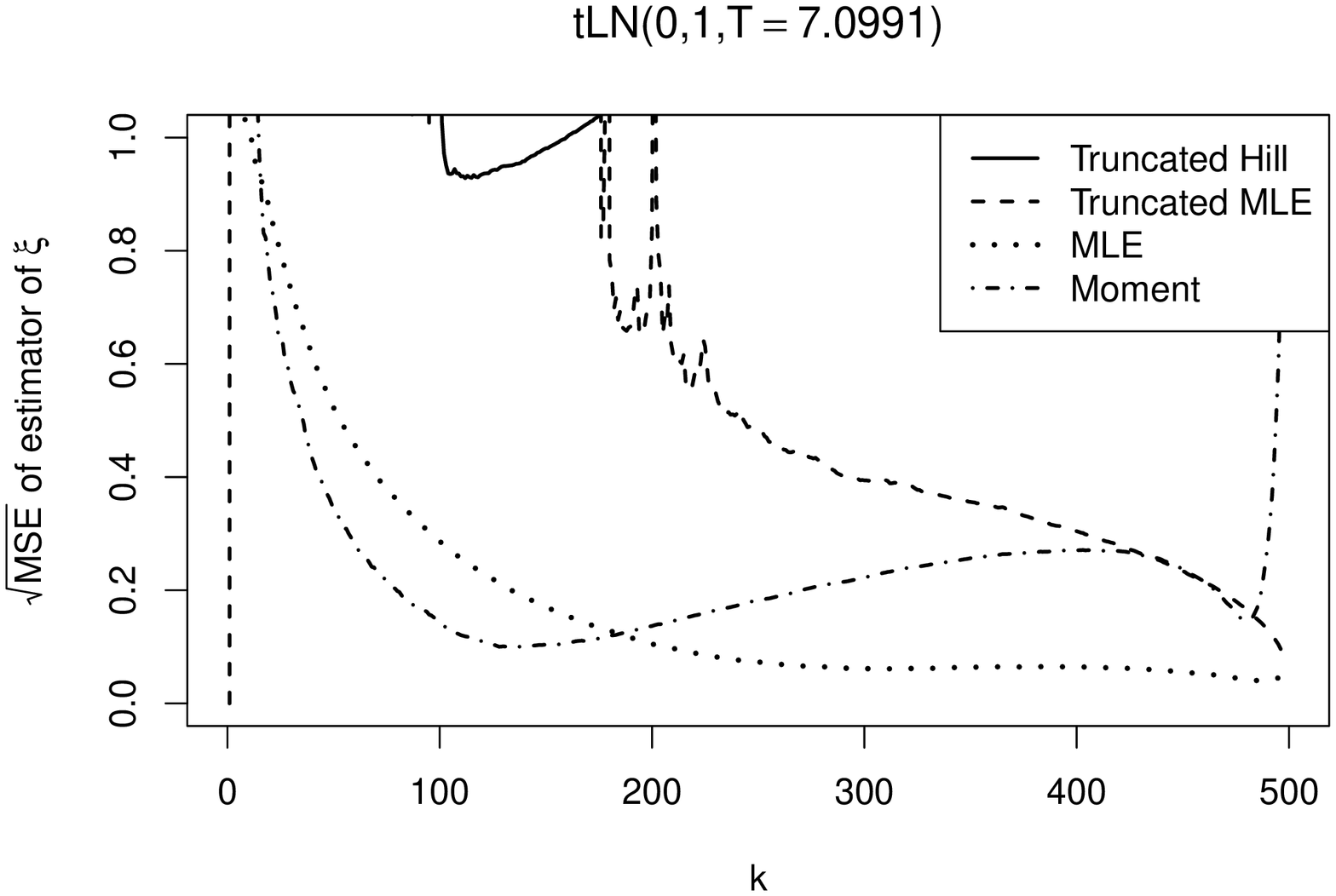}\\   
	\includegraphics[height=6.25cm, angle=270]{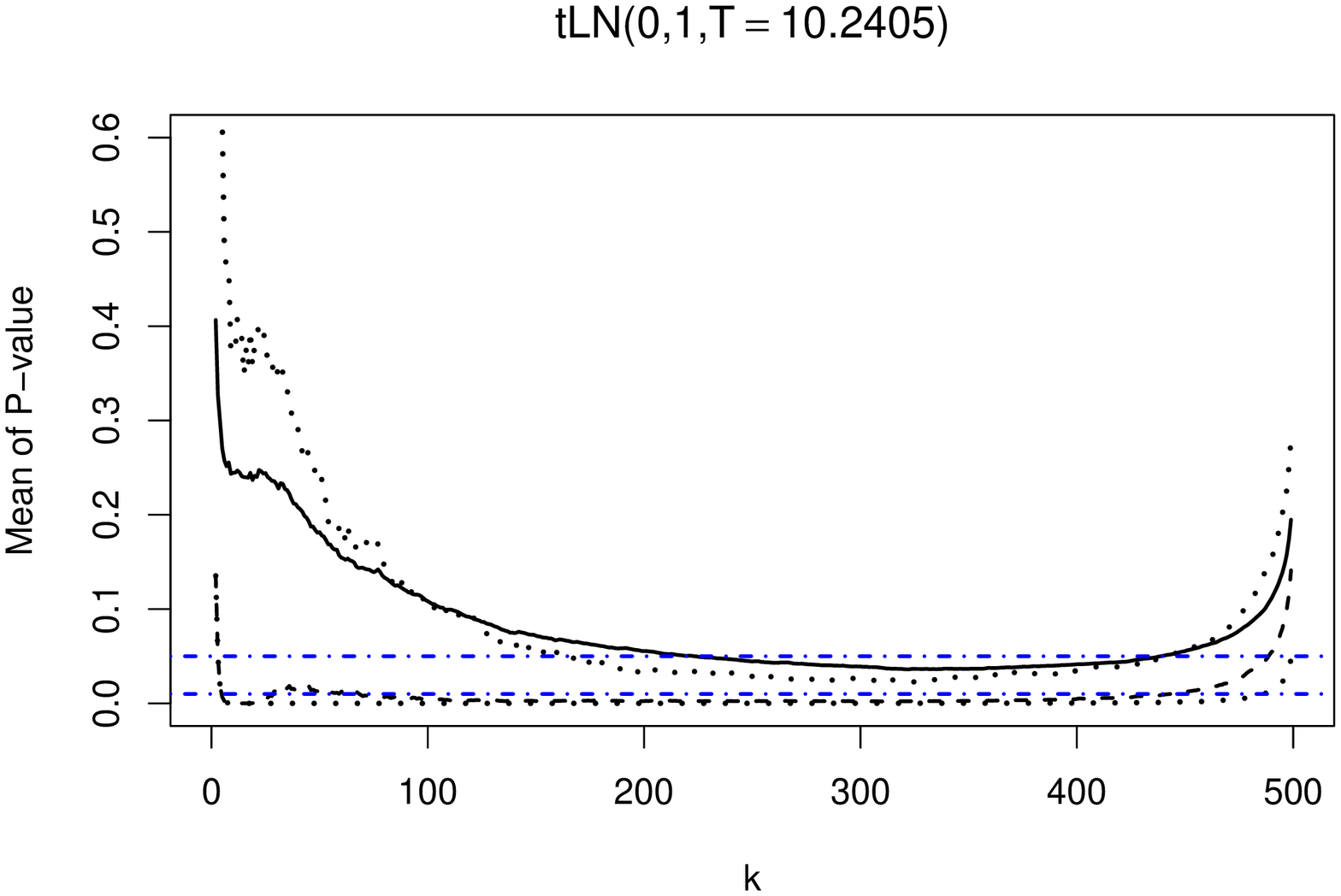}
   \includegraphics[height=6.25cm, angle=270]{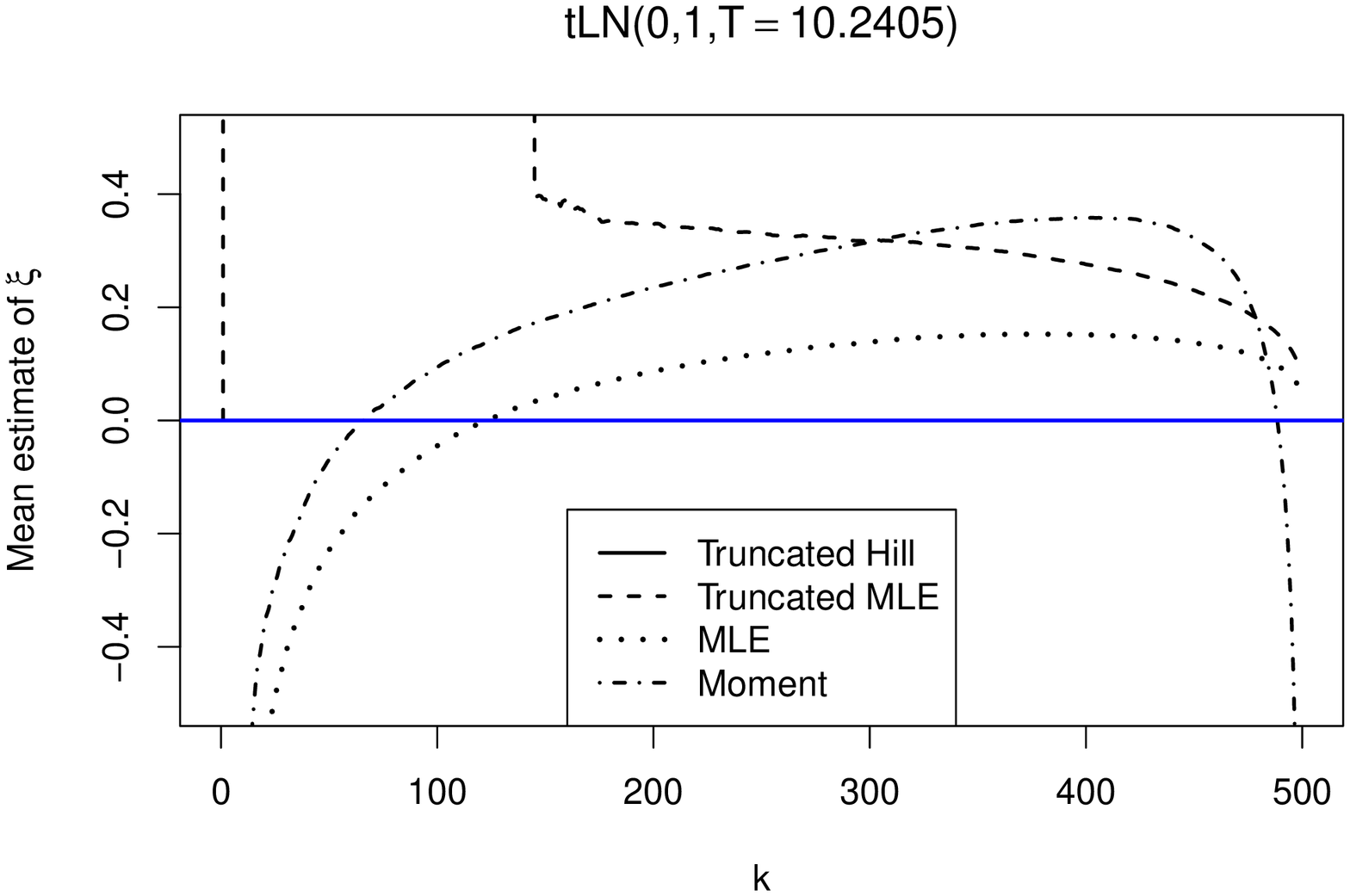}
  \includegraphics[height=6.25cm, angle=270]{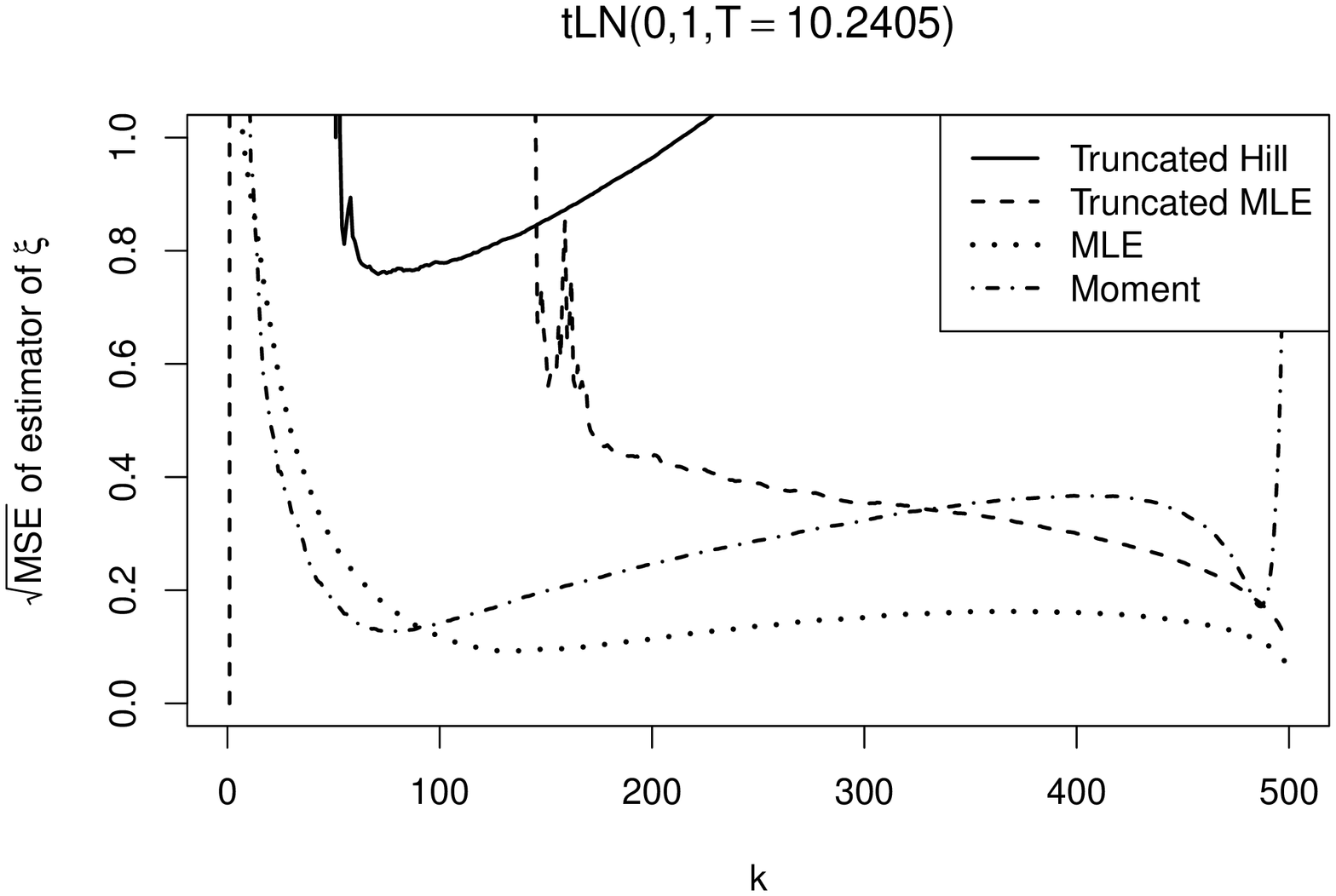}\\
	  \includegraphics[height=6.25cm, angle=270]{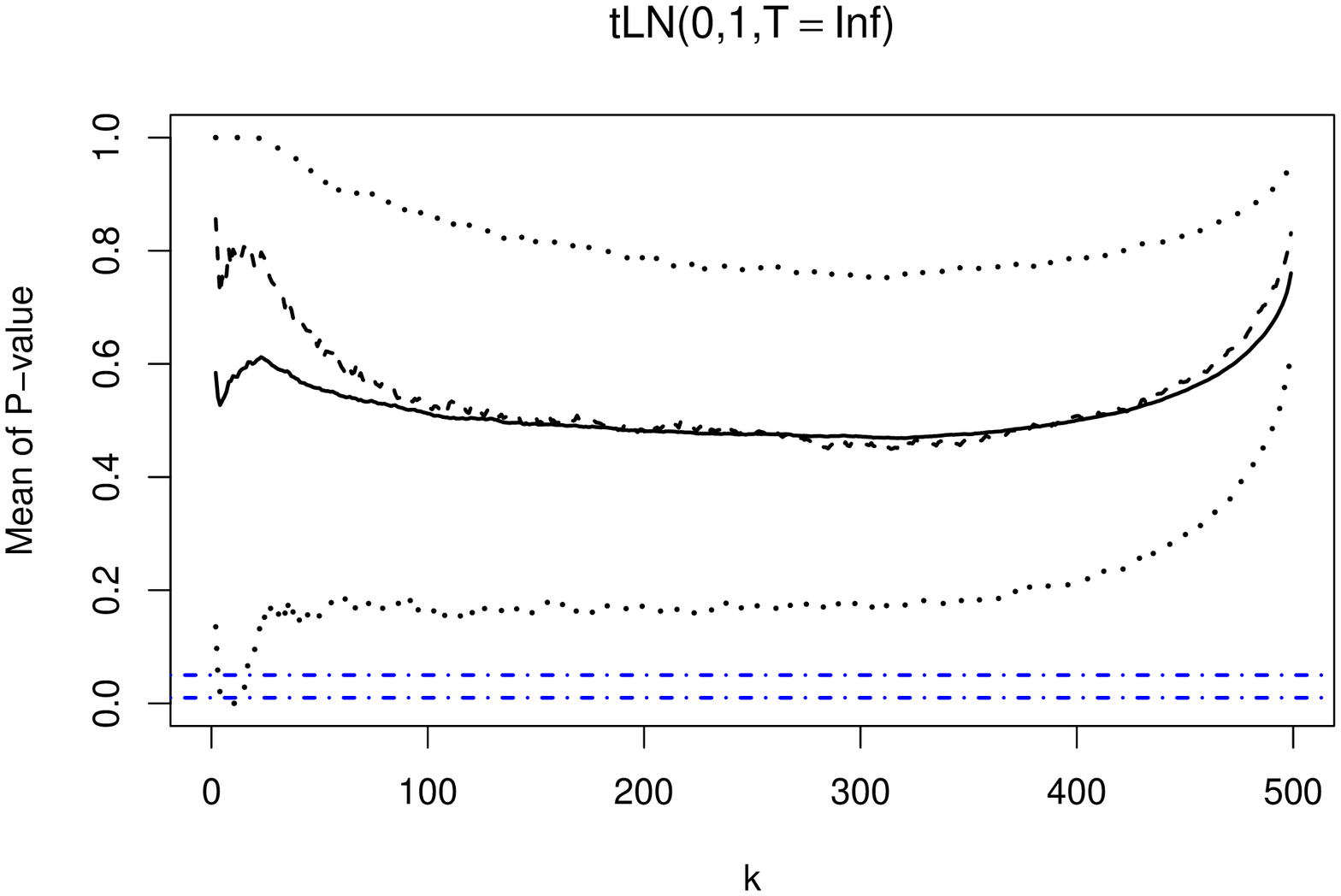}
  \includegraphics[height=6.25cm, angle=270]{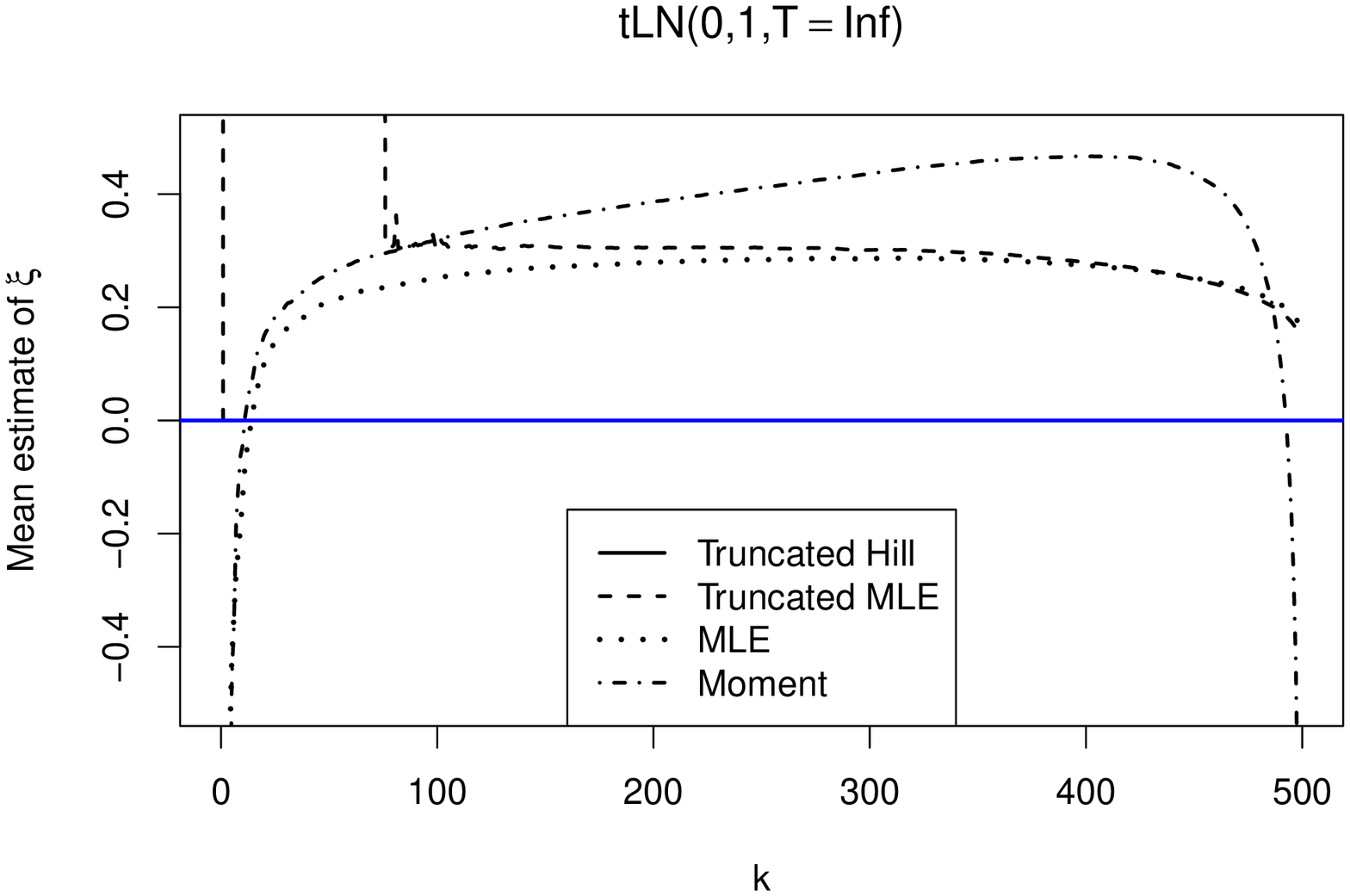}
  \includegraphics[height=6.25cm, angle=270]{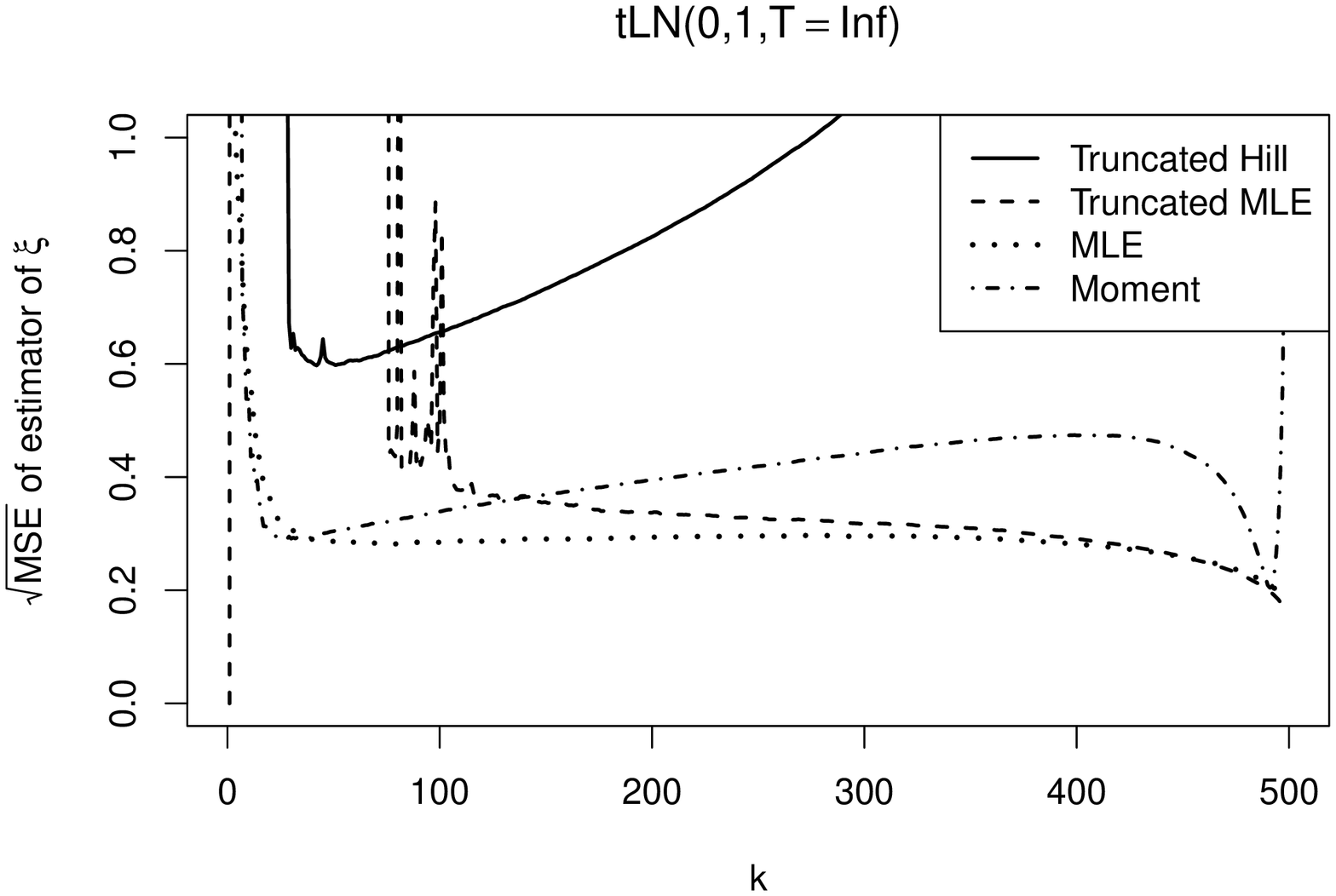}
  \caption{Means and boxplots of P-values for test (left), means (middle) and root MSE (right) of $\hat{\xi}^+_k$, $\hat{\xi}_k$, $\hat{\xi}^{\infty}_k$ and $\hat{\xi}^M_k$  from the standard lognormal distribution truncated at $Q_Y (0.975)$ (top), $Q_Y (0.99)$ (middle) and non truncated (bottom).}
  \end{figure}
	\end{landscape}

\begin{landscape}
   \begin{figure}[!ht]
		\centering
		   \includegraphics[height=6.25cm, angle=270]{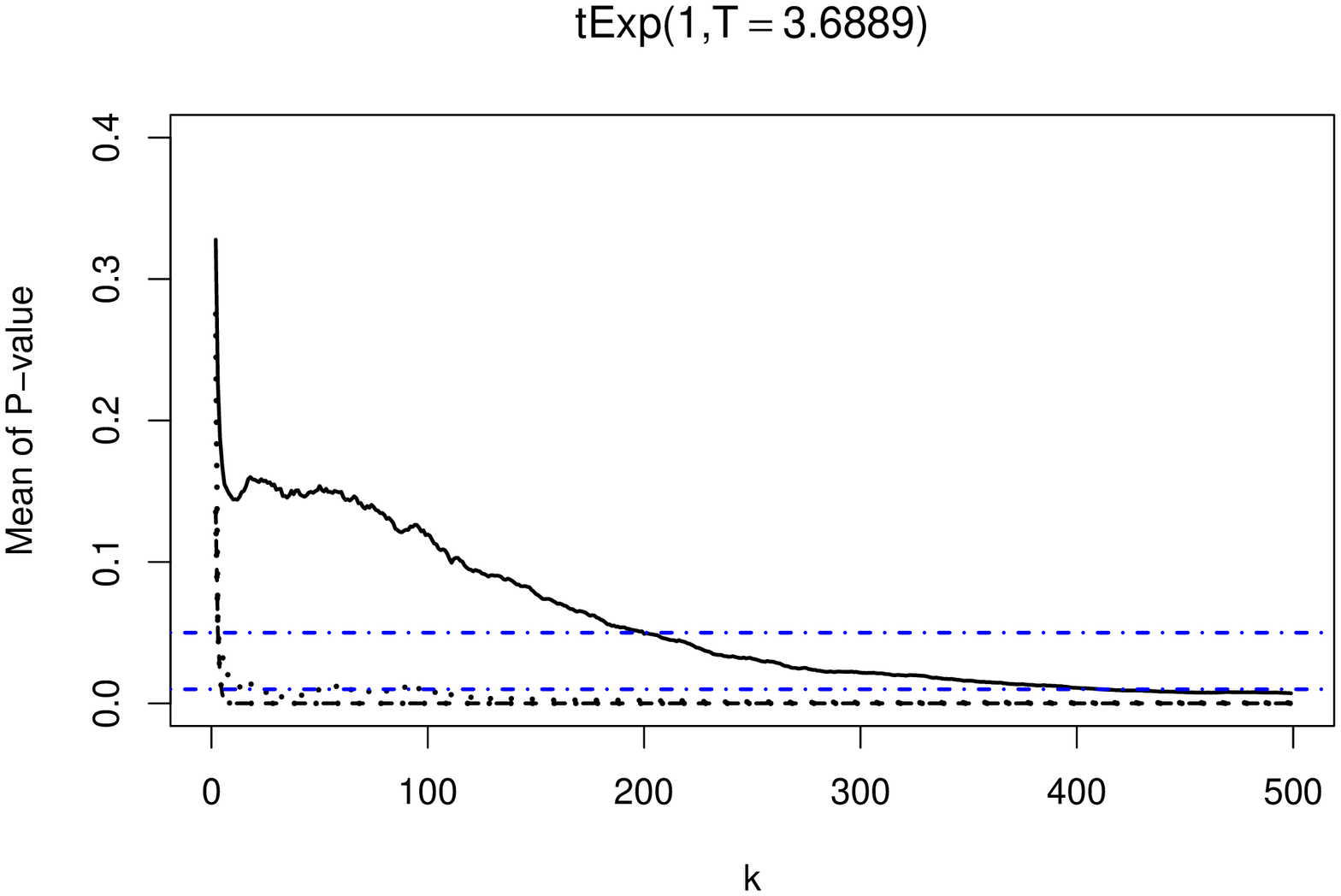}
	   \includegraphics[height=6.25cm, angle=270]{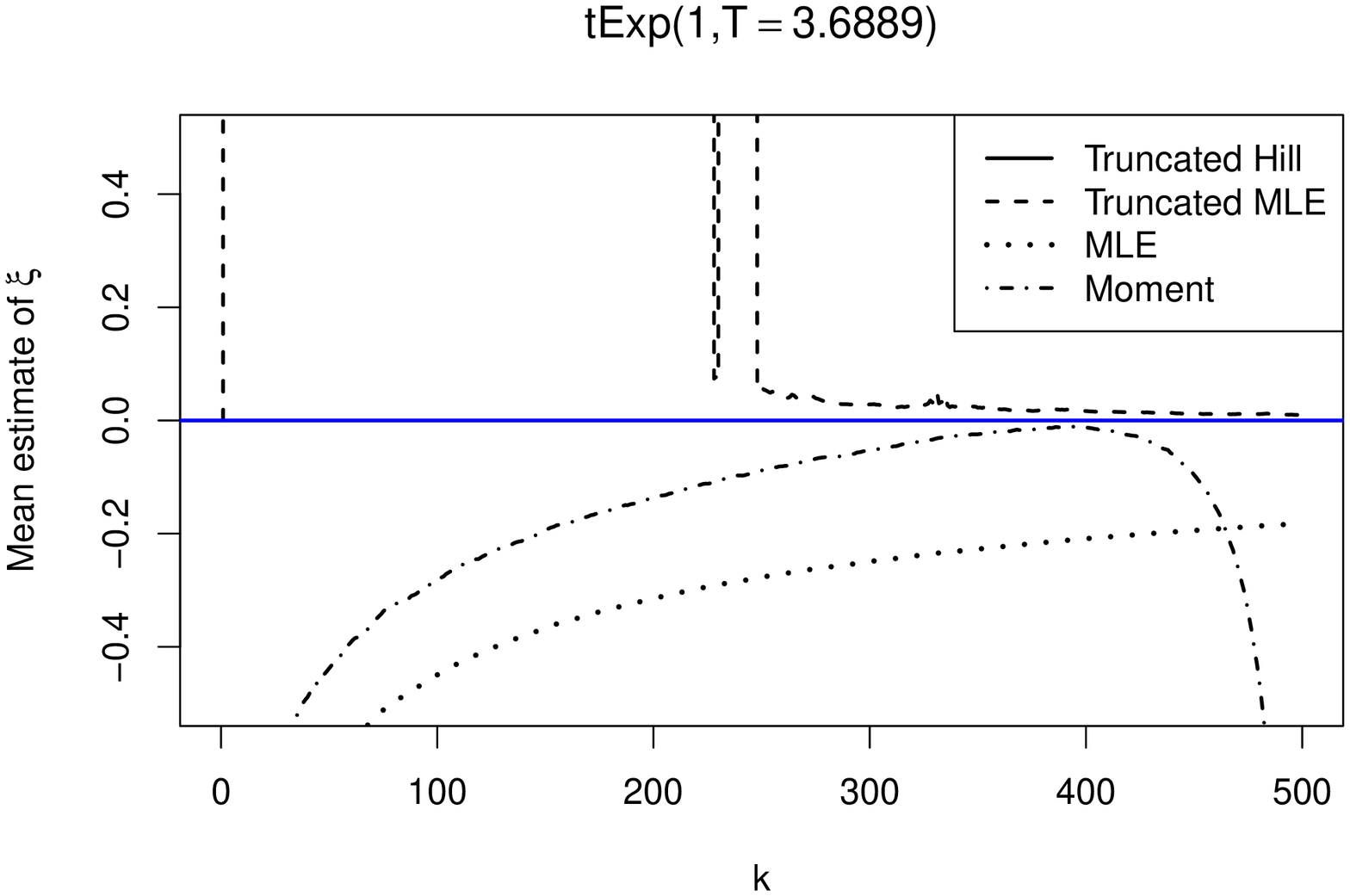}
  \includegraphics[height=6.25cm, angle=270]{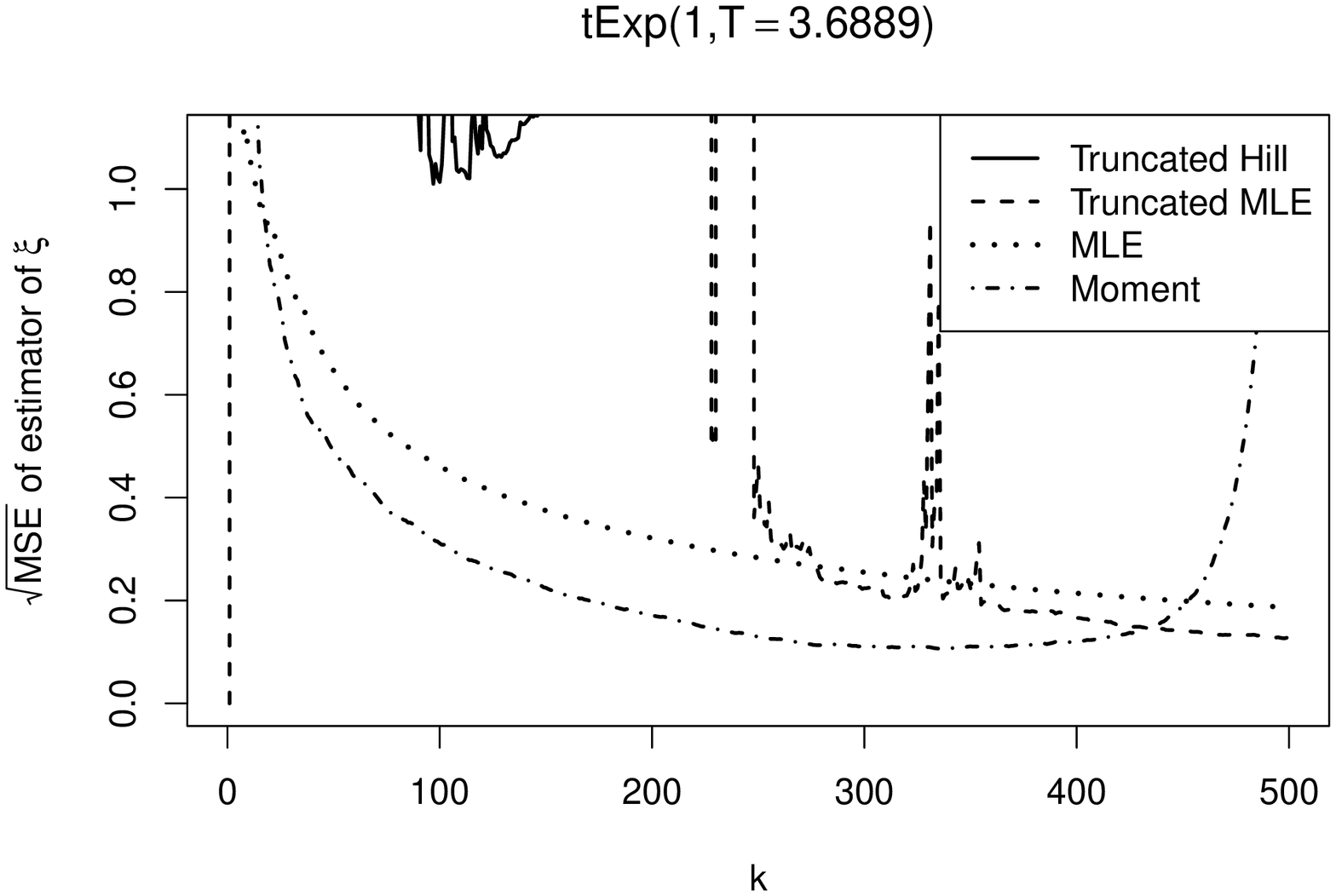}\\   
	\includegraphics[height=6.25cm, angle=270]{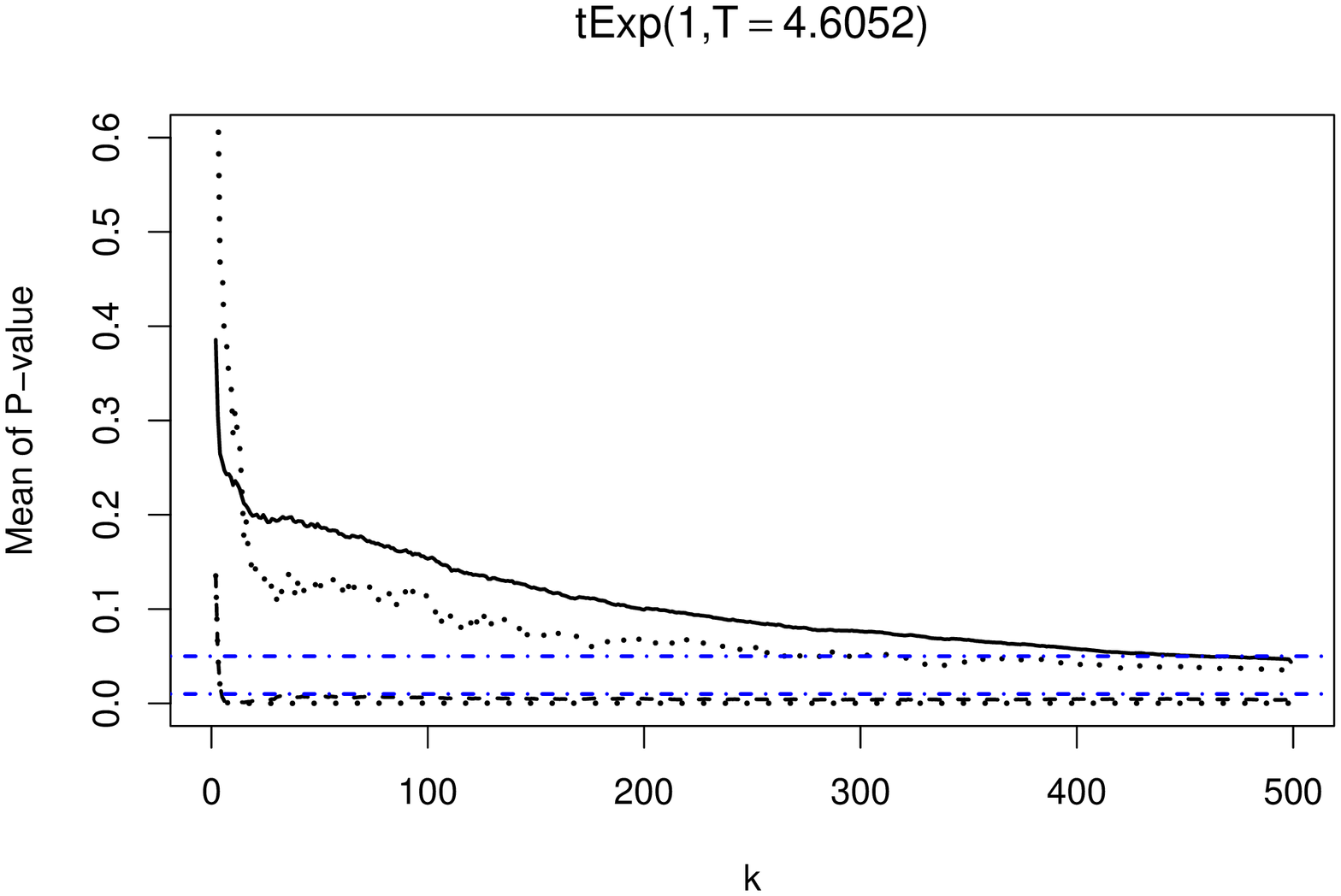}
   \includegraphics[height=6.25cm, angle=270]{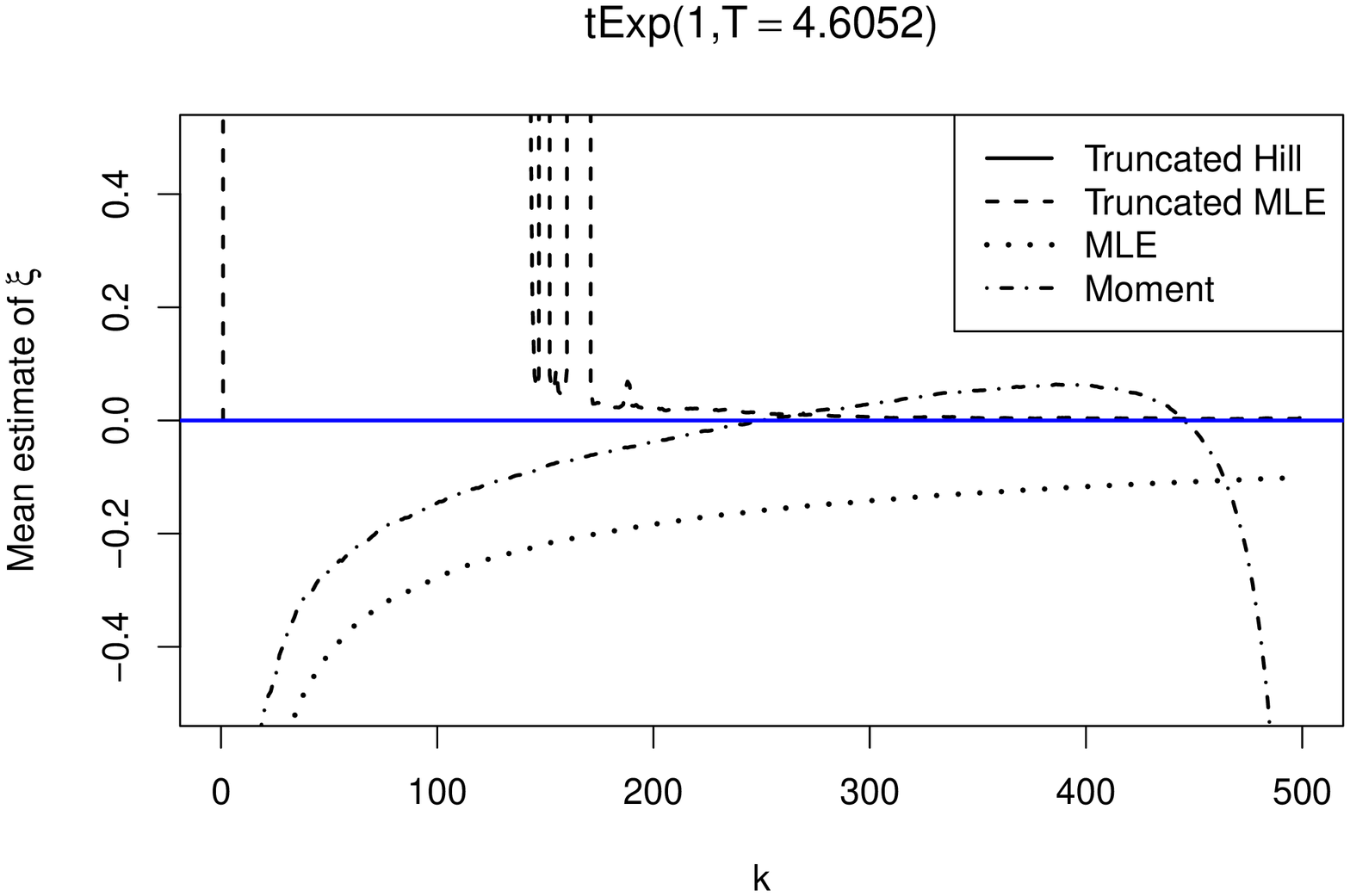}
  \includegraphics[height=6.25cm, angle=270]{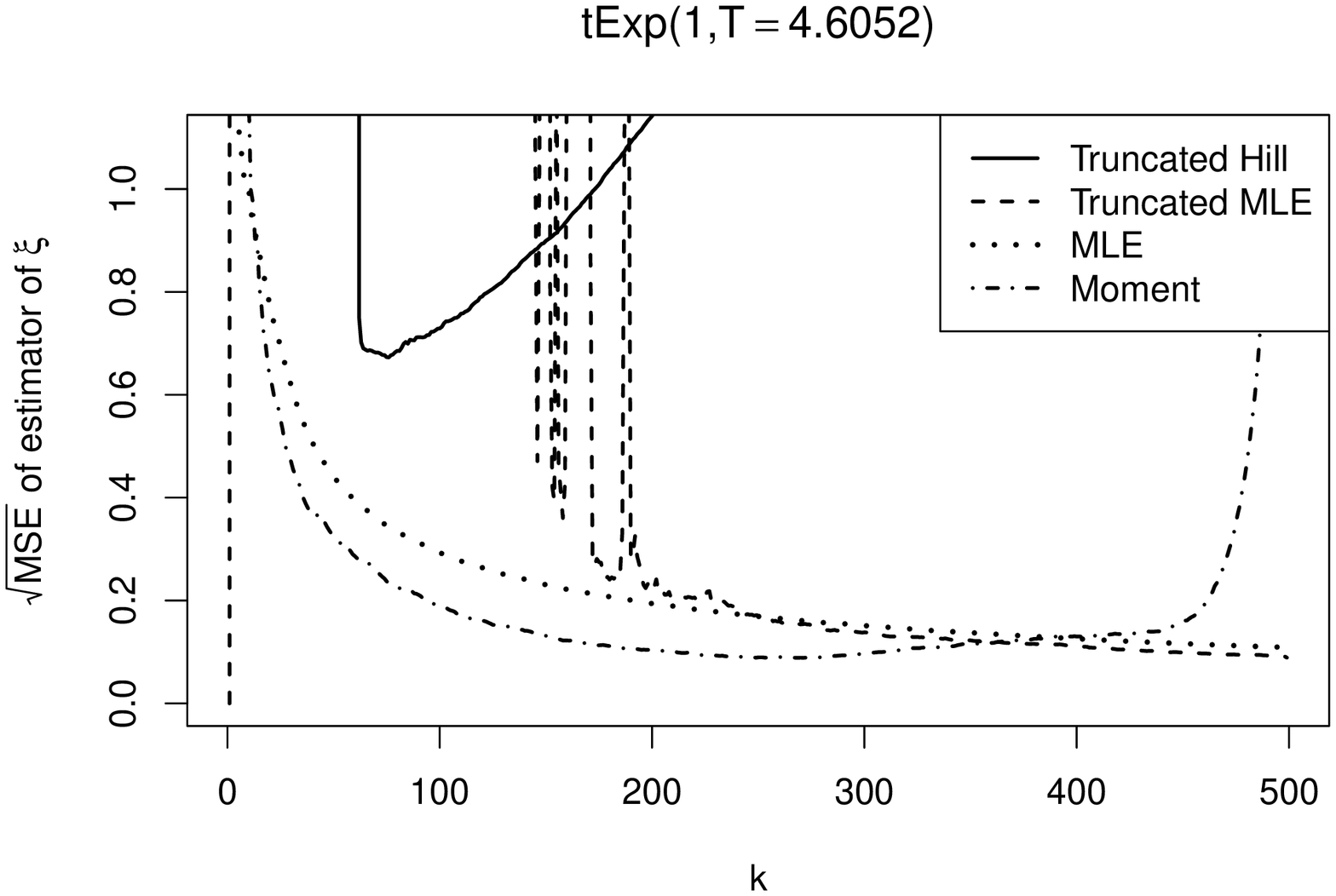}\\
	  \includegraphics[height=6.25cm, angle=270]{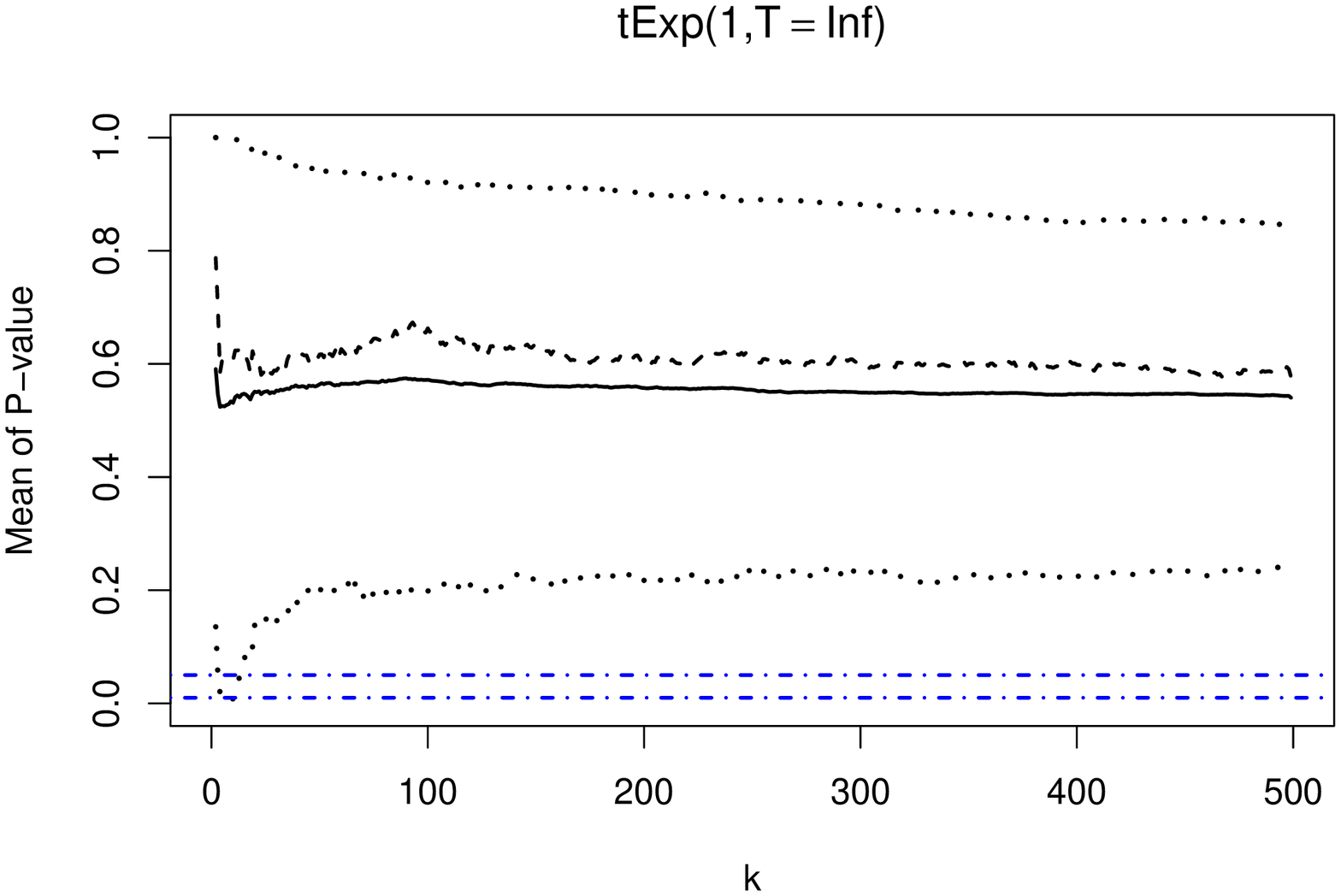}
  \includegraphics[height=6.25cm, angle=270]{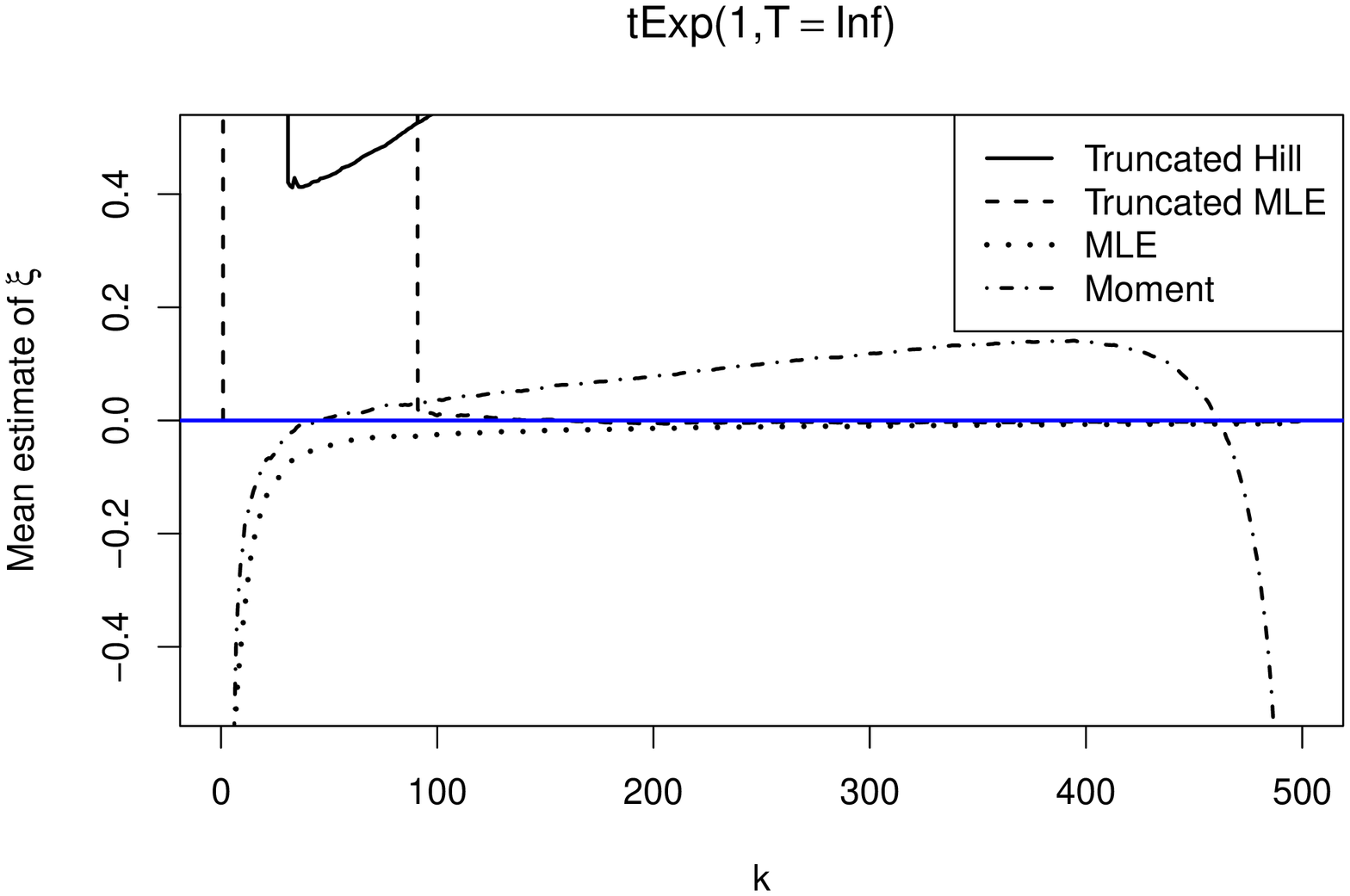}
  \includegraphics[height=6.25cm, angle=270]{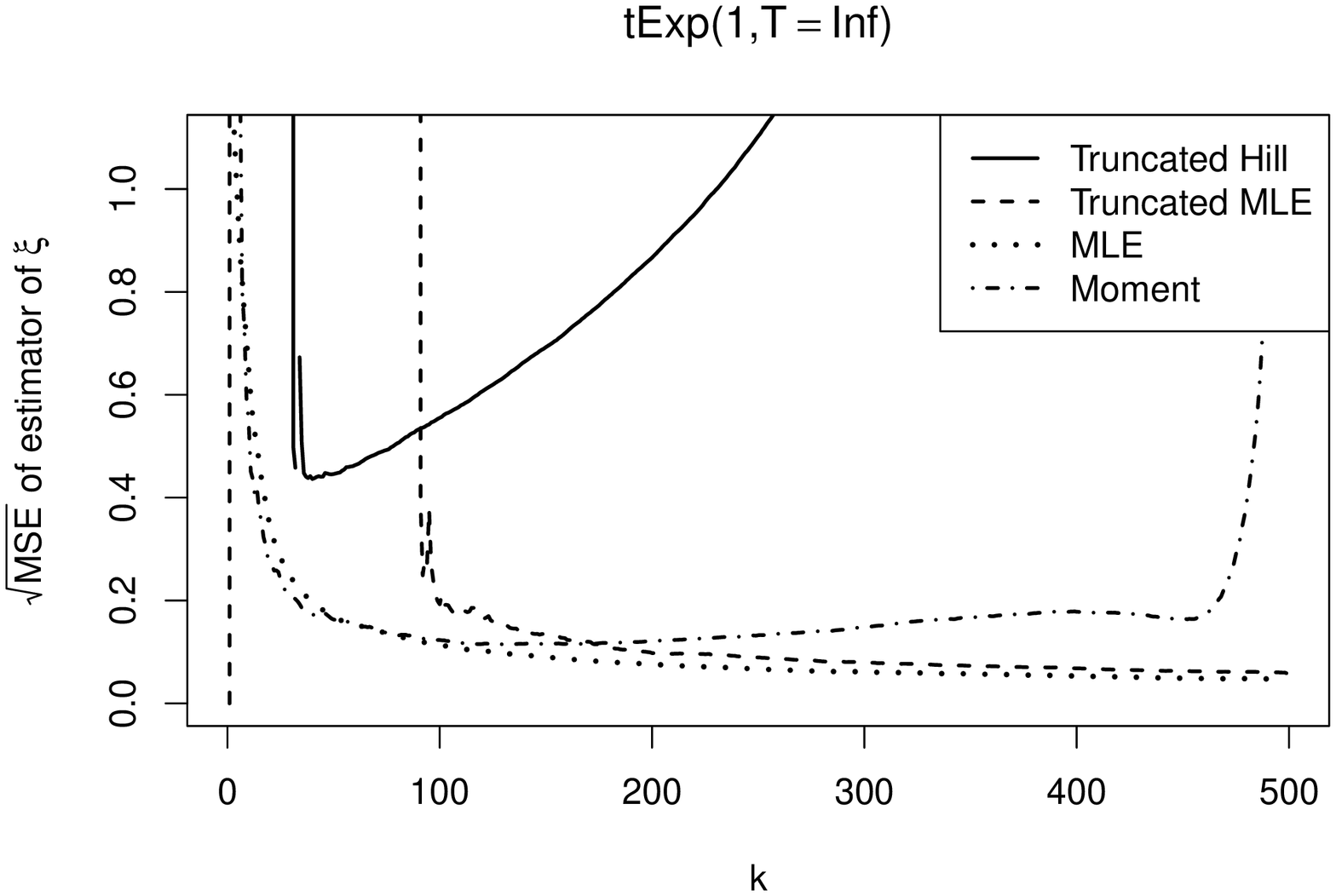}
  \caption{Means and boxplots of P-values for test (left), means (middle) and root MSE (right) of $\hat{\xi}^+_k$, $\hat{\xi}_k$, $\hat{\xi}^{\infty}_k$ and $\hat{\xi}^M_k$  from the standard exponential distribution truncated at $Q_Y (0.975)$ (top), $Q_Y (0.99)$ (middle) and non truncated (bottom).}
  \end{figure}
	\end{landscape}

\begin{landscape}
   \begin{figure}[!ht]
		\centering
		   \includegraphics[height=6.25cm, angle=270]{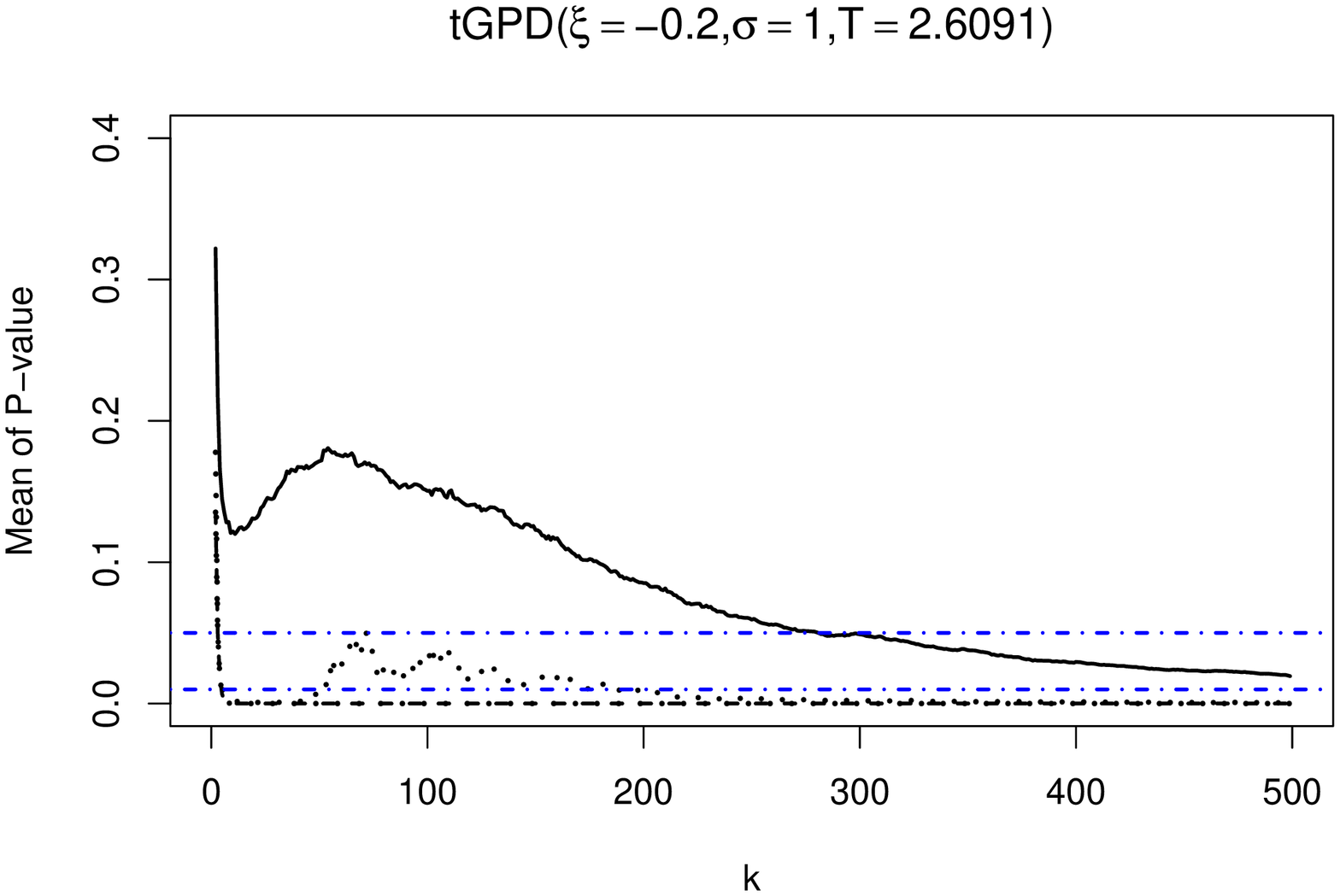}
	   \includegraphics[height=6.25cm, angle=270]{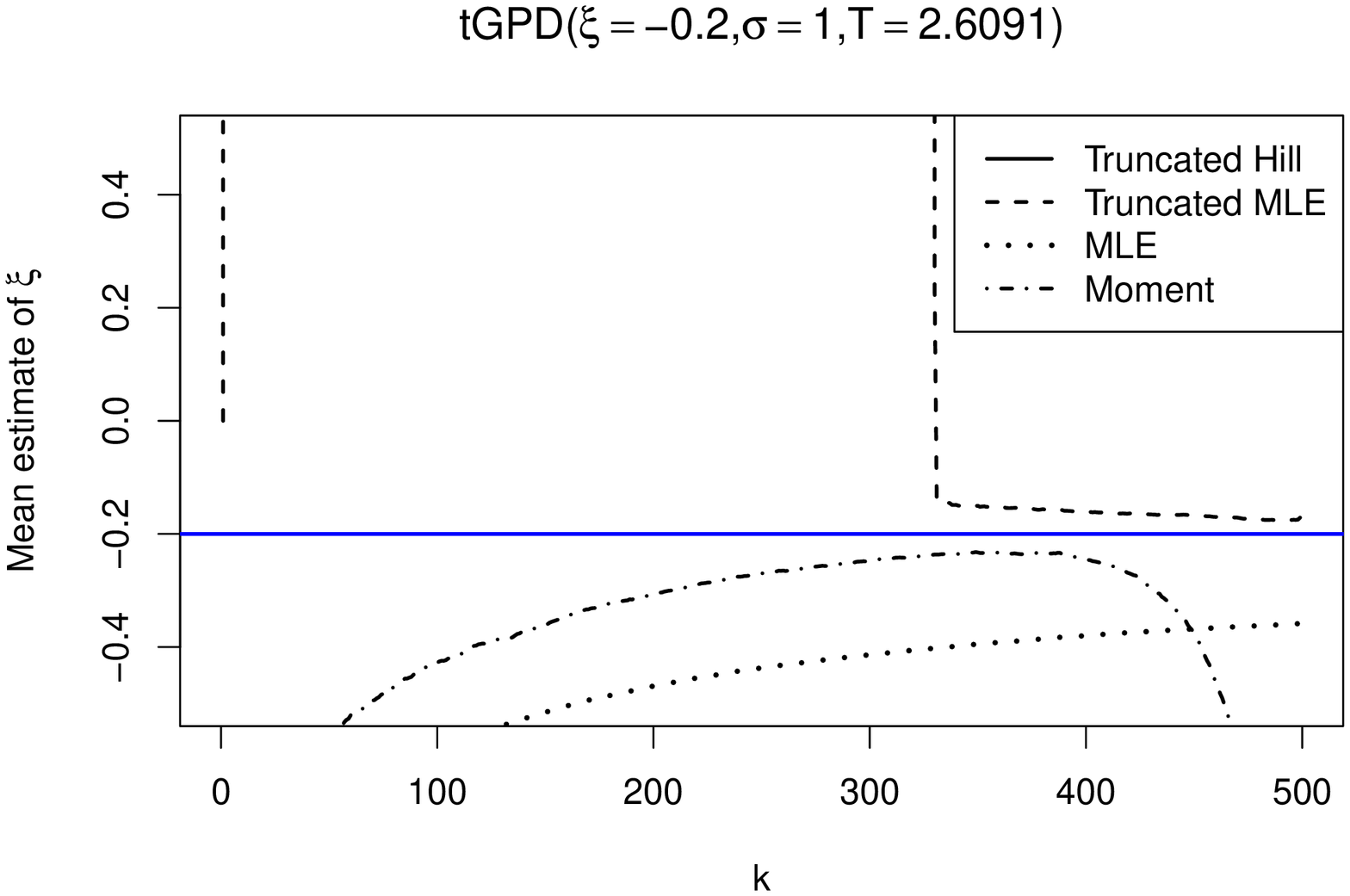}
  \includegraphics[height=6.25cm, angle=270]{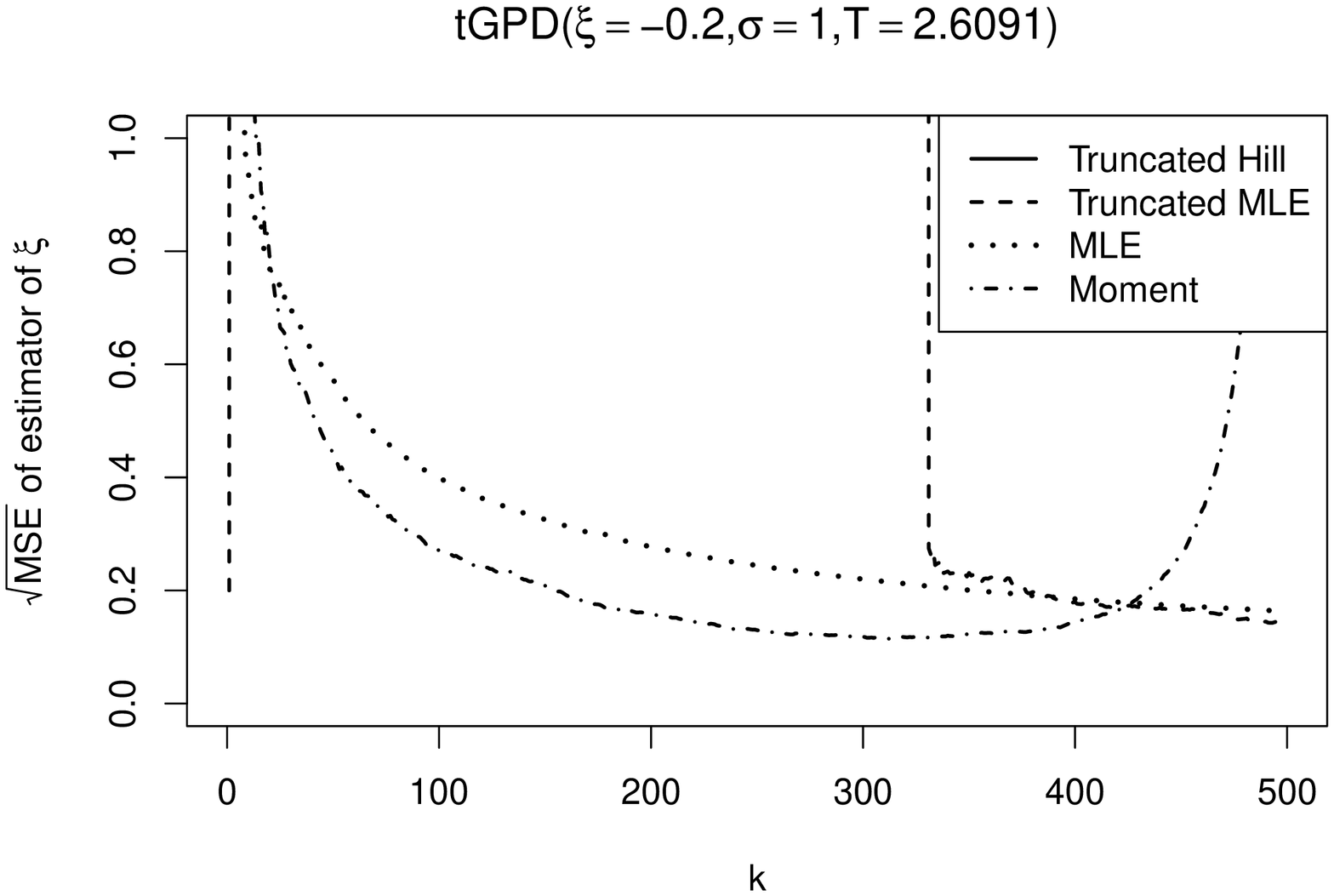}\\   
	\includegraphics[height=6.25cm, angle=270]{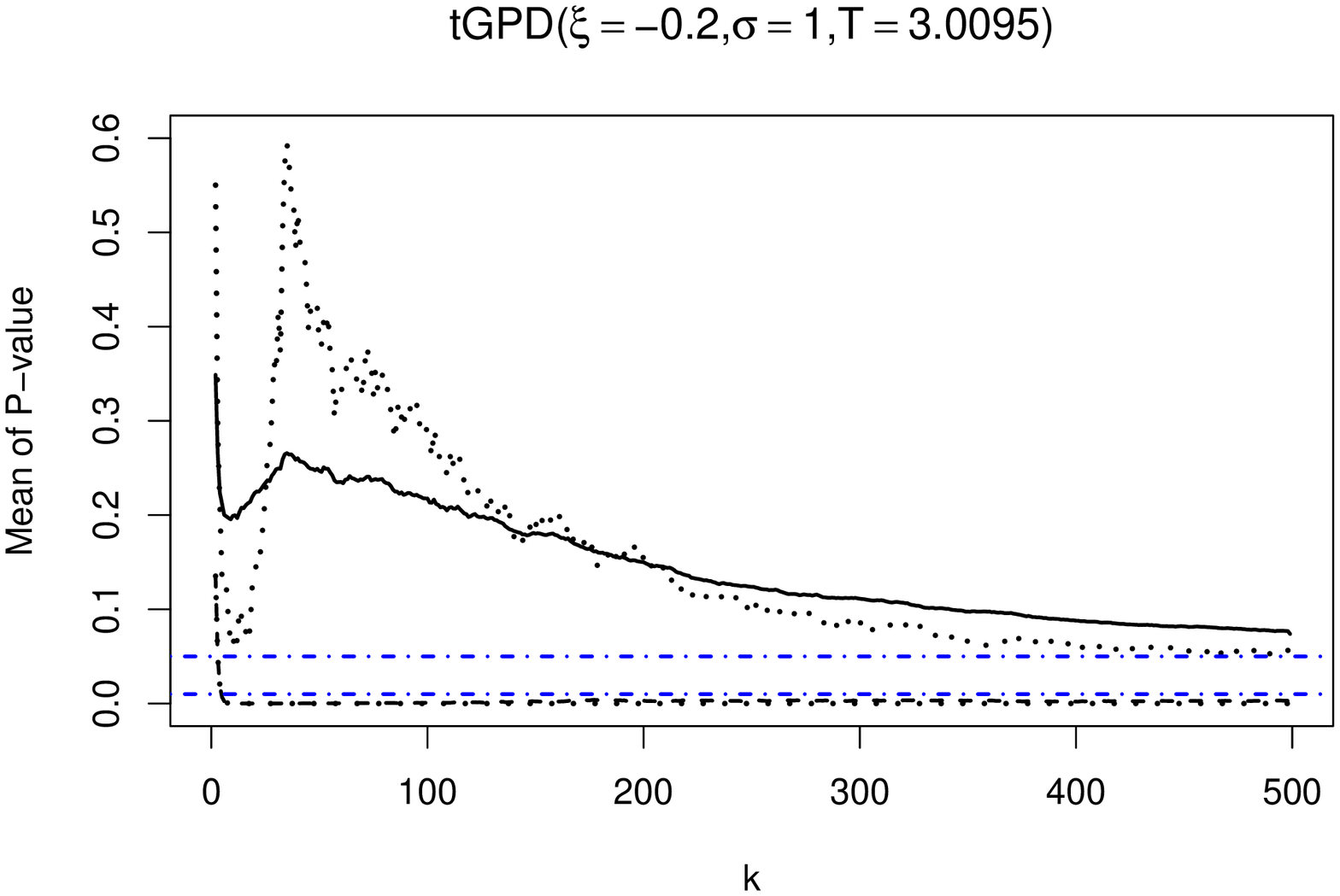}
   \includegraphics[height=6.25cm, angle=270]{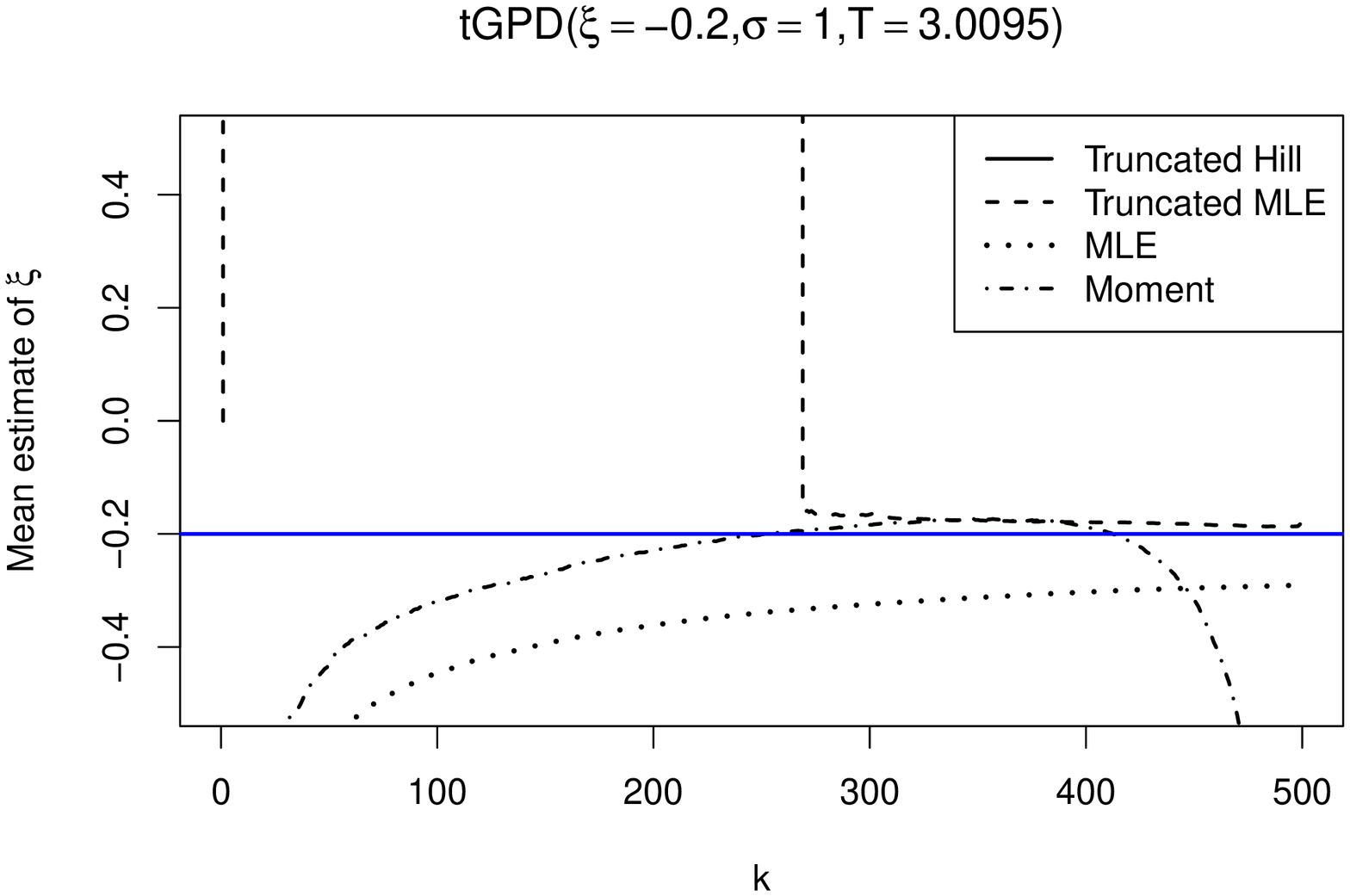}
  \includegraphics[height=6.25cm, angle=270]{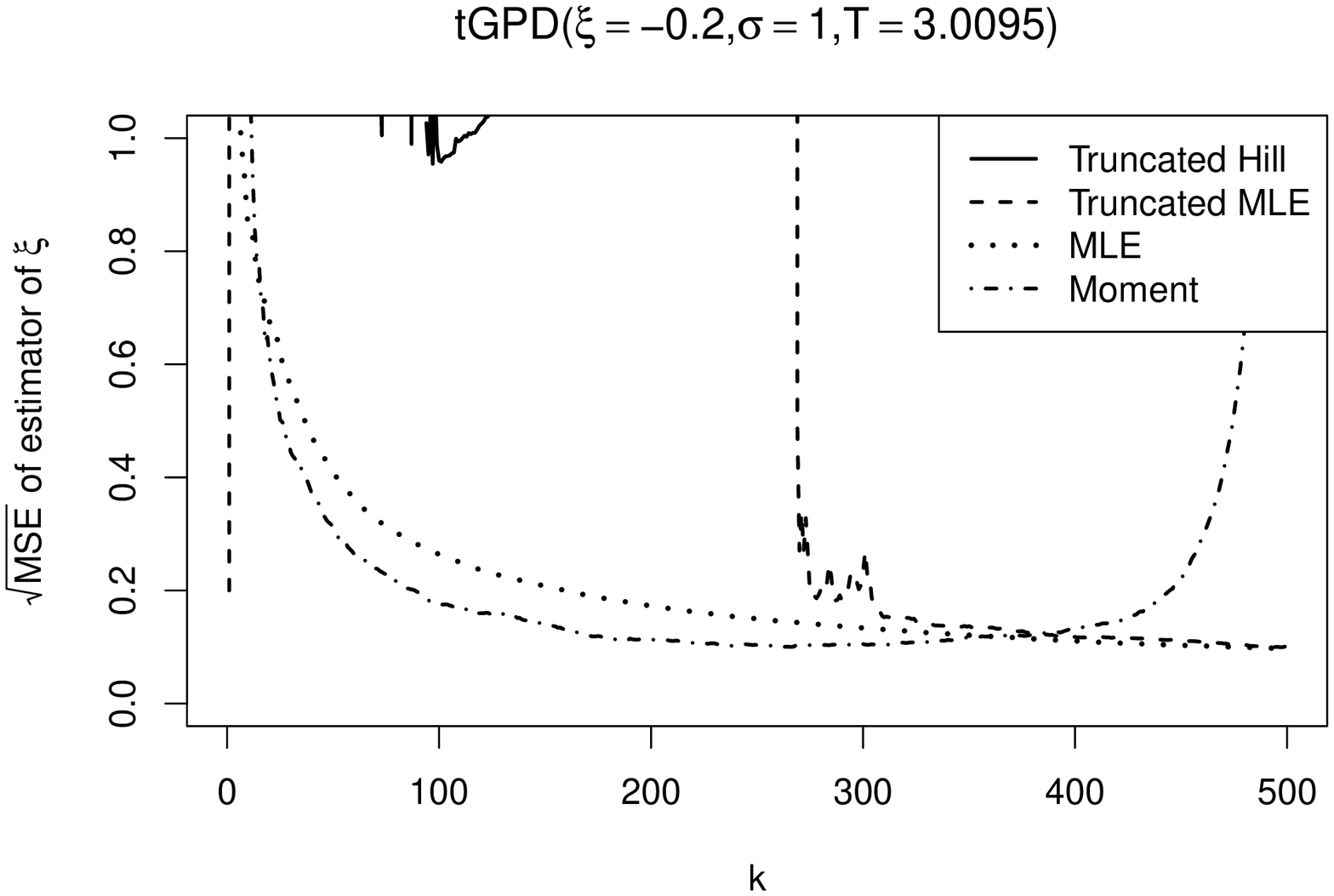}\\
	  \includegraphics[height=6.25cm, angle=270]{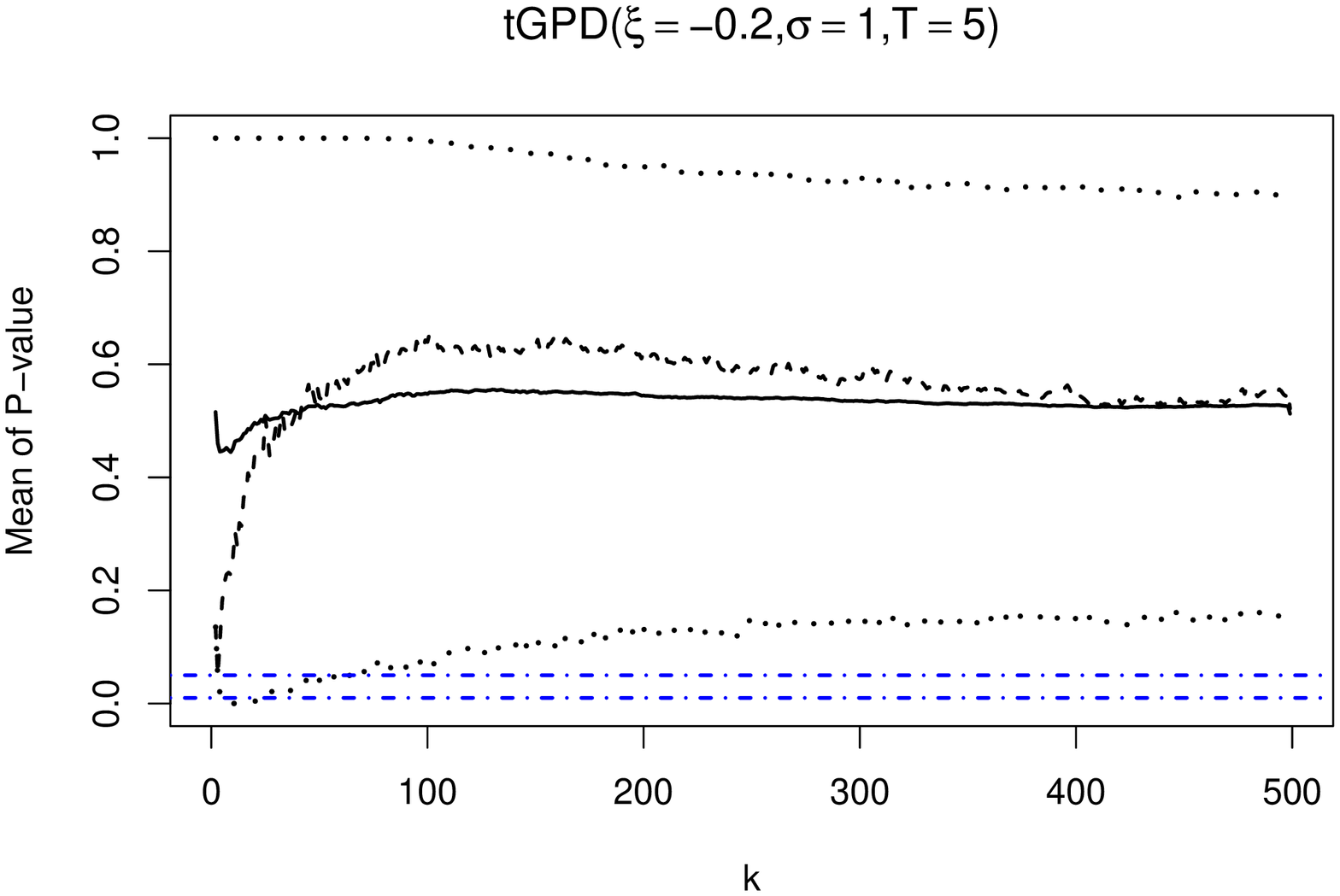}
  \includegraphics[height=6.25cm, angle=270]{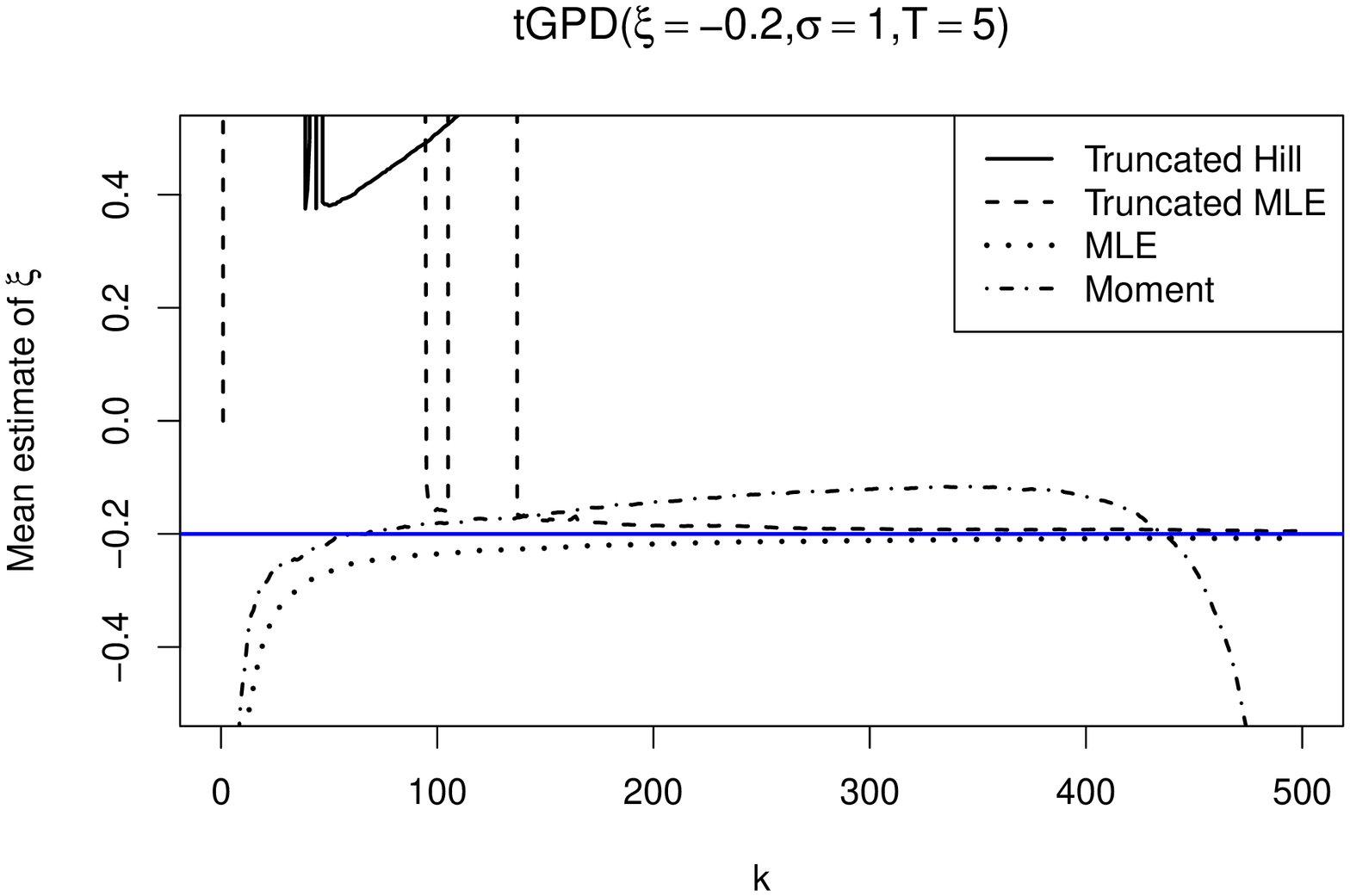}
  \includegraphics[height=6.25cm, angle=270]{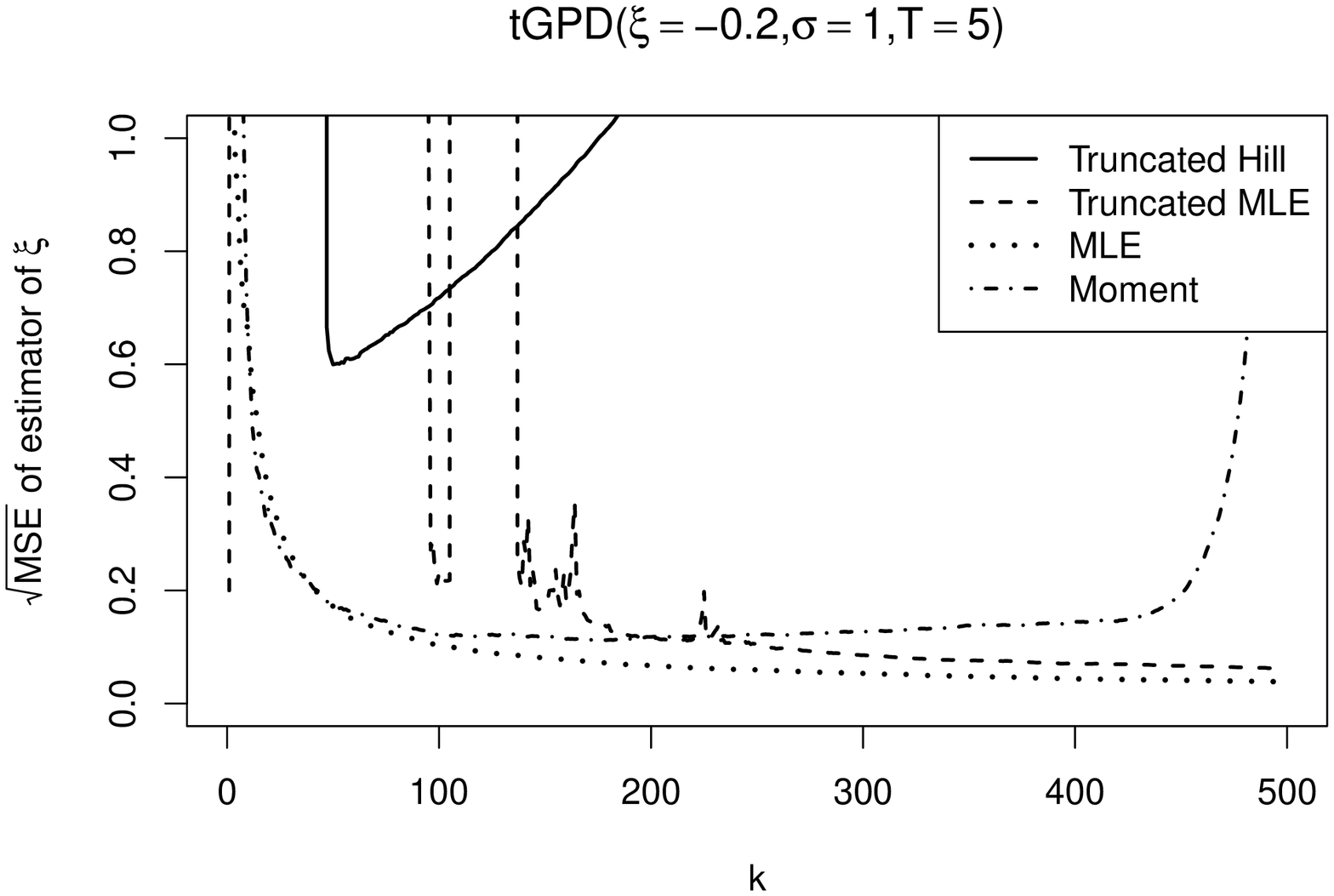}
  \caption{Means and boxplots of P-values for test (left), means (middle) and root MSE (right) of $\hat{\xi}^+_k$, $\hat{\xi}_k$, $\hat{\xi}^{\infty}_k$ and $\hat{\xi}^M_k$  from GPD(-0.2,1) truncated at $Q_Y (0.975)$ (top), $Q_Y (0.99)$ (middle) and  $Q_Y (1)$ (bottom).}\label{fig:sim_xi_last}
  \end{figure}
	\end{landscape}

		\newpage					
   \begin{figure}[!ht]
	\centering
	  \includegraphics[height=0.495\textwidth, angle=270]{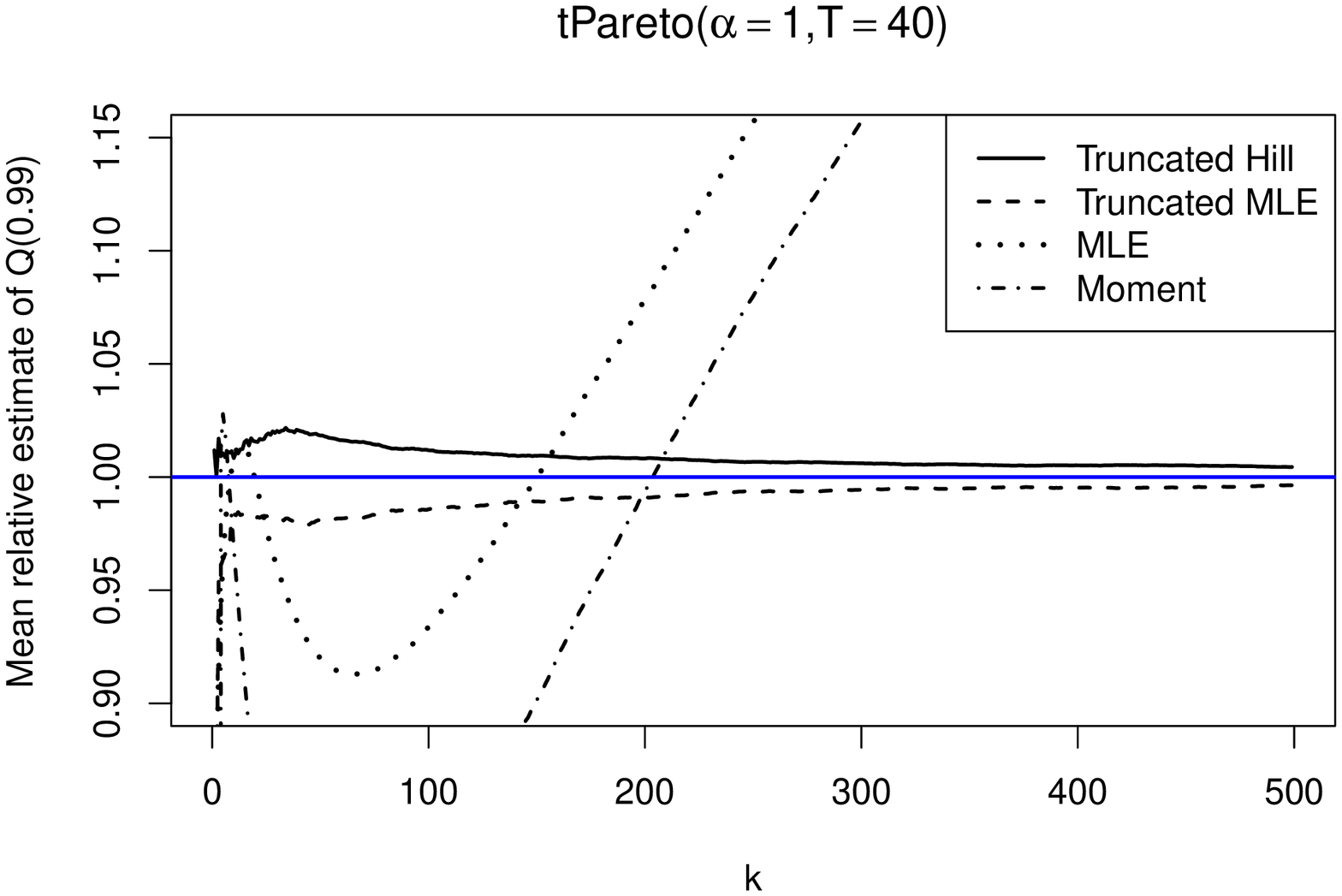}
	 \includegraphics[height=0.495\textwidth, angle=270]{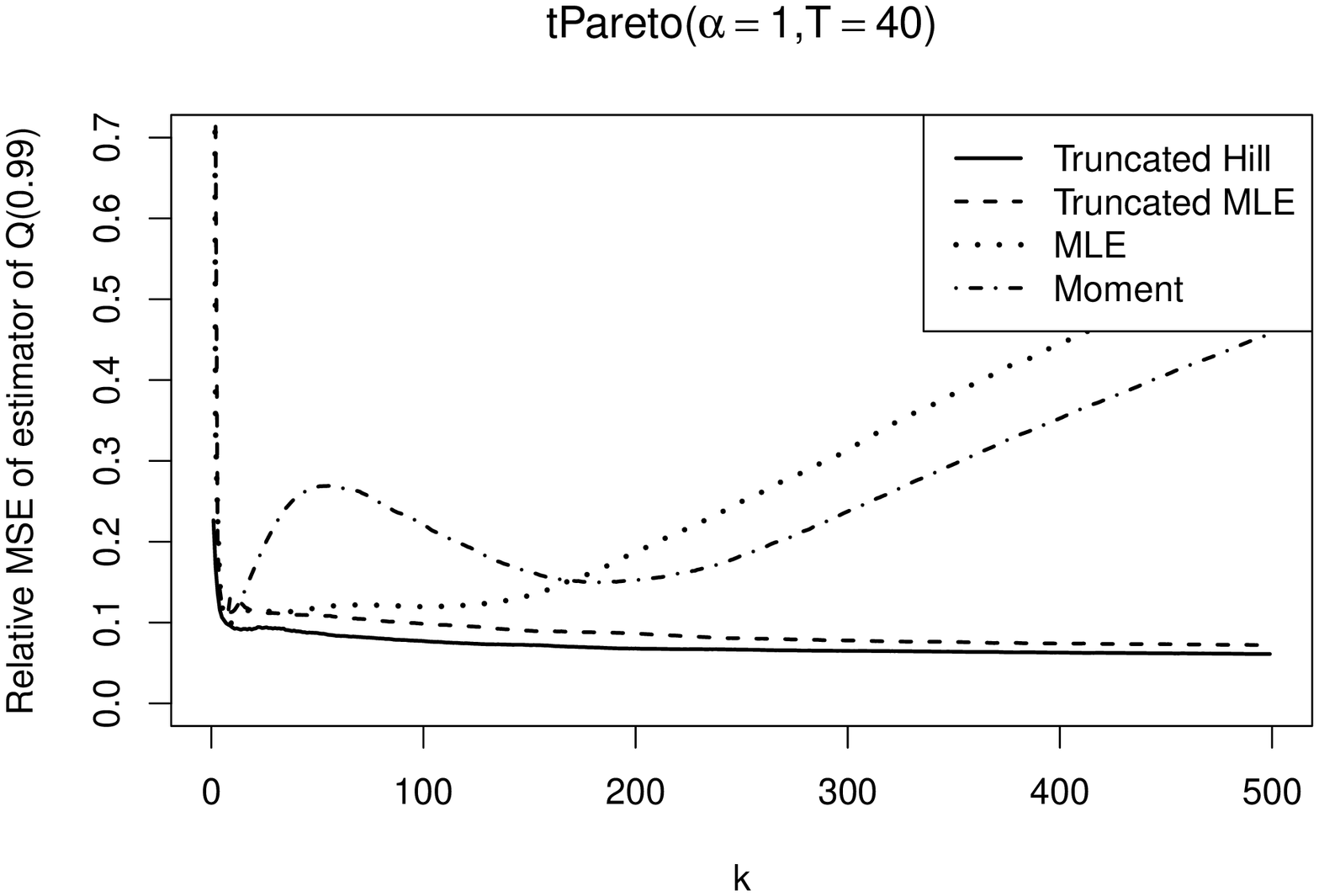}\\
	  \includegraphics[height=0.495\textwidth, angle=270]{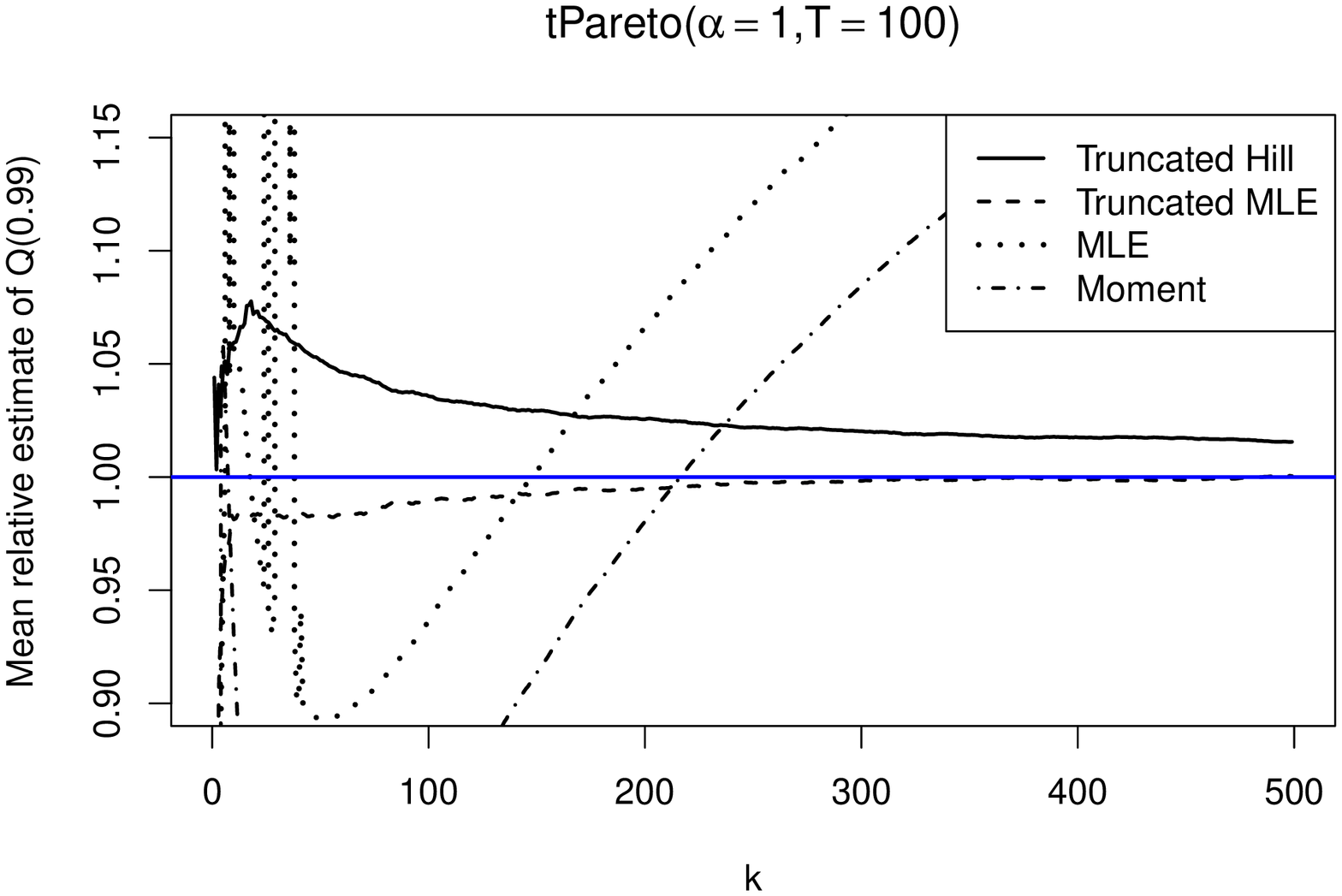}
	 \includegraphics[height=0.495\textwidth, angle=270]{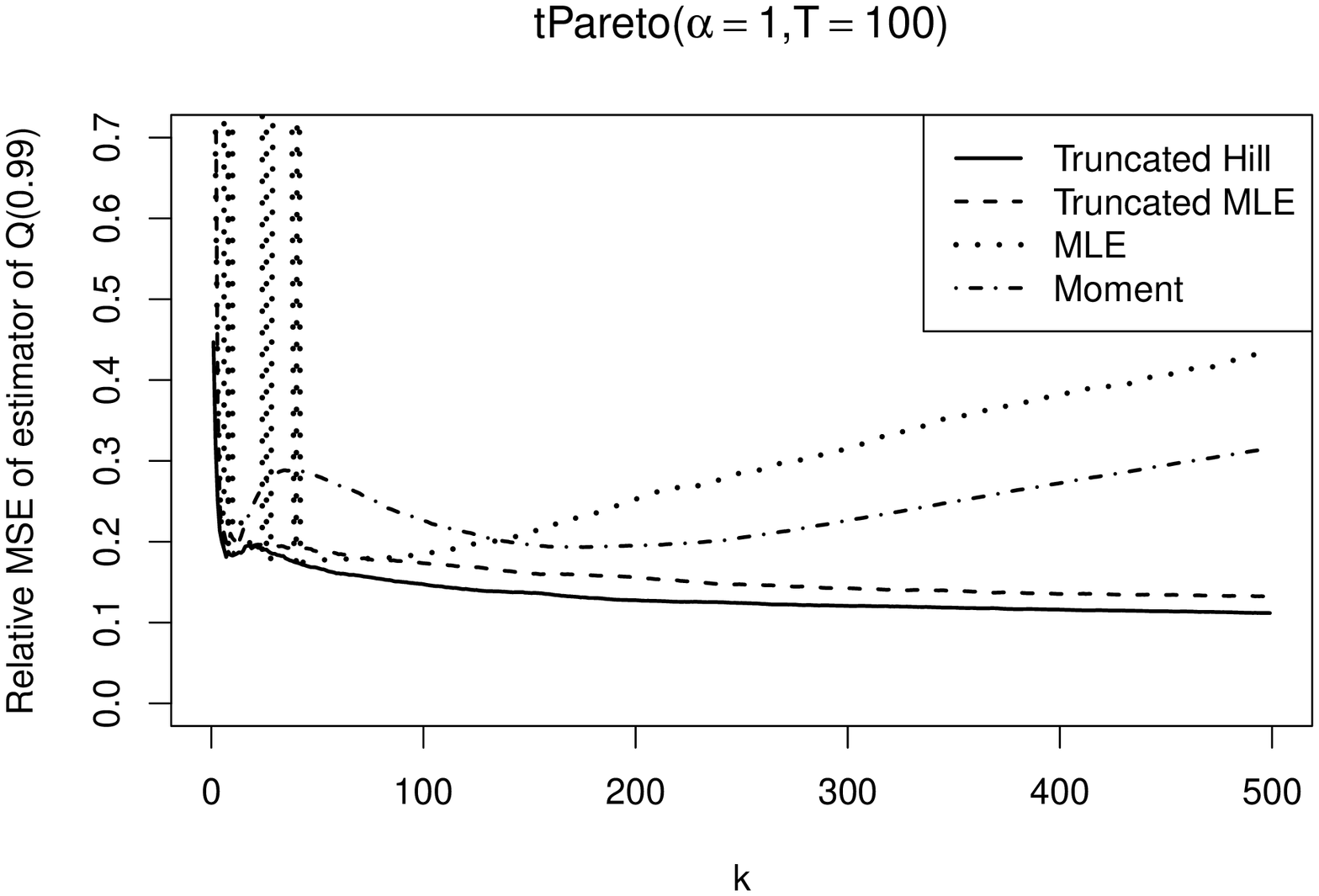}\\
	  \includegraphics[height=0.495\textwidth, angle=270]{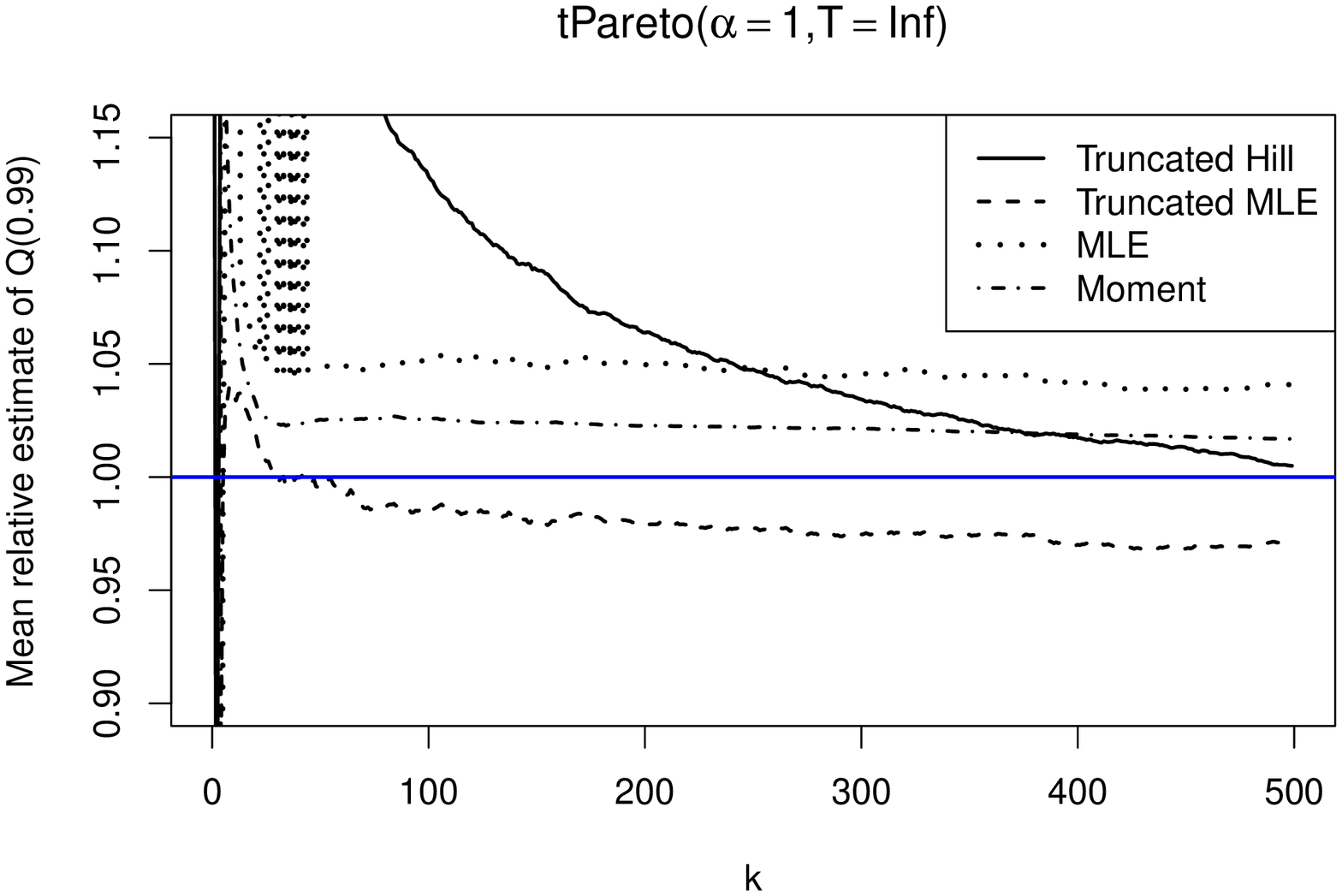}
	 \includegraphics[height=0.495\textwidth, angle=270]{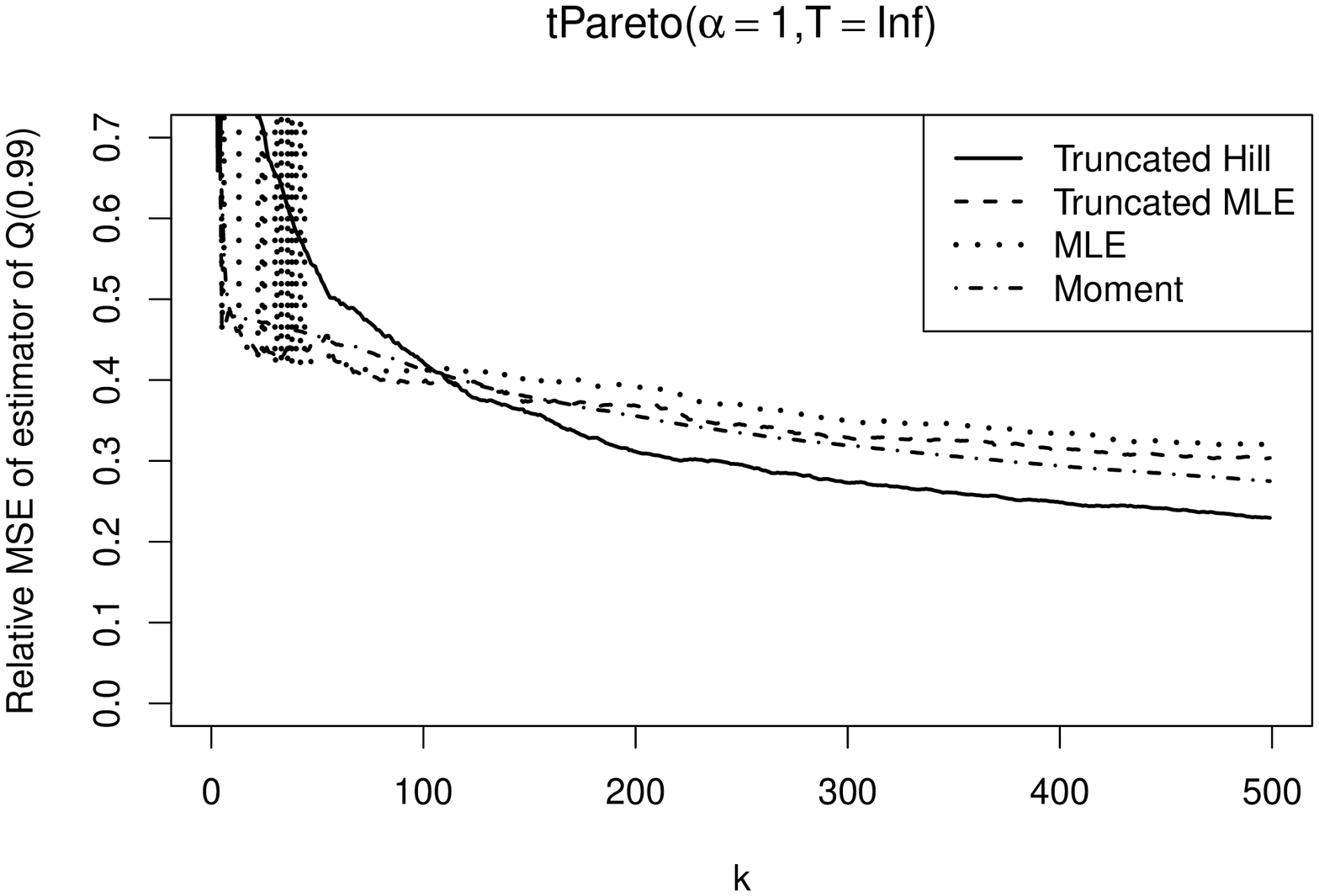}
  \caption{Mean deviations of $\hat{Q}^+_{T,k}(1-p)/Q_T(1-p)$, $\hat{Q}_{T,k}(1-p)/Q_T(1-p)$, $\hat{Q}^{\infty}_{k}(1-p)/Q_T(1-p)$, $\hat{Q}^M_{k}(1-p)/Q_T(1-p)$  and corresponding MSE with $p=0.01$ for the standard Pareto distribution truncated at $Q_Y (0.975)$ (top), $Q_Y (0.99)$ (middle) and non truncated (bottom).}\label{fig:sim_Q_first}
   \end{figure}
	
					\newpage					
   \begin{figure}[!ht]
	\centering
  \includegraphics[height=0.495\textwidth, angle=270]{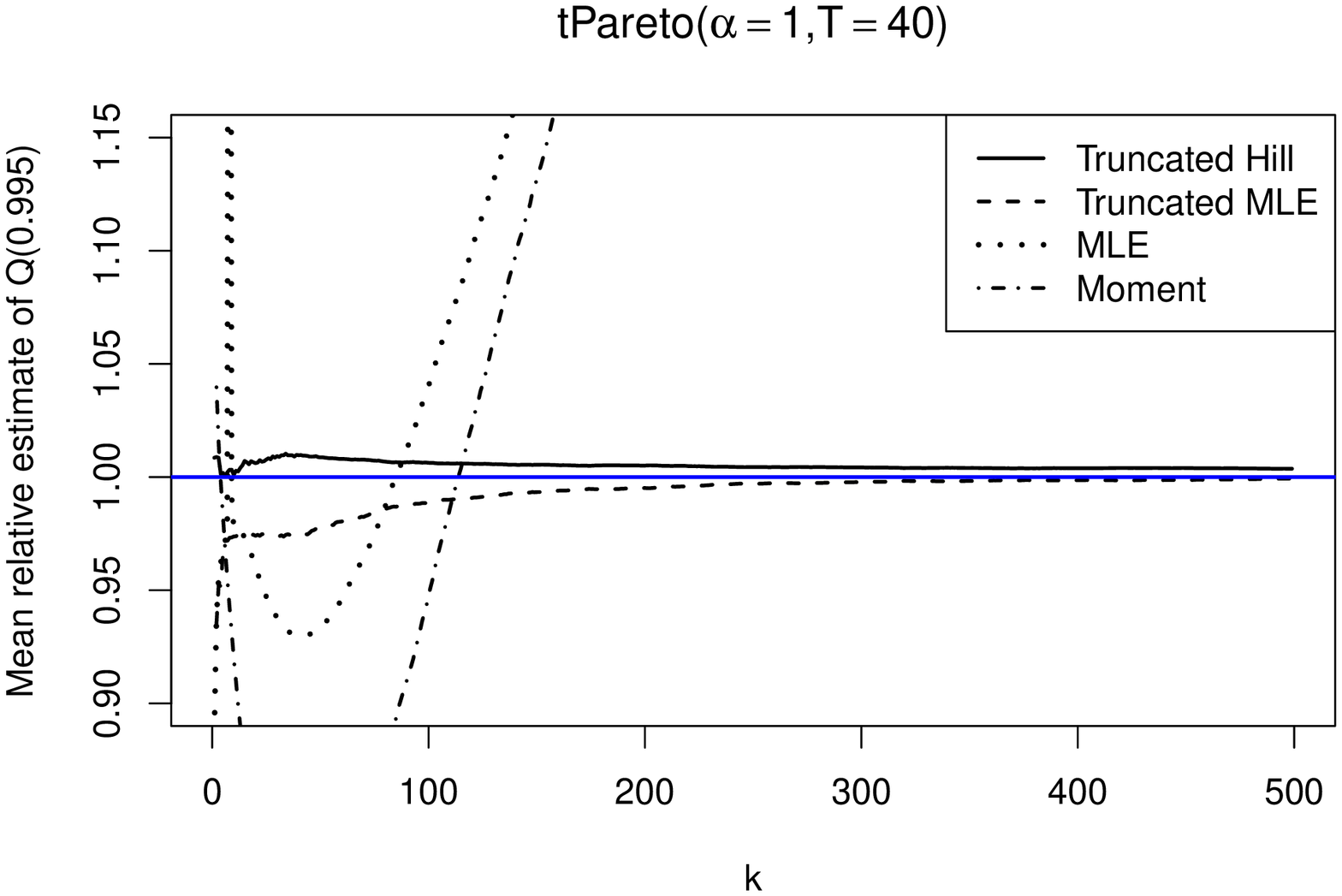}
	 \includegraphics[height=0.495\textwidth, angle=270]{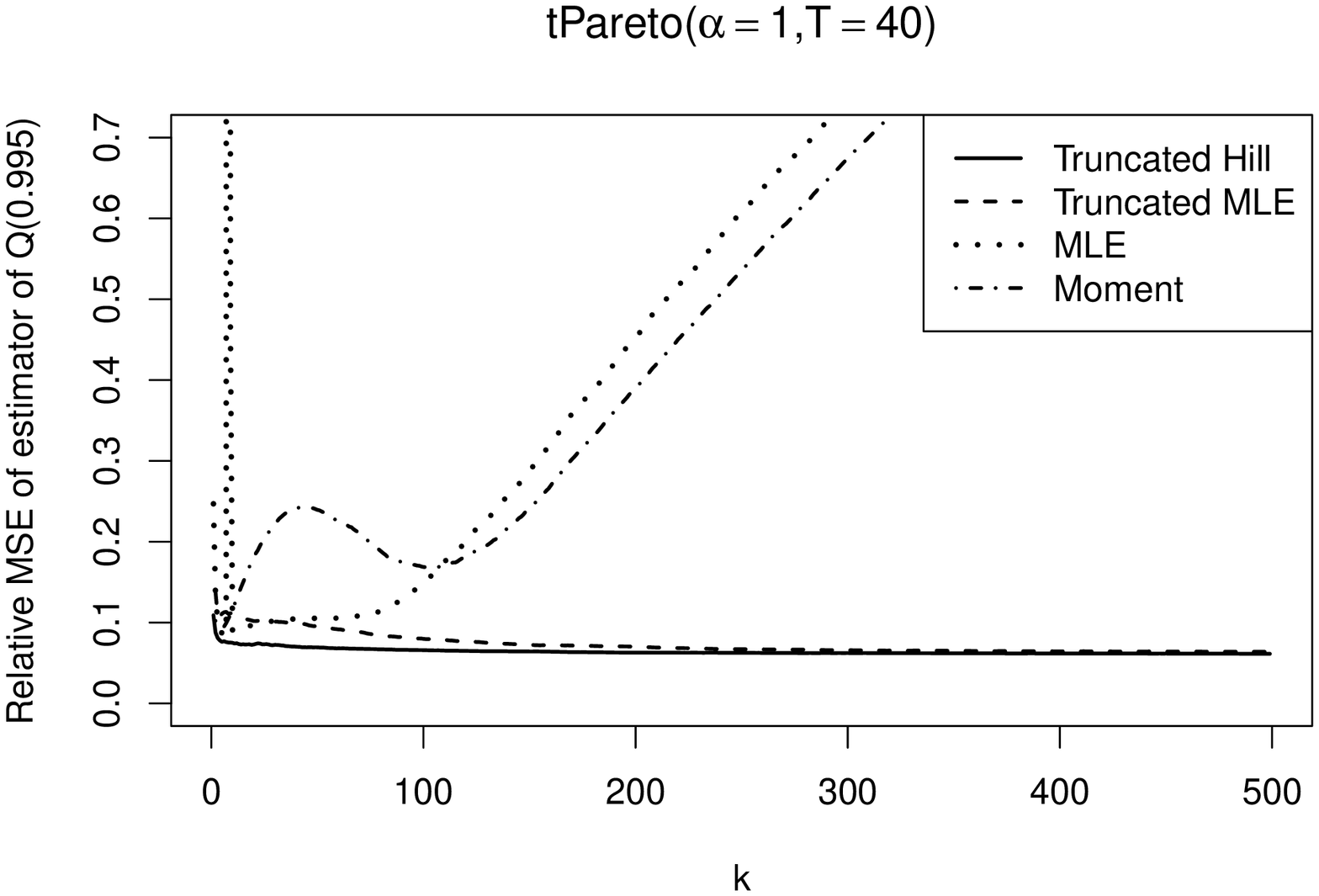}\\
	  \includegraphics[height=0.495\textwidth, angle=270]{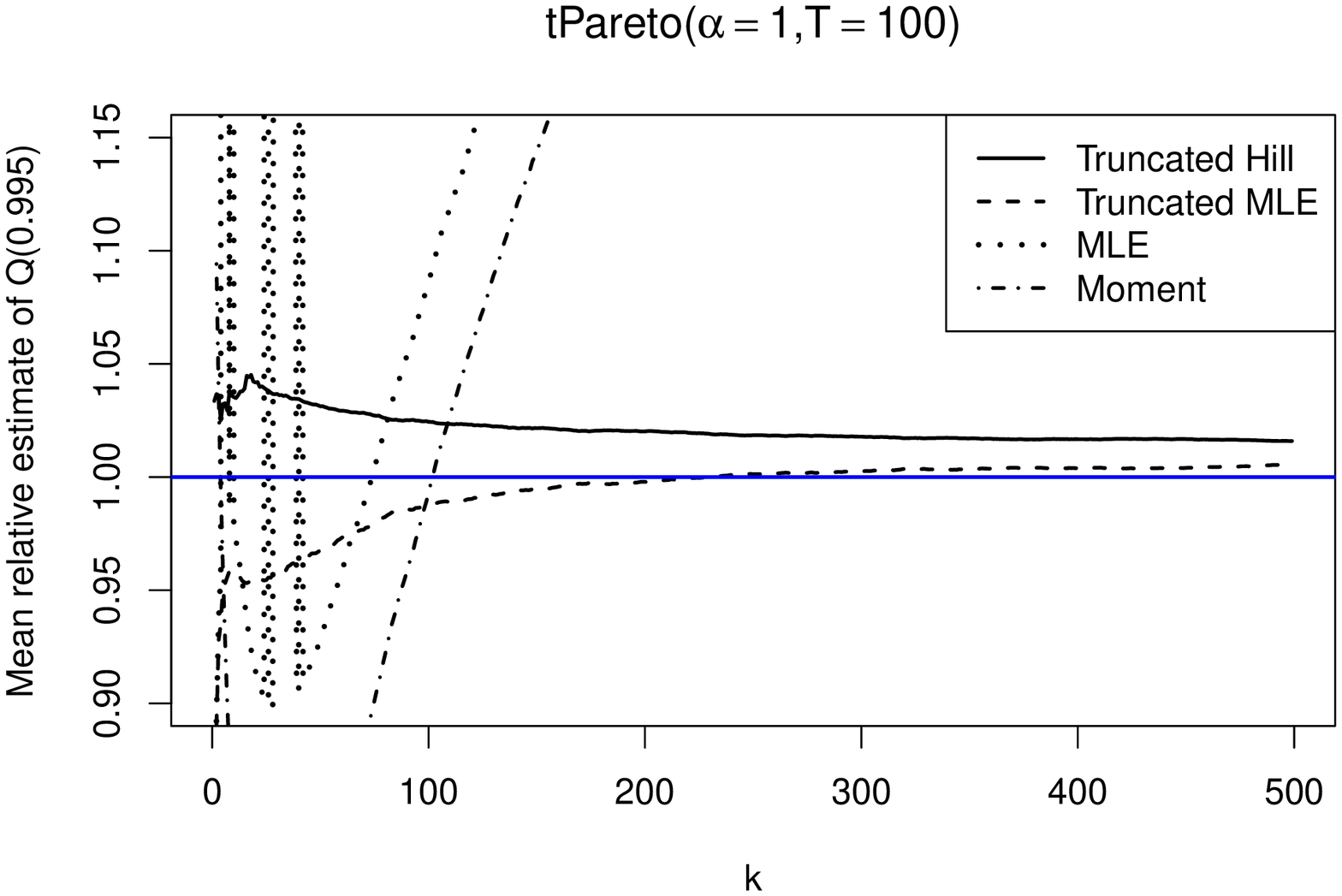}
	 \includegraphics[height=0.495\textwidth, angle=270]{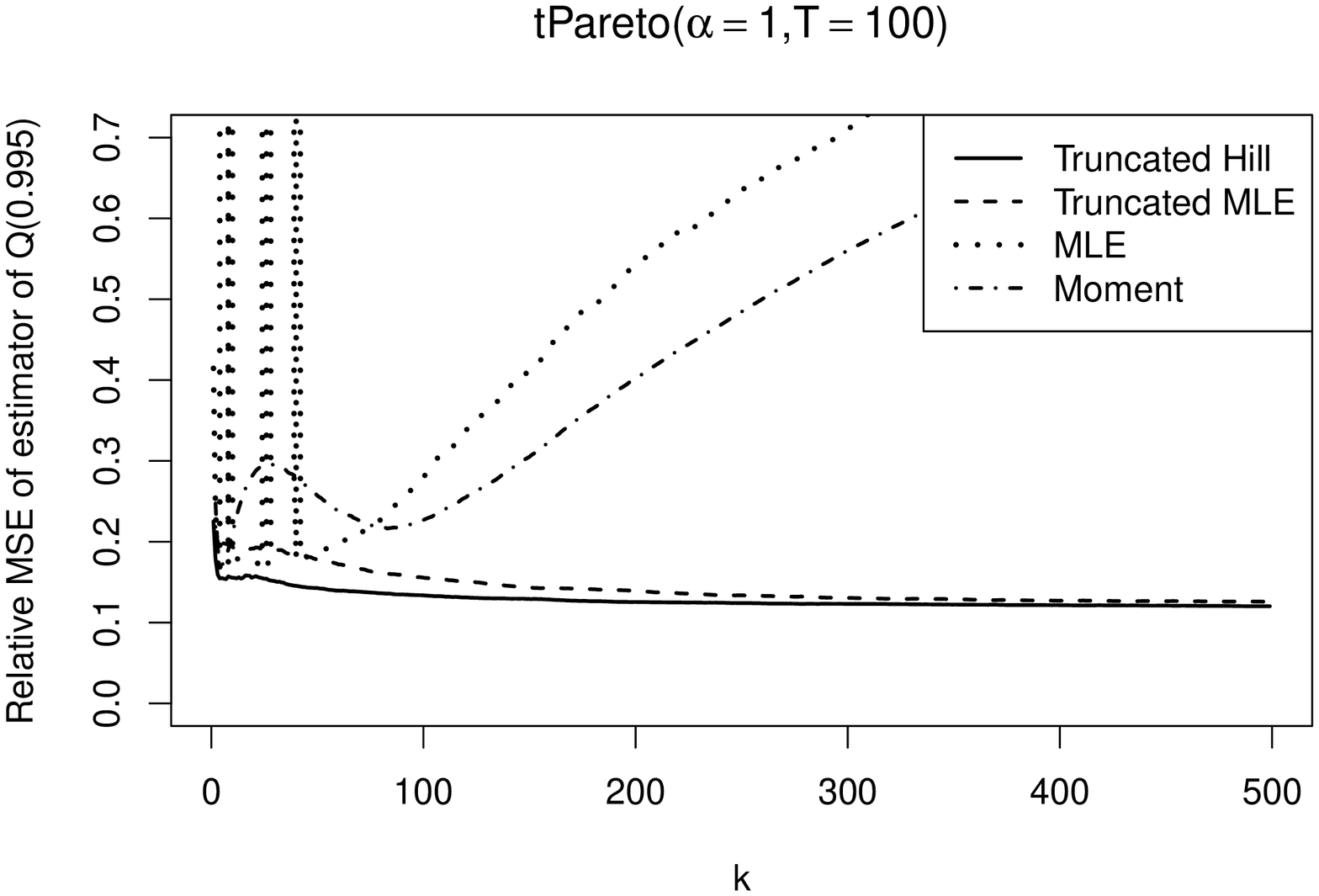}\\
	  \includegraphics[height=0.495\textwidth, angle=270]{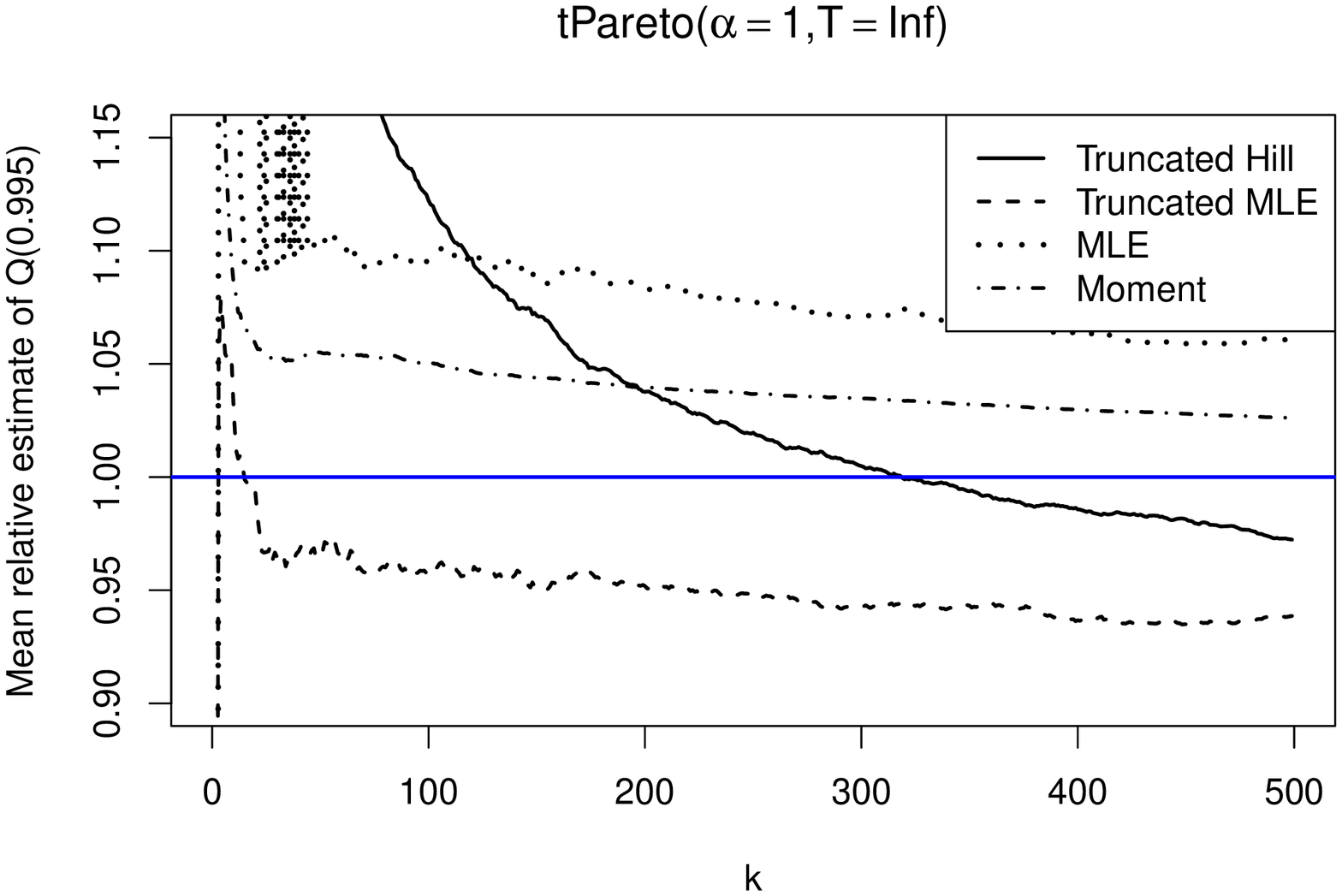}
	 \includegraphics[height=0.495\textwidth, angle=270]{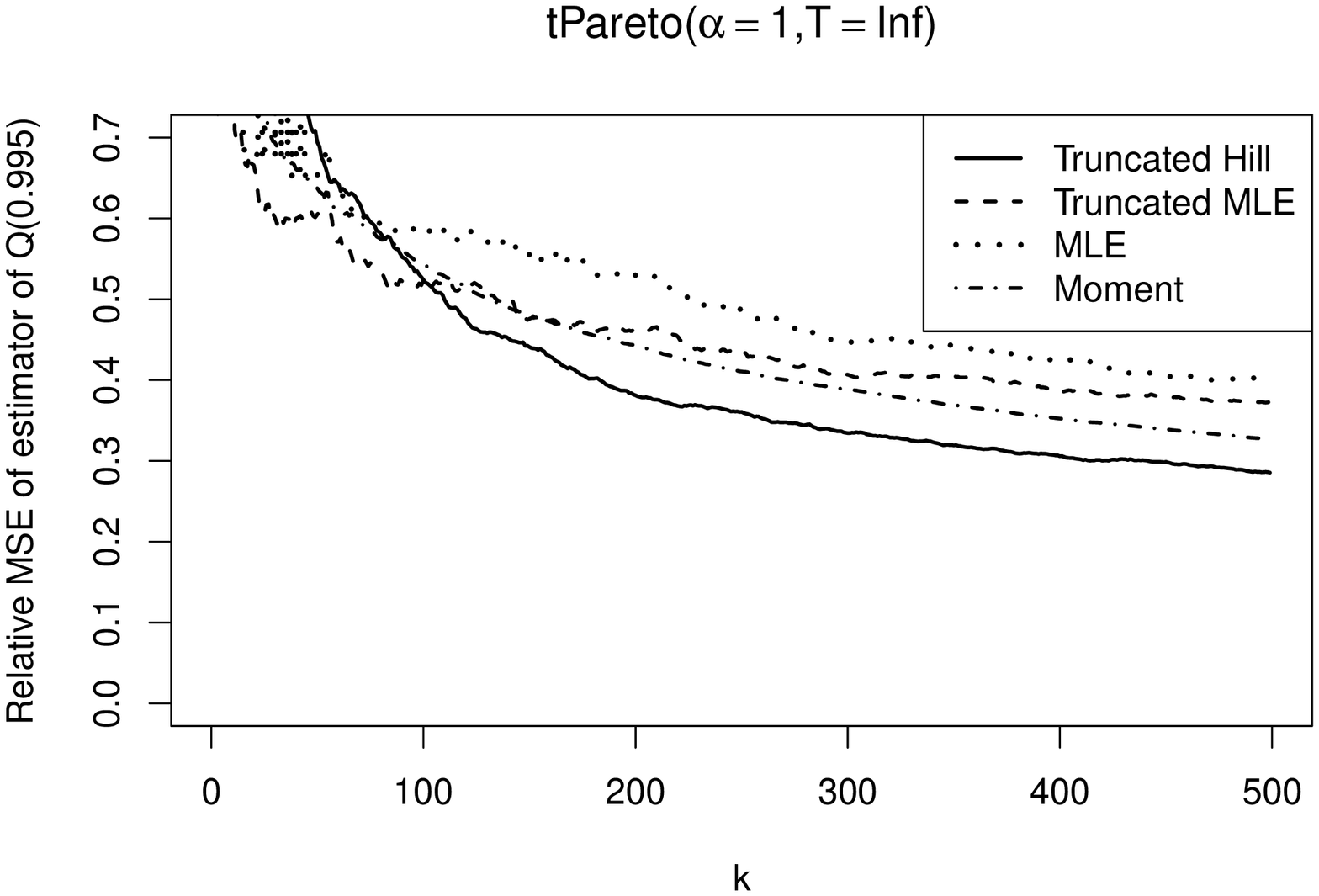}
  \caption{Mean deviations of $\hat{Q}^+_{T,k}(1-p)/Q_T(1-p)$, $\hat{Q}_{T,k}(1-p)/Q_T(1-p)$, $\hat{Q}^{\infty}_{k}(1-p)/Q_T(1-p)$, $\hat{Q}^M_{k}(1-p)/Q_T(1-p)$  and corresponding MSE with $p=0.005$ for the standard Pareto distribution truncated at $Q_Y (0.975)$ (top), $Q_Y (0.99)$ (middle) and non truncated (bottom).}
        \end{figure}

				\newpage													
   \begin{figure}[!ht]
	\centering
	  \includegraphics[height=0.495\textwidth, angle=270]{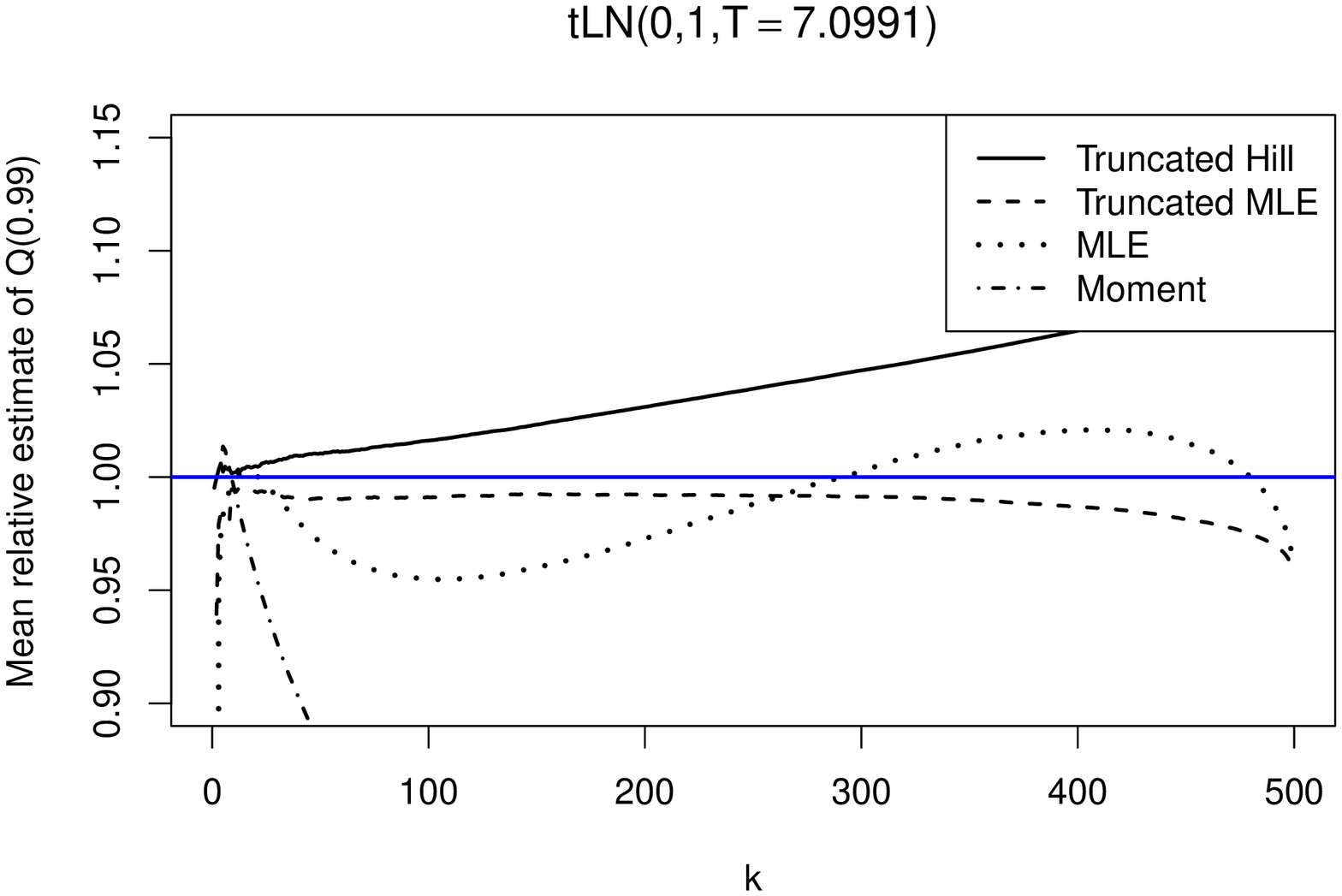}
	 \includegraphics[height=0.495\textwidth, angle=270]{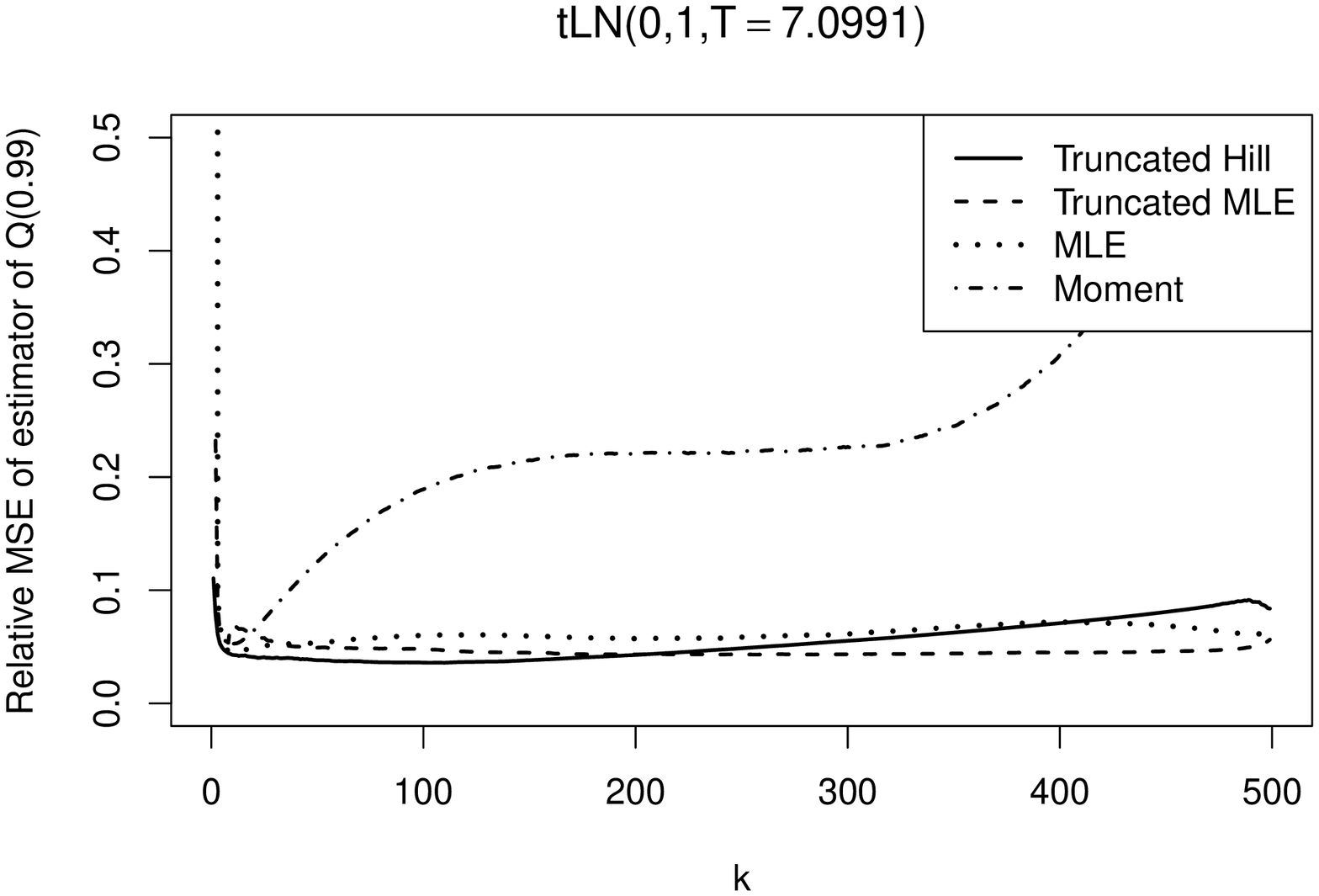}\\
	  \includegraphics[height=0.495\textwidth, angle=270]{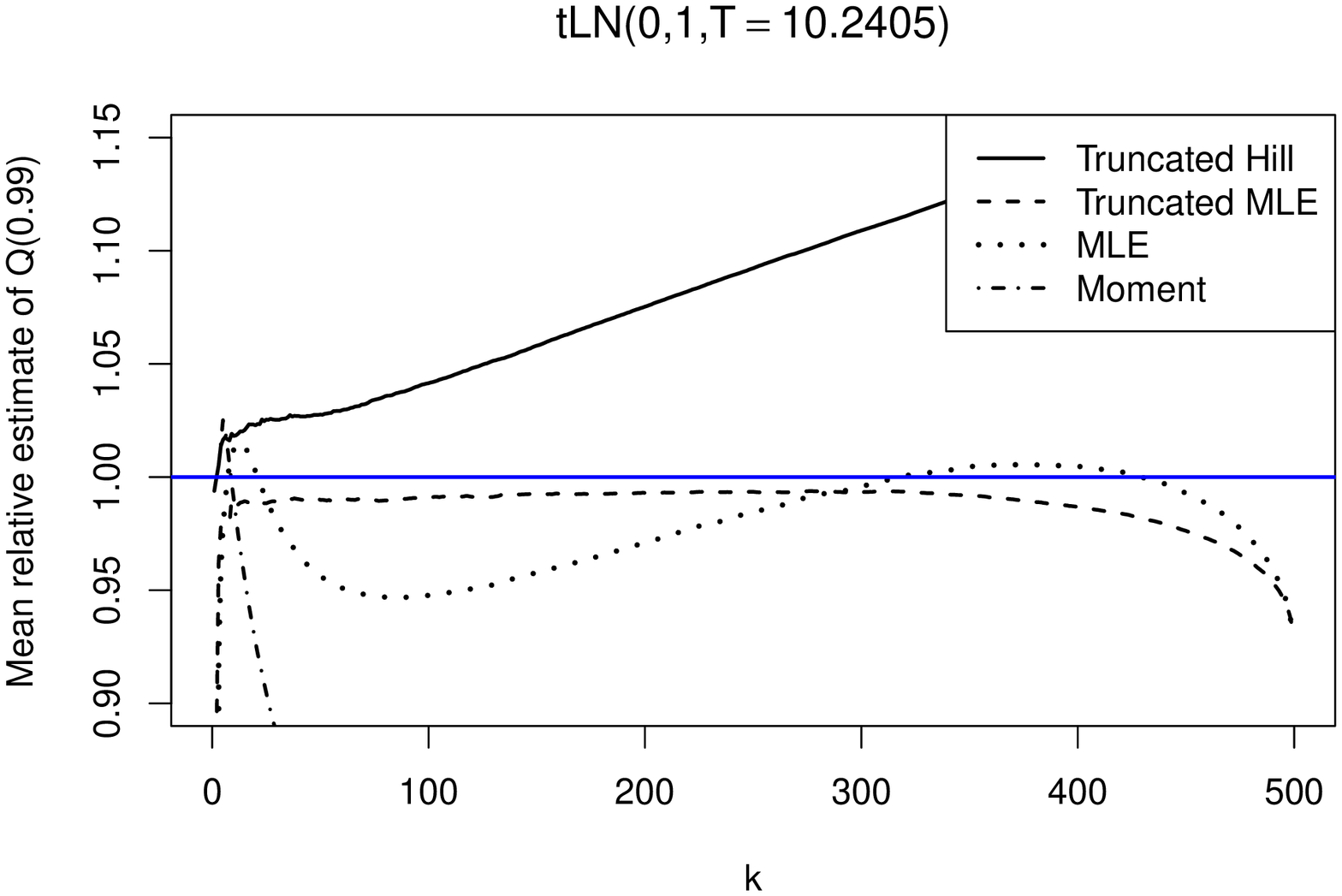}
	 \includegraphics[height=0.495\textwidth, angle=270]{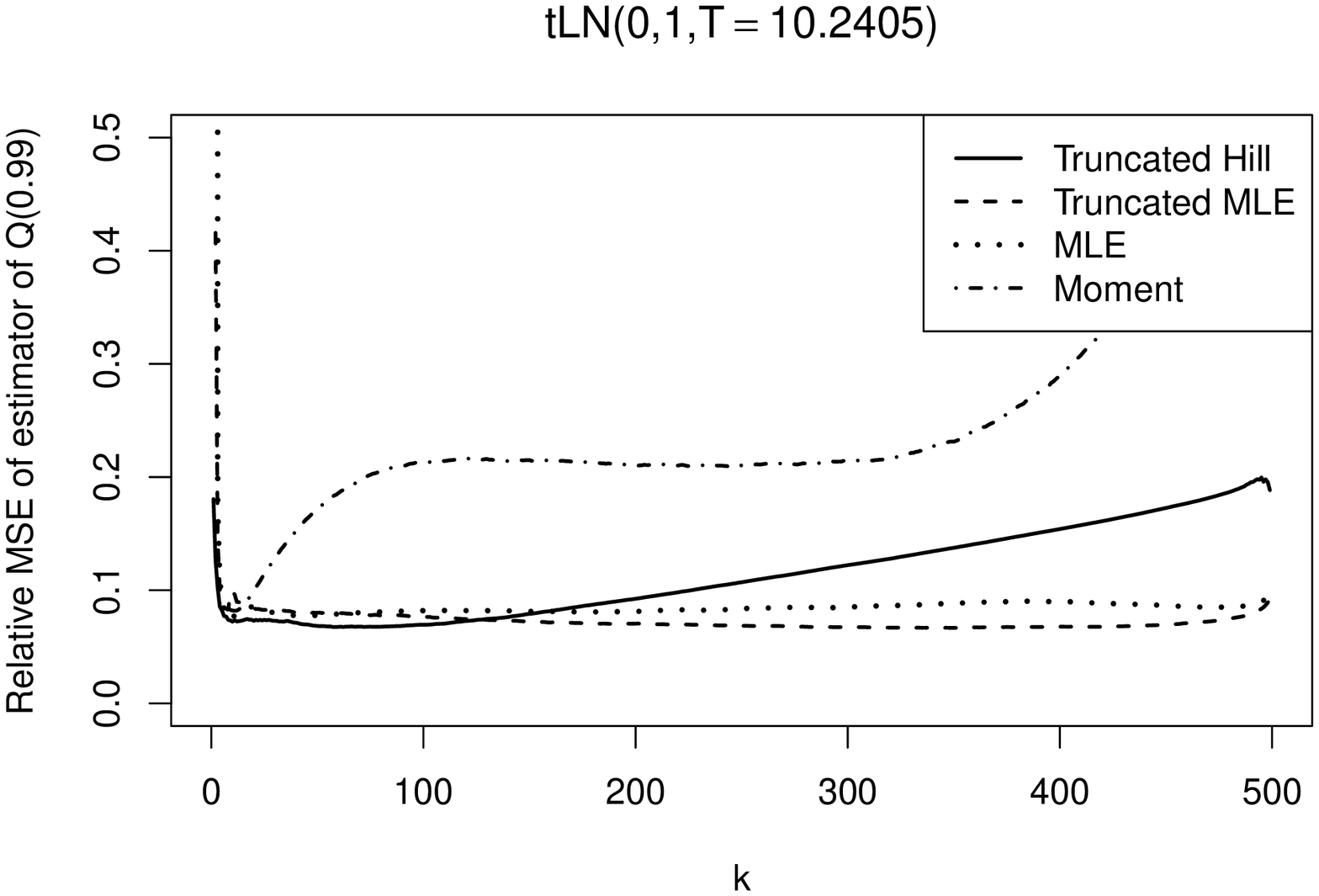}\\
	  \includegraphics[height=0.495\textwidth, angle=270]{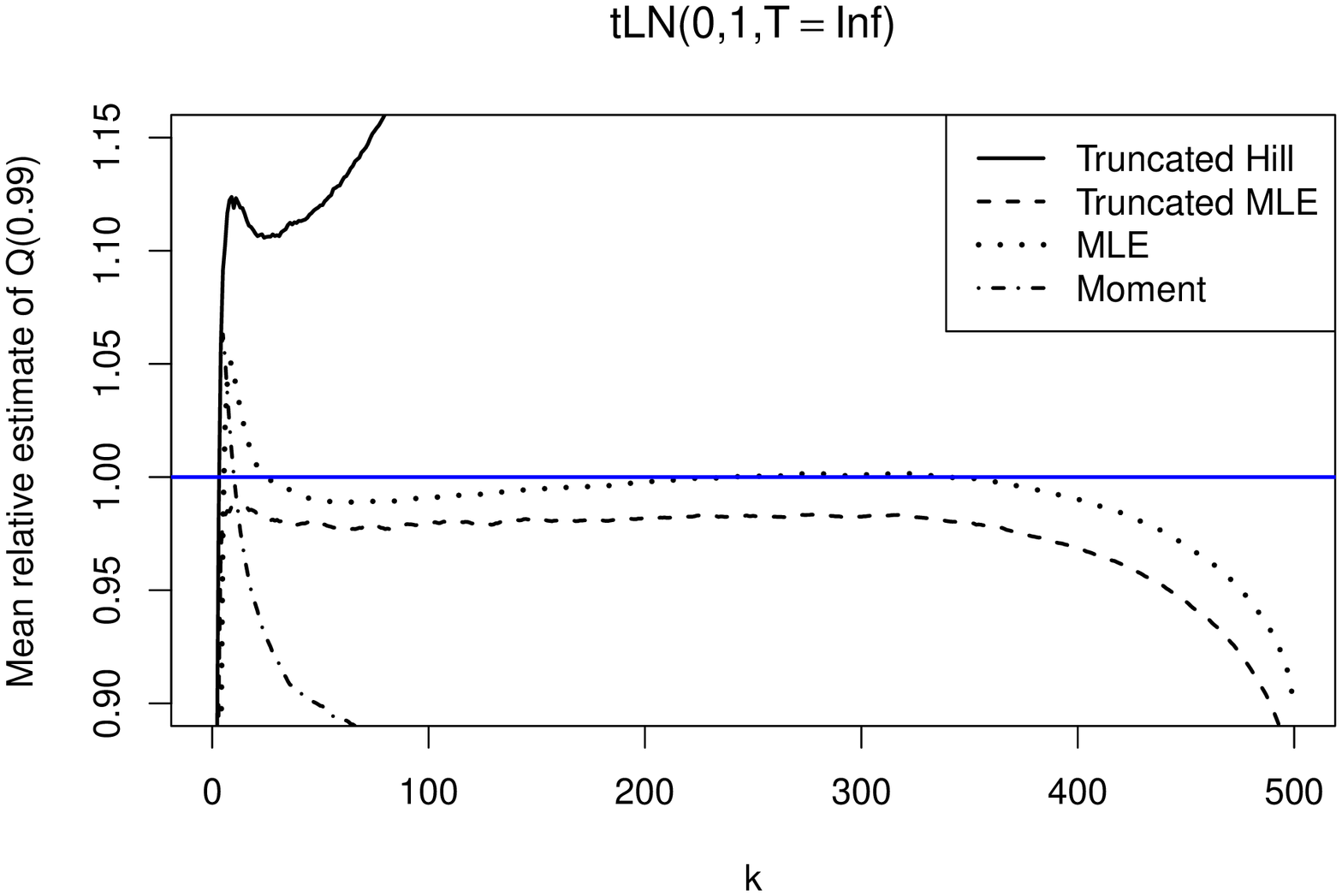}
	 \includegraphics[height=0.495\textwidth, angle=270]{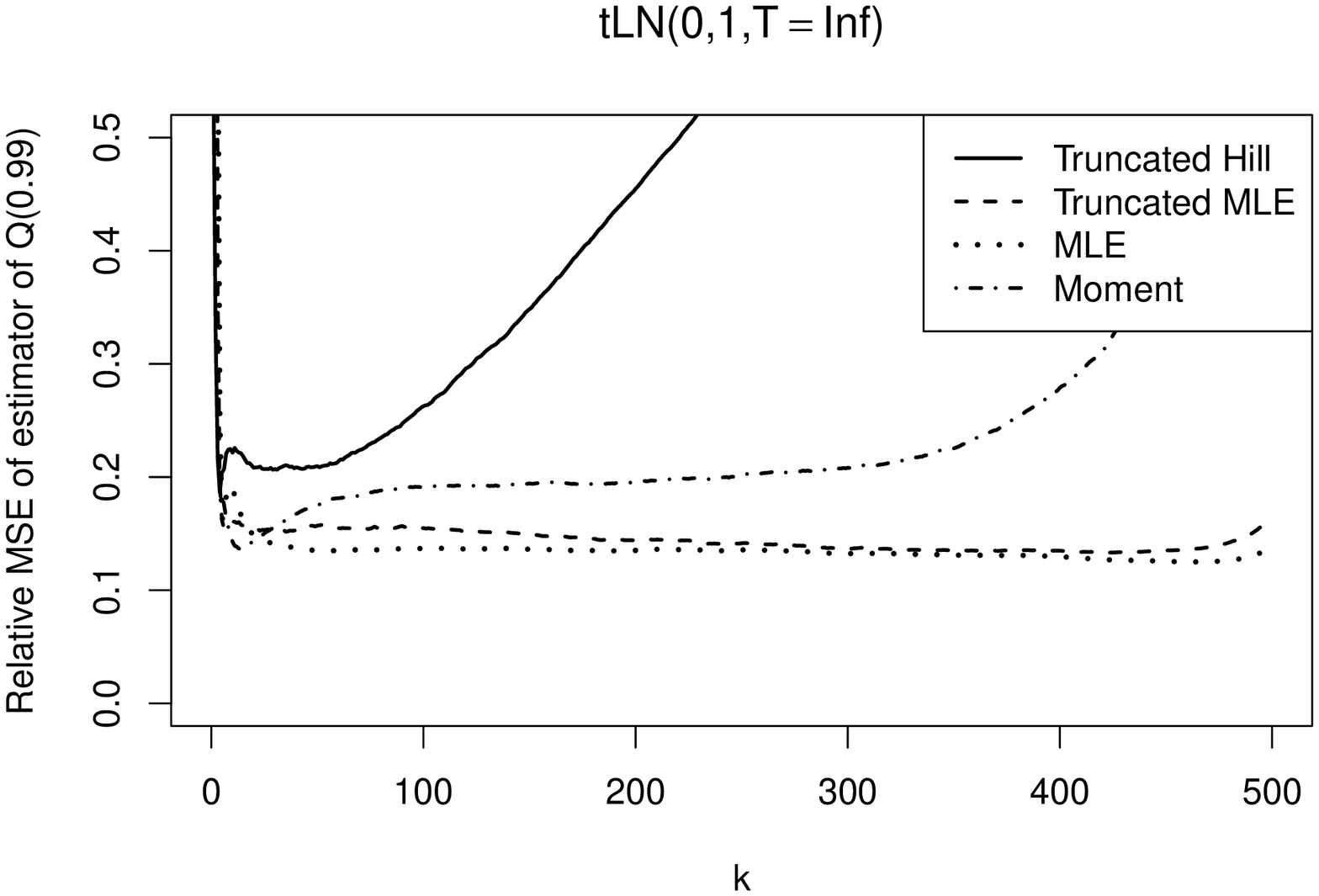}
  \caption{Mean deviations of $\hat{Q}^+_{T,k}(1-p)/Q_T(1-p)$, $\hat{Q}_{T,k}(1-p)/Q_T(1-p)$, $\hat{Q}^{\infty}_{k}(1-p)/Q_T(1-p)$, $\hat{Q}^M_{k}(1-p)/Q_T(1-p)$  and corresponding MSE with $p=0.01$ for the standard lognormal distribution truncated at $Q_Y (0.975)$ (top), $Q_Y (0.99)$ (middle) and non truncated (bottom).}
   \end{figure}
	
			\newpage	
   \begin{figure}[!ht]
	\centering
	  \includegraphics[height=0.495\textwidth, angle=270]{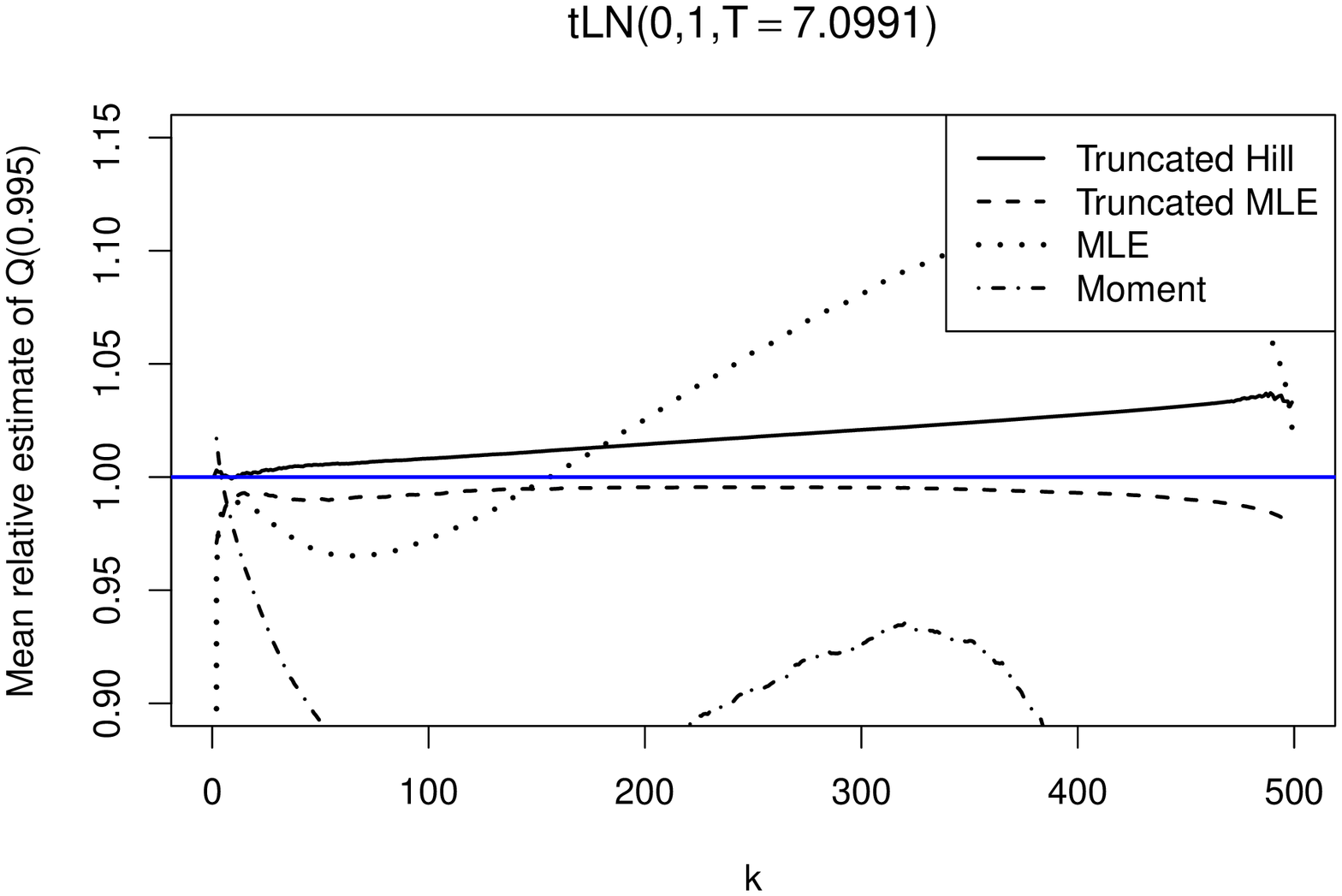}
	 \includegraphics[height=0.495\textwidth, angle=270]{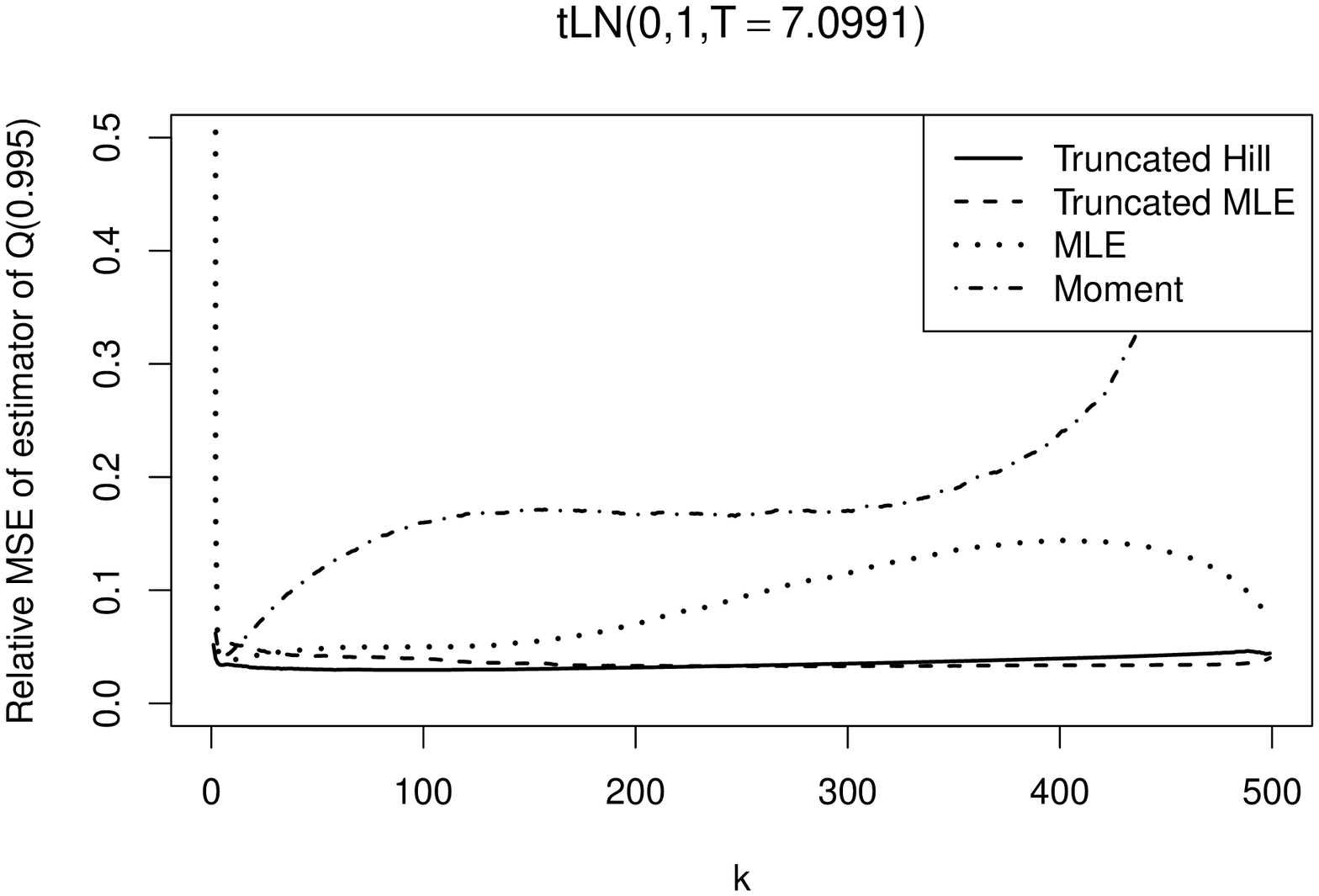}\\
	  \includegraphics[height=0.495\textwidth, angle=270]{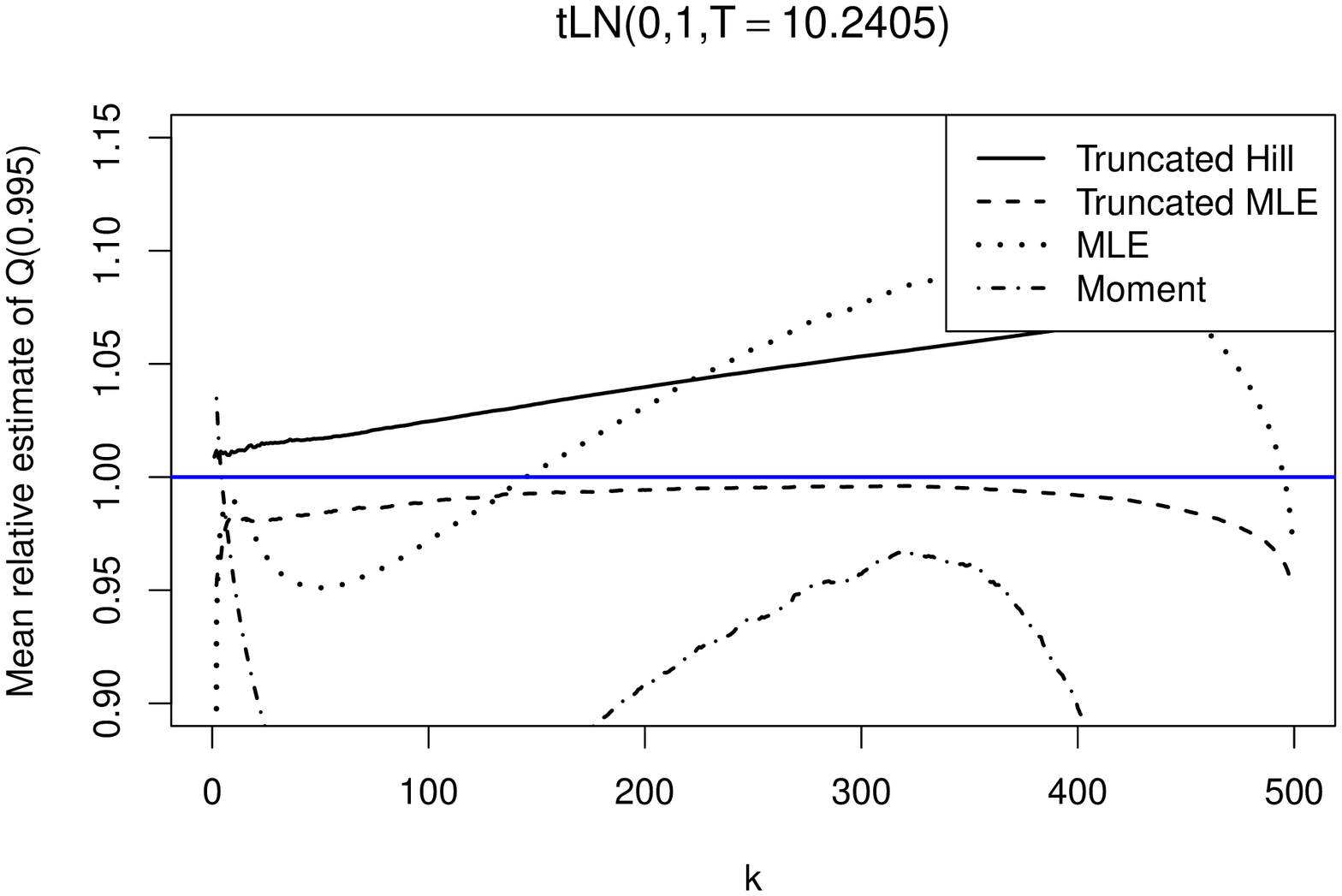}
	 \includegraphics[height=0.495\textwidth, angle=270]{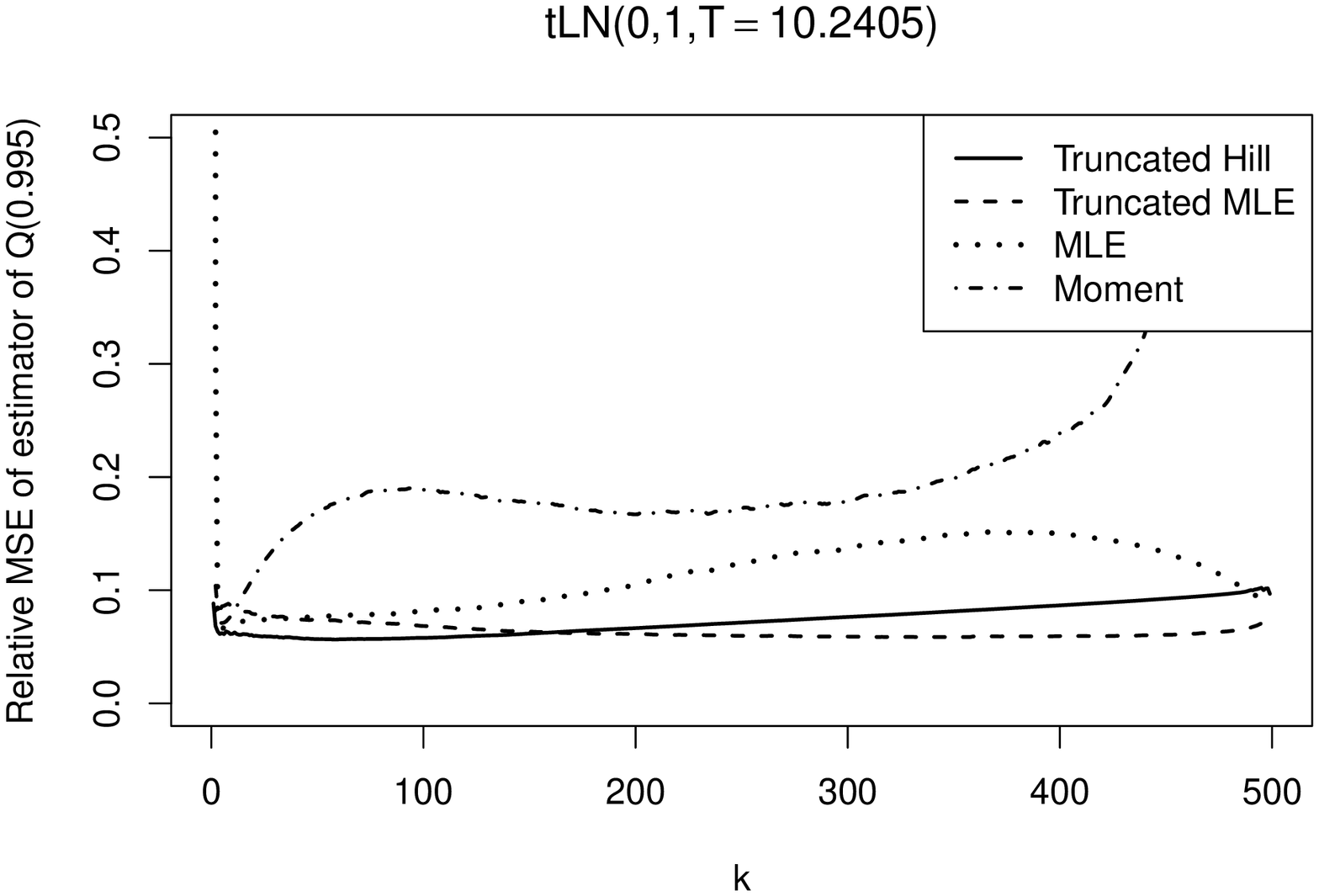}\\
	  \includegraphics[height=0.495\textwidth, angle=270]{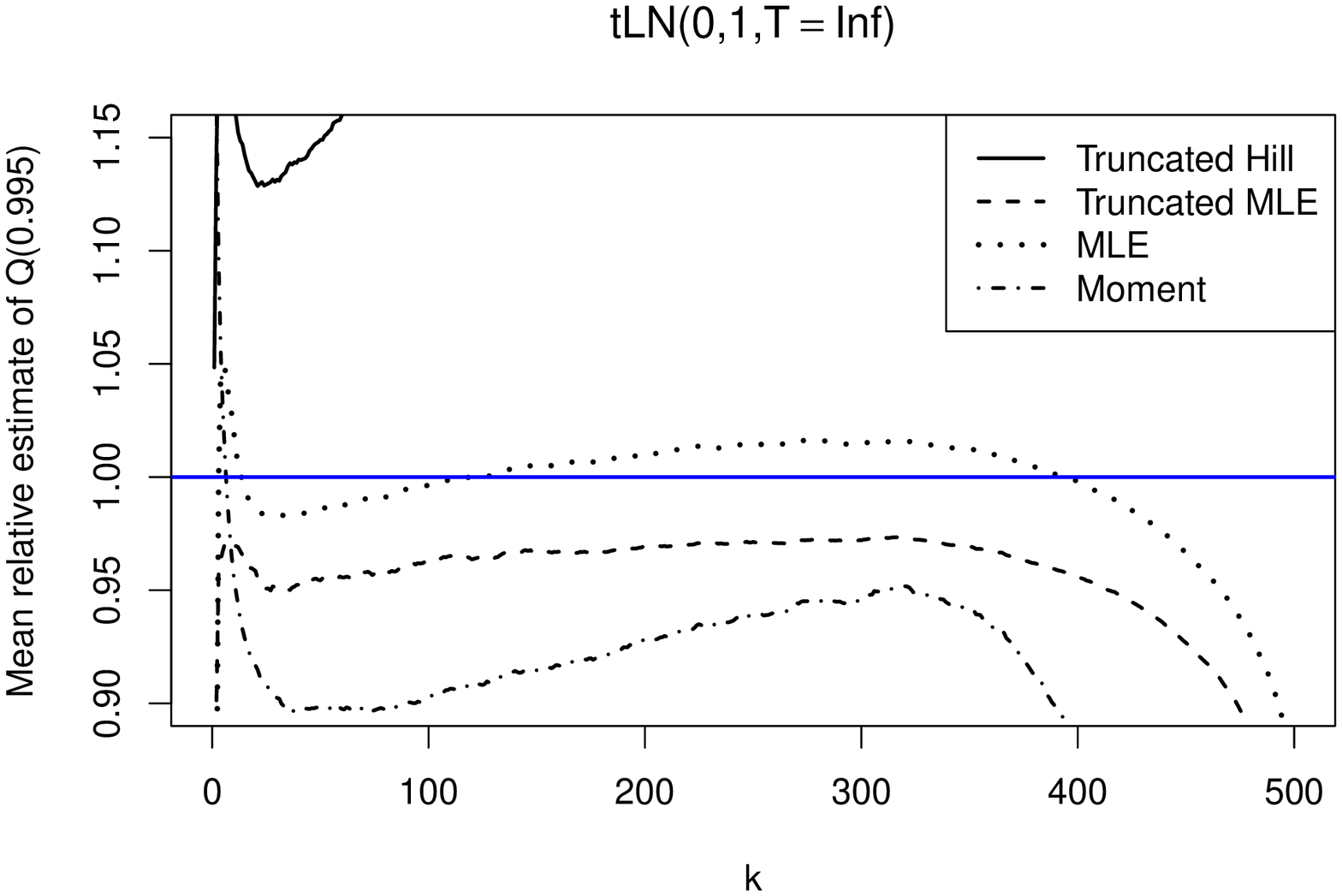}
	 \includegraphics[height=0.495\textwidth, angle=270]{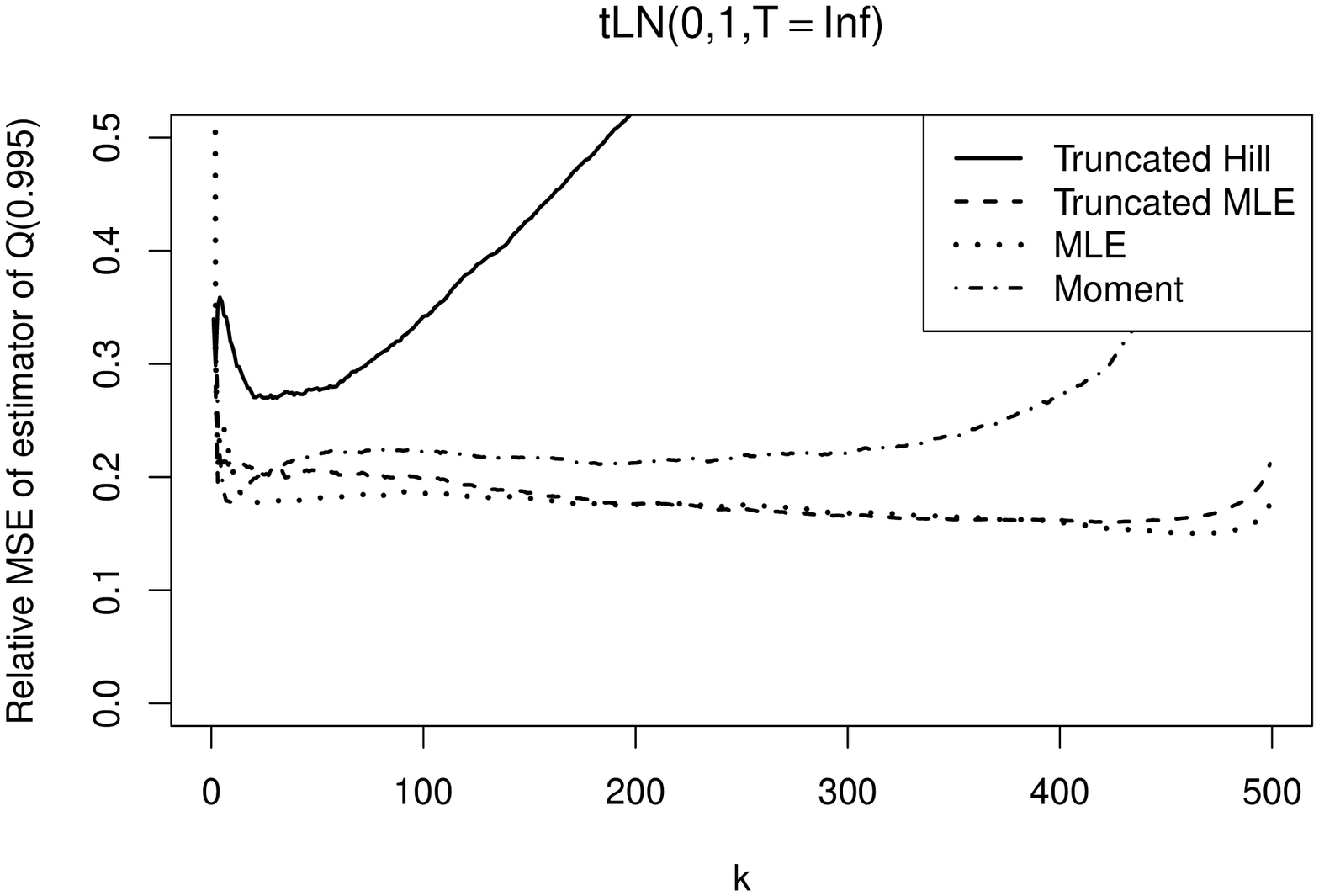}
	  \caption{Mean deviations of $\hat{Q}^+_{T,k}(1-p)/Q_T(1-p)$, $\hat{Q}_{T,k}(1-p)/Q_T(1-p)$, $\hat{Q}^{\infty}_{k}(1-p)/Q_T(1-p)$, $\hat{Q}^M_{k}(1-p)/Q_T(1-p)$  and corresponding MSE with $p=0.005$ for the standard lognormal distribution truncated at $Q_Y (0.975)$ (top), $Q_Y (0.99)$ (middle) and non truncated (bottom).}
        \end{figure}
				
				\newpage			
   \begin{figure}[!ht]
	\centering
	  \includegraphics[height=0.495\textwidth, angle=270]{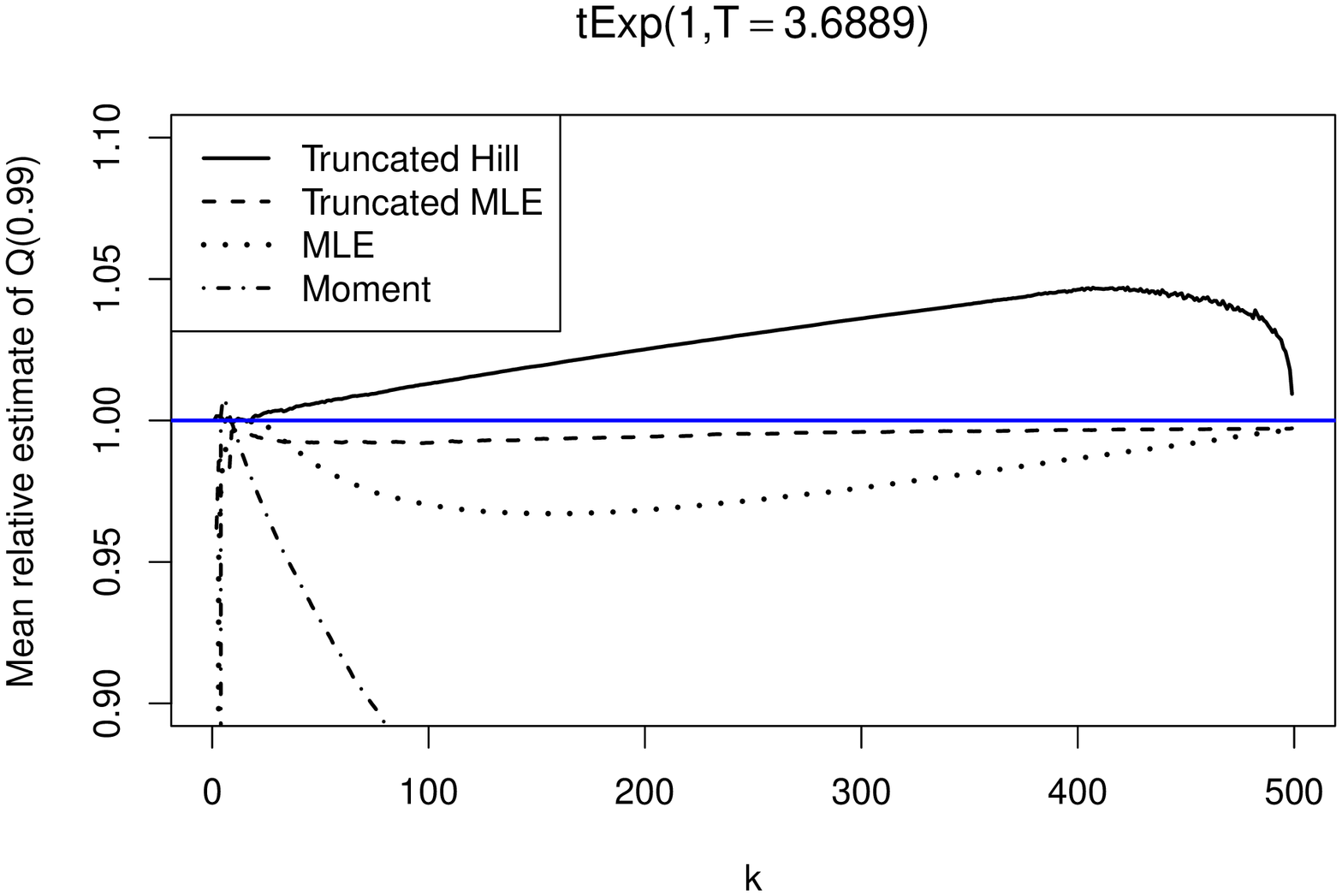}
	 \includegraphics[height=0.495\textwidth, angle=270]{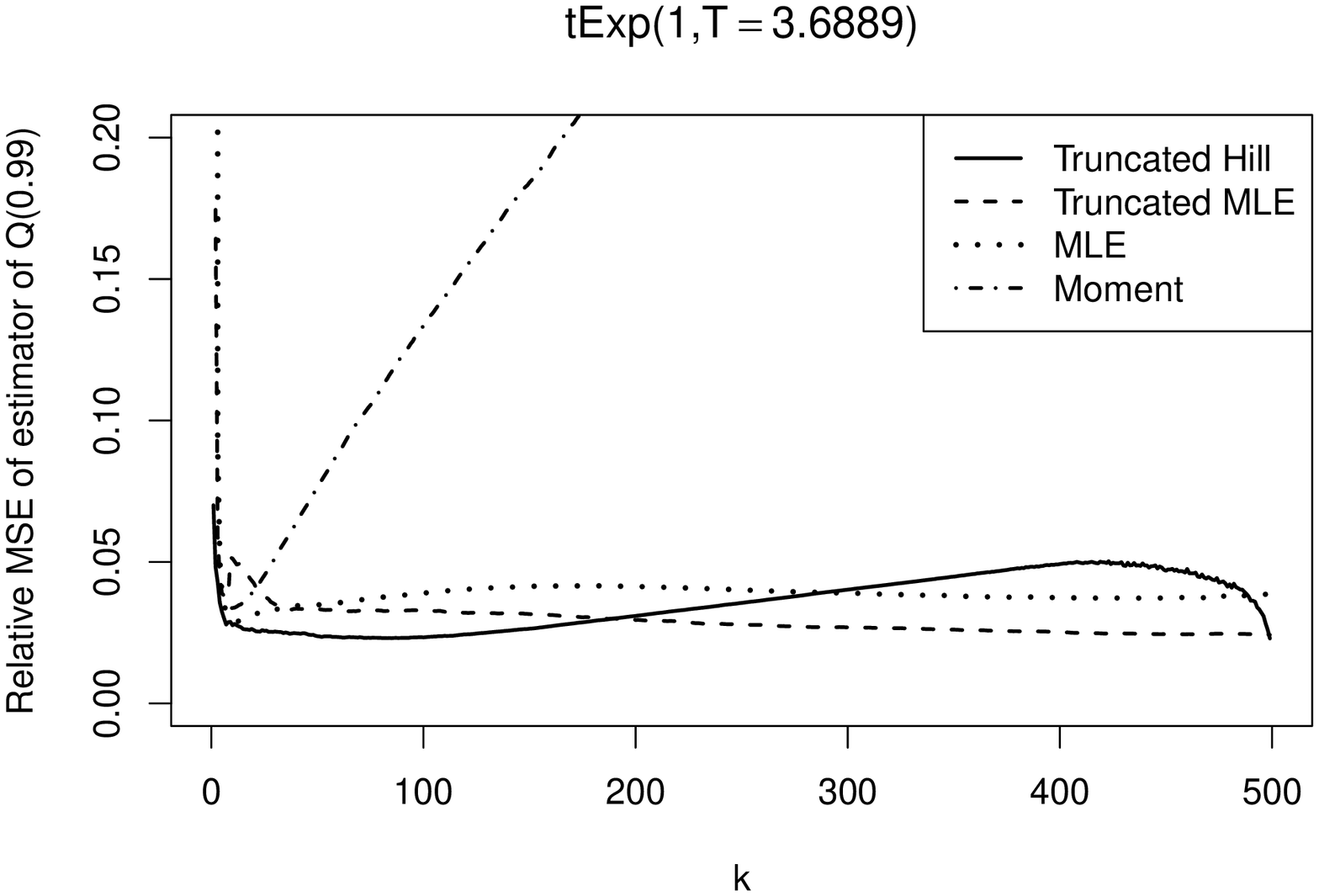}\\
	  \includegraphics[height=0.495\textwidth, angle=270]{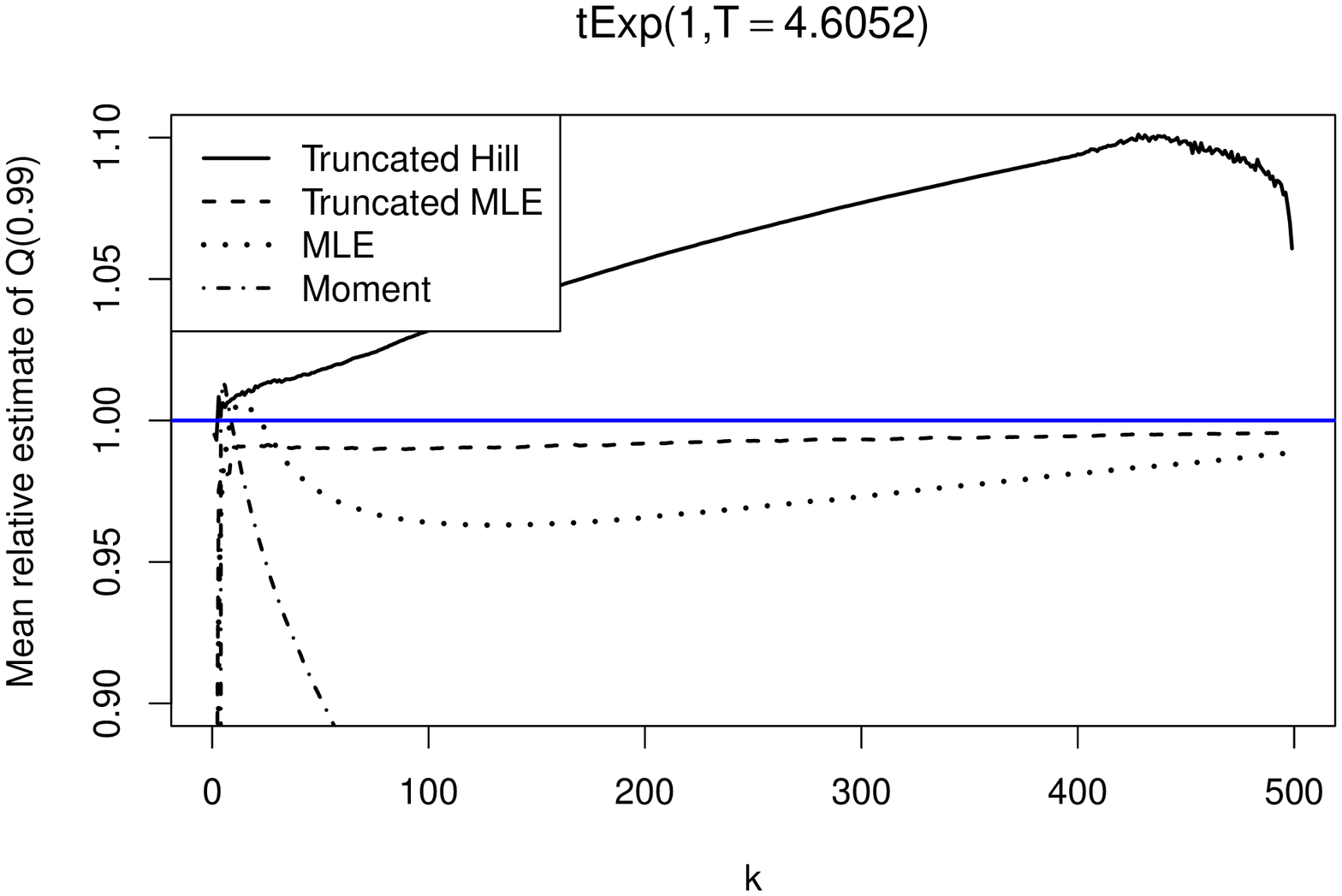}
	 \includegraphics[height=0.495\textwidth, angle=270]{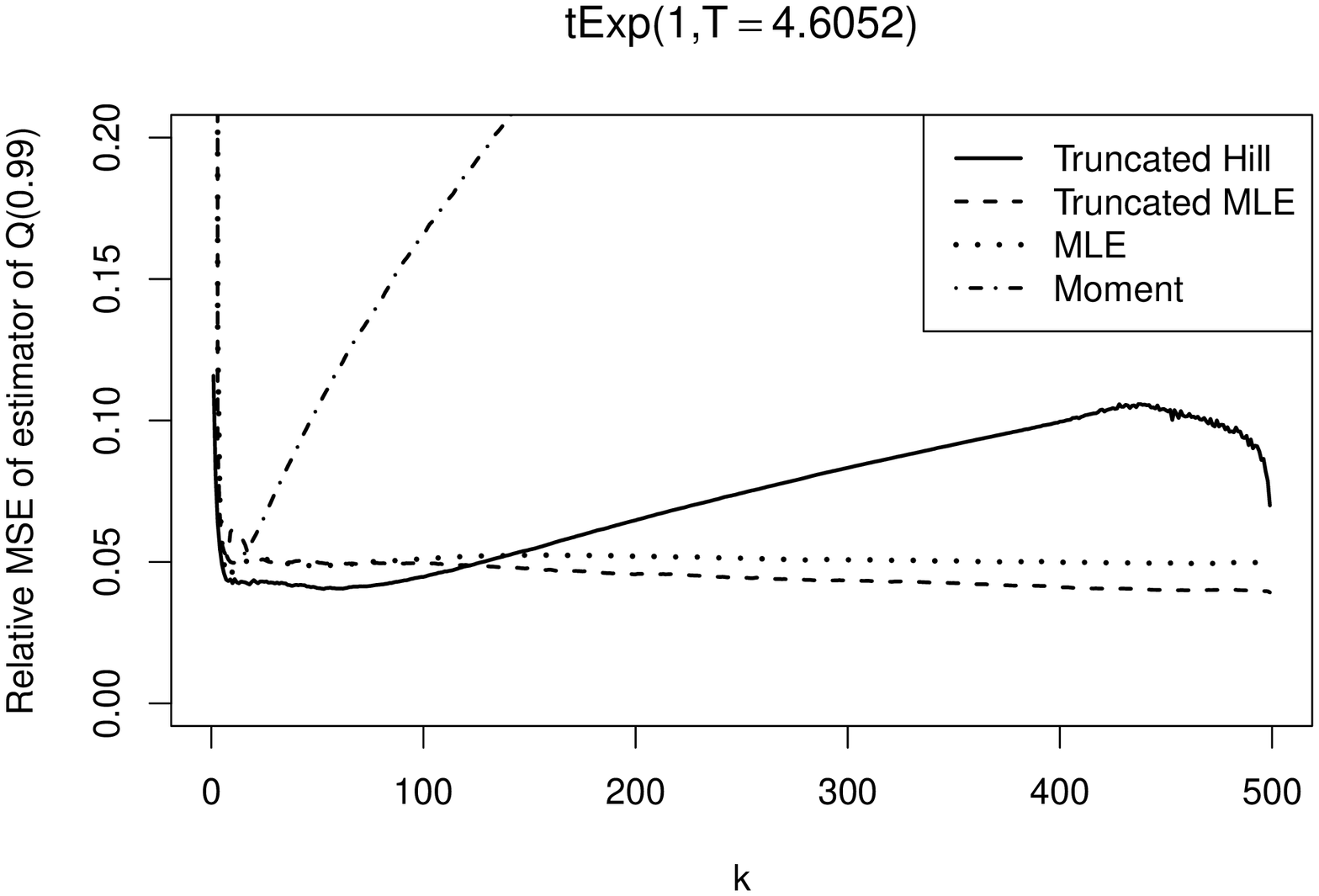}\\
	  \includegraphics[height=0.495\textwidth, angle=270]{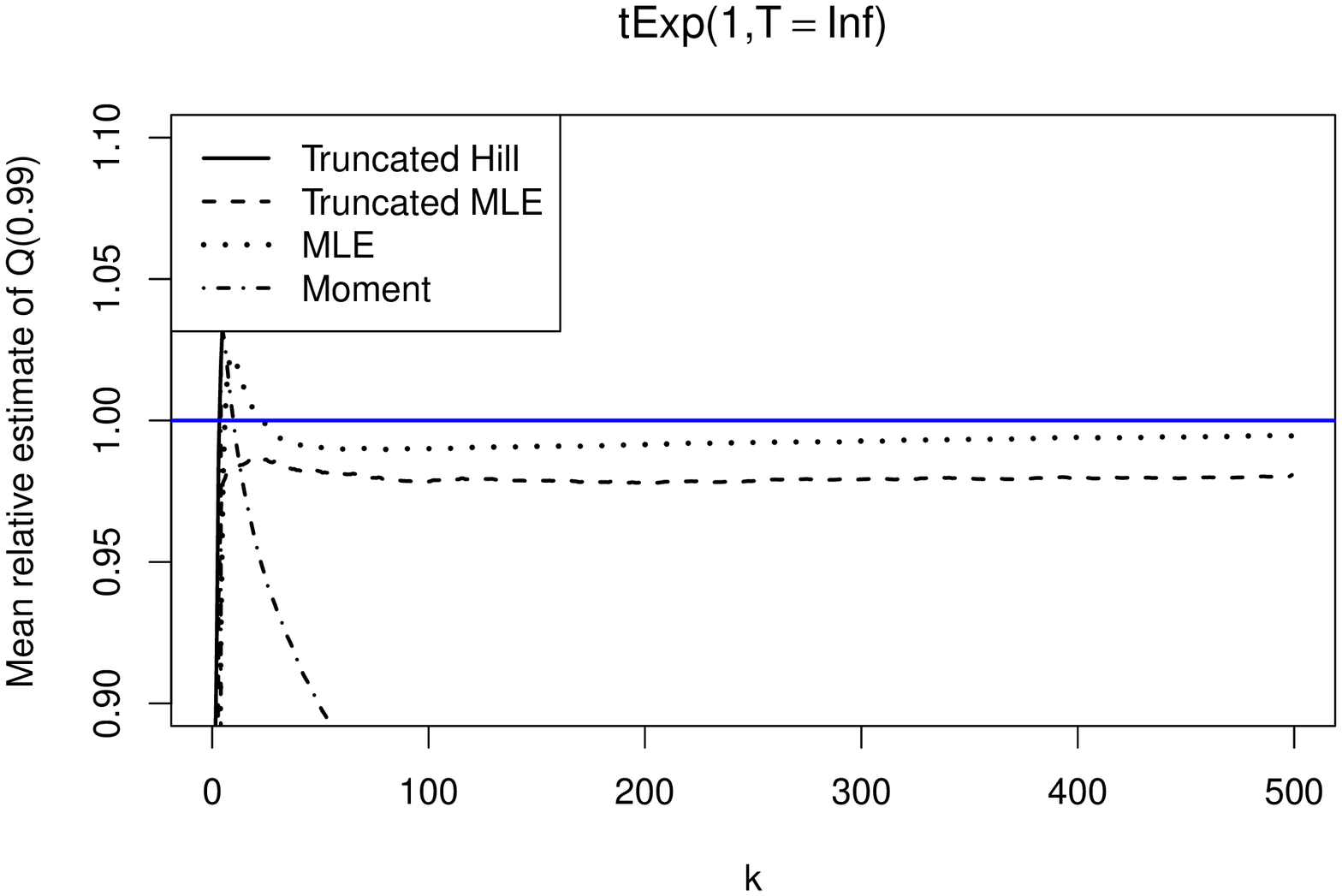}
	 \includegraphics[height=0.495\textwidth, angle=270]{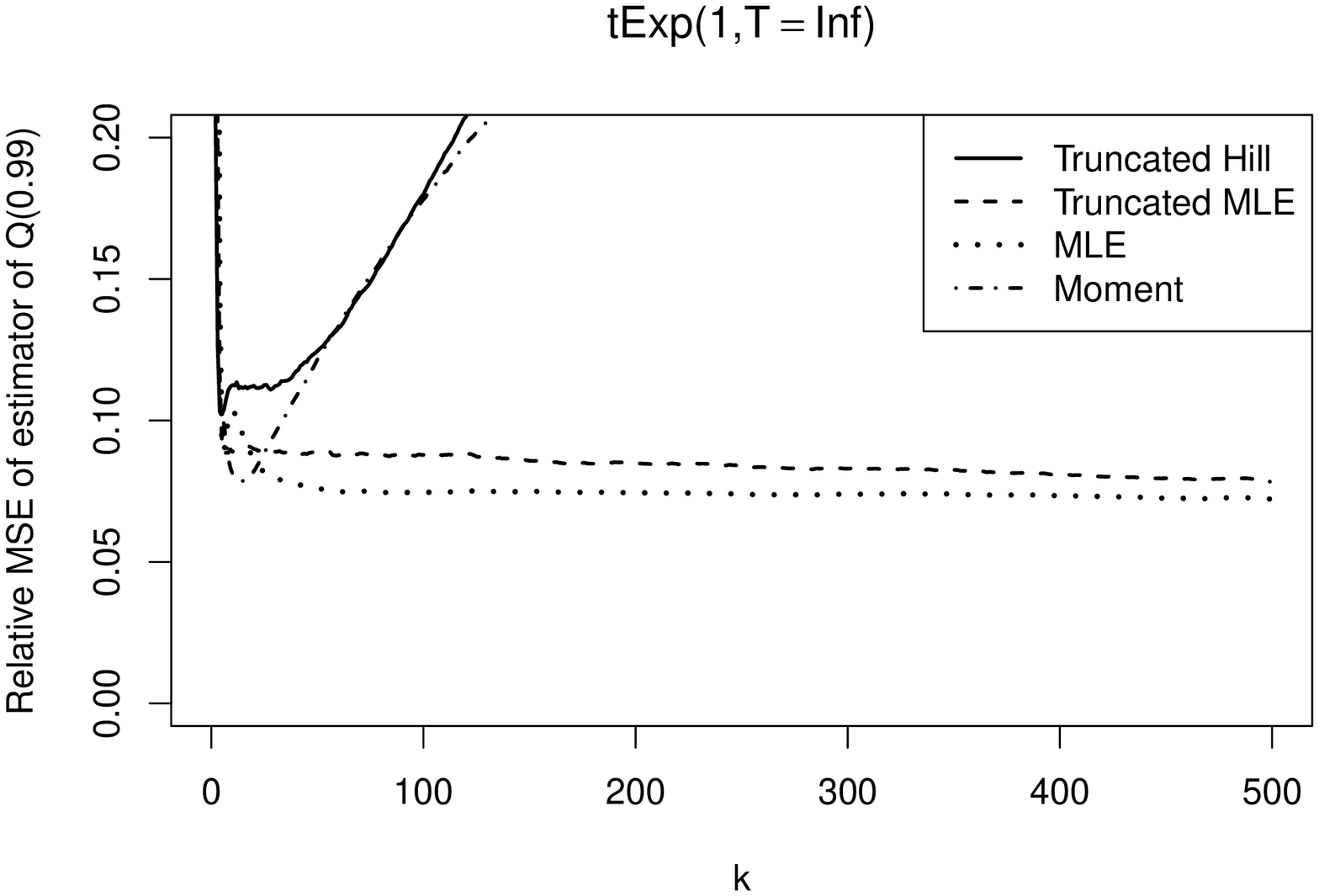}
  \caption{Mean deviations of $\hat{Q}^+_{T,k}(1-p)/Q_T(1-p)$, $\hat{Q}_{T,k}(1-p)/Q_T(1-p)$, $\hat{Q}^{\infty}_{k}(1-p)/Q_T(1-p)$, $\hat{Q}^M_{k}(1-p)/Q_T(1-p)$  and corresponding MSE with $p=0.01$ for the standard exponential distribution truncated at $Q_Y (0.975)$ (top), $Q_Y (0.99)$ (middle) and non truncated (bottom).}
   \end{figure}
				
			\newpage							
   \begin{figure}[!ht]
	\centering
	  \includegraphics[height=0.495\textwidth, angle=270]{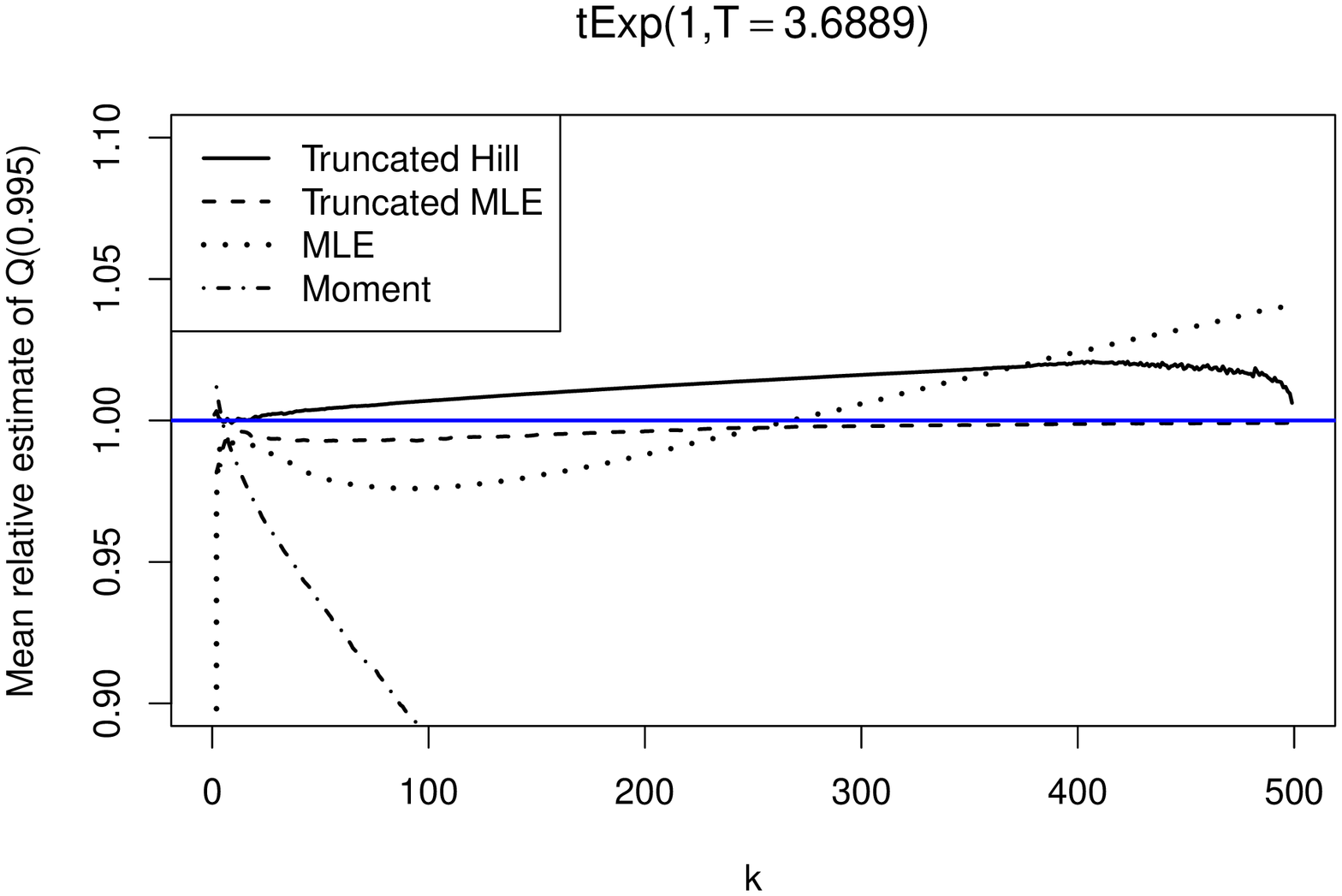}
	 \includegraphics[height=0.495\textwidth, angle=270]{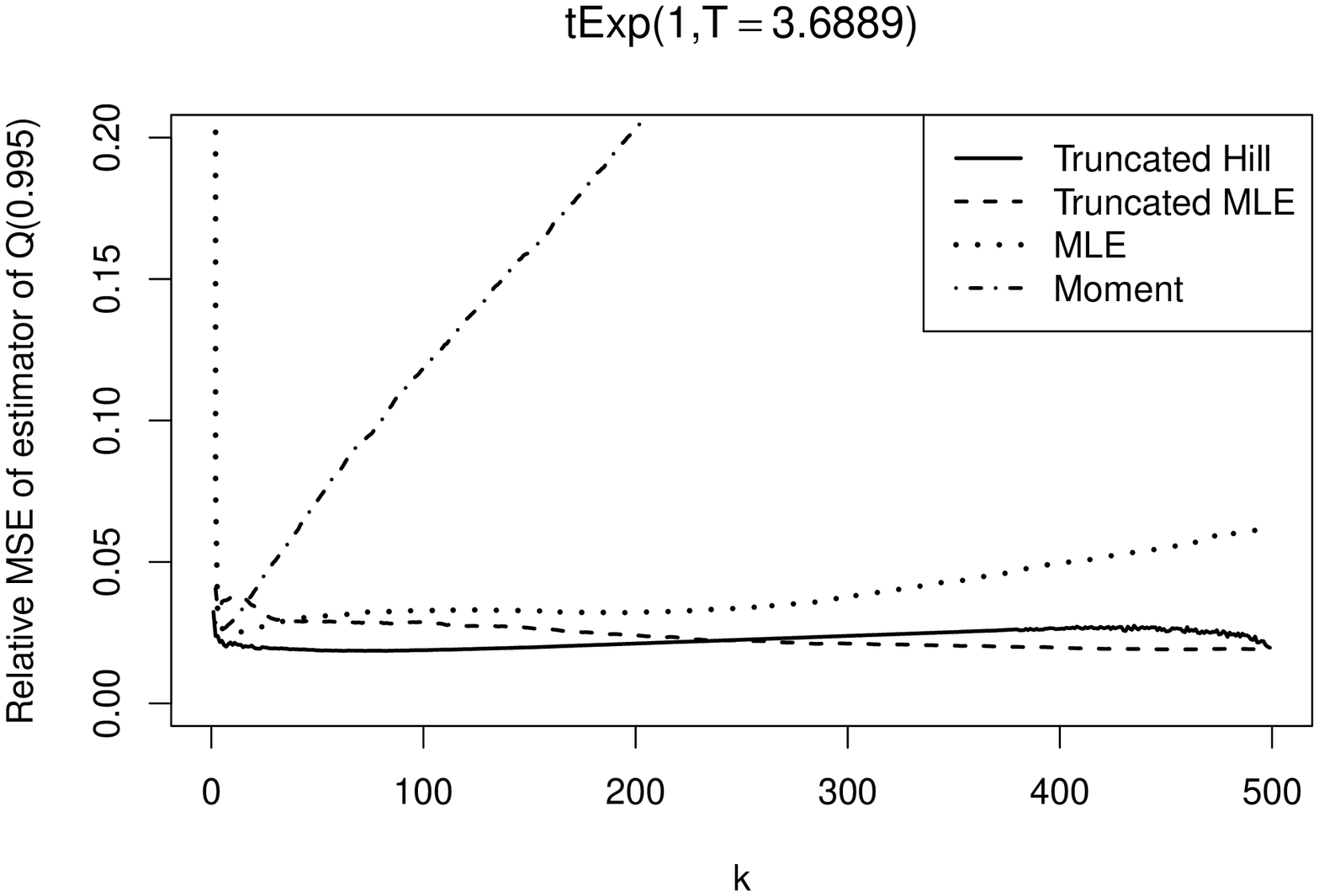}\\
	  \includegraphics[height=0.495\textwidth, angle=270]{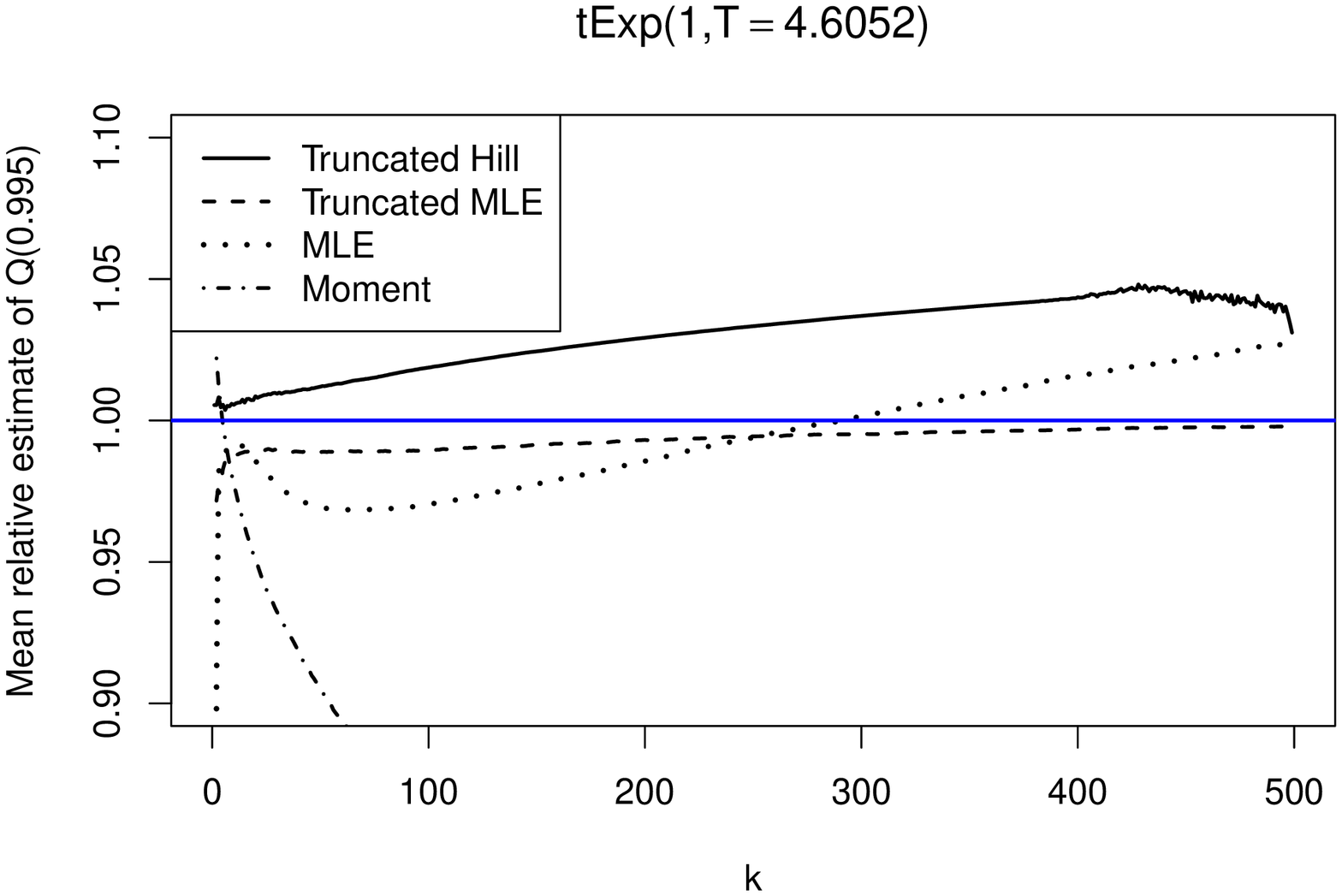}
	 \includegraphics[height=0.495\textwidth, angle=270]{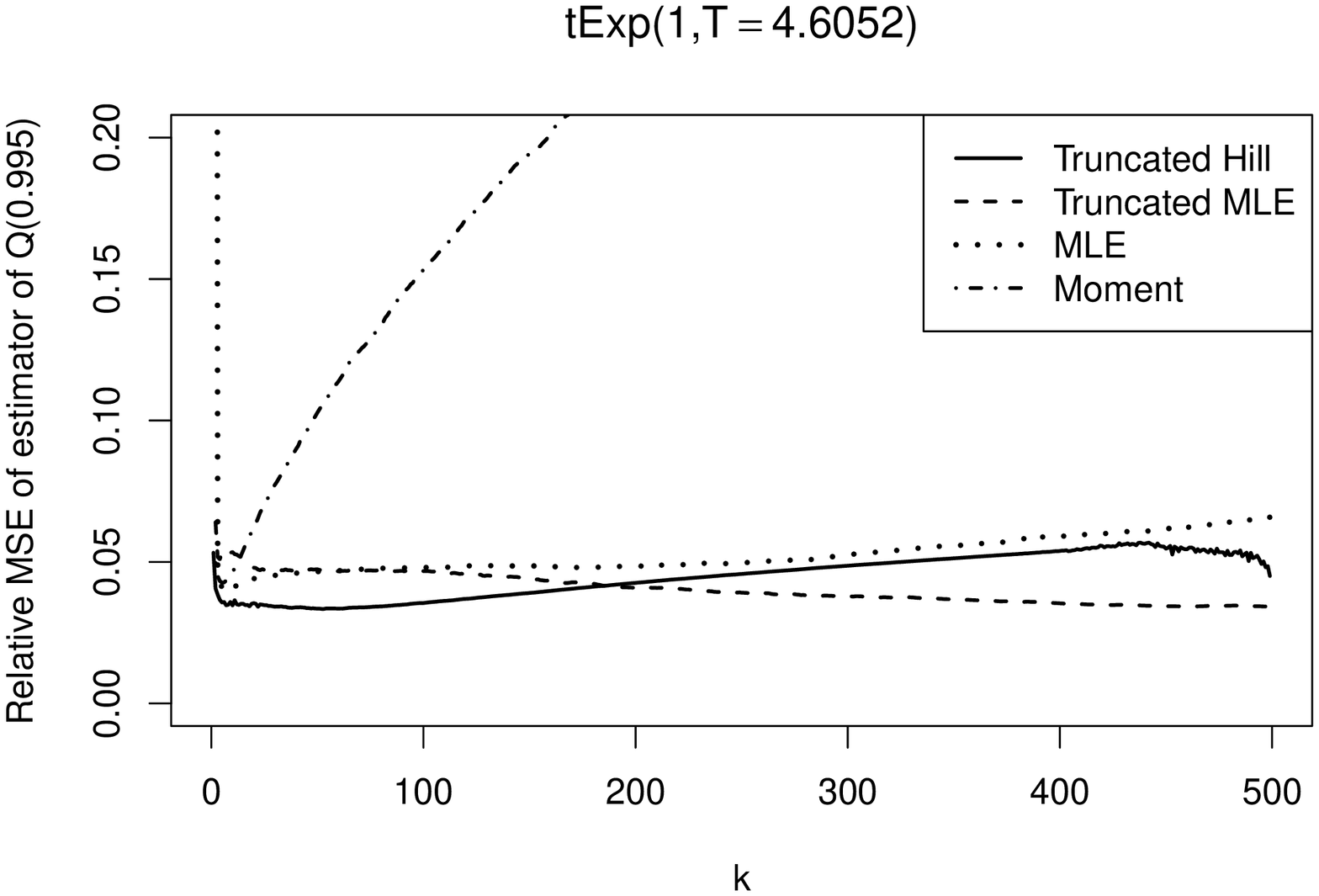}\\
	  \includegraphics[height=0.495\textwidth, angle=270]{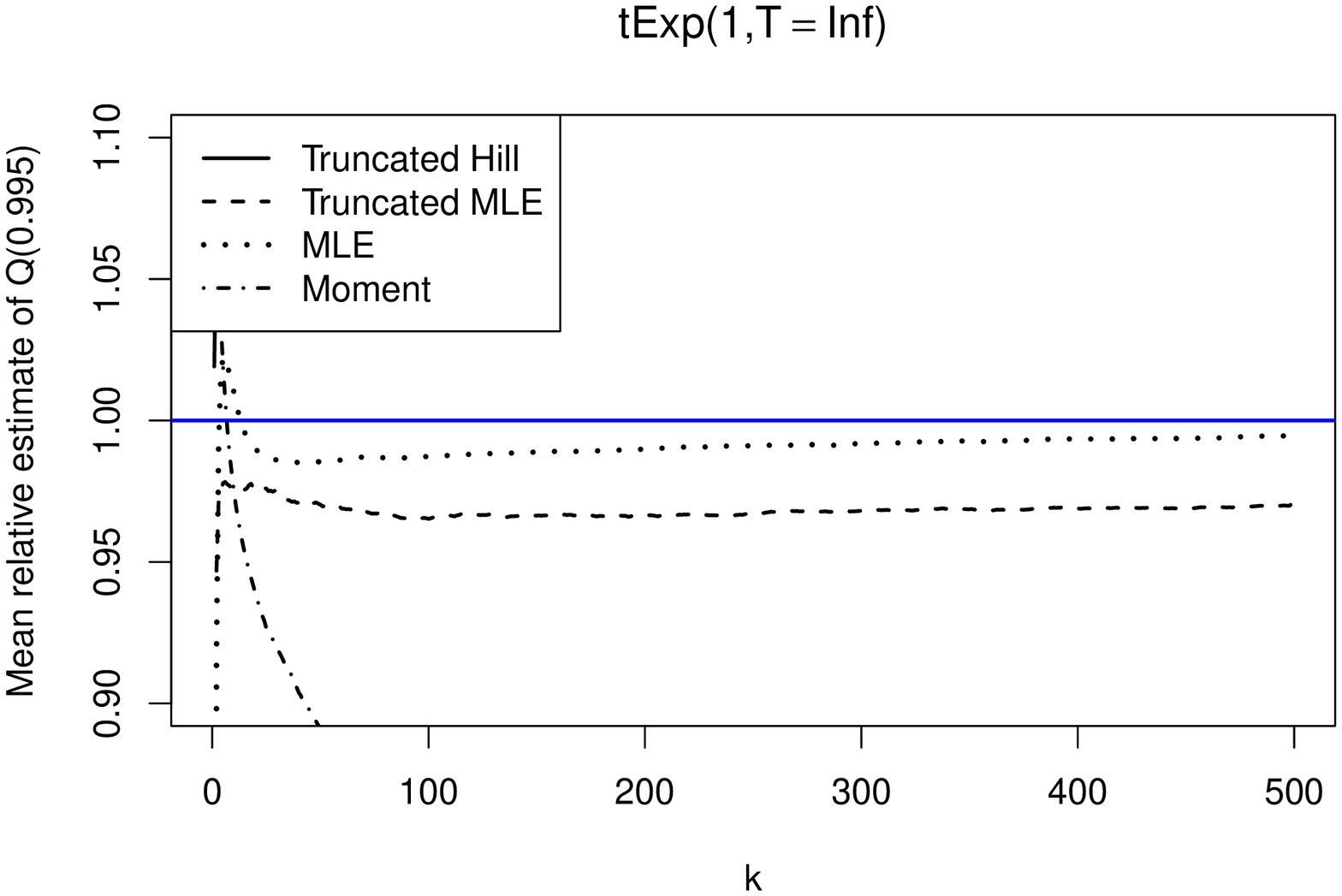}
	 \includegraphics[height=0.495\textwidth, angle=270]{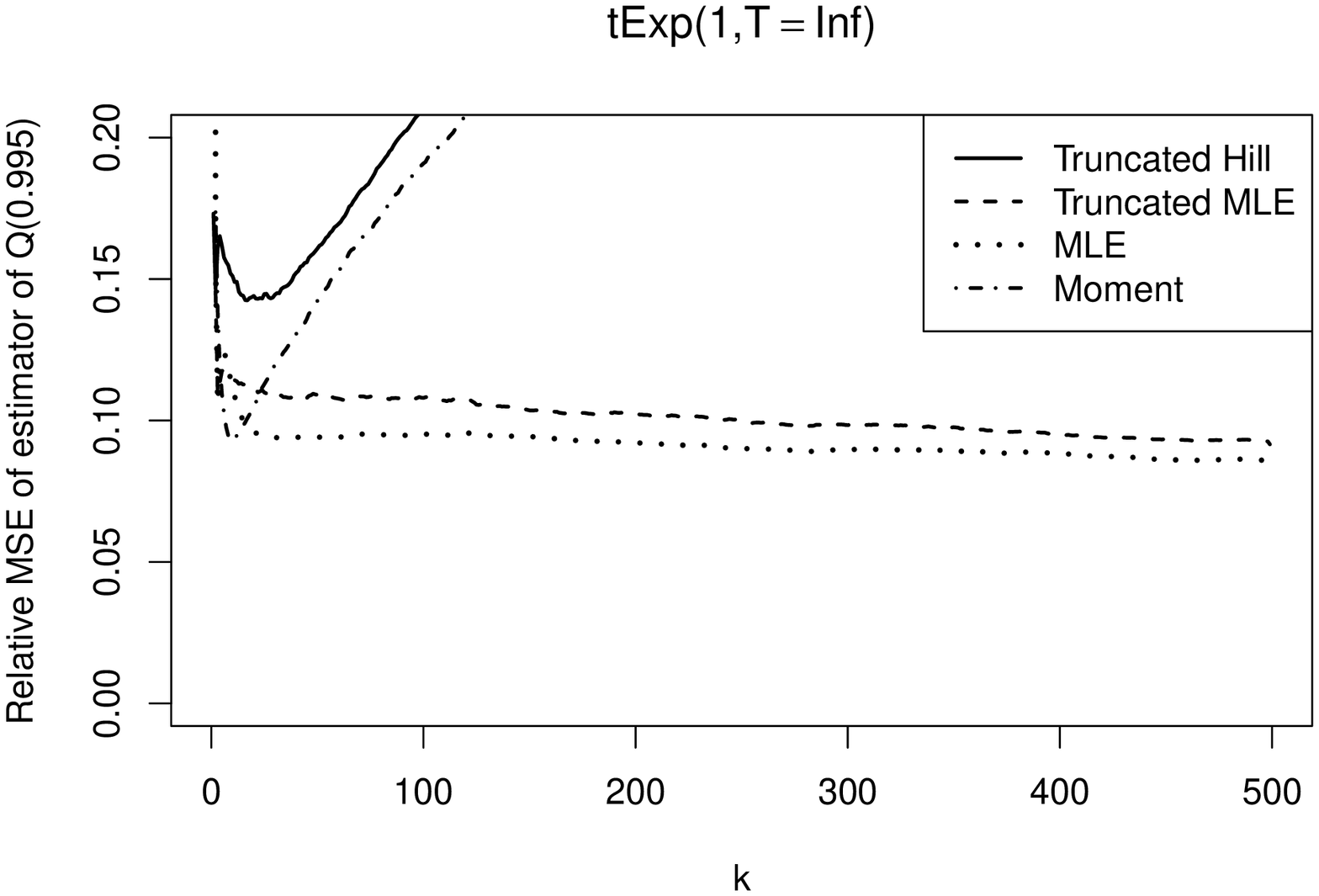}
	  \caption{Mean deviations of $\hat{Q}^+_{T,k}(1-p)/Q_T(1-p)$, $\hat{Q}_{T,k}(1-p)/Q_T(1-p)$, $\hat{Q}^{\infty}_{k}(1-p)/Q_T(1-p)$, $\hat{Q}^M_{k}(1-p)/Q_T(1-p)$  and corresponding MSE with $p=0.005$ for the standard exponential distribution truncated at $Q_Y (0.975)$ (top), $Q_Y (0.99)$ (middle) and non truncated (bottom).}
        \end{figure}

				\newpage						
	\begin{figure}[!ht]
	\centering
  \includegraphics[height=0.495\textwidth, angle=270]{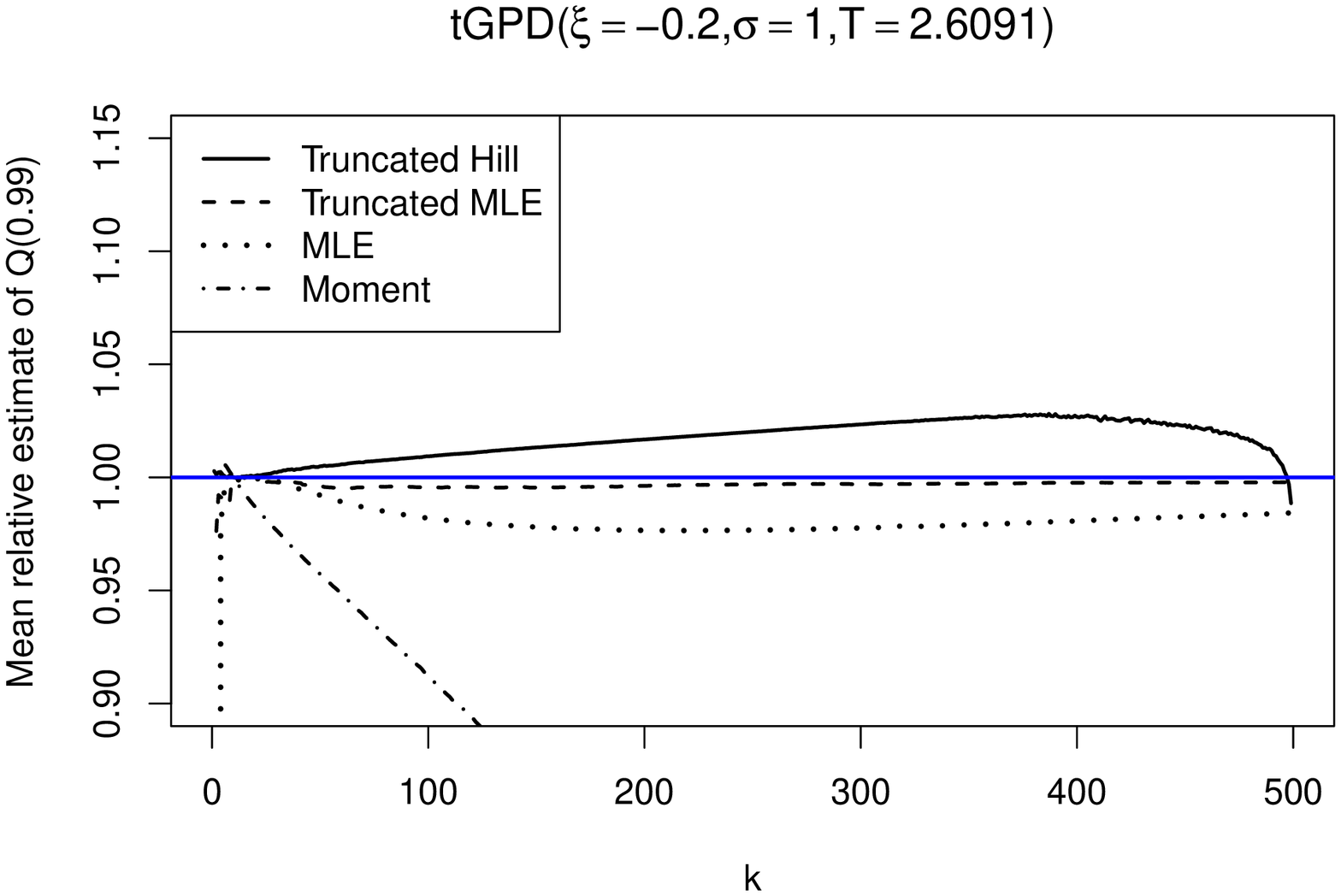}
	 \includegraphics[height=0.495\textwidth, angle=270]{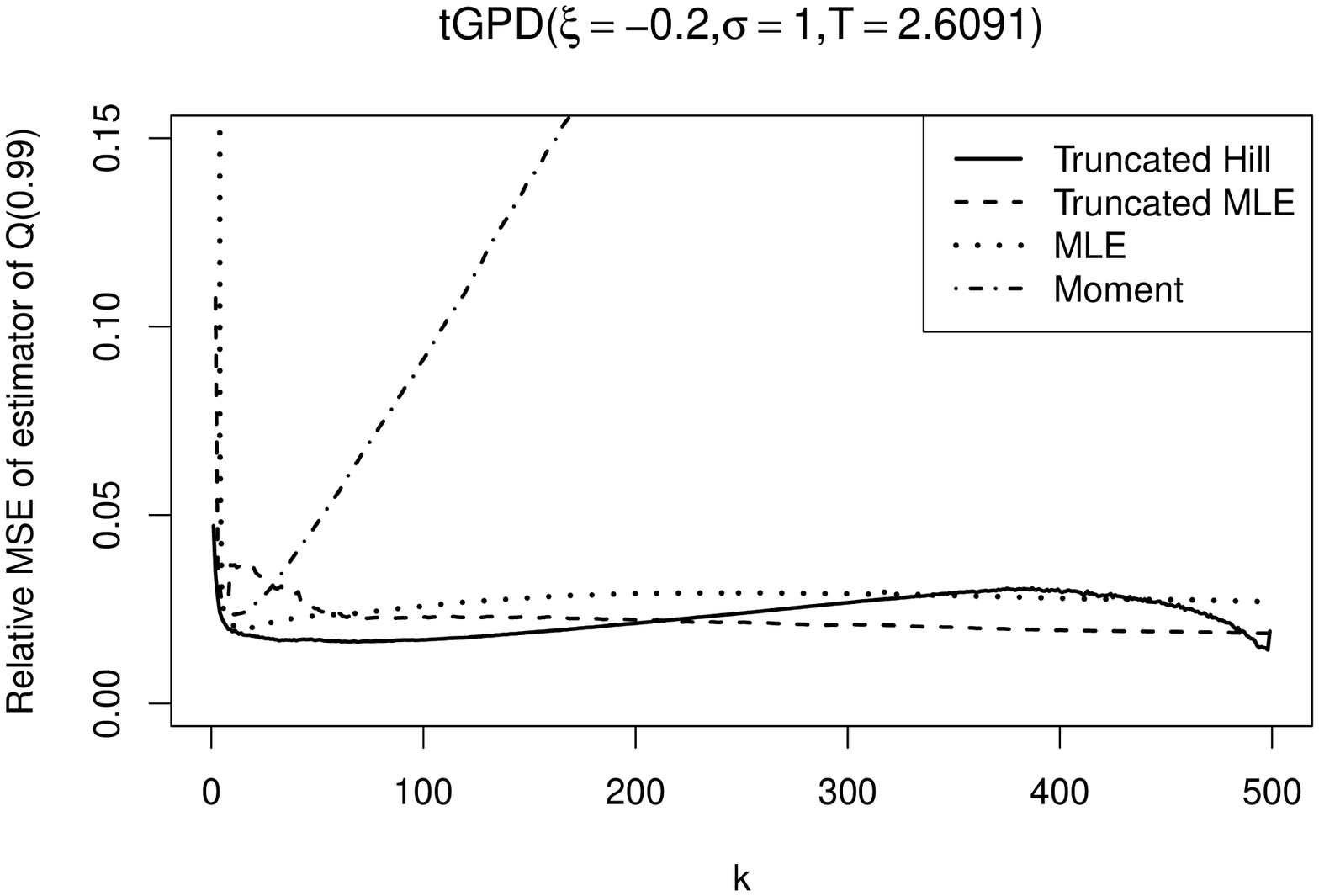}\\
	  \includegraphics[height=0.495\textwidth, angle=270]{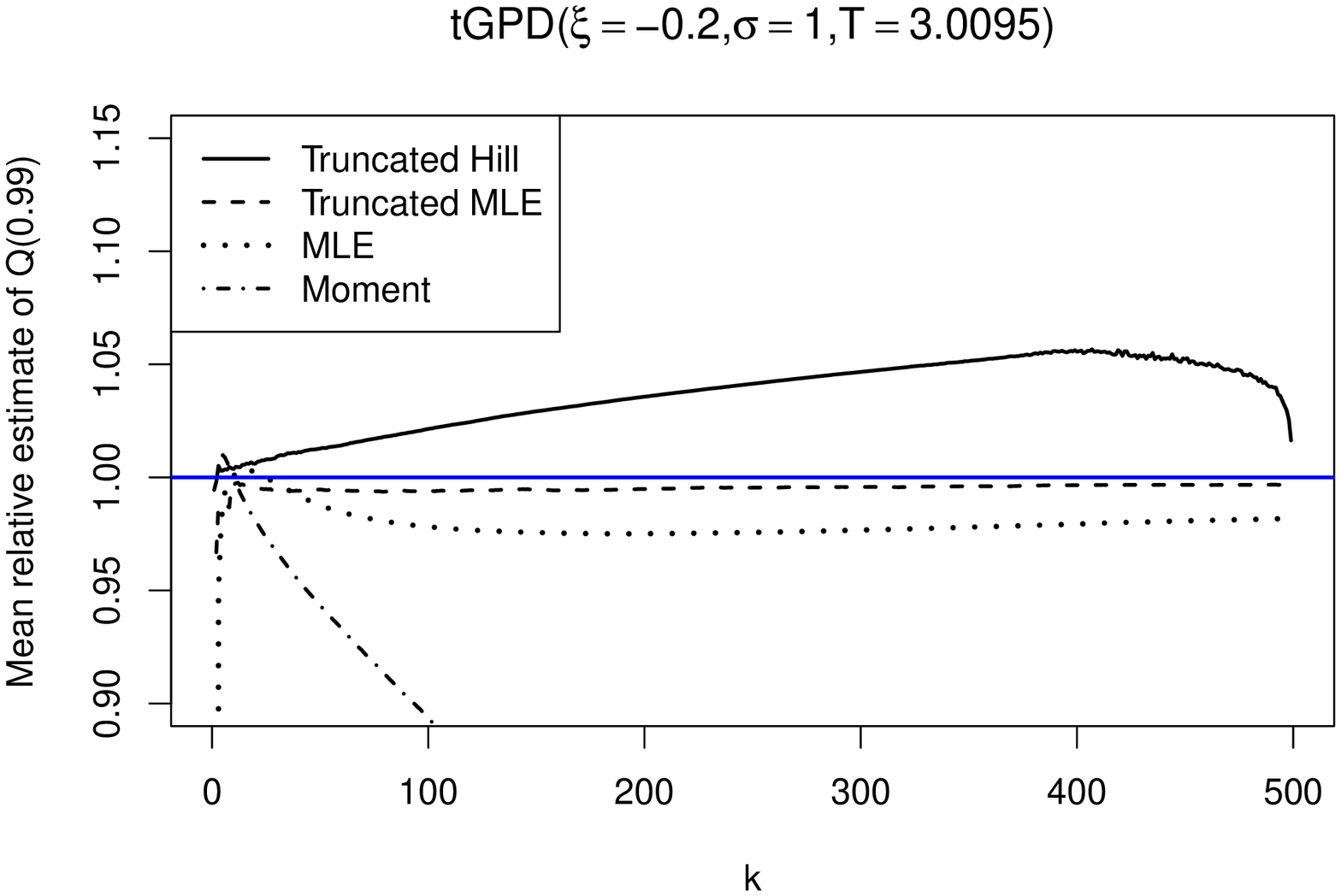}
	 \includegraphics[height=0.495\textwidth, angle=270]{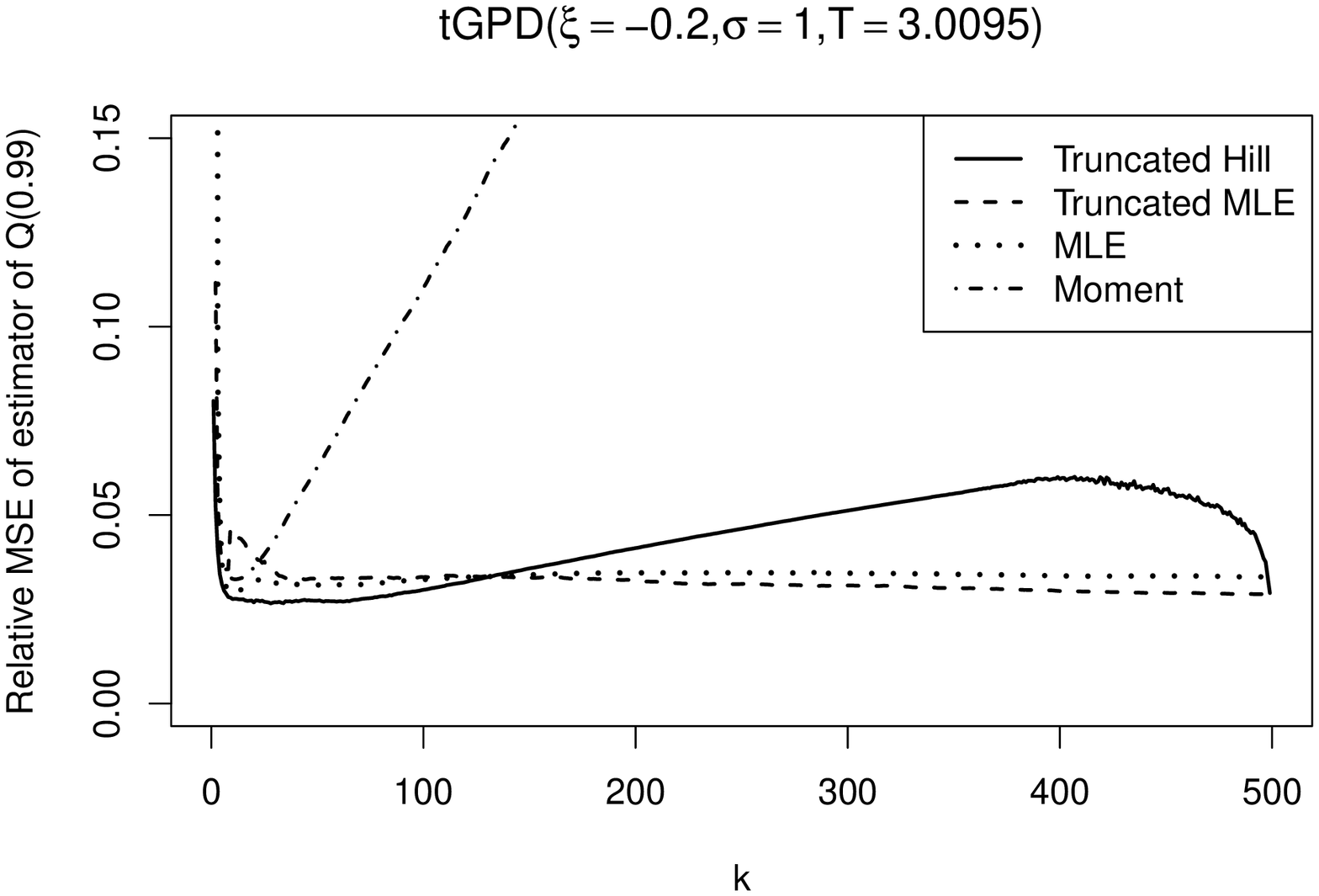}\\
	  \includegraphics[height=0.495\textwidth, angle=270]{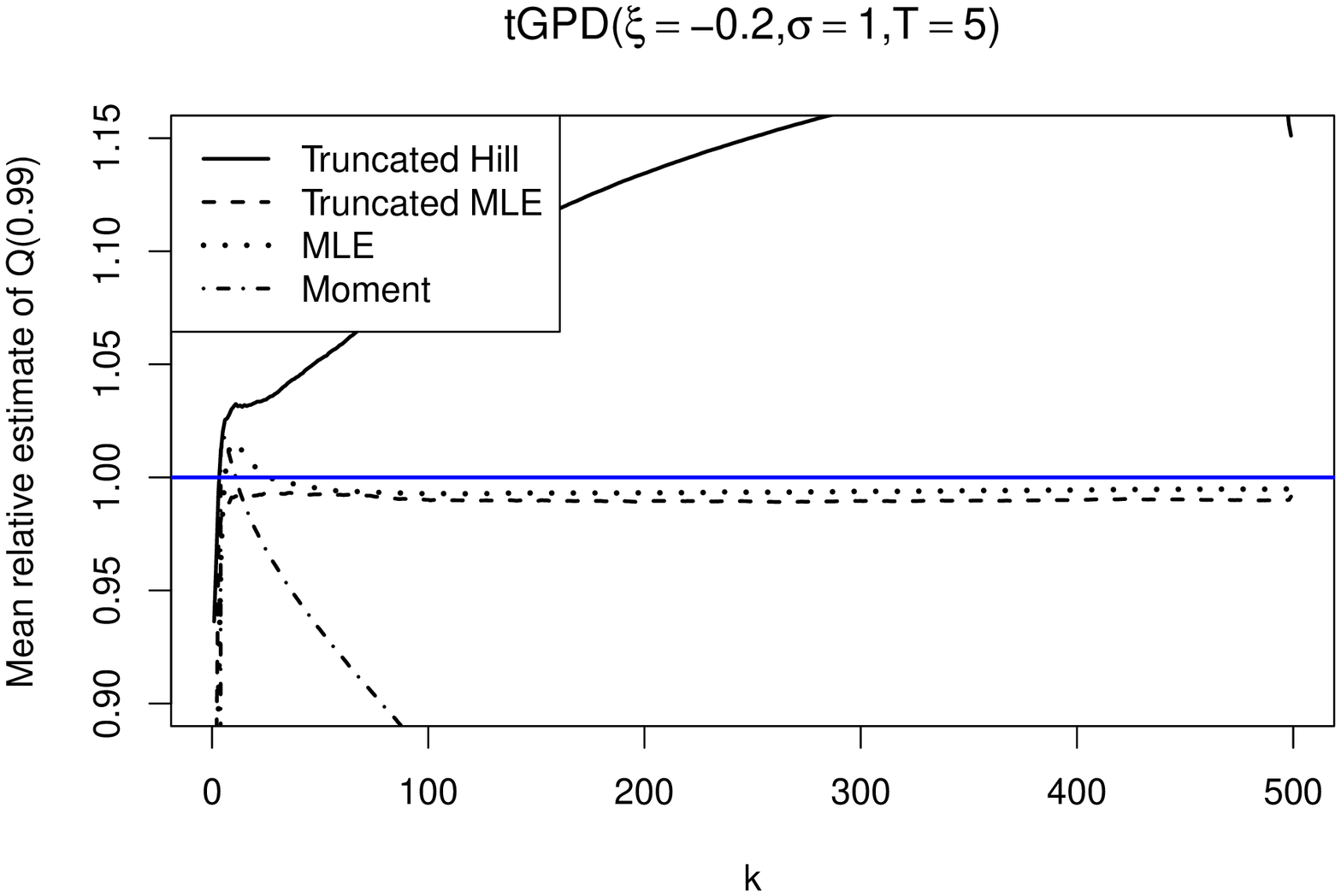}
	 \includegraphics[height=0.495\textwidth, angle=270]{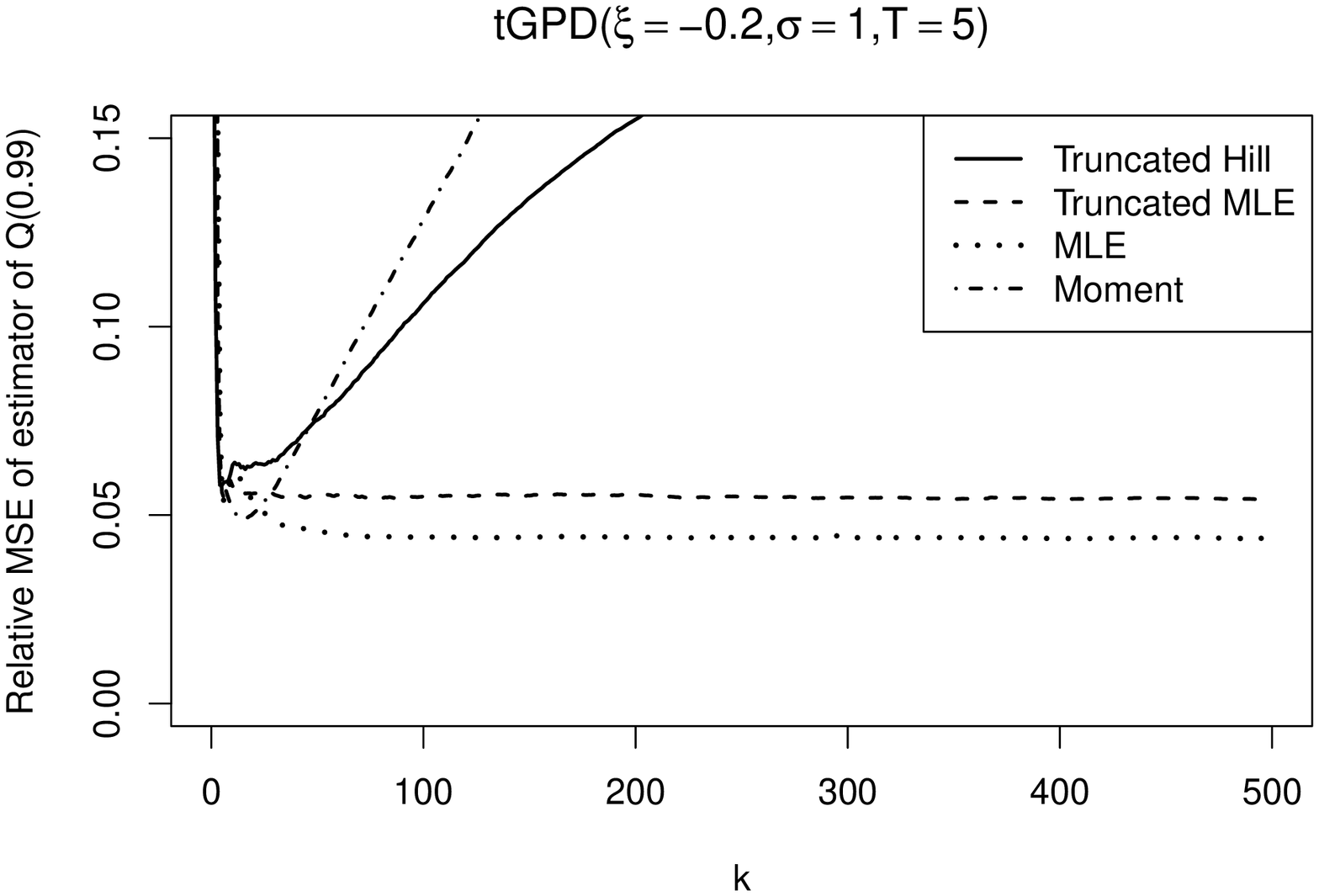}
  \caption{Mean deviations of $\hat{Q}^+_{T,k}(1-p)/Q_T(1-p)$, $\hat{Q}_{T,k}(1-p)/Q_T(1-p)$, $\hat{Q}^{\infty}_{k}(1-p)/Q_T(1-p)$, $\hat{Q}^M_{k}(1-p)/Q_T(1-p)$  and corresponding MSE with $p=0.01$ for GPD(-0.2,1) truncated at $Q_Y (0.975)$ (top), $Q_Y (0.99)$ (middle) and $Q_Y (1)$ (bottom).}
   \end{figure}
				
				\newpage	
	\begin{figure}[!ht]
	\centering
	  \includegraphics[height=0.495\textwidth, angle=270]{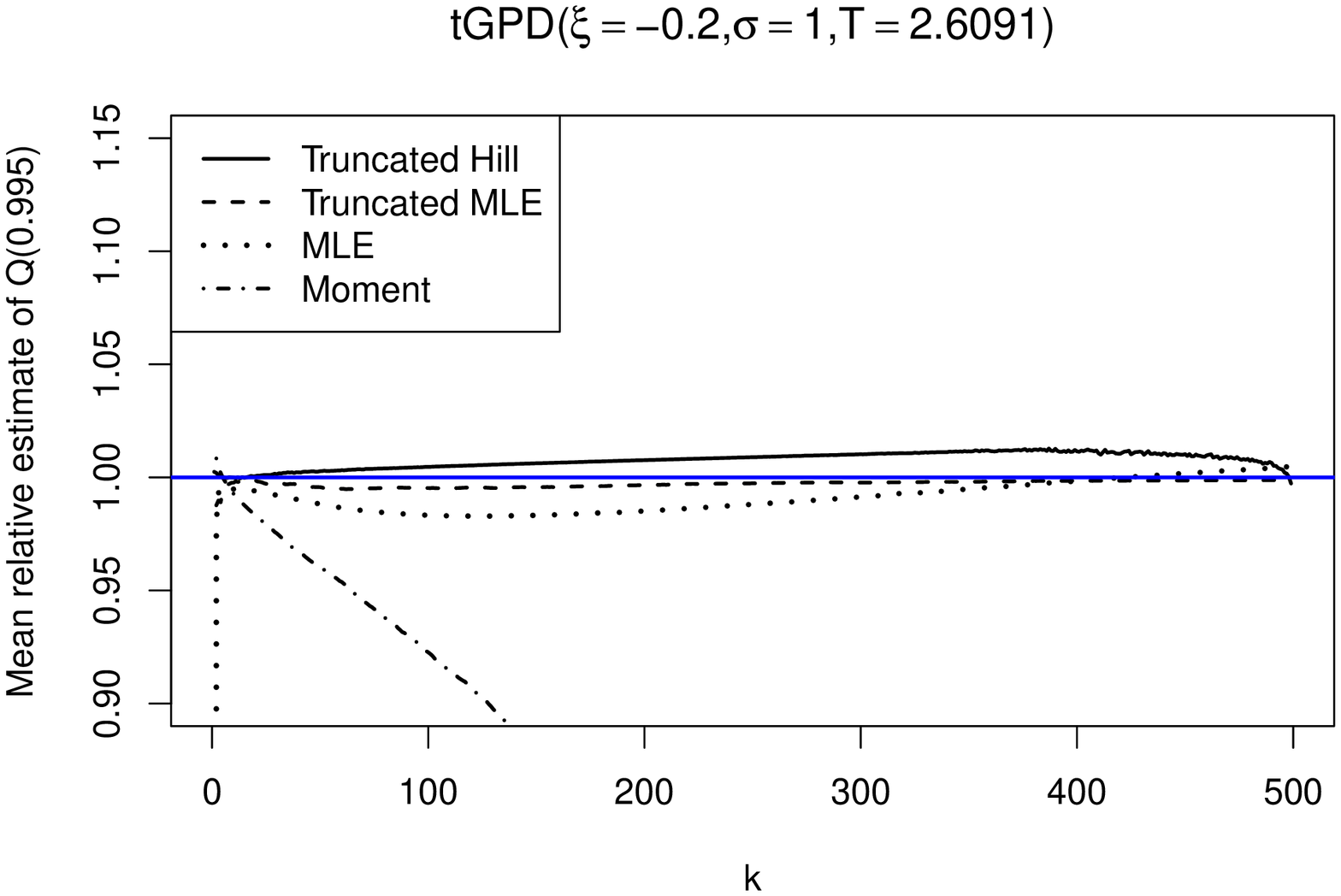}
	 \includegraphics[height=0.495\textwidth, angle=270]{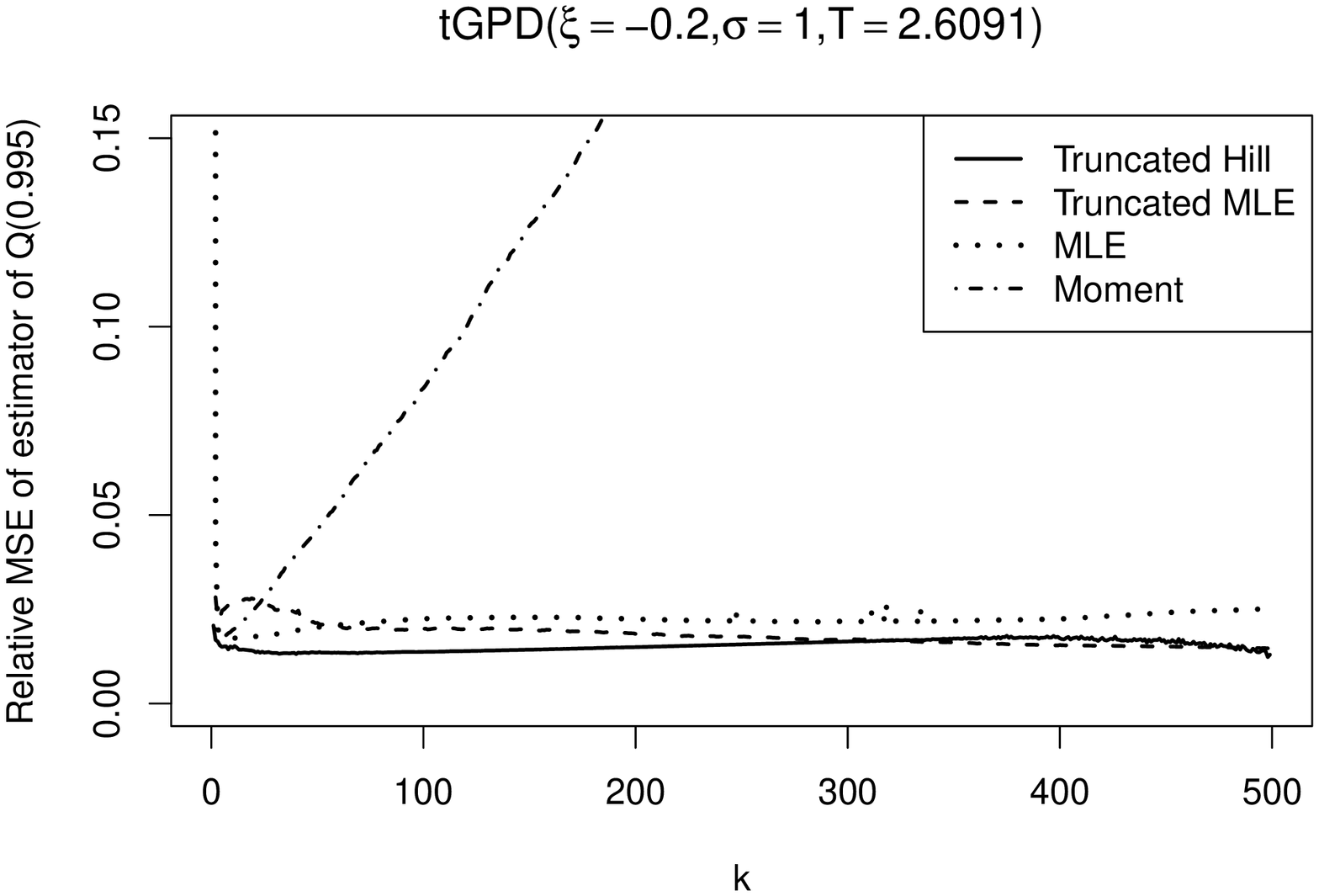}\\
	  \includegraphics[height=0.495\textwidth, angle=270]{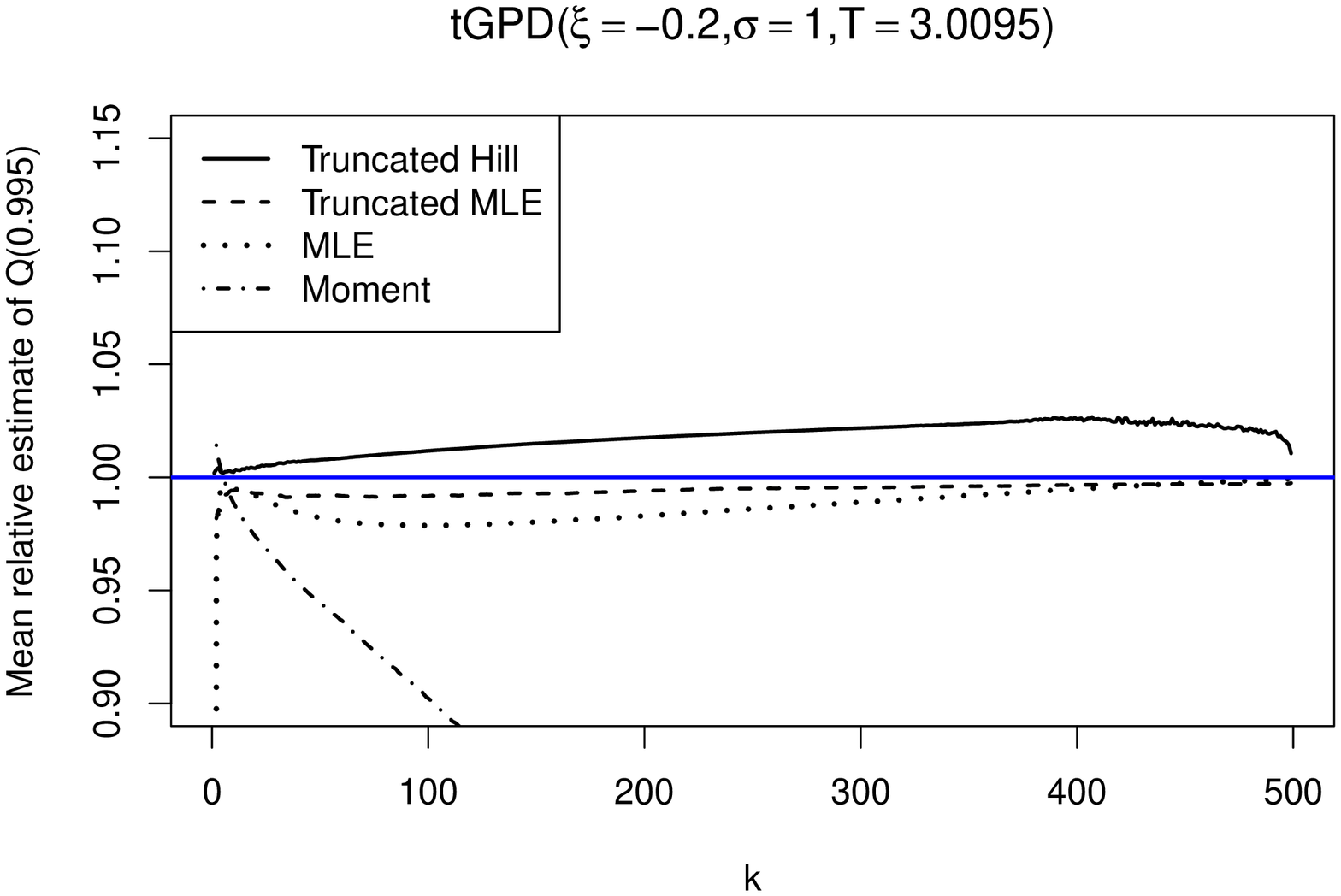}
	 \includegraphics[height=0.495\textwidth, angle=270]{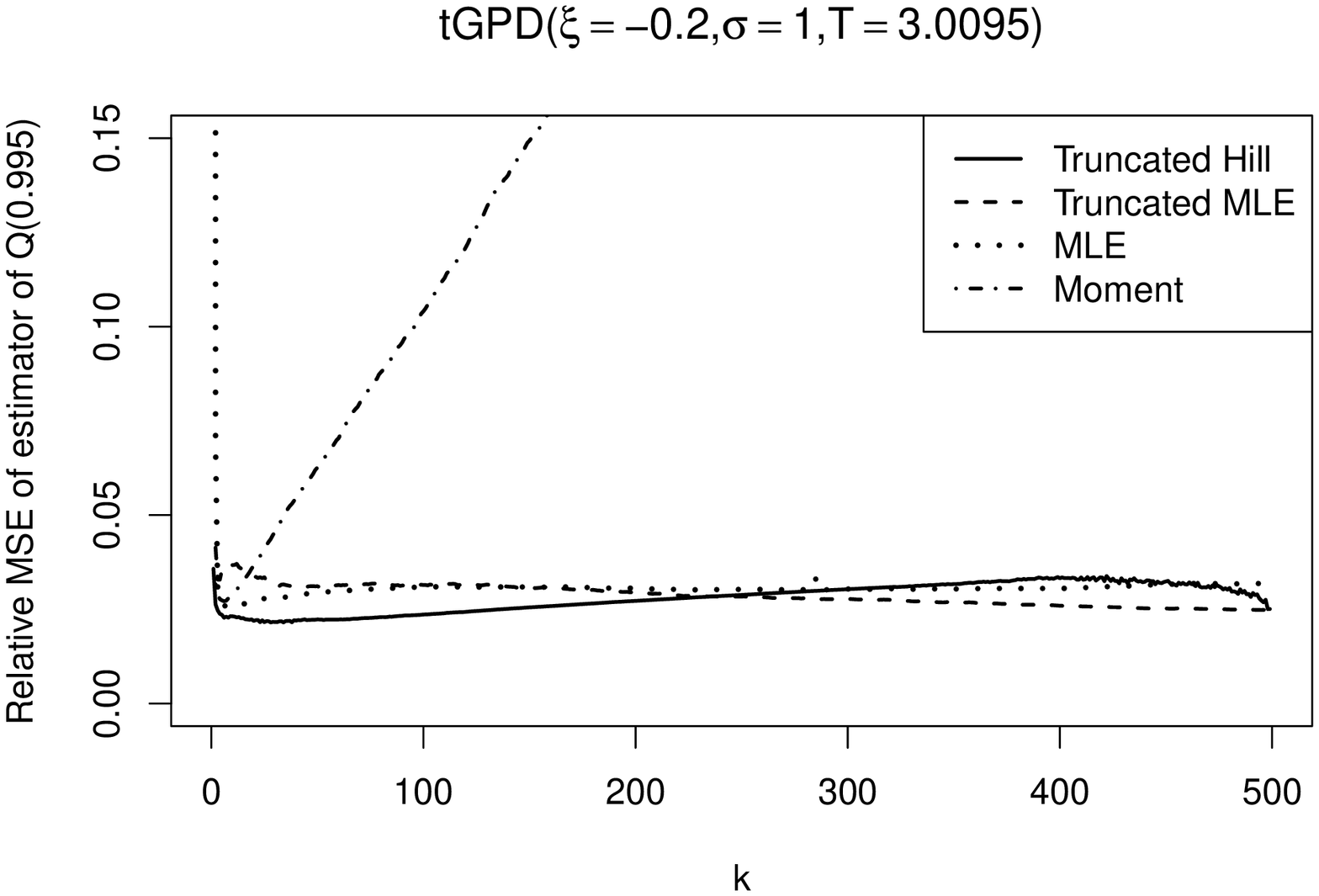}\\
	  \includegraphics[height=0.495\textwidth, angle=270]{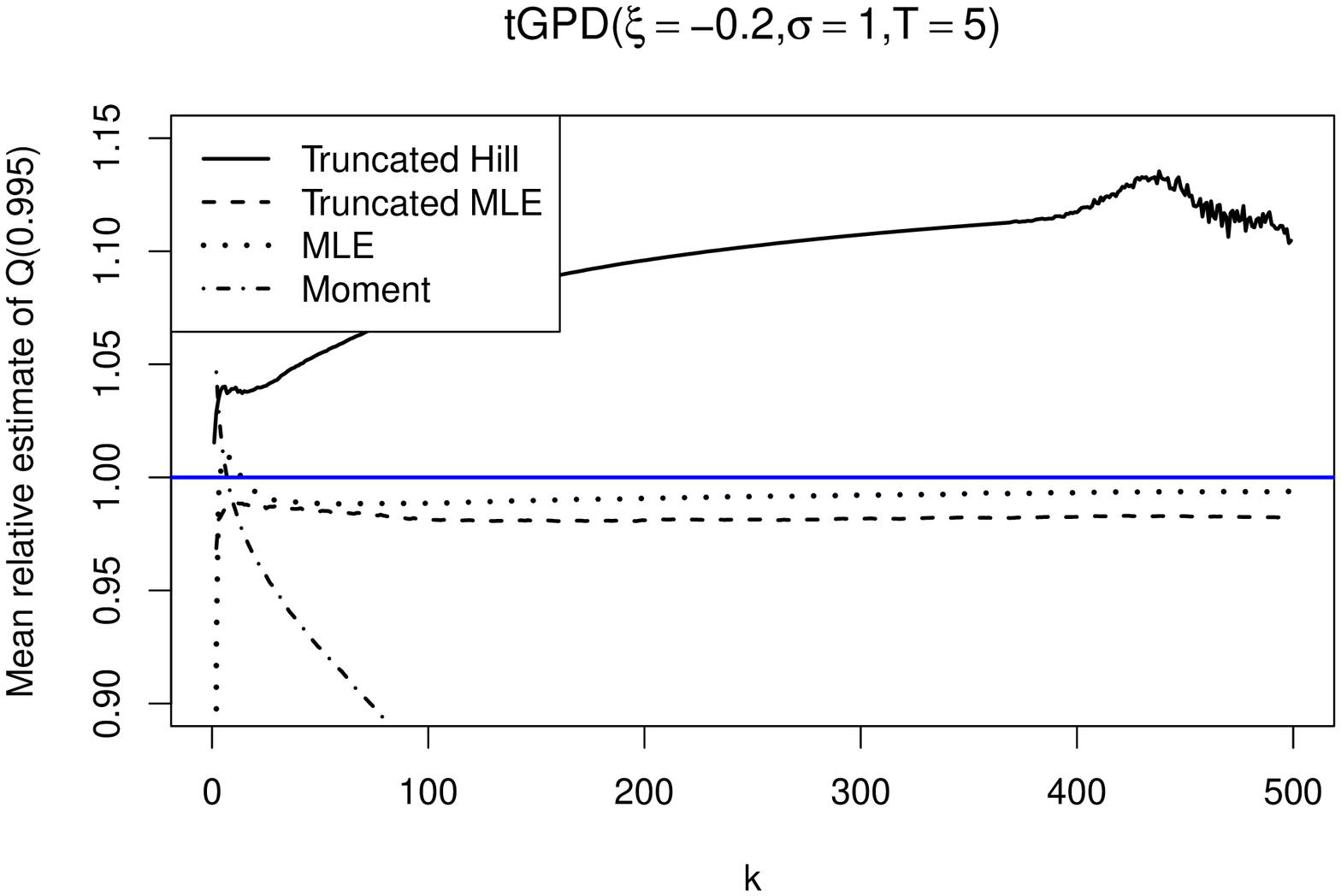}
	 \includegraphics[height=0.495\textwidth, angle=270]{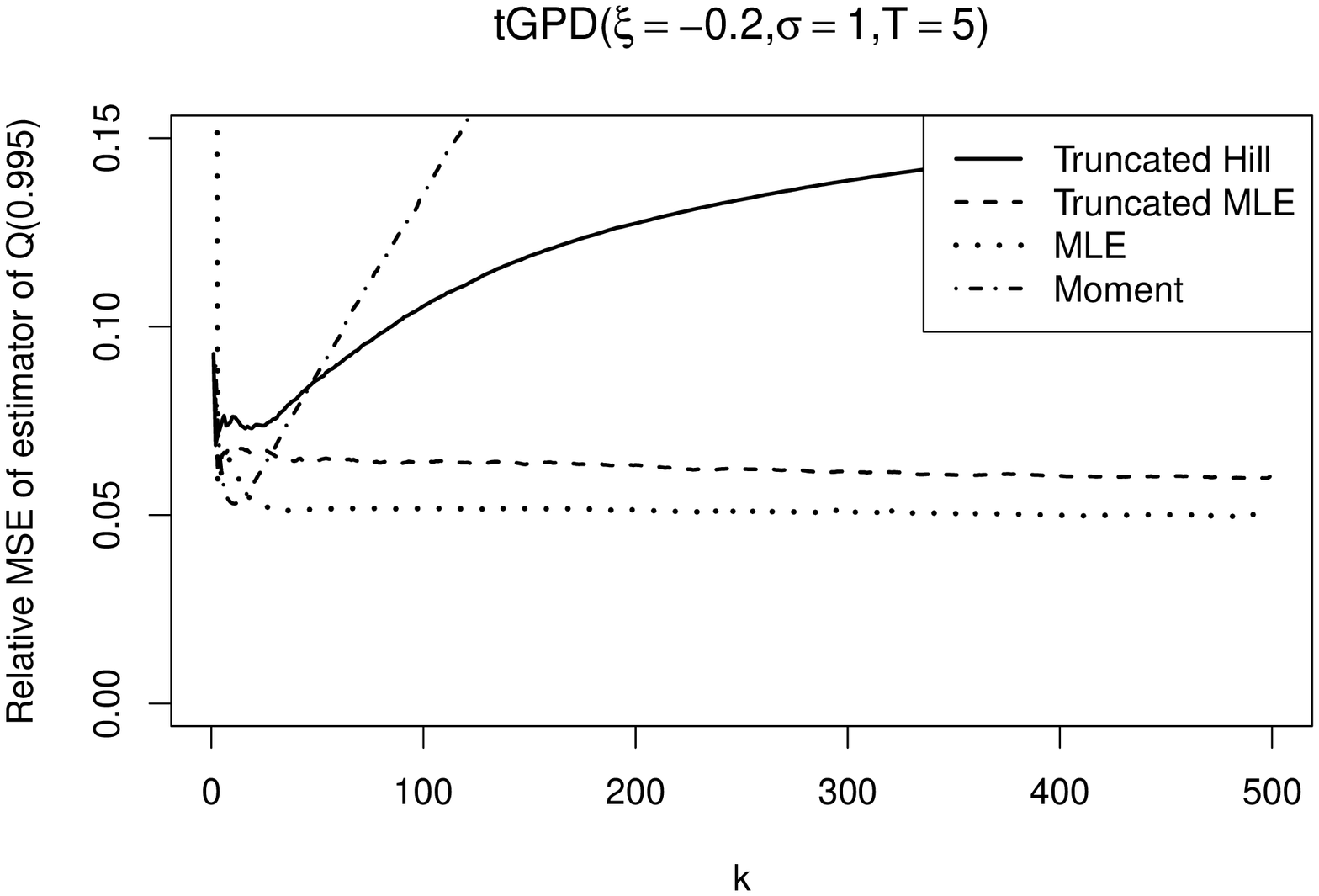}
  \caption{Mean deviations of $\hat{Q}^+_{T,k}(1-p)/Q_T(1-p)$, $\hat{Q}_{T,k}(1-p)/Q_T(1-p)$, $\hat{Q}^{\infty}_{k}(1-p)/Q_T(1-p)$, $\hat{Q}^M_{k}(1-p)/Q_T(1-p)$  and corresponding MSE with $p=0.005$ for GPD(-0.2,1) truncated at $Q_Y (0.975)$ (top), $Q_Y (0.99)$ (middle) and  $Q_Y (1)$ (bottom).}\label{fig:sim_Q_last}
        \end{figure}

		\end{document}